\documentclass[%
amsmath,amssymb,
superscriptaddress,
longbibliography,
nobibnotes,
aps,
prl,
reprint,
]{revtex4-2}
\usepackage{amsthm}
\usepackage{graphicx}
\usepackage{bm}
\usepackage[x11names]{xcolor}
\usepackage{hyperref}
\hypersetup{
  bookmarks=true,
  bookmarksopen=true,
  hidelinks=true,
  colorlinks=true,
  citecolor=Blue4,
  linkcolor=Blue4,
  urlcolor=Blue4,
}

\usepackage{csquotes}
\MakeOuterQuote{"}
\usepackage{physics}
\usepackage[capitalize]{cleveref}

\crefname{section}{Sec.}{Secs.}
\crefname{subsection}{Sec.}{Secs.}
\crefname{subsubsection}{Sec.}{Secs.}
\usepackage{mathtools}
\usepackage{nccmath}

\usepackage{enumerate}

\usepackage{tikz}
\usepackage{tikz-cd}
\usetikzlibrary{babel,decorations.pathreplacing}

\usepackage{booktabs}
\usepackage{threeparttable}
\usepackage{multirow}

\usepackage{enumitem}

\usepackage[makeroom]{cancel}

\usepackage{stmaryrd} %

\usepackage[super]{nth}

\usepackage{orcidlink}

\definecolor{myblue}{RGB}{30,100,200}
\definecolor{myorange}{RGB}{220,110,20}
\definecolor{mygreen}{RGB}{30,150,80}
\definecolor{swapgreen}{HTML}{405B24}
\definecolor{myred}{RGB}{190,30,30}
\definecolor{mygray}{RGB}{120,120,120}
\definecolor{mypurple}{RGB}{130,40,180}
\definecolor{mybrown}{RGB}{160,90,30}
\definecolor{lightblue}{RGB}{220,235,255}
\definecolor{lightorange}{RGB}{255,235,210}
\definecolor{lightgreen}{RGB}{210,240,220}
\definecolor{lightpurple}{RGB}{238,225,250}

\renewcommand{\selectlanguage}[1]{}

\usepackage{titletoc}

\newtheorem{theorem}{Theorem}

\newtheorem{lemma}{Lemma}
\newtheorem{corollary}{Corollary}

\newtheorem{remark}{Remark}

\newtheorem{proposition}{Proposition}
\newtheorem{conjecture}{Conjecture}

\usepackage{xparse}
\NewDocumentCommand{\evalat}{sO{\big}mm}{%
  \IfBooleanTF{#1}
   {\mleft. #3 \mright|_{#4}}
   {#3#2|_{#4}}%
}

\newcommand{\hilsymb}{\mathcal{H}}
\newcommand{\hilone}{\hilsymb_m}
\newcommand{\fockn}{\hilone^n}

\newcommand{\fock}{\mathcal{H}_m}

\newcommand{\focknbasis}{\mathcal{B}_m^n}

\newcommand{\orb}{\operatorname{Orb}}

\newcommand{\End}{\operatorname{End}}

\newcommand{\uni}{\operatorname{U}}

\newcommand{\spa}{\operatorname{span}}

\newcommand{\schop}{\mathcal{S}_{\rm op}(\fock)}

\renewcommand{\ip}[2]{\langle #1 , #2 \rangle}
\newcommand{\ipbig}[2]{\big\langle #1 , #2 \big\rangle}

\newcommand{\gram}{\operatorname{Gram}}

\usepackage{dsfont}
\newcommand{\id}{\mathds{1}}

\newcommand{\NN}{\mathbb{N}}

\AtBeginDocument{\setcounter{secnumdepth}{3}}

\newcommand{\nocontentsline}[3]{}
\newcommand{\tocless}[1]{\begingroup\let\addcontentsline=\nocontentsline #1\endgroup}

\begin{document}

\title{Orbit dimensions in linear and Gaussian quantum optics}

\author{Eliott Z. Mamon\,\orcidlink{0009-0003-9186-2019}}
\affiliation{Laboratoire d’Informatique de Paris 6, CNRS, Sorbonne Université, 4 Place Jussieu, 75005 Paris, France}
\email{eliott.mamon@gmail.com}

\date{\today}

\begin{abstract}
We study the dimension of the manifold of quantum states (called orbit) that a given bosonic state can reach under linear or quadratic Hamiltonian evolutions. That is, we investigate how many directions in the Hilbert space a state can explore in these subuniversal regimes. After showcasing a simple way to compute orbit dimensions, we find that these topological quantities reveal fundamental insights into the structure of attainable state spaces (e.g., boson bunching does not increase the number of accessible directions) with multifaceted consequences. First, we illustrate how they can alone yield no-go results for some transformations. We then propose ways to probe orbit dimensions using homodyne and heterodyne measurements on pure states, or photon counters on two copies of general states. We also relate orbit dimensions to the number of directions accessible to bosonic variational circuits. Next, we study links between orbit dimensions and the resource theory of genuine non-Gaussianity (resp. $P$ nonclassicality), and prove that free states coincide with a unique minimal-dimension orbit in the pure multimode case, under Gaussian (resp. displaced-linear) unitaries. We then extend this result to a mixed-state setting, provided that an alternative convex-roof-based definition of orbit dimensions is taken; however, we show that those fail to be fixed-mode monotones under the respective free operations. Our entire framework is proven to hold in both discrete and continuous-variable settings and can be used with Fock as well as phase-space representations such as the Wigner or stellar representations. Overall, this work offers a new perspective on the structure of reachable quantum states of light, which can help practitioners understand limitations and sources of expressivity and non-Gaussianity (or $P$ nonclassicality) in bosonic quantum information protocols such as quantum machine learning.
\end{abstract}

\maketitle

\section{Introduction}

Bosonic quantum systems are ubiquitous in a wide range of physical platforms. In such systems, quantum information resides in a number $m$ (assumed to be finite) of bosonic modes, which are described by an infinite-dimensional Hilbert space known as the \textit{Fock space}. Quantum optical platforms are such examples, where the modes correspond to distinct properties (spectral, temporal, angle of polarization, etc.) of the electromagnetic field. This is the setting of both discrete-variable \cite{kok_linear_2007,Aaronson-ComputationalComplexity-2013} and continuous-variable \cite{lloyd_quantum_1999,Braunstein-QuantumInformation-2005,weedbrook_gaussian_2012,Adesso-ContinuousVariable-2014,hamilton_gaussian_2017,Serafini-QuantumContinuous-2023,Walschaers-NonGaussianQuantum-2021} platforms, which are equally explored as a basis for quantum computing.

Ideally, a quantum computing platform should be \textit{universally controllable}, i.e., have the ability to unitarily evolve any state in Fock space into any other state, by a suitable driving sequence \cite{dalessandro_introduction_2021}. This is not the case for linear or even Gaussian quantum optics, as the continuous group of accessible unitaries $G$ is only \textit{finite-dimensional} (meaning informally that its elements $U\in G$ can be specified by a finite number of real parameters) and therefore cannot make an initial quantum state explore an infinite-dimensional state space. In fact, since $\dim(G)$ is only \textit{quadratic} in $m$, it cannot implement arbitrary unitary transformations deterministically on some subspace of Fock space of finite but \textit{exponentially} large dimension, either \cite{zanardi_universal_2004}.
Recovering universality thus requires the ability to enact additional non-Gaussian unitaries \cite{lloyd_quantum_1999} (which is experimentally challenging) or to use intermediate measurements \cite{knill_scheme_2001,kok_linear_2007,VanMeter-GeneralLinearoptical-2007} (which may trade determinism for low success rates). Other strategies using measurements to increase the expressivity of gates include adaptive gates \cite{chabaud_quantum_2021} (which increase delay and losses) or adaptive state injection \cite{Monbroussou-QuantumAdvantage-2025} (which sacrifices purity).

In the meantime, a subuniversal unitary group $G$ is only able to send an implementable state $\ket{\psi}$ to a \textit{strict} subset of state space, called the \textit{orbit} of $\ket{\psi}$ under $G$.
Furthermore, preparing specific bosonic states remains much more of a delicate, case-by-case issue \cite{dellanno_multiphoton_2006}. It is therefore of crucial importance to gain insight into the structure of these orbits, and on how they partition the state space.

In the literature, several \textit{invariants} under passive linear \cite{migdal_which_2014,parellada_no-go_2023,Parellada-LieAlgebraic-2026} and Gaussian \cite{chabaud_stellar_2020,Hahn-AssessingNonGaussian-2026,Mele-SymplecticRank-2026} quantum optics have been identified.
While somewhat interpretable and useful to rule out certain state conversions or to verify experimental setups \cite{Rodari-ObservationLie-2025}, these invariants do not always help to convey a better picture of the orbits themselves.

In this work, we showcase the \textit{dimension} of an orbit (as a submanifold in state space), which is an invariant under $G$ that has the added benefit of directly revealing a topological property of orbits themselves.
While the concept of orbit dimensions has found uses in quantum theory (e.g., in entanglement theory of symmetric qubits \cite{cenci_symmetric_2010,lyons_minimum_2005}), the fact that they can be easily evaluated at the Lie algebraic level (and without studying the stabilizer subgroup explicitly \cite{dalessandro_introduction_2021}) is seldom explored.

The work is structured as follows.
In \cref{sec:q-optics-intro}, we recall the basic setup of quantum optics in Fock space, quadratic Hamiltonians, and corresponding unitary orbits, after which we state the main structural theorem.
In \cref{sec:main-text-props-of-orbit-dimensions}, we collect some properties that orbit dimensions satisfy.
In \cref{sec:evaluation-orbit-dimensions}, we first explain how to evaluate orbit dimensions analytically or numerically in Fock or phase-space representations and provide some examples, after which we discuss their experimental evaluation.
In \cref{sec:rel-RTs}, we explore the relation between orbit dimensions and some resource theories of quantum optics.
The results are discussed in \cref{sec:disc-and-outlook} and concluded in \cref{sec:CCL}.
In the end, the appendix contains all the technical proofs and necessary background on the needed concepts.

\section{Structure of orbits and orbit dimensions}\label{sec:q-optics-intro}

\subsection{States of quantum optics}
Consider a bosonic system of $m$ modes, whose Hilbert space is the $m$-mode \textit{Fock space} $\fock:=\oplus_{n\in\NN}\,\fockn$, where $\fockn := \spa(\focknbasis)$, with the \textit{$n$-photon Fock basis states} $\focknbasis := \{\ket{n_1,\dots,n_m}\ \ |\ \ n_i\geq0\,,\ \, n_1 + \cdots + n_m = n \}$.
We also write occupation numbers as $\bm{n}:=(n_1,\dots,n_m)$ and $|\bm{n}|:=n_1+\cdots+n_m$.
Mixed states are described by density operators, i.e., non-negative operators $\rho$ with $\Tr[\rho]=1$.
Pure states can be represented either in the "ket" picture, as normalized elements $\ket{\psi}\in\fock$, or in the "ketbra" picture, as rank-$1$ density operators $\ketbra{\psi}$.
The infinite-dimensionality of $\fock$ demands a twofold technical care when stating properties about states. First, some kets and density operators are not physically meaningful due to having infinite average energy.
We hereafter focus on Schwartz states (states whose Fock basis coefficients decay faster than polynomially), which we call \textit{physical states} and which avoid this technicality.
Note that this class of states includes in particular finite support states (in the Fock basis), as well as, for instance, coherent states and cat states.
Second, topologies on these state spaces must be specified. %
We refer to \cref{sec:SM-Schwartz-spaces-and-proof-of-lsc} for details on these Schwartz spaces and their topologies.

We recall the definitions of canonical operators. The $k^{\rm th}$ \textit{annihilation operator} $a_k$ acts on the Fock basis as $a_k \ket{n_1,\dots,n_k,\dots,n_m}:=\sqrt{n_k} \ket{n_1,\dots,n_k-1,\dots,n_m}$ (for $n_k=0$ it is by definition the zero vector), the \textit{creation operator} $a^\dagger_k$ acts as $a^\dagger_k \ket{n_1,\dots,n_k,\dots,n_m}:=\sqrt{n_k+1} \ket{n_1,\dots,n_k+1,\dots,n_m}$, and they satisfy
\begin{equation}\label{eq:CCR}
[a_k,a^\dagger_l]=\delta_{kl}\id\,.
\end{equation}

\subsection{Unitary groups and Lie algebras of quantum optics}

We begin by introducing some Hamiltonians on Fock space.
For all $1\leq k < l \leq m$, we let
\allowdisplaybreaks[1]
\begin{align}
e_{kl} &:= \frac{1}{2} (a^\dagger_k a_l + a^\dagger_l a_k)\,,\label{eq:def-generator-e}\\
E_{kl} &:= \frac{i}{2} (a^\dagger_k a_l - a^\dagger_l a_k)\,,\label{eq:def-generator-E}\\
r_{kl} &:= \frac{1}{2} (a^\dagger_k a^\dagger_l + a_k a_l)\,,\label{eq:def-generator-r}\\
R_{kl} &:= \frac{i}{2} (a^\dagger_k a^\dagger_l - a_k a_l)\,,\label{eq:def-generator-R}
\intertext{and for all $1\leq k \leq m$,}
N_{k} &:= a^\dagger_k a_k\,,\label{eq:def-generator-N}\\
s_{k} &:= \frac{1}{2} ( a^{\dagger\,2}_k + a_k^{\,2} )\,,\label{eq:def-generator-s}\\
S_{k} &:= \frac{i}{2} ( a^{\dagger\,2}_k - a_k^{\,2} )\,,\label{eq:def-generator-S}\\
q_{k} &:= \frac{1}{\sqrt{2}} ( a^{\dagger}_k + a_k )\,,\label{eq:def-generator-q}\\
p_{k} &:= \frac{i}{\sqrt{2}} ( a^{\dagger}_k - a_k )\,.\label{eq:def-generator-p}
\end{align}
Various levels of experimental abilities in quantum optics correspond to different subsets of the above Hamiltonians that one can enact on the system. Indeed, each of them generates (by exponentiation) a unitary associated with the following physical interpretations \cite{luis_quantum_1995,weedbrook_gaussian_2012}:
$e_{kl}$ generates a $\pi/2$ phase-shifting beam splitter (between modes $k$ and $l$), $E_{kl}$ generates a non-phase-shifting beam splitter, $N_{k}$ generates a phase shifter (on mode $k$), $r_{kl},R_{kl}$ generate different kinds of two-mode squeezing unitaries, $s_{k},S_{k}$ generate different kinds of single-mode squeezing unitaries, and $p_k,q_k$ generate single-mode displacements along the real and imaginary axes in the phase space of mode $k$.
A given subset $\mathcal{X}$ of these Hamiltonians generates a corresponding unitary subgroup $G$ of implementable unitaries (inside the infinite-dimensional group of unitary operators on $\fock$), as $G:=\{ e^{-i t_1 H_1}\cdots e^{-i t_a H_a} \ |\ a\in \NN,\ t_j\in \mathbb{R},\ H_j \in \mathcal{X} \}$.
In fact, the same group is obtained if in the above set the finite list $\mathcal{X}$ is replaced by $i \mathfrak{g}$, where $\mathfrak{g}$ denotes the \textit{Lie algebra} generated by $i \mathcal{X}$ \cite{dalessandro_introduction_2021}, i.e., the smallest set of skew-Hermitian operators containing $i \mathcal{X}$ that is closed under linear combinations and under the commutator $[A,B]:=AB-BA$.
We restrict our focus to four physically relevant granularities of $\mathcal{X}$: \textit{passive linear quantum optics} (PLO) only allows beam splitters and phase shifters, \textit{displaced passive linear quantum optics} (DPLO) extends PLO with displacements, \textit{active linear quantum optics} (ALO) extends PLO with squeezing, and \textit{Gaussian quantum optics} (GO) extends PLO with both displacements and squeezing.

\setlength{\tabcolsep}{4pt}
\begin{table*}[]
\begin{tabular}{@{}lll@{}}
\toprule
Unitary group $G$ & Basis $\mathcal{B}_{\mathfrak{g}}$ of Lie algebra $\mathfrak{g}$ & $\dim(G)=\dim(\mathfrak{g})$ \\ \midrule
PLO \hspace{4pt} (passive linear quantum optics) & $i\,\{ e_{kl},\ E_{kl},\ N_{k} \}$ & $m^2$ \\
DPLO (displaced passive linear quantum optics) & $i\,\{ e_{kl},\ E_{kl},\ N_{k},\ q_k,\ p_k,\ \id \}$ & $m^2 + 2m + 1$ \\
ALO \hspace{3.5pt} (active linear quantum optics) & $i\,\{ e_{kl},\ E_{kl},\ N_{k},\ r_{kl},\ R_{kl},\ s_{k},\ S_{k},\ \id \}$ & $2m^2 + m + 1$ \\
GO \hspace{8.75pt} (Gaussian quantum optics) & $i\,\{ e_{kl},\ E_{kl},\ N_{k},\ q_k,\ p_k,\ r_{kl},\ R_{kl},\ s_{k},\ S_{k},\ \id\}$ & $2m^2 + 3m + 1$ \\ \bottomrule
\end{tabular}
\caption{Common unitary groups of $m$-mode quantum optics, and bases of their associated Lie algebras.}
\label{tab:Lie-algebra-bases}
\end{table*}

We summarize these groups in \cref{tab:Lie-algebra-bases} by identifying bases $\mathcal{B}_{\mathfrak{g}}$ of their Lie algebras. On each line of this table, the claim that $\mathcal{B}_{\mathfrak{g}}$ is a Lie algebra basis amounts to the claim that the commutator of any two basis elements can be written as a real linear combination of all basis elements, which can be elementarily verified from \cref{eq:CCR}. Note that any one of these sets $\mathcal{B}_{\mathfrak{g}}$ would cease to be a Lie algebra basis if a single element were removed, which can clarify why two kinds of beam splitter ($e_{kl},E_{kl}$) or squeezing generators must be included.

\subsection{Orbits of quantum optics}\label{sec:orbs-of-q-optics}

Orbits of states $\ket{\psi}$ or $\rho$ under a unitary group $G$ are defined as $\orb_G(\ket{\psi}) := \{ U\ket{\psi} \ | \ U \in G \}$ and $\orb_G(\rho) := \{ U\rho U^\dagger \ | \ U \in G \}$.
Due to the group structure of $G$, two orbits are either equal or disjoint, and therefore orbits partition the state space into equivalence classes.

The main structural facts about orbits that we rely on are stated in the following \cref{thm:orbit-dim-maintext}, which we prove in \cref{sec:SM-EMRep-and-proof-of-orbit-structure}.
\begin{theorem}[Orbit structure and dimensions]%
\label{thm:orbit-dim-maintext}
Let $G$ and $\mathcal{B}_\mathfrak{g}=i\{H_1,\dots,H_d\}$ be one of the $m$-mode quantum optical unitary groups and associated Lie algebra bases considered in \cref{tab:Lie-algebra-bases}.
For any physical state $\ket{\psi}$ and physical density operator $\rho$, the orbits $\orb_{G}(\ket{\psi})$ and $\orb_{G}(\rho)$ admit a natural smooth manifold structure, whose (manifold) dimensions may be calculated as
\begin{align}
\dim(\orb_{G}(\ket{\psi})) &= \rank_{\mathbb{R}}(\{ H_1 \ket{\psi},\dots, H_d \ket{\psi}\})\,,\label{eq:concrete-rank-formula-ket}\\
\dim(\orb_{G}(\rho)) &= \rank_{\mathbb{R}}(\{ [H_1,\rho],\dots,[H_d,\rho]\})\,.\label{eq:concrete-rank-formula-density}
\end{align}
\end{theorem}
Without further precision, $G$ will from now on denote one of the considered groups of \cref{tab:Lie-algebra-bases}.

Intuitively, since an orbit is a manifold, its manifold dimension matches the vector space dimension of the tangent space at one of its states, and the latter may be obtained by counting [\cref{eq:concrete-rank-formula-ket,eq:concrete-rank-formula-density}] how many independent directions this state can travel in state space when evolved through all the available independent directions offered by $G$ (see \cref{fig:orbit-tangent-space-diagram}).
\begin{figure}[b]
\centering
\includegraphics[width=0.25\textwidth]{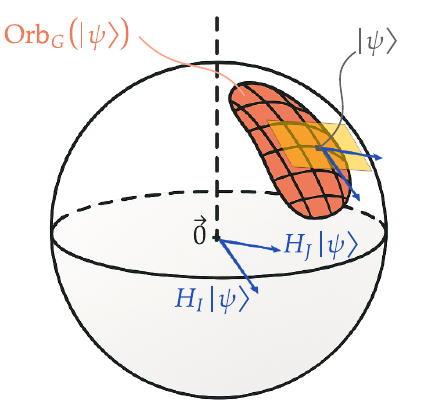}
\caption{Illustration of an orbit of a state $\ket{\psi}$ under group $G$ (red). The tangent space to the orbit at point $\ket{\psi}$ (yellow) encodes the possible movement directions (blue) under small evolutions from the group $G$. The vectors $H_I \ket{\psi}$ ($H_I \in \mathcal{B}_{\mathfrak{g}}$) are translations of a basis of these directions to the origin of the Hilbert space. The dimension of the orbit, equal to the tangent space dimension, is found by counting the number of linearly independent (with respect to~\textit{real} linear combinations) vectors $H_I \ket{\psi}$ [\cref{eq:concrete-rank-formula-ket}].}
\label{fig:orbit-tangent-space-diagram}
\end{figure}

\begin{remark}\label{eq:remark-unkown-if-topologies-coincide-or-not-except-PLO}
Except for the case $G=G_{\rm PLO}$, we do not know whether the considered topologies on these orbits (for which they are manifolds) coincide with the subspace topologies induced by the Hilbert space norms or not. Although we would conjecture that they do coincide, we stress that all the results of this work hold \textit{regardless} of the answer.
\end{remark}

The word "natural" in \cref{thm:orbit-dim-maintext} is made precise in \cref{sec:SM-EMRep-and-proof-of-orbit-structure}, but here it can be understood to mean that, for these manifold structures, \cref{thm:VQC-orbit-dim-maintextinformal} below holds.

In \cref{eq:concrete-rank-formula-ket,eq:concrete-rank-formula-density}, $\rank_{\mathbb{R}}(\cdot):=\dim_{\mathbb{R}}(\spa_{\mathbb{R}}(\cdot))$ denotes the rank of a list of vectors treated as elements of a \textit{real} vector space \footnote{That is, it is the minimal subset size $k$ for which all vectors not in the subset can be written as \textit{real} linear combinations of vectors in the subset. The real rank of a list of vectors is hence greater or equal to its complex rank.}.
The complex Hilbert space $\fock$ with inner product $\braket{\varphi}{\varphi'}$ and for which $\{\ket{\bm{n}}\ |\ \bm{n} \in \NN^m\}$ is an orthonormal basis may also be viewed as a real Hilbert space by defining the inner product $\ip{\varphi}{\varphi'}_{\mathbb{R}}:=\Re \braket{\varphi}{\varphi'} $, for which $\{\ket{\bm{n}},i\ket{\bm{n}}\ |\ \bm{n} \in \NN^m\}$ is an orthonormal basis.
Likewise, the complex Hilbert space $B_2(\fock)$ of Hilbert-Schmidt operators on $\fock$ with inner product $\ip{A}{B}:=\Tr[A^\dagger B]$ and orthonormal basis $\{\ketbra{\bm{n}}{\bm{n}'} \ |\ \bm{n},\bm{n}' \in \NN^m\}$ can also be viewed as a real Hilbert space, by defining the inner product $\ip{A}{B}_{\mathbb{R}}:=\Re \ip{A}{B}$, for which $\{\ketbra{\bm{n}}{\bm{n}'}, i\ketbra{\bm{n}}{\bm{n}'} \ |\ \bm{n},\bm{n}' \in \NN^m\}$ is an orthonormal basis.

Since in a real Hilbert space, the rank of a list of $d$ vectors is equal to the rank of their $d \times d$ Gram matrix \cite[Theorem 7.2.10]{horn_matrix_2013}, %
\cref{eq:concrete-rank-formula-ket,eq:concrete-rank-formula-density} can be rewritten as
\begin{align}
\dim(\orb_{G}(\ket{\psi})) &= \rank( \gram_G(\ket{\psi}) )\,,\label{eq:concrete-rank-gram-formula-ket}\\
\dim(\orb_{G}(\rho)) &= \rank( \gram_G(\rho) )\,,\label{eq:concrete-rank-gram-formula-density}
\end{align}
where the $\dim(G) \times \dim(G)$ real symmetric matrices $\gram_G(\ket{\psi})$ and $\gram_G(\rho)$ are defined as
$[\gram_G(\ket{\psi})]_{I,J}:=\Re \braket{H_{I}\psi}{H_{J}\psi}$
and
$[\gram_G(\rho)]_{I,J}:=\ip{[H_{I},\rho]}{[H_{J},\rho]}$ (the latter is already real).
For pure states, these read after simplifications:
\begin{align}
[\gram_G(\ket{\psi})]_{I,J} &= \mathbb{E}_{\psi}(\{H_{I},H_{J}\})\,,\label{eq:concrete-gram-formula-ket-expression}\\
[\gram_G(\ketbra{\psi})]_{I,J} &= 2\operatorname{Cov}_{\psi}(H_{I},H_{J})\,,\label{eq:concrete-gram-formula-ketbra-expression}
\end{align}
where we denote $\{A,B\}:=\frac{1}{2}(AB + BA)$, $\mathbb{E}_{\psi}(O):=\braket{\psi}{O|\psi}$ and $\operatorname{Cov}_{\psi}(O,O'):=\mathbb{E}_{\psi}(\{O,O'\})- \mathbb{E}_{\psi}(O) \mathbb{E}_{\psi}(O')$.
For general mixed states, simplifications give %
\begin{align}
[\gram_G(\rho)]_{I,J} &= 2\Tr[\{H_I,H_{J}\} \rho^2] - 2 \Tr[H_I \rho H_{J} \rho] \,.\label{eq:concrete-gram-formula-rho-expression}
\end{align}
Note that diagonal entries $[\gram_G(\rho)]_{I,I}$ for displacement Hamiltonians are related to the \textit{quantum coherence scale} \cite{Hertz-QuadratureCoherence-2020,griffet_interferometric_2023}.

\section{Properties of orbit dimensions}\label{sec:main-text-props-of-orbit-dimensions}
Orbit dimensions are integers that can range from $0$ to $\dim(G)$ [see \cref{eq:concrete-rank-formula-ket,eq:concrete-rank-formula-density}].

Let us stress that orbits, although finite-dimensional \textit{manifolds}, are (except for $G_{\rm PLO}$ and finite-support states) not contained inside any finite-dimensional \textit{linear subspace} of the Hilbert space. This is best illustrated with the vacuum orbit $\orb_{G_{\rm GO}}(\ket{\bm{0}})$, as it contains in particular all coherent states $\{\ket{\bm{\alpha}}\}_{\bm{\alpha} \in \mathbb{C}^m}$ which are known to form an (overcomplete) basis of $\fock$ \cite{Glauber-CoherentIncoherent-1963,weedbrook_gaussian_2012}, implying that $\spa_{\mathbb{C}}(\orb_{G_{\rm GO}}(\ket{\bm 0}))=\fock$. We refer to \cref{subsec:SM-infinite-dimensional-span-of-orbits} for the proof for other orbits.
\subsection{Invariance property}
 Because an orbit under $G$ has the structure of a manifold (\cref{thm:orbit-dim-maintext}), its dimension is (tautologically) invariant under $G$ evolutions.
Two states that have different orbit dimensions under $G$ must lie in distinct orbits, and hence cannot be converted into one another by a unitary from $G$.
For example, this is found to be the case between a Fock basis state $\ket{n_1,\dots,n_m}$ and a genuine one-mode superposition state $(\sum_{n_1=0}^{N_{1}} \alpha_{n_{1}} \ket{n_{1}})\otimes\ket{n_2,\dots,n_m}$ over $m$ total modes (see \cref{tab:orbit-dimensions}).

The invariance property can be used as well to show the impossibility of realizing a desired unitary action on a logical subspace deterministically by a global $U\in G$. We illustrate this in \cref{sec:SM-CNOT-dual-rail} for the case of the CNOT gate in linear or Gaussian quantum optics in the dual-rail encoding \cite{knill_scheme_2001}: orbit dimensions alone reveal an obstruction to its deterministic implementation.

For pure states, the orbit dimension in the ket picture is a valid invariant in the ketbra picture [since \cref{eq:concrete-rank-formula-ket} is global-phase invariant], and hence still constitutes a genuine property of a ketbra pure state.
For instance (see \cref{tab:orbit-dimensions}), even though over $m=2$ modes, one-mode superposition states with extra vacuum $\sum_{n_1=0}^{N_{1}} \alpha_{n_{1}} \ket{n_{1},0}$ and $N\geq3$ NOON states $(|N,0\rangle+|0,N\rangle)/\sqrt{2}$ have the same $G_{\rm PLO}$ orbit dimensions in the ketbra picture (3), the fact that their $G_{\rm PLO}$ orbit dimensions in the ket picture differ (3 and 4, respectively) implies that they lie in different $G_{\rm PLO}$ orbits in \textit{both} the ket and ketbra pictures.

Interestingly, for $G_{\rm PLO}$ it is shown in Ref.~\cite{Parellada-LieAlgebraic-2026} that besides its rank, the whole \textit{spectrum} of the covariance matrix of \cref{eq:concrete-gram-formula-ketbra-expression} --- after crucially changing our conventional prefactors of $1/2$ in \cref{eq:def-generator-e,eq:def-generator-E} to $1/\sqrt{2}$ --- is invariant.

Lastly, we stress here that two states having the same orbit dimension is, of course, only a \textit{necessary} condition for their interconvertibility, and not a sufficient one; e.g., two Fock basis states whose occupation-number vectors are not permutations of each other are known to be inequivalent under $G_{\rm PLO}$ \cite{migdal_which_2014}, even though their $G_{\rm PLO}$ orbits can have the same dimension (see \cref{tab:orbit-dimensions}).

\subsection{Genericity on finite-dimensional shells}

The maximal orbit dimension value reached inside a finite-dimensional subspace of physical states is attained for all \textit{generic} states. In particular, for an energy cutoff subspace $\oplus_{0 \leq n \leq N}\,\fockn \subset \fock$, the following holds:
\begin{proposition}\label{prop:generic-orbit-dim}
Let $\ket{\psi}$ be a pure state drawn at random (according to the uniform measure) from the complex unit sphere on the subspace
$\mathcal{H}_{m}^{\leq N}$ of total photon number at most $N$.
Then, with probability one,
\begin{align}\label{eq:generic-orbit-dim-energy-cutoff}
\dim(\orb_{G}(\ket{\psi})) = \dim(G) - \delta_{N=0} m^2 - \delta_{N=1} (m-1)^2\,.
\end{align}
As a consequence, it also holds that the set of states satisfying \cref{eq:generic-orbit-dim-energy-cutoff} is dense in the unit sphere on $\mathcal{H}_{m}^{\leq N}$.
\end{proposition}
The special cases in \cref{eq:generic-orbit-dim-energy-cutoff} stem from the particularly simple structure of the PLO action on zero- and one-photon subspaces. %
We prove this result in \cref{sec:SM-genericity}, by relying on the fact that \cref{eq:concrete-gram-formula-ket-expression,eq:concrete-gram-formula-ketbra-expression,eq:concrete-gram-formula-rho-expression} are \textit{analytic} functions on the sphere, and by constructing explicit state families $\ket{\psi_{m,N}}$ in these subspaces that achieve the claimed maximal orbit dimensions.
States with submaximal orbit dimension are in this sense the exception rather than the norm, yet many states of interest in quantum optics are of this kind, owing to their highly structured form; see, e.g., \cref{tab:orbit-dimensions}.

\subsection{Robustness from below}

Given a state of a certain orbit dimension $r$, there exists an $\epsilon>0$ and a finite number of energy moments, such that a disturbance of the state that shifts these moments by less than $\epsilon$ guarantees the new orbit dimension to be at least as large as $r$. In other words:

\begin{proposition}\label{prop:lsc-orbit-dim}
On physical states and density operators, the orbit dimension maps $\ket{\psi} \mapsto \dim(\orb_{G}(\ket{\psi}))$ and $\rho \mapsto \dim(\orb_{G}(\rho))$ are \textit{lower semicontinuous} with respect to the Schwartz topologies.
\end{proposition}

We refer to \cref{sec:SM-Schwartz-spaces-and-proof-of-lsc} for a detailed proof.

\subsection{Implications for number of directions explored in state space by bosonic variational quantum circuits}

Consider a variational quantum circuit (VQC),
\begin{equation}\label{eq:VQC-orbit-dim-maintextinformal-definition}
U(\bm{\theta}) := W_{p} e^{-i \theta_p H_{p}} W_{p-1} \cdots W_{1} e^{-i \theta_1 H_{1}} W_{0}\,,
\end{equation}
where each $H_{k}$ is an element of the Lie algebra basis $\mathcal{B}_\mathfrak{g}$ and each $W_k$ is a fixed unitary in $G$.
Given an initial state $\ket{\psi}$ (resp. $\rho$), let $\ket{\psi_{\rm out}(\bm{\theta})} := U(\bm{\theta}) \ket{\psi}$ [resp. $\rho_{\rm out}(\bm{\theta}) := U(\bm{\theta}) \rho U(\bm{\theta})^\dagger$] be the corresponding output state of the VQC.
At parameter values $\bm{\theta}$, denote by $\operatorname{localdim}_{\ket{\psi}}(\bm{\theta})$ [resp. $\operatorname{localdim}_{\rho}(\bm{\theta})$] the number of directions explored in state space by the output state $\ket{\psi_{\rm out}(\bm{\theta}')}$ [resp. $\rho_{\rm out}(\bm{\theta}')$] when infinitesimally varying the $p$ parameters $\bm{\theta}'=(\theta_1',\dots,\theta_p')$ around $\bm{\theta}$.
\begin{theorem}[Orbit dimensions upper-bound local dimensions of bosonic VQCs]\label{thm:VQC-orbit-dim-maintextinformal}
If the initial state $\ket{\psi}$ (resp. $\rho$) is an exponential-decay state, then at all parameter values $\bm{\theta} \in \mathbb{R}^p$:
\begin{align}
\operatorname{localdim}_{\ket{\psi}}(\bm{\theta}) &\leq \dim(\orb_{G}(\ket{\psi}))\,,\\
\text{resp.}\qquad\
\operatorname{localdim}_{\rho}(\bm{\theta}) &\leq \dim(\orb_{G}(\rho))\,.
\end{align}
For a pure state in the ketbra picture, this can be equivalently stated in terms of the quantum Fisher information matrix $\operatorname{QFIM}_\psi(\bm{\theta})$ of $\ket{\psi_{\rm out}(\bm{\theta})}$ as
\begin{equation}
\rank( \operatorname{QFIM}_\psi(\bm{\theta}) ) \leq \dim(\orb_{G}(\ketbra{\psi}))\,.
\end{equation}
\end{theorem}
For technical reasons, we required the Fock basis coefficients decay of $\ket{\psi}$ ($\rho$) to be subtly stronger than general physical states (exponential, rather than just faster than polynomial).
Note that the quantity $\operatorname{localdim}_{\rho}(\bm{\theta})$ is called the "number of degrees of freedom" of the output state in Ref.~\cite{Monbroussou-QuantumAdvantage-2025}.
A more general and formal version of \cref{thm:VQC-orbit-dim-maintextinformal} is proven in \cref{subsec:proof-thm-2}.
The reason why this theorem is not an immediate feature of \cref{thm:orbit-dim-maintext} is that we must do without the fact that topologies on orbits are compatible with the ambient ones [since we do not prove this fact in this work, cf. \cref{eq:remark-unkown-if-topologies-coincide-or-not-except-PLO}]; the way we are able to circumvent this gap in the proof is by again exploiting the property of \textit{analyticity} of the relevant functions.
An illustrative consequence of this theorem is deferred to \cref{sec:disc-and-outlook} (see \cref{fig:vqc-input-state-rank-scaling-dense}).

\section{Evaluating orbit dimensions}\label{sec:evaluation-orbit-dimensions}

\subsection{Calculations in the Fock representation}

For states with a finite expansion in the Fock basis, their orbit dimensions can be evaluated numerically or with symbolic computation through direct evaluation of either \cref{eq:concrete-rank-formula-ket,eq:concrete-rank-formula-density} or \cref{eq:concrete-rank-gram-formula-ket,eq:concrete-rank-gram-formula-density,eq:concrete-gram-formula-ket-expression,eq:concrete-gram-formula-ketbra-expression,eq:concrete-gram-formula-rho-expression}. %
This scales efficiently for state families whose expansion size grows as $\mathcal{O}(\text{poly}(m))$.
A symbolic implementation is provided in Ref.~\cite{Bosonic-Orbit-Dimensions-Github-cff}.

For some states, closed-form formulas can also be derived analytically.

Consider, for example, a Fock basis state $\ket{\psi}=\ket{\bm{n}}$ and the group $G=G_{\rm PLO}$. The generators act on $\ket{\psi}$ as
\begin{align}
e_{kl}\ket{\psi} &=
\begin{aligned}[t]
&\frac{1}{2}\sqrt{(n_k+1)n_l} \ket{\dots,n_k+1,\dots,n_l-1,\dots}\\
+\ &\frac{1}{2}\sqrt{n_k(n_l+1)} \ket{\dots,n_k-1,\dots,n_l+1,\dots}\,,
\end{aligned}\\
E_{kl}\ket{\psi} &=
\begin{aligned}[t]
\frac{i}{2}&\sqrt{(n_k+1)n_l} \ket{\dots,n_k+1,\dots,n_l-1,\dots}\\
-\ \frac{i}{2}&\sqrt{n_k(n_l+1)} \ket{\dots,n_k-1,\dots,n_l+1,\dots}\,,
\end{aligned}\\
N_{k}\ket{\psi} &= n_k \ket{\dots,n_k,\dots,n_l,\dots}\,.
\end{align}

Denote by $u$ the number of \textit{unoccupied} modes in $\ket{n_1,\dots,n_m}$.
Observe that (i) $e_{kl}\ket{\psi}$ and $E_{kl}\ket{\psi}$ are zero if and only if $n_k=n_l=0$, and hence %
there are $\binom{m}{2} - \binom{u}{2}$ nonzero vectors of each type; (ii) different vectors $e_{kl}\ket{\psi}$ and $e_{k'l'}\ket{\psi}$ or $E_{kl}\ket{\psi}$ and $E_{k'l'}\ket{\psi}$ are complex-orthogonal, since they are superpositions of product states whose $j^{\rm th}$ slot states are complex-orthogonal for $j$ a noncommon mode between $\{k,l\}$ and $\{k',l'\}$; (iii) all $e_{kl}\ket{\psi}$ and $N_{k'}\ket{\psi}$ or $E_{kl}\ket{\psi}$ and $N_{k'}\ket{\psi}$ are likewise complex-orthogonal (here due to the orthogonality of the $j^{\rm th}$ slots for any $j\in\{k,l\}$); (iv) vectors $E_{kl}\ket{\psi}$ are real-orthogonal to all $e_{kl}\ket{\psi}$ and $N_k\ket{\psi}$ since their coefficients in the Fock basis are purely imaginary for the former and purely real for the latter; (v) all $N_k\ket{\psi}$ are real-collinear, and at least one of them is nonzero (except for the vacuum state $\ket{\psi}=\ket{0}$).
Therefore, these $m^2$ vectors consist of $2[\binom{m}{2} - \binom{u}{2}] + \delta_{u\neq m}$ nonzero and real-orthogonal vectors, and the rest are zero, which by \cref{eq:concrete-rank-formula-ket} establishes %
that
\begin{equation}\label{eq:fock-state-ket-orbit-dim-PLO}
\dim(\orb_{G_{\rm PLO}}(\ket{\bm{n}})) = m(m-1) - u(u-1) + \delta_{u \neq m}\,.
\end{equation}

One may reason similarly for the other groups $G$, and in the ketbra picture as well (the only difference in that case is that phase-shifter generators $N_k$ do not contribute). Other classes of states can also be tackled without major complications, this time favoring the Gram matrix approach of \cref{eq:concrete-rank-gram-formula-ket,eq:concrete-rank-gram-formula-density,eq:concrete-gram-formula-ket-expression,eq:concrete-gram-formula-ketbra-expression,eq:concrete-gram-formula-rho-expression} and evaluating its rank by applying standard linear algebraic manipulations to the matrix.
We collect some examples obtained in this way in \cref{tab:orbit-dimensions}.

\renewcommand{\arraystretch}{1.25} %
\begin{table*}[!t]
\begin{threeparttable}
\begin{tabular}{@{}llll@{}}
\toprule
\multirow[t]{2}{*}{Unitary} & \multirow[t]{2}{*}{Fock basis state} & \multirow[t]{2}{*}{One-mode superposition state} & \multirow[t]{2}{*}{NOON state ($N\geq3$)}\\
group $G$ & $\ket{n_1,\dots,n_m}$ & $(\sum_{n_1=0}^{N_{1}} \alpha_{n_{1}} \ket{n_{1}})\otimes\ket{n_2,\dots,n_m}$ & $\frac{1}{\sqrt{2}}(\ket{N,0} + \ket{0,N})\otimes\ket{n_3,\dots,n_m}$\\
 \midrule
\multirow[t]{2}{*}{$G_{\rm PLO}$} & $m(m-1) - u(u-1) + \delta_{\ket{\psi}} \delta_{u \neq m}$ & $m(m-1) - u(u-1) + 1 + \delta_{\ket{\psi}} \delta_{u \neq m-1}$ & $m(m - 1) - u(u - 1) + 1 + \delta_{\ket{\psi}}$ \\[4pt]
\multirow[t]{2}{*}{$G_{\rm DPLO}$} & $m(m+1) - u(u-1) + \delta_{\ket{\psi}}$ & $m(m+1) - u(u-1) + 1 + \delta_{\ket{\psi}}$ & $m(m+1) - u(u-1) + 1 + \delta_{\ket{\psi}}$ \\[4pt]
\multirow[t]{2}{*}{$G_{\rm ALO}$} & $2m^2 - u(u-1) + \delta_{\ket{\psi}}$ & $2m^2 - u(u-1) + 1 + \delta_{\ket{\psi}}$ & $2m^2 - u(u-1) + 1 + \delta_{\ket{\psi}}$ \\[4pt]
\multirow[t]{2}{*}{$G_{\rm GO}$} & $2m(m+1) - u(u-1) + \delta_{\ket{\psi}}$ & $2m(m+1) - u(u-1) + 1 + \delta_{\ket{\psi}}$ & $2m(m+1) - u(u-1) + 1 + \delta_{\ket{\psi}}$ \\
\bottomrule
\end{tabular}
\end{threeparttable}
\caption{Orbit dimensions, under common quantum optics groups, of some classes of pure states $\ket{\psi}$. Values are stated for both ket \big[$\dim(\orb_{G}(\ket{\psi}))$\big] and ketbra \big[$\dim(\orb_{G}(\ketbra{\psi}))$\big] pictures, with $\delta_{\ket{\psi}}=1$ for the ket picture and $0$ for the ketbra picture. For all cases of $\ket{\psi}$ listed, the number $u$ is defined as the number of \textit{unoccupied} modes in the product part of the state that is a Fock basis state, i.e., in $\ket{\psi} = (\cdots)\otimes\ket{n_k,\dots,n_m}$, $u:=|\{i \in \{ k, \dots, m\} \ | \ n_i = 0 \}|$. The one-mode superposition state is assumed to be genuine, i.e., to contain at least two nonzero terms in its superposition.}
\label{tab:orbit-dimensions}
\end{table*}

\subsection{Calculations in phase-space representations}
States in quantum optics are often alternatively described through phase-space representations, which map states $\ket{\psi}$ or $\rho$ to certain functions on the phase space $\mathbb{C}^m\cong\mathbb{R}^{2m}$.
Two examples are the \textit{stellar-function} representation $\ket{\psi} \leftrightarrow f^\star_\psi(\bm{z})$ of pure states \cite{chabaud_stellar_2020,Chabaud-ResourcesBosonic-2023}, and the \textit{Wigner-function} representation $\rho \leftrightarrow W_\rho(\bm{q},\bm{p})$ of mixed states \cite{weedbrook_gaussian_2012}.
As both of these correspondences are linear and one-to-one, it directly follows that \cref{eq:concrete-rank-formula-ket,eq:concrete-rank-formula-density} remain valid in the stellar representation and the Wigner representation respectively, with the elements $H_k \ket{\psi}$ and $[H_k,\rho]$ replaced with the action $\widehat{H_k}^{(\star)} f^\star_\psi$ or $\widehat{H_k}^{(\mathrm{W})} W_\rho$ of a corresponding differential operator $\widehat{H_k}$ onto the state function.
That is, orbit dimensions can then be calculated via \cref{eq:concrete-rank-formula-ket,eq:concrete-rank-formula-density} as the rank of lists of \textit{functions}:
\begin{align*}
\dim(\orb_{G}(\ket{\psi})) &= \rank_{\mathbb{R}}(\{ \widehat{H_1}^{(\star)}\! f^\star_\psi,\dots, \widehat{H_d}^{(\star)}\! f^\star_\psi\})\,,\\
\dim(\orb_{G}(\rho)) &= \rank_{\mathbb{R}}(\{ \widehat{H_1}^{(\mathrm{W})}\! W_\rho,\dots,\widehat{H_d}^{(\mathrm{W})}\! W_\rho\})\,,
\end{align*}
which can be more convenient in cases of infinite-support states.
Notably, this can be used to show the non-$G_{\rm GO}$-convertibility between pairs of states of infinite stellar ranks \cite{chabaud_stellar_2020,Walschaers-NonGaussianQuantum-2021} (which cannot be assessed with the concept of stellar rank alone), as shown in \cref{subsec:SM-stellar-representation-examples}.
Note that since the above correspondences are in fact Hilbert space isomorphisms (they preserve inner products), the values of Gram matrix entries are independent of the (Fock or phase-space) representation in which they are calculated.
We refer to \cref{sec:SM-phase-space-representations} for details and example calculations.

\subsection{Experimental estimation of orbit dimensions}\label{subsec:experimental-schemes}

\subsubsection{Pure states: Homodyne and heterodyne schemes}
For pure states, \cref{eq:concrete-rank-gram-formula-ket,eq:concrete-rank-gram-formula-density,eq:concrete-gram-formula-ket-expression,eq:concrete-gram-formula-ketbra-expression} relate orbit dimensions to the rank of Gram matrices whose entries can each be estimated via expectation values of some Hermitian operators of the form $\{ H_I, H_J \}$ for the ket picture [\cref{eq:concrete-gram-formula-ket-expression}] and $\{ H_I, H_J \}$ or $H_I$ in the ketbra picture [\cref{eq:concrete-gram-formula-ketbra-expression}].
In both cases, this involves $\mathcal{O}(\dim(\mathfrak{g})^2) = \mathcal{O}(m^4)$ polynomial observables of degree at most $4$ in $a_k,a^\dagger_k$ or equivalently in $q_k,p_k$ [via \cref{eq:def-generator-q,eq:def-generator-p}].
Such constant-degree polynomials can be in principle estimated using heterodyne measurement (by postprocessing by each polynomial's antinormal form), but perhaps less obviously, also using homodyne measurements:
\begin{proposition}
\label{prop:homodyne-measurements-maintext}
Joint measurements of homodyne observables $\cos(\theta_k)q_k + \sin(\theta_k)p_k$ over four fixed modes, with each phase $\theta_k$ running through any fixed list of five distinct angles modulo $\pi$ [e.g., $(0, \pi/2, \pi/4, -\pi/4, \pi/5)$], suffice to estimate (via appropriate postprocessing) expectation values of all degree-$4$ polynomials in these four modes.
\end{proposition}
The details of this procedure and its proof are given in \cref{sec:measuring-pure-gram-matrix-entries}.
The method outlined there, of estimating constant-degree polynomial observables using joint local homodyne measurements involving a finite number of phase settings in each mode, may be of independent interest.
Shadow-based estimation techniques for bosonic systems \cite{Becker-ClassicalShadow-2024,Gandhari-PrecisionBounds-2024,Thomas-SheddingLight-2025} may also be considered to estimate the list of degree-$4$ observables.

\subsubsection{Mixed states: Two-copy PNRD scheme}
For general mixed states, entries of $\gram_G(\rho)$ involve second-order moments of the state [\cref{eq:concrete-gram-formula-rho-expression}]. This suggests \cite{brun_measuring_2004} attempting to relate them to expectation values of observables defined on two parallel copies of the state ($\rho\otimes\rho$); however, it %
is not clear how the measurement of those observables is to be implemented.
Still, let us theoretically relate orbit dimensions under $G$ to experimental data, assuming access to two copies of the state and photon-number resolving detectors (PNRDs).
By the second-order approximation of small-time evolutions $U_{H}^{(t)}(\rho) := \exp(-i H t) \rho \exp(i H t) = \rho - it [H,\rho] - \frac{t^2}{2}[H,[H,\rho]] + \mathcal{O}(t^3)$ on physical states $\rho$, entries of the Gram matrix can be rewritten
as
\begin{align}\label{eq:gram-matrix-second-order-derivative-expression}
\begin{aligned}[t]
&[\gram_G(\rho)]_{I,J}
=\\
&\frac{1}{2}\left(
\evalat[\Big]{\frac{d^2}{dt^2}}{t=0} \beta_{I,J}^\rho(t)
- \evalat[\Big]{\frac{d^2}{dt^2}}{t=0} \beta_{I,0}^\rho(t)
- \evalat[\Big]{\frac{d^2}{dt^2}}{t=0} \beta_{0,J}^\rho(t)
\right)\,,
\end{aligned}
\end{align}
with $\beta_{I,J}^\rho(t) := \ip{U_{H_I}^{(t)}(\rho)}{U_{H_J}^{(t)}(\rho)}$, $\beta_{I,0}^\rho(t) := \ip{U_{H_I}^{(t)}(\rho)}{\rho}$, and $\beta_{0,J}^\rho(t) := \ip{\rho}{U_{H_J}^{(t)}(\rho)}$.
We refer to \cref{sec:SM-swap-test-second-derivatives} for a more detailed derivation of \cref{eq:gram-matrix-second-order-derivative-expression}.
By evolving two parallel copies of $\rho$ for time $t$ under $H_I$ and $H_J$, respectively, the quantum optical \textit{SWAP test} \cite{volkoff_ancilla-free_2022,griffet_interferometric_2023} (using a cascade of $m$ beam
splitters and PNRDs) then provides experimental access to the quantity $\beta_{I,J}^\rho(t)$ (in expectation value), as depicted in \cref{fig:evolved-swap-test}.
Having obtained estimates of $\beta_{I,J}^\rho(t)$ for several small times $t \approx 0$, one can deduce an estimate of its second-order derivative at $t=0$ (although this estimator will not be unbiased in general). Estimating such derivatives for $\beta_{I,0}^\rho(t)$ and $\beta_{0,J}^\rho(t)$ can be done analogously (by evolving only one of the two copies before the SWAP test). Combining these three estimates through \cref{eq:gram-matrix-second-order-derivative-expression} thus yields an estimate of the entry $[\gram_G(\rho)]_{I,J}$.

\begin{figure}[h]
\centering
\includegraphics{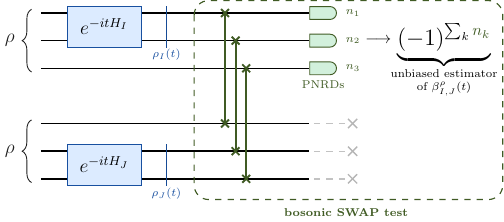}

\caption{Illustration (for an $m=3$ mode state $\rho$, and for $H_I$, $H_J$ acting on modes $\{1,2\}$ and $\{2,3\}$) of the proposed scheme for the estimation of $\beta_{I,J}^\rho(t)$ using PNRDs, by performing the bosonic SWAP test \cite{volkoff_ancilla-free_2022,griffet_interferometric_2023} on two copies of the state $\rho$ that have undergone different evolutions (under $H_I$ and $H_J$, which can be one-mode or two-mode, cf. \cref{tab:Lie-algebra-bases}) for the same time $t$. This bosonic SWAP test consists of a cascade (dark green vertical lines) of non-phase-shifting 50:50 beam splitters $e^{+ i ({\pi/2}) E_{k\, (m+k)}}$ [cf. \cref{eq:def-generator-E}] between the $k^{\rm th}$ mode of the first and the $k^{\rm th}$ mode of the second system, followed by photon counters (PNRDs) measuring only one of the two copies and returning (as a postprocessing) the parity of the total number of photons observed.}
\label{fig:evolved-swap-test}
\end{figure}

\subsubsection{Sensitivity to noise}
In practice, due to experimental noise, the preparation of a target state $\rho$ (pure or mixed) produces a nearby mixed state $\rho_{\mathrm{expt}}$. \Cref{prop:lsc-orbit-dim} at least guarantees that under a sufficiently low level of "physical noise" (i.e., affecting energy moments in a controlled way), $\rank(\gram_G(\rho_{\mathrm{expt}})) \geq \dim(\orb_{G}(\rho))$; but in view of \cref{prop:generic-orbit-dim} and moreover of shot noise, equality may seem hopeless.
However, after having constructed the matrix $\gram_{G,\mathrm{expt}}(\rho_{\mathrm{expt}})$ from experimental shots, evaluating its rank is done using a numerical routine that counts the number of eigenvalues greater than a tolerance threshold $\tau$.
Consequently, if $\gram_{G,\mathrm{expt}}(\rho_{\mathrm{expt}})$ and $\gram_G(\rho)$ are $\epsilon$-close in each entry (this $\epsilon$ accounts for both preparation and shot noise) with $\epsilon< \lambda_{\mathrm{min}}^+/[2\dim(\mathfrak{g})]$, then the numerical rank evaluation on the experimentally built matrix is at least guaranteed to yield $\dim(\orb_{G}(\rho))$ when $\tau \in (\dim(\mathfrak{g}) \epsilon,\ \lambda_{\mathrm{min}}^+ - \dim(\mathfrak{g}) \epsilon)$ (see \cref{sec:SM-numerical-rank-noisy-Gram-matrix} for details); where $\lambda_{\mathrm{min}}^+$, which we refer to as a "spectral gap" parameter, denotes the least positive eigenvalue of $\gram_G(\rho)$.

\section{Relation between orbit dimensions and existing resource theories of quantum optics}\label{sec:rel-RTs}

\subsection{Resource theories of quantum optics}\label{subsec:RTs-of-qo}
We begin by briefly recalling three existing \textit{resource theories} (RTs) of quantum optics that have received considerable attention in the literature. Resource theories \cite{Chitambar-QuantumResource-2019} are useful frameworks to capture situations in which subclasses of quantum states and operations, deemed as \textit{free states} and \textit{free operations}, are singled out for being easy to obtain experimentally, while at the same time being insufficient to perform useful quantum information protocols or avoid classical simulation. A \textit{monotone} in an RT is then a state property $f(\rho) \in \mathbb{R}$ such that for all free operations $\mathcal{E}$ (and for all states $\rho$), $f(\mathcal{E}(\rho)) \leq f(\rho)$.

First, let us denote by $\mathcal{G}$ the set of \textit{Gaussian states}, which consists of all states $\rho$ whose Wigner function $W_\rho(\cdot)$ is a Gaussian function in phase space, and by $\mathcal{G}^{\rm pure}$ the subset of those that are pure.
In the RT of \textit{genuine non-Gaussianity}, free states are defined as the "continuous" convex hull of $\mathcal{G}$, which we denote by $\operatorname{conv}(\mathcal{G})$; it is the set of all convex mixtures
\begin{equation}\label{eq:convex-mixture-Gaussian-states}
\rho = \int d \lambda \, p_\lambda \, \rho_\lambda
\end{equation}
of Gaussian states $\rho_\lambda$ by a (possibly continuous) probability distribution $p_\lambda$ over $\mathcal{G}$ \cite{Takagi-ConvexResource-2018}.
In fact, it holds \cite{Takagi-ConvexResource-2018,Albarelli-ResourceTheory-2018} that those always admit pure-state decompositions, i.e., the Gaussian states $\rho_\lambda$ in \cref{eq:convex-mixture-Gaussian-states} can be assumed to be pure, hence
\begin{equation}
\operatorname{conv}(\mathcal{G}) = \operatorname{conv}(\mathcal{G}^{\rm pure})\,.
\end{equation}
The states outside of $\operatorname{conv}(\mathcal{G})$ are termed \textit{genuine non-Gaussian} states \cite{Takagi-ConvexResource-2018}.

Second, in the RT of \textit{Wigner negativity}, free states are defined as the set $\mathcal{W}_+$ of states $\rho$ whose Wigner function $W_\rho(\cdot)$ is non-negative; as such, $\mathcal{G} \subseteq \mathcal{W}_+$.

Third, in the RT of \textit{(optical) nonclassicality}, free states are defined as the set $\mathcal{P}_+$ of states $\rho$ whose Glauber-Sudarshan $P$ function \cite{Glauber-CoherentIncoherent-1963,Sudarshan-EquivalenceSemiclassical-1963} is non-negative. We refer to the states outside of $\mathcal{P}_+$ as \textit{$P$ nonclassical} states. Since the $P$ function of a quantum state $\rho$ is a quasidistribution decomposition of $\rho$ into coherent states, i.e.,
\begin{equation}\label{eq:convex-mixture-classical-states}
\rho = \int d^2\bm{\alpha} \, P(\bm{\alpha}) \, \ketbra{\bm{\alpha}} \,,
\end{equation}
and since coherent states are pure, it follows \cite{Hillery-ClassicalPure-1985} that the pure subset $\mathcal{P}_+^{\rm pure}$ is exactly the coherent states, %
and hence that \cite{Yadin-OperationalResource-2018,Tan-NonclassicalLight-2019,Tan-ResourceTheories-2019}
\begin{equation}
\mathcal{P}_+ = \operatorname{conv}(\mathcal{P}_+^{\rm pure})\,.
\end{equation}

In the realm of \textit{pure} states, these notions simplify considerably: Wigner-non-negative pure states coincide with Gaussian pure states \cite{Hudson-WhenWigner-1974,Soto-WhenWigner-1983}, and, as just said, $P$-classical pure states coincide with coherent states.

\subsection{Minimal orbit dimensions}\label{subsec:minimal-orbdims}

All Gaussian pure states (resp. coherent states) live in the same orbit under $G_{\rm GO}$ (resp. $G_{\rm DPLO}$), that of the vacuum $\ketbra{\bm{0}}$ \cite{weedbrook_gaussian_2012,Walschaers-NonGaussianQuantum-2021},
and these orbits have dimension
\begin{align}
\dim(\orb_{G_{\rm GO}}(\ketbra{\bm{0}})) &= m(m+3)\,,\label{eq:GO-orbdim-vac}\\
\dim(\orb_{G_{\rm DPLO}}(\ketbra{\bm{0}})) &= 2m\,,\label{eq:DPLO-orbdim-vac}
\end{align}
as per the first column of \cref{tab:orbit-dimensions} (for $u=m$).
Thus, a $G_{\rm GO}$ (resp. $G_{\rm DPLO}$) orbit dimension of a pure state found to be greater than $m(m+3)$ ($2m$) acts as a witness of its non-Gaussianity or Wigner negativity (resp. its nonclassicality).

Interestingly, as covered by our next two results, we show that (i) the above orbit dimensions of \cref{eq:GO-orbdim-vac,eq:DPLO-orbdim-vac} are \textit{minimal} among all states $\rho$ (\cref{thm:minimal-orbdim-statements-main}), and
(ii) among \textit{pure} states, they are the \textit{unique} orbits of minimal dimension in the multimode setting (\cref{cor:minimal-orbdim-purestates-characterization-main}).

\begin{theorem}[Characterization of minimal orbit dimensions]
\label{thm:minimal-orbdim-statements-main}
Let $\rho$ be a physical state, and consider one of the four $m$-mode quantum optical unitary groups $G$ of \cref{tab:Lie-algebra-bases}.
Then,
\begin{enumerate}[label=(\arabic*)]
\item the following lower bounds hold:
\begin{equation}\label{eq:minimal-centered-orbdim-general-lowerbounds-main}
\dim(\orb_{G}(\rho)) \geq \dim(G)^{\scriptscriptstyle\!-} - K(\rho) \geq \dim(G)^{\scriptscriptstyle\!-} -\, m^2\,,
\end{equation}
where $K(\rho) := k_1^2 + \cdots + k_r^2$ denotes the sum of the squares of the multiplicities in the symplectic spectrum of $\rho$;

\item
and the following equivalence holds:
\begin{equation}\label{eq:minimal-orbdim-purestates-characterization-proof-GO-equiv-main}
\begin{gathered}
\dim(\orb_G(\rho)) = \dim(G)^{\scriptscriptstyle\!-} - m^2\\
\iff\\
\rho \sim_G \sigma,\\
\end{gathered}
\end{equation}
where $\sigma$ is a state that is stabilized by all of $G_{\rm PLO}$, i.e.,
$\sigma = \oplus_{n \in \mathbb{N}} \, p_n \id_{\fockn}$ for some $p_n \in \mathbb{R}\,.$
\end{enumerate}
\end{theorem}
In the above statements, we denoted $\dim(G)^{\scriptscriptstyle\!-} := \dim(G) - \delta_{G \in (G_{\rm DPLO},G_{\rm ALO},G_{\rm GO})}$ \footnote{This appears as a reflection of the fact that the identity Hamiltonian, which appears for these three groups' Lie algebra bases of \cref{tab:Lie-algebra-bases}, never contributes to the orbit dimension in the density operator picture.}.
We refer to \cref{subsec:SM-minimal-orbit-dimensions} for a proof of \cref{thm:minimal-orbdim-statements-main}, which hinges on a study of the structure of the covariance matrices that remain invariant under the action of Gaussian unitaries.

\subsubsection{Pure states: Coincidence with pure free states for $m\geq 2$}

Since in the multimode setting, the only pure state that is of the form of $\sigma$ [in \cref{{eq:minimal-orbdim-purestates-characterization-proof-GO-equiv-main}}] is the vacuum state, we immediately obtain the following corollary:
\begin{corollary}[Characterization of pure states of minimal orbit dimension]\label{cor:minimal-orbdim-purestates-characterization-main}
In the \emph{multimode} setting $(m\geq 2)$, there is a \emph{unique} \emph{pure} $G$-orbit of minimal dimension [$\dim(G)^{\scriptscriptstyle\!-} - m^2$],
and it coincides exactly with the \emph{pure free states} (over $m$ modes) of the following resource theories (RTs) of quantum optics:
\begin{enumerate}[label=(\arabic*)]
\item for $G=G_{\rm GO}$: the RT of \emph{genuine non-Gaussianity} or the RT of \emph{Wigner negativity};
\item for $G=G_{\rm DPLO}$: the RT of \emph{optical nonclassicality}.
\end{enumerate}
\end{corollary}

\Cref{cor:minimal-orbdim-purestates-characterization-main} grants the property of orbit dimension under $G_{\rm GO}$ (resp. $G_{\rm DPLO}$) the status of a \textit{universal} witness of non-Gaussianity (resp. nonclassicality) for multimode pure states.
Note that we do not investigate whether such characterization can be turned into an \textit{efficient} practical tester of non-Gaussianity for pure states \cite{Girardi-ItGaussian-2025} under appropriate assumptions (which would provide lower bounds on the spectral gap parameter $\lambda_{\mathrm{min}}^+$).

\subsubsection{Mixed states: Nonequivalence with free states}

One may naturally wonder next if this equivalence remains true for mixed states, i.e., if a mixed state being free in these RTs is equivalent to the minimality of the associated $G$-orbit dimension.
The answer is negative, as per the following \cref{rem:minimal-GO-orbdim-vs-Gaussianity-main}.
Let us recall here that \textit{thermal states} \cite{Adesso-ContinuousVariable-2014,Tan-NonclassicalLight-2019} are the states of the form $\tau_{\bm{\beta}} = \tau_{\beta_1} \otimes \cdots \otimes \tau_{\beta_m}$, where for $\beta>0$, $\tau_\beta$ denotes the single-mode thermal state given by $\tau_\beta := (1-e^{-\beta}) \sum_{n=0}^{\infty} e^{- \beta n} \ketbra{n}$.
We refer to \textit{equal-temperature} thermal states for those where $\beta_1 = \cdots = \beta_m$, and to \textit{non-equal-temperature} thermal states otherwise. Note that thermal states $\tau_{\bm{\beta}}$ are Gaussian, and are equal-temperature if and only if their symplectic eigenvalues are all equal \cite{weedbrook_gaussian_2012}.
\begin{remark}
[For mixed states, minimal $G_{\rm GO}$ ($G_{\rm DPLO}$)-orbit dimension is neither necessary nor sufficient for freeness in RT of genuine non-Gaussianity (nonclassicality)]
\label{rem:minimal-GO-orbdim-vs-Gaussianity-main}
The family of states $\sigma$ in \cref{thm:minimal-orbdim-statements-main} [\cref{eq:minimal-orbdim-purestates-characterization-proof-GO-equiv-main}] contains in particular all equal-temperature thermal states. However, note that this family of states $\sigma$
\begin{enumerate}
\item[(i)] does not contain states $G_{\rm GO}$ equivalent to non-equal-temperature thermal states

(because symplectic eigenvalues are $G_{\rm GO}$ invariants \cite{Adesso-ContinuousVariable-2014});
\item[(ii)] contains states outside of both $\operatorname{conv}(\mathcal{G})$ and $\mathcal{P}_+$, such as $\id_{\fockn}/\dim(\fockn)$ for $n\geq1$

(because the latter are only supported on $\fockn$, which contains no pure Gaussian states for $n\geq1$
\footnote{If a state supported on an $n\geq1$ $\fockn$ were in $\operatorname{conv}(\mathcal{G})$ or $\mathcal{P}_+$, its associated function $p_\lambda$ or $P(\bm{\alpha})$ in \cref{eq:convex-mixture-Gaussian-states,eq:convex-mixture-classical-states} would be a probability distribution, which would have to be supported in $\fockn$ as well [cf. \cref{eq:counterexample-convex-roof-zero-orthogonal-support}], but that is impossible since the only pure Gaussian state in a $\fockn$ is the vacuum, with $n=0$.}).
\end{enumerate}

Thus, in the space of mixed states, the subset of states of minimal $G_{\rm GO}$ (resp. $G_{\rm DPLO}$) orbit dimension has no inclusion relation with $\operatorname{conv}(\mathcal{G})$ (resp. $\mathcal{P}_+$).%
\end{remark}

\subsubsection{Convex-roof based orbit dimensions}

To alleviate this mismatch, we may introduce an alternative notion of "orbit dimension" for mixed states, based on a (continuous) \textit{convex roof} construction on top of the pure state definition. That is, we keep the definition of $\dim(\orb_G(\rho))$ of \cref{sec:orbs-of-q-optics} only for pure states $\rho$, and we extend it on mixed states as
\begin{equation}\label{eq:orbdim-cr-def}
\operatorname{dimOrb}_G^{\rm cr} (\rho):= \min_{\{p_\lambda,\rho_\lambda\}} \max_{\rho_\lambda} \dim(\orb_G(\rho_\lambda)) \,,
\end{equation}
where the minimum is over all convex decompositions $\rho = \int d \lambda \, p_\lambda \, \rho_\lambda $ of $\rho$ into \textit{pure} states $\rho_\lambda$, and the maximum (technically an essential supremum, to discard measure-zero maxima) is over all pure states $\rho_\lambda$ in the decomposition.
Note that this redefinition for mixed states is analogous to how the symplectic rank $\aleph^{\rm cr}(\rho)$ in Ref.~\cite{Mele-SymplecticRank-2026} departs from the ordinary symplectic rank $\aleph(\rho)$ \cite{Adesso-ContinuousVariable-2014,Serafini-QuantumContinuous-2023} for mixed states by taking a convex-roof redefinition; see \cref{subsec:relations-with-stellar-or-symplectic-ranks}.

It easily follows from the definition of \cref{eq:orbdim-cr-def} that the set of states $\rho$ that have minimal $\operatorname{dimOrb}_G^{\rm cr}(\rho)$ is exactly the (continuous) convex hull of the set of pure states $\rho$ that have minimal $\dim(\orb_G(\rho))$. Consequently, %
the above \cref{cor:minimal-orbdim-purestates-characterization-main} directly yields the following:
\begin{corollary}
[Equivalence between minimal convex-roof $G_{\rm GO}$ ($G_{\rm DPLO}$)-orbit dimension and freeness in RT of genuine non-Gaussianity (nonclassicality)]\label{cor:minimal-cr-orbdim--characterization-main}
In the \emph{multimode} setting $(m\geq 2)$, the space of (physical) states $\rho$ that achieve the minimal convex-roof orbit dimension [$\operatorname{dimOrb}_G^{\rm cr}(\rho) = \dim(G)^{\scriptscriptstyle\!-} - m^2$] coincides exactly with the \emph{free states} (over $m$ modes) of the following resource theories (RTs) of quantum optics:
\begin{enumerate}[label=(\arabic*)]
\item for $G=G_{\rm GO}$: the RT of \emph{genuine non-Gaussianity};
\item for $G=G_{\rm DPLO}$: the RT of \emph{optical nonclassicality}.
\end{enumerate}
\end{corollary}
For the case of $G=G_{\rm GO}$, this is analogous to the situation of stellar rank and (convex-roof) symplectic rank, which are also equivalent to freeness in the RT of genuine non-Gaussianity \cite{Chabaud-HolomorphicRepresentation-2022,Mele-SymplecticRank-2026}.

\subsection{Orbit dimensions are generally not monotones}

Given this equivalence, one may further ask if these two convex-roof orbit dimensions are \textit{monotones} under the free operations of the matching RTs --- just as the stellar rank \cite{Chabaud-HolomorphicRepresentation-2022} and the (convex-roof) symplectic rank \cite{Mele-SymplecticRank-2026} have been proven to be in the RT of genuine non-Gaussianity.

First, we stress here that when the number of modes $m$ is allowed to change, the answer is clearly negative. This is because the channel $\rho \mapsto \rho \otimes \ketbra{\bm{0}}$ that appends extra modes in the vacuum state always increases $\dim(\orb_{G_{\rm GO}}(\rho))$ (as more modes means more generators --- cf. \cref{tab:Lie-algebra-bases} --- which in turn offers, for any $\rho$, new evolution directions in the Hilbert space), and this channel is considered free in both RTs of genuine non-Gaussianity and nonclassicality \cite{Takagi-ConvexResource-2018,Yadin-OperationalResource-2018}.

If we now restrict the study to a fixed number of modes $m$ (and only free channels that have this number of input and output modes), we may still ask whether monotonicity holds. It turns out that the answer remains negative:
\begin{theorem}[Orbit dimensions are not fixed-mode monotones]\label{thm:orbdims-are-not-fixed-mode-monotones}
For $G=G_{\rm GO}$ (resp. $G=G_{\rm DPLO}$), the convex-roof orbit dimensions $\rho \mapsto \operatorname{dimOrb}_{G}^{\rm cr}(\rho)$ [as well as the ordinary orbit dimensions $\rho \mapsto \dim(\orb_{G}(\rho))$] are not fixed-mode monotones under the free operations in the RT of genuine non-Gaussianity (resp. nonclassicality).
\end{theorem}
We prove \cref{thm:orbdims-are-not-fixed-mode-monotones} in \cref{subsec:SM-fixed-mode-nonmonotonicity} by exhibiting a simple example of a two-mode state whose orbit dimensions increase under the action of a two-mode CPTP map that is free in both RTs.

\subsection{Differences and links between orbit dimensions and stellar or symplectic ranks}\label{subsec:relations-with-stellar-or-symplectic-ranks}
As the stellar and symplectic ranks are two other properties of bosonic states that are also integer-valued invariants under $G_{\rm GO}$ unitaries (and that have undergone considerable study in the literature), it is natural to wonder about possible equivalences or links with them. We thus clarify in this section the differences and some partial relations with these concepts.

First, let us recall the definitions of stellar and symplectic ranks.
The \textit{stellar rank} of a pure state $\rho$ can be loosely thought of as the minimal number of photon additions needed to obtain $\rho$ from a Gaussian state; it is then extended to mixed states via a convex-roof construction \cite{chabaud_stellar_2020,Walschaers-NonGaussianQuantum-2021,Chabaud-ClassicalSimulation-2021,Chabaud-HolomorphicRepresentation-2022}.
The \textit{covariance matrix} of $\rho$ is the matrix $V_\rho \in \mathbb{R}^{2m \times 2m}$ defined by
$[V_\rho]_{kl} := \operatorname{Cov}_{\rho}(r_k,r_l)$, where $\bm{r} := (q_1,p_1,\dots,q_m,p_m)^\intercal$. Then, as per Williamson's theorem \cite{Adesso-ContinuousVariable-2014,Nicacio-WilliamsonTheorem-2021}, there exists a unique $m$-tuple $\bm{\sigma}:=(\sigma_1,\sigma_2,\dots,\sigma_m)$ with $\sigma_i \leq \sigma_{i+1}$ for all $i$ (called the \textit{symplectic spectrum} of $\rho$), such that for some $U \in G_{\rm ALO}$, $V_{U \rho U^\dagger}$ takes on the diagonal form $V_{U \rho U^\dagger}=\operatorname{diag}(\sigma_1,\sigma_1,\sigma_2,\sigma_2,\dots,\sigma_m,\sigma_m)$. Each distinct value $\eta$ in $\bm{\sigma}$ is called a \textit{symplectic eigenvalue} of $\rho$, and its \textit{multiplicity} is the number of times that it appears in $\bm{\sigma}$. Finally, one definition of symplectic rank, which we refer to as the \textit{ordinary symplectic rank} and denote by $\aleph(\rho)$, is the number of symplectic eigenvalues of $V_\rho$ (counted with multiplicities) that are strictly greater than $1/2$ \cite{Adesso-ContinuousVariable-2014,Serafini-QuantumContinuous-2023}. A second definition of symplectic rank is the one adopted in the more recent work of Ref.~\cite{Mele-SymplecticRank-2026}, which we refer to as the \textit{convex-roof symplectic rank} and denote by $\aleph^{\rm cr}(\rho)$: it is defined by keeping the previous definition only on pure states, and extending to mixed states via a convex-roof construction.

Let us stress that all presented notions of $G$-orbit dimensions are \textit{not equivalent} to neither the notion of stellar rank, nor those of symplectic rank. That is, for all four groups $G$ of \cref{tab:Lie-algebra-bases}, one can find pairs of states (even among pure states) that have equal $G$-orbit dimensions but different stellar rank (resp. symplectic rank), and vice versa. Such examples are easy to come by (e.g., among the first two columns of \cref{tab:orbit-dimensions}), and we provide explicit example families in \cref{subsec:SM-orbit-dimensions-vs-stellar-symplectic-ranks}.

Nevertheless, inequalities that relate these properties may still hold.
We have already seen above that orbit dimensions are lower-bounded by $\dim(G)^{\scriptscriptstyle\!-}$ minus the square sum of the symplectic multiplicities, i.e., the first inequality in \cref{eq:minimal-centered-orbdim-general-lowerbounds-main}.
As the ordinary symplectic rank is also related to these multiplicities, this can easily provide us an inequality involving this ordinary symplectic rank as well (see \cref{subsec:SM-minimal-orbit-dimensions} for the proof):
\begin{corollary}[Lower bound on orbit dimensions in terms of ordinary symplectic rank]\label{cor:orbdim-lowerbound-symplectic-rank-main}
Let $\rho$ be a physical state, and consider one of the four $m$-mode quantum optical unitary groups $G$ of \cref{tab:Lie-algebra-bases}.
Then,
\begin{equation}\label{eq:orbdim-lowerbound-symplectic-rank-claim-main}
\dim(\orb_{G}(\rho)) \geq \dim(G)^{\scriptscriptstyle\!-} - m^2 + 2 \, \aleph(\rho) \, [m - \aleph(\rho)]\,.
\end{equation}
\end{corollary}

We do not investigate general inequalities relating orbit dimensions to the stellar rank, or to $\aleph^{\rm cr}(\rho)$.

\section{Discussion and Outlook}\label{sec:disc-and-outlook}

It is quite remarkable that orbit dimensions of Fock basis states $\ket{\bm{n}}$
(see first column of \cref{tab:orbit-dimensions}) do not depend on the particular values of the occupation numbers $n_i$, but only on whether they are zero or nonzero.
This shows that while bosons can (unlike fermions) bunch into common modes, doing so does not increase their number of explorable directions in Hilbert space.
An illustrative example is given by passive linear optics with $n=m$ photons in $m$ modes, where there is a quadratic difference in the scaling of $\dim(\orb_{G_{\rm PLO}}(\ketbra{\bm{n}}))$ between the separate state $\bm{n}=(1,1,\dots,1)$ [$m(m-1) = \mathcal{O}(m^2)$] and the bunched state $\bm{n}=(m,0,\dots,0)$ [$2m-2 = \mathcal{O}(m)$]. This shows (via \cref{thm:VQC-orbit-dim-maintextinformal}) a quadratic difference in dimensional expressivity \cite{Funcke-DimensionalExpressivity-2021a,Haug-CapacityQuantum-2021} between variational quantum circuits (VQCs) with these two input states, which we illustrate in \cref{fig:vqc-input-state-rank-scaling-dense}. Besides, one can notice that this higher orbit dimension state $\bm{n}=(1,1,\dots,1)$ is precisely the one that would provide computational hardness when used as the input state of boson sampling, while computational hardness is lost if the bunched state is used instead \cite{Aaronson-ComputationalComplexity-2013}.
Whether orbit dimensions can be used to better understand which regimes of boson sampling have computational advantages is an interesting direction for future work.

\begin{figure}[t]
\centering
\includegraphics[trim=0 7pt 0 0]{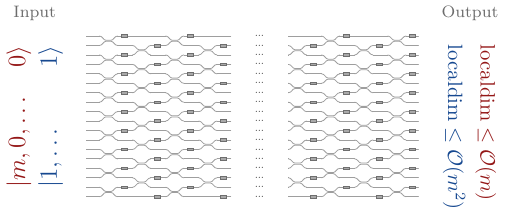}

\caption{Illustration of the quadratic difference in dimensional expressivity [which is equivalently the local dimension or the quantum Fisher information matrix (QFIM) rank of the parametrized output state $\ket{\psi_{\rm out}(\bm{\theta})}$, see \cref{thm:VQC-orbit-dim-maintextinformal}], due to differences in orbit dimensions [cf. \cref{eq:fock-state-ket-orbit-dim-PLO}], that a variational quantum circuit made of parametrized beam splitters (bent optical paths) and phase shifters (gray boxes) can reach over $m$ modes (irrespective of the circuit's architecture or parameter count), when using a \textit{bunched} (red) versus \textit{separate} (blue) input Fock state of $n=m$ photons.}
\label{fig:vqc-input-state-rank-scaling-dense}
\end{figure}

Taking superpositions over a constant number of modes, perhaps surprisingly, only appears to give a limited increase to orbit dimensions.
For instance, orbit dimensions for the NOON state family do not show a dependency in $N$ for $N\geq3$ (see last column of \cref{tab:orbit-dimensions}). There is a slight increase throughout the $N=1,2,3$ range however, since for $N=1,2$ the states lie in the orbits of the corresponding Fock basis states with $\ket{1,0}$ and $\ket{1,1}$ in the first two modes (for $N=2$ this is the Hong-Ou-Mandel effect \cite{Hong-MeasurementSubpicosecond-1987}), which gives an increase in orbit dimensions of $2u$ from $N=1$ to $N=2$ and a further increase of $1$ from $N=2$ to $N=3$ (see \cref{tab:orbit-dimensions}); but this dependency is not too surprising, considering the special structures of single-photon transformations \cite{Moyano-Fernandez-LinearOptics-2017} and of $N=2$ photon interference \cite{Tichy-InterferenceIdentical-2014, DeGliniasty-SimpleRules-2024}.
Likewise, taking finite superpositions of $N$ Fock basis states that differ just in one mode, only increases the orbit dimensions by $1$ (compared with the single Fock basis states), independently of $N\geq 2$ (see middle column of \cref{tab:orbit-dimensions}).

Orbit dimensions in the single-mode setting have a somewhat degenerate structure:
the vacuum as well as higher Fock basis states all have the same orbit dimension in the ketbra picture (e.g., $4$ for $G_{\rm GO}$), although these orbits are disjoint (due to having different stellar ranks \cite{chabaud_stellar_2020}). %
In the multimode setting however, the vacuum has different orbit dimensions than all higher Fock basis states, and $\ket{\bm{n}}$,$\ket{\bm{n'}}$ with different numbers of unoccupied modes $u,u'$ have different orbit dimensions, except in the case $u, u' \in \{0,1\}$.

Since the upper bounds of \cref{thm:VQC-orbit-dim-maintextinformal} provide an input-state-dependent limitation on the dimensional expressivity \cite{Funcke-DimensionalExpressivity-2021a,Haug-CapacityQuantum-2021} of bosonic VQCs (via \cref{thm:VQC-orbit-dim-maintextinformal}), such as those used for quantum machine-learning (QML) applications, calculating the orbit dimension of a circuit's input state therefore provides the minimal number of parameters $p$ that the circuit should contain to fall in the so-called \textit{overparametrized} regime of QML, which is linked to loss landscapes with fewer local minima \cite{Haug-CapacityQuantum-2021,Larocca-TheoryOverparametrization-2023,Anschuetz-QuantumVariational-2022}.
Stated as bounds on the rank of the quantum Fisher information matrix (QFIM) of a parametrized quantum state, the upper bounds of \cref{thm:VQC-orbit-dim-maintextinformal} are also informative in the context of bosonic quantum metrology, where this QFIM rank indicates locally the number of parameter combinations that can be estimated simultaneously \cite{Liu-QuantumFisher-2019,Mihailescu-MetrologicalSymmetries-2025}. As an example, in distributed quantum sensing, the notion of perfect privacy (the parties can only learn a global function of the parameters but not the parameters themselves) implies that the QFIM of the shared state must have rank $1$ \cite{Shettell-PrivateNetwork-2022,Bugalho-PrivateRobust-2025,Hassani-PrivacyNetworks-2025}; hence studying which global states have an orbit dimension (under the relevant group $G$ of local unitaries) of $1$ provides a way to infer which states can achieve such privacy.

The schemes outlined in \cref{subsec:experimental-schemes} to measure experimentally the dimension of a state's orbit, despite their sensitivity to noise, provide an enticing prospect of directly probing a global geometrical feature of the accessible state manifold. Further studying which physical assumptions on state classes ensure a spectral gap scaling that makes these schemes efficient for estimating these orbit dimensions, or merely bounding them, is an interesting further direction of research.

Under the lens of the resource theories (RTs) of quantum optics that are genuine non-Gaussianity (for $G_{\rm GO}$) and nonclassicality (for $G_{\rm DPLO}$), we found that for multimode pure states, free states are characterized as the unique orbit of minimal dimension (\cref{cor:minimal-orbdim-purestates-characterization-main}). This is, to the best of our knowledge, a novel insight in these RTs, and may be helpful in furthering the general understanding of how these RTs structure the pure state space.
While convex-roof orbit dimensions then propagate this insight over to mixed states (\cref{cor:minimal-cr-orbdim--characterization-main}), the fact that they fail to be even fixed-mode monotones (\cref{thm:orbdims-are-not-fixed-mode-monotones}) may appear as a limitation of the framework of orbit dimensions as a whole. However, that orbit dimensions can increase under these "free" fixed-mode operations is not unnatural, as in the realm of mixed states, some "noise" operations may be regarded as free (such as the one in the proof of \cref{thm:orbdims-are-not-fixed-mode-monotones}), but it is intuitive and well-established that noise can generally increase the (dimensional) expressivity \cite{Garcia-Martin-EffectsNoise-2024,Wu-RandomnessEnhancedExpressivity-2024,Monbroussou-QuantumAdvantage-2025}.
To sum it up, when working over a fixed number of modes, orbit dimensions provide both a quantifier of a state's expressivity (under a quantum optical unitary group), and a way to quantify its amount of quantum resource (genuine non-Gaussianity or nonclassicality), although for the latter case they fail to faithfully capture even the fixed-mode order relations in these RTs.

\section{Conclusion}\label{sec:CCL}
In this work, we have put forward the notion of orbit dimensions as a central concept in the study of bosonic systems undergoing quadratic Hamiltonian evolutions, such as linear or Gaussian quantum optics.
Despite its simplicity, we have demonstrated that this framework
provides insights into the structure of reachable states starting from a given state, and notably of the set of non-Gaussian and $P$-nonclassical states. Having established theoretically its unified validity in all usual discrete or continuous-variable treatments, provided example calculations, and discussed amenability to experiments, we have cast orbit dimensions as a general investigation toolkit for the study of dimensional expressivity of bosonic variational quantum circuits, and of state conversion.
The validity of the presented framework hinges on the fact that it is limited to dealing with Hamiltonian evolution occurring on a fixed number of modes and excluding any measurement. The no-go results that the framework can yield were only obtained for deterministic and exact state conversion. Exploring possible extensions to more general settings is an interesting direction of work, but goes beyond the scope of this paper.

\paragraph{Acknowledgments} E.Z.M. is supported by the Grant ANR-22-PNCQ-0002. E.Z.M. thanks Léo Monbroussou for many discussions on earlier projects which motivated this work, Elham Kashefi and Mario Sigalotti for insightful feedback, and Frédéric Grosshans, Ulysse Chabaud, Mattia Walschaers, Pérola Milman, Gerard Milburn, Bettina Kazandjian, Verena Yacoub, and Hugo Thomas for fruitful discussions.
The author used ChatGPT and Codex (GPT-5.3 to GPT-5.6-Sol) to assist with the exploration of proof directions for Theorem 3, the search and analysis of examples and counterexamples appearing in the proofs of Proposition 1 and Theorem 4, and the example computations in \cref{subsec:SM-stellar-representation-examples}. The author takes full responsibility for the content and correctness of this work.

\paragraph{Data availability} There are no publicly available research data supporting this paper. A symbolic implementation of the formulas of \cref{thm:orbit-dim-maintext} is openly available \cite{Bosonic-Orbit-Dimensions-Github-cff}. Requests for further information or data should be sent to the author.

\let\oldaddcontentsline\addcontentsline
\renewcommand{\addcontentsline}[3]{}

\let\addcontentsline\oldaddcontentsline

\clearpage

\appendix

\counterwithin{theorem}{section}
\counterwithin{definition}{section}
\counterwithin{lemma}{section}
\counterwithin{corollary}{section}
\counterwithin{fact}{section}
\counterwithin{remark}{section}
\counterwithin{example}{section}
\counterwithin{proposition}{section}
\counterwithin{conjecture}{section}

\renewcommand{\thetheorem}{\thesection\arabic{theorem}}
\renewcommand{\thedefinition}{\thesection\arabic{definition}}
\renewcommand{\thelemma}{\thesection\arabic{lemma}}
\renewcommand{\thecorollary}{\thesection\arabic{corollary}}
\renewcommand{\thefact}{\thesection\arabic{fact}}
\renewcommand{\theremark}{\thesection\arabic{remark}}
\renewcommand{\theexample}{\thesection\arabic{example}}
\renewcommand{\theproposition}{\thesection\arabic{proposition}}
\renewcommand{\theconjecture}{\thesection\arabic{conjecture}}

\tableofcontents

\vspace{1cm}

\section{Example of impossibility of CNOT gate in dual-rail encoding}\label{sec:SM-CNOT-dual-rail}
Dual-rail encoding consists of encoding a qubit state as a state in the two-dimensional logical subspace $\spa\{\ket{0}_L:=\ket{1,0},\ket{1}_L:=\ket{0,1}\}$ of Fock space.
Then, two-qubit states correspond to the four-dimensional logical subspace $\spa\{\ket{00}_L,\ket{01}_L,\ket{10}_L,\ket{11}_L\}$, where $\ket{00}_L = \ket{1,0,1,0}$, $\ket{01}_L = \ket{1,0,0,1}$, $\ket{10}_L = \ket{0,1,1,0}$, and $\ket{11}_L = \ket{0,1,0,1}$ are specific $m=4,n=2$ Fock basis states.

The logical CNOT gate on this four-dimensional subspace is the unitary $\mathrm{CNOT}_L$ defined to act on the basis states as
\begin{equation}
\begin{aligned}
\ket{00}_L \mapsto \ket{00}_L\,,\\
\ket{01}_L \mapsto \ket{01}_L\,,\\
\ket{10}_L \mapsto \ket{11}_L\,,\\
\ket{11}_L \mapsto \ket{10}_L\,.
\end{aligned}
\end{equation}
In particular, it must satisfy by linearity:
\begin{align}
\mathrm{CNOT}_L \ket{+0}_L = \ket{\Phi^+}_L\,,
\end{align}
where $\ket{+0}_L = \frac{1}{\sqrt{2}}(\ket{00}_L+\ket{10}_L)$ and $\ket{\Phi^+}_L = \frac{1}{\sqrt{2}}(\ket{00}_L+\ket{11}_L)$.

However, as can be verified using symbolic computation \cite{Bosonic-Orbit-Dimensions-Github-cff}, the orbit dimensions of $\ket{+0}_L$ and $\ket{\Phi^+}_L$ under $G_{\rm GO}$ are equal to different values, 39 and 38, respectively (or 38 and 37 in the ketbra picture). This illustrates that the concept of orbit dimensions alone implies the fact that the unitary gate $\mathrm{CNOT}_L$ cannot belong to the four-mode Gaussian unitary group $G_{\rm GO}$ (and hence not to any of its subgroups, such as $G_{\rm PLO}$, either).

For illustrative purposes, we also report the smallest positive eigenvalues $\lambda_{\mathrm{min}}^+$ of the Gram matrices associated with these four-mode states $\ket{+0}_L$ and $\ket{\Phi^+}_L$: they are $1/4$ and $1/2$ (or $1/2$ and $1$ in the ketbra picture).

\section{Background on general unitary group representations}

Let us recall notions of group representations, and of irreducible subspaces (see, e.g., Ref.~\cite[Chap. 3]{Folland-CourseAbstract-2016}).
Given a group $G$, a \textit{(unitary) group representation} of $G$ onto a (possibly infinite-dimensional) complex Hilbert space $\mathcal{H}$, is a map $\phi:G \to \uni(\mathcal{H})$ from $G$ to the group of unitary operators on $\mathcal{H}$ that is a group homomorphism, i.e., fulfills $\phi(g_1 g_2) = \phi(g_1) \phi(g_2)$ for all $g_1,g_2 \in G$. A closed subspace $\mathcal{V} \subseteq \mathcal{H}$ is said to be an \textit{invariant subspace} for the representation if $\phi(g)(\mathcal{V}) \subseteq \mathcal{V}$ for all $g \in G$.
If $\mathcal{V}$ is an invariant subspace, then the map $G \to \uni(\mathcal{V})$ given by $g \mapsto \left.\phi(g)\right|_{\mathcal{V}}$ still defines a representation of $G$, called the \textit{subrepresentation} of $\phi$ on $\mathcal{V}$.
The representation $\phi$ is said to be \textit{irreducible} if its only closed invariant subspaces are $\{0\}$ and $\mathcal{H}$.
Two representations $\phi_1:G \to \uni(\mathcal{H}_1)$ and $\phi_2:G \to \uni(\mathcal{H}_2)$ of $G$ are said to be \textit{equivalent} if there exists a Hilbert space isomorphism $U:\mathcal{H}_1 \to \mathcal{H}_2$ such that $\phi_2(g) = U \circ\phi_1(g)\circ U^{-1}$ for all $g \in G$.

Lastly, a representation $\phi:G \to \uni(\mathcal{H})$ is said to be \emph{completely reducible without multiplicities} when there exists a family of pairwise-orthogonal closed subspaces $(\mathcal{H}_\lambda)_{\lambda}$ of $\mathcal{H}$ satisfying
\begin{equation}\label{eq:orbit-closed-span-hilbert-space-decomposition}
\mathcal{H} = \bigoplus_{\lambda} \mathcal{H}_\lambda\,,
\end{equation}
such that, for every $\lambda$, $\dim(\mathcal{H}_\lambda) \geq 1$, $\mathcal{H}_\lambda$ is an invariant subspace for $\phi$ and the associated subrepresentation of $\phi$ on $\mathcal{H}_\lambda$ is irreducible,
and such that for every $\lambda\neq\lambda'$, the subrepresentations $\phi_\lambda : G \to \uni(\mathcal{H}_\lambda)$ and $\phi_{\lambda'} : G \to \uni(\mathcal{H}_{\lambda'})$ are not equivalent representations of $G$.

Given a point $p \in \mathcal{H}$, we denote by $\orb_G(p) := \{\phi(g) \cdot p \,|\, g \in G \}$ the \textit{orbit} of $p$ under the group representation $\phi$; the representation $\phi$ of $G$ should be clear from the context when using the notation $\orb_G(p)$.

We need the following lemma, which is a direct consequence of the above definitions (and of Schur's lemma):
\begin{lemma}\label{lem:orbit-closed-span}
Let $\phi:G \to \uni(\mathcal{H})$ be a unitary group representation of a group $G$ on a Hilbert space $\mathcal{H}$.
Suppose that $\phi$ is \emph{completely reducible without multiplicities}, via a decomposition of the Hilbert space $\mathcal{H}$ of the form of \cref{eq:orbit-closed-span-hilbert-space-decomposition}.

Then, for every $p \in \mathcal{H}$, the closed linear span of the orbit of $p$ is given by
\begin{equation}\label{eq:orbit-closed-span-mainclaim}
\overline{\spa_{\mathbb{C}}(\orb_G(p))} \ = \bigoplus_{\substack{\lambda;\\ P_\lambda(p) \,\neq\, 0}} \mathcal{H}_\lambda\,,
\end{equation}
where $P_\lambda$ denotes the orthogonal projection onto $\mathcal{H}_\lambda$.
\begin{proof}
Fix $p \in \mathcal{H}$, and denote $\mathcal{H}_p := \overline{\spa_{\mathbb{C}}(\orb_G(p))}$.
First, the closed subspace $\mathcal{H}_p$ is an invariant subspace under $\phi$.
Indeed, let $q \in \mathcal{H}_p$. By definition, there exists a sequence $(q_n)$ in $\spa_{\mathbb{C}}(\orb_G(p))$ such that $q_n \underset{n\to\infty}{\longrightarrow} q$, with each $q_n$ written as a finite linear combination of the form
$q_n = \sum_{i=1}^{k^{(n)}} c_i^{(n)} \phi(g_i^{(n)}) \cdot p$, with $k_n \in \mathbb{N}$, $c_i^{(n)} \in \mathbb{C}$, $g_i^{(n)} \in G$.
For all $n$ and all $g \in G$, we have
\begin{align}
\phi(g) \cdot q_n
&= \sum_{i=1}^{k^{(n)}} c_i^{(n)} \phi(g) \phi(g_i^{(n)}) \cdot p\\
&= \sum_{i=1}^{k^{(n)}} c_i^{(n)} \phi(g\, g_i^{(n)}) \cdot p\,,
\end{align}
using the linearity of $\phi(g)$ and the fact that $\phi$ is a group representation, which shows that $\phi(g) \cdot q_n \in \spa_{\mathbb{C}}(\orb_G(p))$.
Because $\phi(g) \cdot q_n \underset{n\to\infty}{\longrightarrow} \phi(g) \cdot q$ [since $\phi(g)$ is continuous, being a unitary operator], it means that $\phi(g) \cdot q \in \overline{\spa_{\mathbb{C}}(\orb_G(p))} = \mathcal{H}_p$ (for all $g \in G$), which establishes that $\mathcal{H}_p$ is an invariant subspace.

Now, let $P_p$ denote the orthogonal projection onto $\mathcal{H}_p$.
Since $\mathcal{H}_p$ is an invariant subspace for $\phi$ and since $\phi(g)$ is unitary, we have for all $g$:
\begin{equation}\label{eq:orbit-closed-span-commutation-relation-with-Pp}
P_p \circ \phi(g) = \phi(g) \circ P_p\,,
\end{equation}
For all $\lambda,\mu$, let us define the maps $T_{\mu\lambda} := \left.P_\mu P_p\right|_{\mathcal{H}_\lambda}$. Applying $P_\mu$ to both sides of \cref{eq:orbit-closed-span-commutation-relation-with-Pp} and restricting to $\mathcal{H}_\lambda$ gives $T_{\mu\lambda}\,\phi_\lambda(g) = \phi_\mu(g)\,T_{\mu\lambda}$. Schur's lemma (see, e.g., Ref.~\cite[Theorem 3.5]{Folland-CourseAbstract-2016}) thus applies and gives that $T_{\mu\lambda} = 0$ for $\mu \neq \lambda$ and $T_{\lambda\lambda} = c_\lambda \id_{\mathcal{H}_\lambda}$, for some $c_\lambda \in \mathbb{C}$. It hence follows that for all $\lambda$, $\left.P_p\right|_{\mathcal{H}_\lambda} = c_\lambda \id_{\mathcal{H}_\lambda}$.
Therefore,
\begin{equation}\label{eq:orbit-closed-span-expression-of-Pp}
P_p = \bigoplus_\lambda c_\lambda P_\lambda\,.
\end{equation}
The fact that $P_p$ is a projection enforces $c_\lambda \in \{0,1\}$, and we in fact claim that $c_\lambda = 1$ if and only if $P_\lambda(p) \neq 0$, otherwise $c_\lambda = 0$.
Indeed, since we have $P_\lambda(p) = P_\lambda P_p(p) = c_\lambda P_\lambda(p)$ [using respectively that $p \in \mathcal{H}_p$ and \cref{eq:orbit-closed-span-expression-of-Pp}], $P_\lambda(p) \neq 0$ implies that $c_\lambda = 1$.
Furthermore, since by the invariance of the subspace $\mathcal{H}_\lambda$ we also have
\begin{equation}\label{eq:orbit-closed-span-commutation-relation-with-Plambda}
P_\lambda \circ \phi(g) = \phi(g) \circ P_\lambda\,,
\end{equation}
assuming that $P_\lambda(p) = 0$ and applying \cref{eq:orbit-closed-span-commutation-relation-with-Plambda} to $p$ gives that for all $g \in G$, $P_\lambda (\phi(g) \cdot p) = 0$,
i.e., $\orb_G(p) \subseteq \mathcal{H}_\lambda^\perp$. This implies (by taking linear span and closure, respectively)
\begin{equation}\label{eq:orbit-closed-span-cite-inclusion-perp}
\mathcal{H}_p \subseteq \mathcal{H}_\lambda^\perp\,.
\end{equation}
This implies that $\left.P_p\right|_{\mathcal{H}_\lambda} = 0$, i.e., $c_\lambda = 0$, thus establishing the claim.

Therefore, \cref{eq:orbit-closed-span-expression-of-Pp} simplifies to
\begin{equation}
P_p \ = \bigoplus_{\substack{\lambda;\\ P_\lambda(p) \,\neq\, 0}} P_\lambda\,,
\end{equation}
and hence taking the image of both sides establishes \cref{eq:orbit-closed-span-mainclaim}.
\end{proof}
\end{lemma}

\section{Background on manifolds}

We roughly introduce some necessary notions of finite-dimensional smooth manifolds, but refer the reader to \cite{lee_introduction_2012} for the precise definitions.

In short, a finite-dimensional \textit{smooth manifold} $X$ is a space which "locally looks like" $\mathbb{R}^r$ for some fixed integer $r$.
The \textit{dimension} $r$ of a manifold may be thought of as the number of independent real numbers needed to locally describe a patch of the manifold. For example, the unit sphere $S^2 \subset \mathbb{R}^3$ is a smooth embedded submanifold of $ \mathbb{R}^3$, and its dimension is 2.
The smooth structures make it possible to define the notion of an infinitesimal first-order approximation of the space $X$ around a point $x \in X$, called the \textit{tangent space} $T_x X$ of $X$ at point $x$, and this tangent space $T_x X$ has the structure of a real vector space of dimension $r$ (see \cref{fig:orbit-tangent-space-diagram} for a visualization of a tangent space). They also make it possible to define a notion of smooth maps $f$ between two smooth manifolds $X$ and $Y$, and therefore of their first-order derivative $D f(x): T_x X \to T_{f(x)}Y$, as a \textit{linear} map between respective tangent spaces (which can be thought of as a coordinate-free generalization of the concept of a Jacobian matrix for a differentiable map $f:\mathbb{R}^r \to \mathbb{R}^s$).

We also briefly recall concepts related to constant-rank maps.
Let $f: X \to Y$ be a smooth map between two smooth manifolds $X$ and $Y$. $f$ is said to be a \textit{constant-rank $q$} map if the linear maps $D f(x)$ are all of the same rank $q$ (for all $x \in X$).
A map $f$ of constant rank $q$ is called a \textit{smooth immersion} if $f$ is injective and $q=\dim(X)$ [i.e., all the linear maps $D f(x)$ are injective], a \textit{smooth submersion} if $f$ is surjective and $q=\dim(Y)$ [i.e., all the linear maps $D f(x)$ are surjective], and a \textit{diffeomorphism} if it is both a smooth immersion and a smooth submersion.

\textit{Analytic} manifolds are analogous to smooth manifolds, but allow the concept of an \textit{analytic map} $f: X \to Y$ to be well-defined, i.e., it is defined in local coordinates as the property that $f$ agrees with its Taylor series locally around every point of $X$.

\textit{Lie groups} are manifolds which are at the same time groups, and whose group operations are smooth (and in fact analytic). The tangent space at the identity element of a Lie group $G$ is called the \textit{Lie algebra} of $G$, denoted $\mathfrak{g}$, and it has the structure of a real vector space, as well as a Lie bracket operation $[.,.]:\mathfrak{g} \times \mathfrak{g} \to \mathfrak{g}$. Since it is a tangent space of the Lie group (as a manifold), it holds that $\dim(\mathfrak{g}) = \dim(G)$.

\section{Background on extended metaplectic representation and proof of Theorem \ref{thm:orbit-dim-maintext}}\label{sec:SM-EMRep-and-proof-of-orbit-structure}

Throughout this section (and this section only), the $m$-mode quantum optical unitary groups "$G$"$\subset \uni(\fock)$ such as those studied in the main text (\cref{tab:Lie-algebra-bases}) will be renamed correspondingly $G_{\fock}$ (e.g., the PLO unitary group on $\fock$ is now denoted $G_{\fock,\rm PLO}$).

\subsection{Unitary Gaussian quantum optics through a group representation}\label{subsec:background-on-optics-through-a-repr}

These unitary groups from the main text are all subgroups of the Gaussian unitary group, $G_{\rm GO}$, the group generated by quadratic Hamiltonians (degree $\leq 2$) in the canonical operators.
In Gaussian quantum optics, it is common to describe these unitaries in terms of their action on the $m$-mode phase space. That is, elements of the group $G_{\fock,\rm GO}$ are in practice specified by a pair $(\bm{d},S)$ --- where $\bm{d} \in \mathbb{R}^{2m}$ is a displacement vector and $S \in \mathrm{Sp}(2m,\mathbb{R})$ is a real symplectic matrix --- via a map $f: \mathbb{R}^{2m} \times \mathrm{Sp}(2m,\mathbb{R}) \to G_{\fock,\rm GO} \subseteq \uni(\fock)$ \cite{weedbrook_gaussian_2012}.
From a theoretical point of view, (i) it is technically desirable that the map $f$ be a group homomorphism (as we see later), (ii) furthermore global phases should be treated correctly (as they are part of $G_{\fock,\rm GO}$), and (iii) there should be a unique principled way to choose such a mapping $f$; however, these three requirements fail to hold for the above assignment $f:(\bm{d},S) \mapsto U$ due to some remaining freedom of global phase to assign to each unitary $U$ \cite{arvind_real_1995}. The mathematical tool that resolves these requirements is the so-called \textit{extended metaplectic representation}. %
For an introduction to the metaplectic representation in the context of quantum optics, we refer the reader to Ref.~\cite{arvind_real_1995}. A more mathematical treatment of the extended metaplectic representation may be found in Ref.~\cite[Chap. 4]{Folland-HarmonicAnalysis-1989}.

What this provides us is a finite-dimensional Lie group $G_0$, and a unitary representation $\phi:G_0 \to \uni(\fock)$ of $G_0$, such that the image $\phi(G_0)$ is exactly the Gaussian unitary group $G_{\fock,\rm GO}$ [and with $\dim(G_0)=\dim(G_{\fock,\rm GO}) = 2m^2 + 3m + 1$].
The three other unitary groups considered in the main text ($G_{\fock,\rm PLO}$, $G_{\fock,\rm DPLO}$, and $G_{\fock,\rm ALO}$) may then also be obtained as the image of three respective subgroups of $G_0$ under this representation $\phi$.

Moreover, this representation $\phi$ is well-behaved enough to allow one to study its "infinitesimal" evolutions, as we formalize next.

\subsection{General properties of strongly continuous representations, and structure of orbits}
We start by outlining the concepts of \textit{smoothness} and \textit{analyticity} associated with unitary representations, as they will both be needed later. For further details, we generally refer the reader to Ref.~\cite[Sec. 12.2.11]{Moretti-SpectralTheory-2017}.
Here, let $\phi:G\to \uni(\mathcal{H})$ be a strongly continuous unitary representation of a finite-dimensional Lie group $G$ onto a (possibly infinite-dimensional) Hilbert space $\mathcal{H}$. For a fixed $p \in \mathcal{H}$, denote by $f_p:G\to \mathcal{H}$ the \textit{orbit map} (for $p$), defined as $f_p(g):=\phi(g) \cdot p$.
The above requirement of $\phi$ being \textit{strongly continuous} is defined to mean that the orbit maps $f_p$ are all continuous (for all $p \in \mathcal{H}$).

We recall that a Lie group $G$ comes with not only a smooth manifold structure but in fact an analytic manifold structure (cf., e.g., Ref.~\cite[Chap. 2]{Varadarajan-LieGroups-1984}), which enables the notions of \textit{smooth maps} and \textit{analytic maps} between $G$ and other analytic manifolds to be well-defined.
The spaces $\mathcal{H}^\infty$ (resp. $\mathcal{H}^\omega$) of \textit{smooth vectors} (resp. \textit{analytic vectors}) of the representation $\phi$ are defined as the set of vectors $p \in \mathcal{H}$ for which the orbit map $f_p:G \to \mathcal{H}$ is smooth (resp. analytic). Note that $\mathcal{H}^\omega \subseteq \mathcal{H}^\infty$.

It is always possible to induce from the representation $\phi:G \to \uni(\mathcal{H})$ of the group $G$, a \textit{derived representation} of the Lie algebra $\mathfrak{g}$ of $G$, which is defined as
\begin{equation}\label{eq:SC-repr-def-of-its-derivative}
\phi'(X) \cdot p := \evalat[\Big]{\frac{d}{dt}}{t=0} \Big( \phi(e^{tX}) \cdot p \Big)
\end{equation}
for all $X \in \mathfrak{g}$ and for at least all $p \in \mathcal{H}^\infty$.
To be explicit, what is meant here is that for all such $X$ and $p$, as $t \to 0$ the function $t\mapsto [\phi(e^{tX}) \cdot p - p]/t$ converges (with respect to the norm of $\mathcal{H}$) to some point in $\mathcal{H}$, which is denoted $\phi'(X) \cdot p$.
Note that \cref{eq:SC-repr-def-of-its-derivative} also takes an exponentiated form,
\begin{equation}\label{eq:SC-repr-def-of-its-derivative-exponentiated-form}
e^{t \phi'(X)} = \phi(e^{tX})\,,
\end{equation}
valid for all $t \in \mathbb{R}$, $X \in \mathfrak{g}$ \cite[p. 743]{Moretti-SpectralTheory-2017}.

We now collect general properties of these objects that will be useful to us. First, $\mathcal{H}^\infty$ and $\mathcal{H}^\omega$ are vector subspaces of $\mathcal{H}$, and they are both dense in $\mathcal{H}$ \cite[Prop. 12.85]{Moretti-SpectralTheory-2017}. The maps $\phi'(X)$ ($X \in \mathfrak{g}$) are linear, all preserve the subspace $\mathcal{H}^\infty$ \cite[Theorems 12.79.c and 12.81]{Moretti-SpectralTheory-2017}, and with such restriction to $\mathcal{H}^\infty$ the derived representation thereby gives an $\mathbb{R}$-linear map $\phi':\mathfrak{g} \to \End(\mathcal{H}^\infty)$ \cite[Theorems 12.79.d and 12.81]{Moretti-SpectralTheory-2017}. Next, for $X \in \mathfrak{g}$, the operator $\phi'(X)$ is in general well-defined on an even bigger domain than $\mathcal{H}^\infty$ (but still not all of $\mathcal{H}$); i.e., the right-hand side of \cref{eq:SC-repr-def-of-its-derivative} is meaningful for some other points $p$ not in $\mathcal{H}^\infty$. In practice, given an operator $\phi'(X)$, and a point $p \in \mathcal{H}$, we may know that the operator $\phi'(X)$ is well-defined at $p$, but not know yet whether $p$ is a smooth vector of the representation $\phi$ ($p \in \mathcal{H}^\infty$) or not. The following sufficient condition helps to assess this: if $[\phi'(X)]^k$ is well-defined at $p$ for all $X$ in a basis of $\mathfrak{g}$ and for all $k \in \mathbb{N}$, then $p \in \mathcal{H}^\infty$ \cite[Theorem 1.1]{Goodman-AnalyticEntire-1969}.
Similarly, for a $p \in \mathcal{H}$, it holds that $p \in \mathcal{H}^\omega$ if and only if for all $c_1,\dots,c_d \in \mathbb{R}$, the point $p$ is an \textit{analytic vector for the operator} $A:= -i \sum_{j=1}^d c_j \phi'(X_j)$ (with $\{ X_1,\dots,X_d \}$ a basis of $\mathfrak{g}$) \cite[Prop. 12.84]{Moretti-SpectralTheory-2017}. In this criterion, the notion of "$p$ is an analytic vector for the operator $A$" is defined as $A$ along with all its powers $A^n$ ($n \in \mathbb{N}$) are well-defined at $p$, and satisfy $\sum_{n=0}^\infty \frac{t^n}{n!} \norm{A^n p} < +\infty$ for some $t>0$.

We are now able to state the last general property, \cref{thm:SC-repr-orbit-structure}, which concerns the manifold structure of orbits and their dimension:

\begin{theorem}\label{thm:SC-repr-orbit-structure}
Let $G$ be a finite-dimensional Lie group, let $\phi:G\to \uni(\mathcal{H})$ be a strongly continuous unitary representation of $G$ onto a (possibly infinite-dimensional) Hilbert space $\mathcal{H}$, and consider the associated representation $\phi':\mathfrak{g} \to \End(\mathcal{H}^\infty)$ of the Lie algebra $\mathfrak{g}$ of $G$.
Then, for every $p \in \mathcal{H}^\infty$: the orbit $\orb_G(p) := \{ \phi(g) \cdot p \,|\, g \in G \}$ inherits the structure of a smooth manifold, whose dimension can be computed as
\begin{equation}\label{eq:SC-repr-orbit-dim-formula-claim}
\dim(\orb_G(p)) = \rank_{\mathbb{R}}(\{ \phi'(X_1) \cdot p,\dots, \phi'(X_k) \cdot p\})\,
\end{equation}
with $\{ X_1,\dots,X_k \}$ any basis of $\mathfrak{g}$,
and this smooth manifold structure on the orbit is (uniquely) characterized by the fact that it makes the orbit map $f_p:G\to \orb_G(p)$ be a smooth submersion.

\begin{proof}
For a fixed $p \in \mathcal{H}^\infty$, define the stabilizer subgroup $G_p$ of $G$ as $G_p := \{ g \in G \,|\, \phi(g) \cdot p = p \}$, and define the associated quotient space $G/G_p := \{ [g] \,|\, g \in G \}$, where $[g] := \{ g \, s \,|\, s \in G_p \}$.
Since $G_p$ is a closed subgroup of $G$ (this follows from the strong-continuity of $\phi$), a version of the closed subgroup theorem \cite[Theorem 21.17]{lee_introduction_2012} provides us with the following: (i) $G_p$ is a Lie group in its own right, whose Lie algebra can be identified as the space $\mathfrak{g}_p := \{ X \in \mathfrak{g} \,|\, \forall t \in \mathbb{R}\ e^{tX} \in G_p \}$,
(ii) the quotient space $G/G_p$ is a finite-dimensional topological manifold of dimension $\dim(G/G_p)=\dim(G) - \dim(G_p)$, and admits a unique smooth manifold structure such that the quotient map $\pi:G\to G/G_p$ defined by $\pi(g):=[g]$ is a smooth submersion.

Next, define the map $F_p:G/G_p \to \orb_G(p)$ by $F_p([g]) := \phi(g) \cdot p$. Note that the expression $\phi(g) \cdot p$ only depends on the equivalence class $[g]$ of $g$ (since if $g' = g s$ with $s \in G_p$, then $\phi(g') \cdot p = \phi(g) \phi(s) \cdot p = \phi(g) \cdot p$), which makes the map $F_p$ well-defined.
One may elementarily check that the map $F_p$ is bijective. Consequently, the topology and smooth structure of $G/G_p$ can be transported onto $\orb_G(p)$ through this bijection $F_p$, thus endowing $\orb_G(p)$ with a smooth manifold structure that makes $F_p$ a diffeomorphism. In other words, via the map $F_p$, the orbit $\orb_G(p)$ can be simply identified with the smooth manifold $G/G_p$.
Lastly, the map $f_p = F_p \circ \pi$ is a smooth submersion, being a composition of a smooth submersion ($\pi$) with a diffeomorphism ($F_p$).

Let now $f_p':\mathfrak{g} \to \mathcal{H}$ be the linear map defined by $f_p'(X) := \phi'(X) \cdot p$. It remains to establish \cref{eq:SC-repr-orbit-dim-formula-claim}, i.e., to show that $\dim(\orb_G(p)) = \rank_{\mathbb{R}}(f_p')$.
To this end, we claim that $\ker(f_p') = \mathfrak{g}_p$. Indeed, if $X \in \mathfrak{g}_p$, it means that
\begin{equation}\label{eq:proof-kerfpprime-equals-gp-eq1}
\forall t \in \mathbb{R},\ \phi(e^{tX}) \cdot p = p\,,
\end{equation}
but since $p \in \mathcal{H}^\infty$, \cref{eq:proof-kerfpprime-equals-gp-eq1} can be differentiated at $t=0$ [cf. \cref{eq:SC-repr-def-of-its-derivative}] to give
\begin{equation}
\phi'(X) \cdot p = 0\,,
\end{equation}
and thus $X \in \ker(f_p')$. Conversely, suppose that $X \in \ker(f_p')$. Denoting the curve $\gamma(t):= \phi(e^{tX}) \cdot p$, it satisfies $\gamma(0) = p$, and is a smooth map $\gamma:\mathbb{R} \to \mathcal{H}$ --- as it is a composition of $f_p$ (smooth since $p \in \mathcal{H}^\infty$) and the smooth map $t \mapsto e^{tX}$. For a fixed $t$, consider now the curve $\sigma(s):=\gamma(t+s)$. By composition, $\sigma$ is again a smooth map, and differentiating it at $s$ gives $\dot{\sigma}(s) = \dot{\gamma}(t+s)$. But since (due to $\phi$ being a group representation) we have $\sigma = \phi(e^{tX}) \circ \gamma$ with $\phi(e^{tX})$ a linear map, and linear maps are their own derivatives, we also have $\dot{\sigma}(s) = \phi(e^{tX}) \cdot \dot{\gamma}(s)$. We thus obtained $\dot{\gamma}(t+s) = \phi(e^{tX}) \cdot \dot{\gamma}(s)$. For $s=0$, this gives
\begin{equation}\label{eq:SC-repr-orbit-structure-proof-gammatdot-at-t-and-at-0-relation}
\dot{\gamma}(t) = \phi(e^{tX}) \cdot \dot{\gamma}(0)\,.
\end{equation}
But $\dot{\gamma}(0) = \phi'(X) \cdot p = 0$ [since $X \in \ker(f_p')$], hence $\dot{\gamma}(t) = \phi(e^{tX}) \cdot 0 = 0$ for all $t$. Thus, $\gamma: \mathbb{R} \to \mathcal{H}$ is a curve in the Hilbert space $\mathcal{H}$ that has zero derivative everywhere. But this necessarily means that $\gamma(t)$ is constant. Indeed, for any $w \in \mathcal{H}$, by letting $L_w(v):= \langle w, v \rangle$, we have that the map $L_w: \mathcal{H} \to \mathbb{C}$ is a linear and continuous (since bounded, by Cauchy-Schwarz) map; therefore the curve $g_w := L_w \circ \gamma : \mathbb{R} \to \mathbb{C}$ is also differentiable everywhere, with $\dot{g_w}(t) = (L_w \circ \dot{\gamma})(t) = \langle \dot{\gamma}(t), w \rangle = 0$, hence we know that for all $t$, $g_w(t)=g_w(0)$, thus $\langle \gamma(t) - \gamma(0), w \rangle = 0$ (for all $w \in \mathcal{H}$), which implies that $\gamma(t)=\gamma(0) = p$.

This establishes that $X \in \mathfrak{g}_p$, which concludes the proof that $\ker(f_p') = \mathfrak{g}_p$.

Putting things together, we therefore have
\begin{align}
\dim(\orb_G(p)) &= \dim(G/G_p)\\
&= \dim(G) - \dim(G_p)\label{eq:orb-stab-dim-formula-generic}\\
&= \dim(\mathfrak{g}) - \dim(\ker(f_p'))\\
&= \rank_{\mathbb{R}}(f_p')\,,
\end{align}
where the last equality uses the rank-nullity theorem from linear algebra, for the real-linear map $f_p':\mathfrak{g} \to \mathcal{H}$ from the finite-dimensional real vector space $\mathfrak{g}$ to the (possibly infinite-dimensional) vector space $\mathcal{H}$ (seen as a real vector space).
\end{proof}
\end{theorem}

We now give a remark on the case of $G$ being compact, which is the case for the group of passive linear optics $G_{\rm PLO}$ [cf. \cref{eq:remark-unkown-if-topologies-coincide-or-not-except-PLO} in the main text]:
\begin{remark}\label{rem:appendix-Gcompact-case}
If $G$ is compact, then the topology on an orbit provided by \cref{thm:SC-repr-orbit-structure} coincides with the subset topology inherited from $\mathcal{H}$.

Without such assumption, the former is at least finer than the latter (meaning that closeness in the orbit's former topology implies closeness in the Hilbert space $\mathcal{H}$).
\end{remark}
\paragraph{Proof of \cref{rem:appendix-Gcompact-case}.}
Indeed, without yet assuming the compactness of $G$, we have the following situation, for a given point $p \in \mathcal{H}$:
\begin{equation}\label{diag:orbit-map-decomposition}
\begin{tikzcd}
G \arrow[rrd, "f_p"] \arrow[ddd, "\pi"'] & & \\
                                         & & \mathcal{H} \\
                                         & & \orb_G(p) \arrow[u, "i"'] \\
G/G_p \arrow[rru, "F_p"'] & &
\end{tikzcd}
\end{equation}
Here, $i: q \mapsto q$ denotes the inclusion map, for which, importantly, we consider the topology on its domain $\orb_G(p)$ to be the one provided by \cref{thm:SC-repr-orbit-structure}, while we consider the topology on its codomain $\mathcal{H}$ to be the standard Hilbert space topology.
We now suppose that $p \in \mathcal{H}^\infty$. In this diagram, we have already seen that the maps $\pi$, $F_p$, and $f_p$ are all smooth (see also proof of \cref{thm:SC-repr-orbit-structure}); the only map that is potentially nonsmooth is the map $i$.

In fact, the map $i$ is at least guaranteed to be continuous.
To show this, we first show that the map $\tilde{F}_p := i \circ F_p : G/G_p \to \mathcal{H}$ is continuous.
We have, for any subset $U \subseteq \mathcal{H}$:
\begin{align}
(\tilde{F}_p^{-1})(U)
&= \pi\left[ \pi^{-1}\left( \tilde{F}_p^{-1}(U) \right) \right]\\
&= \pi\left[ (\tilde{F}_p \circ \pi)^{-1}(U) \right]\\
&= \pi\left[ (f_p)^{-1}(U) \right]\,,\label{eq:relation-between-two-topologies-discussion--cite1}
\end{align}
where the first equality is by surjectivity of $\pi$.
But since $\pi$ is a smooth submersion, it is an open map (e.g., \cite[Prop. 4.28]{lee_introduction_2012}), i.e., sends open sets to open sets. Thus, if $U$ is open in $\mathcal{H}$, then by the continuity of $f_p$, $(f_p)^{-1}(U)$ is open in $G$, and hence by the openness of $\pi$, $\pi\left[ (f_p)^{-1}(U) \right]$ is open in $G/G_p$. This shows, by \cref{eq:relation-between-two-topologies-discussion--cite1}, that $(\tilde{F}_p)^{-1}(U)$ is open in $G/G_p$ for all open $U \subseteq \mathcal{H}$, and hence that $\tilde{F}_p$ is continuous.
Therefore, the map $i$ can be written as the composition $i=\tilde{F}_p \circ F_p^{-1}$ of a continuous map ($\tilde{F}_p$) with a smooth map ($F_p^{-1}$), which establishes that $i$ is continuous. The continuity of $i$ establishes the second claim of the remark.

Let us introduce the map $$\tau_p : \orb_G(p) \to \orb_G(p)$$ to act as the identity $q \mapsto q$, but with its domain being regarded with its topology provided by \cref{thm:SC-repr-orbit-structure} while its codomain is regarded with the subset topology from $\mathcal{H}$.
The map $\tau_p$ is bijective by definition.
Lastly, we also introduce $\hat{F}_p : G/G_p \to \orb_G(p)$ to be the same map as $F_p$ but with the codomain $\orb_G(p)$ regarded with the subset topology from $\mathcal{H}$. The map $\hat{F}_p$ is continuous (due to $\tilde{F}_p$ being continuous).
Since we have $\tau_p = \hat{F}_p \circ (F_p)^{-1}$, it is a continuous map by composition (recall that $F_p$ is a diffeomorphism, hence its inverse is indeed continuous).
Our map $\tau_p$ is therefore a continuous bijection.
Now, if $G$ is compact, then so is the domain of $\tau_p$ (being the image of the compact $G$ under the continuous map $F_p \circ \pi$). Moreover, the codomain of $\tau_p$ is Hausdorff (as a subset of the Hausdorff space $\mathcal{H}$).
Our map $\tau_p$ is therefore a continuous bijection from a compact space to a Hausdorff space. This implies (see, e.g., Ref.~\cite[Corr. A.36]{Tu-IntroductionManifolds-2011}) that $\tau_p$ is a homeomorphism, which proves (in the case of $G$ compact) that the two topologies considered on the set $\orb_G(p)$ coincide.$\qed$

\subsection{Properties of the extended metaplectic representation $\phi$}
We now turn our attention back to our representation of interest, the extended metaplectic representation $\phi:G_0 \to \uni(\fock)$ of the extended metaplectic group $G_0$.
This representation $\phi$ is such that for all $X \in \mathfrak{g}_0$ ($\mathfrak{g}_0$ denoting the Lie algebra of the Lie group $G_0$), the operator $\phi'(X)$ is a skew-Hermitian polynomial in the canonical operators $a_k,a_k^\dagger$ of degree $\leq 2$.
In fact, $\phi'$ is surjective onto all such operators \cite[Chap. 4]{Folland-HarmonicAnalysis-1989}, and since the kernel of $\phi$ is discrete (see proof of \cref{lem:optical-groups-are-images-of-closed-subgroups-under-EMR} below), $\phi'$ is also injective. It is thus a one-to-one linear map between $\mathfrak{g}_0$ and the space of all skew-Hermitian such polynomial operators of degree $\leq 2$.

\subsubsection{Closed preimages of the quantum optical groups}
As the four unitary groups $G_{\fock}$ considered in the main text (\cref{tab:Lie-algebra-bases}) are subgroups of $\uni(\fock)$, it follows (due to $\phi$ being a representation)
that their preimages under $\phi$ form subgroups of the extended metaplectic group $G_0$. In fact, the latter are \textit{closed} subgroups of $G_0$:
\begin{lemma}\label{lem:optical-groups-are-images-of-closed-subgroups-under-EMR}
The preimages $G:=\phi^{-1}(G_{\fock})$ of the four $m$-mode unitary groups $G_{\fock,\rm PLO}$, $G_{\fock,\rm DPLO}$, $G_{\fock,\rm ALO}$, $G_{\fock,\rm GO}$ are all \emph{closed} subgroups $G_{\rm PLO}$, $G_{\rm DPLO}$, $G_{\rm ALO}$, $G_{\rm GO}$ of the extended metaplectic group $G_0$.
\begin{proof}
Let us first present more explicitly the extended metaplectic group $G_0$: it is the Lie group defined as \cite{Folland-HarmonicAnalysis-1989} the semidirect product
\begin{equation}\label{eq:EMG-explicit-semidirect-def}
G_0 := \mathrm{H}_{2m+1} \rtimes \mathrm{Mp}(2m,\mathbb{R})\,,
\end{equation}
where $\mathrm{H}_{2m+1}$ is the \textit{Heisenberg} Lie group of dimension $2m+1$ (whose underlying set is $\mathbb{R}^{2m+1}$), and where $\mathrm{Mp}(2m,\mathbb{R})$ is the \textit{metaplectic} Lie group of dimension $2m^2 + m$.

The definition of such a semidirect product between two Lie groups is such that it is first a smooth manifold defined as simply the Cartesian product $G_0 = \mathrm{H}_{2m+1} \times \mathrm{Mp}(2m,\mathbb{R})$ of both smooth manifolds, on top of which an appropriate group law is defined [which intuitively is the merging of the two group laws that makes the "rotations" from $\mathrm{Mp}(2m,\mathbb{R})$ and the "translations" from $\mathrm{H}_{2m+1}$ compose in the expected way], making it a Lie group.
Moreover, there is a smooth surjective and $2$-to-$1$ Lie group homomorphism $\Pi : \mathrm{Mp}(2m,\mathbb{R}) \to \mathrm{Sp}(2m,\mathbb{R})$ from the metaplectic group to the \textit{symplectic} group $\mathrm{Sp}(2m,\mathbb{R})$.

Now, it is known \cite{arvind_real_1995,Folland-HarmonicAnalysis-1989} that the four unitary groups in the claim are the images under $\phi$ of particularly simple subgroups of $G_0$:
\begin{align}
G_{\fock,\rm PLO} =&\ \phi(G_{\rm PLO})\,,\\ G_{\rm PLO} :=&\ \{\, g \in Z_{\mathrm{H}_{2m+1}} \times \Pi^{-1}(K) \,|\, \phi(g)\cdot\ket{\bm{0}} = \ket{\bm{0}} \,\}\,; \label{eq:optical-groups-are-images-of-closed-subgroups-under-EMR--cite1}\\[6pt]
G_{\fock,\rm DPLO} =&\ \phi(G_{\rm DPLO})\,,\\ G_{\rm DPLO} :=&\ \mathrm{H}_{2m+1} \times \Pi^{-1}(K)\,; \label{eq:optical-groups-are-images-of-closed-subgroups-under-EMR--cite2}\\[6pt]
G_{\fock,\rm ALO} =&\ \phi(G_{\rm ALO})\,,\\ G_{\rm ALO} :=&\ Z_{\mathrm{H}_{2m+1}} \times \mathrm{Mp}(2m,\mathbb{R})\,; \label{eq:optical-groups-are-images-of-closed-subgroups-under-EMR--cite3}\\[6pt]
G_{\fock,\rm GO} =&\ \phi(G_{\rm GO})\,,\\ G_{\rm GO} :=&\ G_0\, \label{eq:optical-groups-are-images-of-closed-subgroups-under-EMR--cite4}\,,
\end{align}
where $Z_{\mathrm{H}_{2m+1}} := \{(\bm{0},\bm{0},t)\ |\ t\in\mathbb{R}\} \subset \mathrm{H}_{2m+1}$ is the central subgroup of $\mathrm{H}_{2m+1}$, $K$ is the subgroup of $\mathrm{Sp}(2m,\mathbb{R})$ defined as $K := \mathrm{Sp}(2m,\mathbb{R}) \cap \mathrm{O}(2m)$, with $\mathrm{O}(2m)$ the \textit{orthogonal} group, and we denote the products as Cartesian products since here we are just focusing on the underlying sets of these subgroups of $G_0$.

Note that our conditioning in \cref{eq:optical-groups-are-images-of-closed-subgroups-under-EMR--cite1} reflects the fact that we take the convention of not including the identity Hamiltonian as part of $G_{\fock,\rm PLO}$ (cf.\ \cref{tab:Lie-algebra-bases}).

In each case of $G_{\fock}$, we have above identified \textit{a} subgroup $G$ of $G_0$ whose image under $\phi$ is $G_{\fock}$, but the preimage $\phi^{-1}(G_{\fock})$ could generally be larger than $G$, as $\phi$ is not quite injective. Indeed, it is clear that
\begin{equation}\label{eq:optical-groups-are-images-of-closed-subgroups-under-EMR--preimage-from-subgroup-of-preimage}
\phi^{-1}(G_{\fock}) = G \ker(\phi)\,,
\end{equation}
where the right-hand side indicates the subgroup of $G_0$ generated by elements of $G$ and of $\ker(\phi)$.
But it is also known that
\begin{equation}\label{eq:optical-groups-are-images-of-closed-subgroups-under-EMR--kerphi}
\ker(\phi) = \left\{
\begin{aligned}
&\big( (\bm{0},\bm{0},k) , e_{\mathrm{Mp}(2m,\mathbb{R})} \big) \,,\, \\
&\big( (\bm{0},\bm{0},k+1/2) , \epsilon_{\mathrm{Mp}(2m,\mathbb{R})} \big)
\end{aligned}
\,\bigg|\, k \in \mathbb{Z} \right\}\,,
\end{equation}
which is a discrete subgroup of $G_0$, where $\ker(\Pi) := \{ e_{\mathrm{Mp}(2m,\mathbb{R})}, \epsilon_{\mathrm{Mp}(2m,\mathbb{R})} \}$; let us justify \cref{eq:optical-groups-are-images-of-closed-subgroups-under-EMR--kerphi} in more detail.
Indeed, consider the associated representation $\tilde{\phi} : \mathrm{H}_{2m+1} \rtimes \mathrm{Sp}(2m,\mathbb{R}) \to \uni(\fock)/\{\id, -\id\}$. This $\tilde{\phi}$ is the representation that descends from $\phi$ after projection, i.e., it is related to $\phi$ via
\begin{equation}\label{eq:optical-groups-are-images-of-closed-subgroups-under-EMR--tildephirelation}
\tilde{\phi} \circ \pi_{G_0} = \pi_{\uni(\fock)} \circ \phi\,,
\end{equation}
with the projection maps $\pi_{G_0} : G_0 \to \mathrm{H}_{2m+1} \rtimes \mathrm{Sp}(2m,\mathbb{R})$ and $\pi_{\uni(\fock)} : \uni(\fock) \to \uni(\fock)/\{\id, -\id\}$.
By making a choice of representatives for all image points of $\tilde{\phi}$ in $\uni(\fock)/\{\id, -\id\}$, the representation $\tilde{\phi}$ is turned into an "almost-representation" $\bar{\phi} : \mathrm{H}_{2m+1} \rtimes \mathrm{Sp}(2m,\mathbb{R}) \to \uni(\fock)$, which is a map that is not quite a representation in the sense that the product rule is only satisfied up to a $\pm 1$ sign. It is the kernel (preimage of $+\id$) of \textit{this} map $\bar{\phi}$ that is explicitly given in Ref.~\cite[p. 196]{Folland-HarmonicAnalysis-1989}:
\begin{equation}\label{eq:optical-groups-are-images-of-closed-subgroups-under-EMR--kerphibar}
\ker(\bar{\phi}) = \big\{ \big( (\bm{0},\bm{0},k) , e_{\mathrm{Sp}(2m,\mathbb{R})} \big) \,|\, k \in \mathbb{Z} \big\}\,.
\end{equation}
First, let us obtain back the kernel of $\tilde{\phi}$ from this kernel of $\bar{\phi}$.
Let $(X,\tilde{A}) \in \mathrm{H}_{2m+1} \rtimes \mathrm{Sp}(2m,\mathbb{R})$, and suppose that $(X,\tilde{A}) \in \ker(\tilde{\phi})$. This means that either $\bar{\phi}((X,\tilde{A})) = +\id$ or $\bar{\phi}((X,\tilde{A})) = -\id$.
In the first case, \cref{eq:optical-groups-are-images-of-closed-subgroups-under-EMR--kerphibar} gives that $(X,\tilde{A}) = \big( (\bm{0},\bm{0},k) , e_{\mathrm{Sp}(2m,\mathbb{R})} \big)$ for some $k \in \mathbb{Z}$.
Suppose now that $\bar{\phi}((X,\tilde{A})) = -\id$.
Decomposing the map $\bar{\phi}$ into its two components $\rho$ (Heisenberg part) and $\bar{\mu}$ (symplectic part), cf. Ref.~\cite[p. 196]{Folland-HarmonicAnalysis-1989}, this case reads $\rho(X) \bar{\mu}(\tilde{A}) = - \id$.
But by the explicit expression of the Heisenberg representation $\rho$ (\cite[Eq. 1.25]{Folland-HarmonicAnalysis-1989}) we have for $t\in \mathbb{R}$ that
\begin{equation}\label{eq:optical-groups-are-images-of-closed-subgroups-under-EMR--rho00t}
\rho( (\bm{0},\bm{0},t) ) = e^{2 \pi i t} \id\,.
\end{equation}
In particular [\cref{eq:optical-groups-are-images-of-closed-subgroups-under-EMR--rho00t} for $t=1/2$], we have $\rho( (\bm{0},\bm{0},1/2) ) = -\id$, hence we can rewrite the above as $\rho( (\bm{0},\bm{0},1/2) ) \rho(X) \bar{\mu}(\tilde{A}) = \id$, and thus (since $\rho$ \textit{is} a representation, as it is only $\bar{\mu}$ that is only an almost-representation) also as
$\rho( (\bm{0},\bm{0},1/2) X) \bar{\mu}(\tilde{A}) = \id$. In other words, we have in this case that $((\bm{0},\bm{0},1/2) X,\tilde{A}) \in \ker(\bar{\phi})$, which by \cref{eq:optical-groups-are-images-of-closed-subgroups-under-EMR--kerphibar} gives $((\bm{0},\bm{0},1/2) X,\tilde{A}) = \big( (\bm{0},\bm{0},k) , e_{\mathrm{Sp}(2m,\mathbb{R})} \big)$, i.e., that $(X,\tilde{A}) = \big( (\bm{0},\bm{0},k-1/2) , e_{\mathrm{Sp}(2m,\mathbb{R})} \big)$ for some $k \in \mathbb{Z}$.
Putting both cases together, we have hence established that:
\begin{equation}\label{eq:optical-groups-are-images-of-closed-subgroups-under-EMR--kerphitilde}
\ker(\tilde{\phi}) = \big\{ \big( (\bm{0},\bm{0},a/2) , e_{\mathrm{Sp}(2m,\mathbb{R})} \big) \,|\, a \in \mathbb{Z} \big\}\,.
\end{equation}
We can deduce the kernel of $\phi$ from this kernel of $\tilde{\phi}$.
Let $(X,A) \in G_0$, and suppose that $(X,A) \in \ker(\phi)$. By \cref{eq:optical-groups-are-images-of-closed-subgroups-under-EMR--tildephirelation}, it implies that $\pi_{G_0}(X,A) := (X,\tilde{A}) \in \ker(\tilde{\phi})$, which shows from \cref{eq:optical-groups-are-images-of-closed-subgroups-under-EMR--kerphitilde} that
\begin{equation}\label{eq:optical-groups-are-images-of-closed-subgroups-under-EMR--kerphi-inclusion-step}
\ker(\phi) \subseteq \left\{
\begin{aligned}
&\big( (\bm{0},\bm{0},a/2) , e_{\mathrm{Mp}(2m,\mathbb{R})} \big) \,,\, \\
&\big( (\bm{0},\bm{0},a/2) , \epsilon_{\mathrm{Mp}(2m,\mathbb{R})} \big)
\end{aligned}
\,\bigg|\, a \in \mathbb{Z} \right\}\,.
\end{equation}
But conversely, $\phi\big( (\bm{0},\bm{0},a/2) , A \big) = \rho( (\bm{0},\bm{0},a/2) ) \mu( A )$ (decomposing this time the map $\phi$ itself into its semidirect product components $\rho$ and $\mu$), which with \cref{eq:optical-groups-are-images-of-closed-subgroups-under-EMR--rho00t} gives $\phi\big( (\bm{0},\bm{0},a/2) , A \big) = e^{i \pi a} \mu( A )$. Since for both cases of $A= e_{\mathrm{Mp}(2m,\mathbb{R})}$ and $A=\epsilon_{\mathrm{Mp}(2m,\mathbb{R})}$, we have $\mu(A) \in \{\id, -\id\}$ (as $\ker(\Pi) = \{ e_{\mathrm{Mp}(2m,\mathbb{R})}, \epsilon_{\mathrm{Mp}(2m,\mathbb{R})} \}$ and $\mu$ is an injective representation \cite[p.~178]{Folland-HarmonicAnalysis-1989}), we obtain that the elements of the right-hand side of \cref{eq:optical-groups-are-images-of-closed-subgroups-under-EMR--kerphi-inclusion-step} are elements of $\ker(\phi)$ if and only if $A = e_{\mathrm{Mp}(2m,\mathbb{R})}$ and $a \in 2 \mathbb{Z}$, or $A = \epsilon_{\mathrm{Mp}(2m,\mathbb{R})}$ and $a + 1 \in 2 \mathbb{Z}$, which exactly establishes \cref{eq:optical-groups-are-images-of-closed-subgroups-under-EMR--kerphi}.

Now, from \cref{eq:optical-groups-are-images-of-closed-subgroups-under-EMR--kerphi}, we can readily check that for all four groups $G$ of \cref{eq:optical-groups-are-images-of-closed-subgroups-under-EMR--cite1,eq:optical-groups-are-images-of-closed-subgroups-under-EMR--cite2,eq:optical-groups-are-images-of-closed-subgroups-under-EMR--cite3,eq:optical-groups-are-images-of-closed-subgroups-under-EMR--cite4}, it holds that $\ker(\phi) \subseteq G$; by \cref{eq:optical-groups-are-images-of-closed-subgroups-under-EMR--preimage-from-subgroup-of-preimage}, this means that we have $\phi^{-1}(G_{\fock}) = G$ in all four cases.

The only important point about these preimage expressions, for us, is that they make explicit the \textit{closedness} of these subsets.
Indeed, note that $K$ is a closed subgroup of $\mathrm{Sp}(2m,\mathbb{R})$ [this follows from $\mathrm{O}(2m)$ being a closed subset of the invertible matrices $\mathrm{GL}(2m,\mathbb{R})$], and hence (by continuity of $\Pi$) $\Pi^{-1}(K)$ is closed in $\mathrm{Mp}(2m,\mathbb{R})$.
Note also that $Z_{\mathrm{H}_{2m+1}}$ is closed in $\mathrm{H}_{2m+1}$, being a linear subspace of the underlying manifold $\mathbb{R}^{2m+1}$.
Of course, it also holds that $\mathrm{Mp}(2m,\mathbb{R})$ is closed in itself, and that $G_0$ is closed in itself.
Therefore, since Cartesian products of closed sets are closed sets, the three subgroups $G$ of \cref{eq:optical-groups-are-images-of-closed-subgroups-under-EMR--cite2,eq:optical-groups-are-images-of-closed-subgroups-under-EMR--cite3,eq:optical-groups-are-images-of-closed-subgroups-under-EMR--cite4} are all closed subgroups of $G_0$, and so is the one of \cref{eq:optical-groups-are-images-of-closed-subgroups-under-EMR--cite1} due to the strong continuity of $\phi$.
\end{proof}
\end{lemma}

\subsubsection{Infinite-dimensional linear span of orbits (except for $G_{\rm PLO}$)}\label{subsec:SM-infinite-dimensional-span-of-orbits}
Let us here give the justification for the infinite-dimensional linear span of orbits (except in the $G_{\rm PLO}$ case) claimed in the main text's \cref{sec:main-text-props-of-orbit-dimensions}:
\begin{proposition}\label{prop:QO-orbit-span-infinite-dim-conditions}
For all nonzero $\ket{\psi} \in \fock$, the following holds:
\begin{enumerate}[label=(\arabic*)]
\item The subspace $\spa_{\mathbb{C}}(\orb_{G_{\rm PLO}}(\ket{\psi}))$ is finite-dimensional if and only if $\ket{\psi}$ has finite support in the Fock basis.
\item The subspaces $\spa_{\mathbb{C}}(\orb_{G}(\ket{\psi}))$ for $G=G_{\rm DPLO}, G_{\rm ALO}, G_{\rm GO}$ are all infinite-dimensional.
\end{enumerate}
\begin{proof}

These are simply consequences of the way in which the restrictions $\left.\phi\right|_{G} :{G}\to \uni(\fock)$
of the extended metaplectic representation for these four groups decompose into irreducible representations, and of \cref{lem:orbit-closed-span}. Indeed, fix a nonzero $p\in \fock$, and let us denote $\mathcal{V}_{p,G} := \spa_{\mathbb{C}}(\orb_{G}(p))$ and $\mathcal{H}_{p,G} := \overline{\mathcal{V}_{p,G}}$.

Consider first $G_{\rm PLO}$. It is known that $\left.\phi\right|_{G_{\rm PLO}}$ is completely reducible (without multiplicities) via the following decomposition \cite[p. 208]{Folland-HarmonicAnalysis-1989} (or see also Refs.~\cite{Aniello-ExploringRepresentation-2006} and \cite[Theorem 5]{Harrow-ChurchSymmetric-2013}):
\begin{align}
\fock = \bigoplus_{n\in \mathbb{N}} \mathcal{H}_{m}^{n}\,,
\end{align}
where recall that $\mathcal{H}_{m}^{n}$ denotes the subspace of $\fock$ of total photon number exactly $n$.
Therefore, by \cref{lem:orbit-closed-span}, we have
\begin{equation}\label{eq:QO-orbit-span-infinite-dim-conditions--decomp-PLO}
\mathcal{H}_{p,G_{\rm PLO}} \ = \bigoplus_{\substack{n \in \mathbb{N};\\ P_n(p) \,\neq\, 0}} \mathcal{H}_{m}^{n}\,,
\end{equation}
where $P_n$ denotes the orthogonal projection onto $\mathcal{H}_{m}^{n}$.
Since each $\mathcal{H}_{m}^{n}$ is finite-dimensional, it follows from \cref{eq:QO-orbit-span-infinite-dim-conditions--decomp-PLO} that
$\mathcal{H}_{p,G_{\rm PLO}}$ is finite-dimensional if and only if $p$ has finite support in the Fock basis. But note that in Hilbert spaces, a subspace is finite-dimensional if and only if its closure is finite-dimensional (this follows from the fact that finite-dimensional subspaces are always closed). Thus, we obtain that $\mathcal{V}_{p,G_{\rm PLO}}$ is finite-dimensional if and only if $p$ has finite support in the Fock basis.

Consider now $G_{\rm DPLO}$. Since it is known that for the \textit{Weyl-Heisenberg} subgroup $G_{\rm WH} \subset G_{\rm DPLO} \subset G_0$ (physically $G_{\rm WH}$ consists of "just displacements" operators), $\left.\phi\right|_{G_{\rm WH}}$ is irreducible \cite[Prop. 1.43]{Folland-HarmonicAnalysis-1989}, then \textit{a fortiori} $\left.\phi\right|_{G_{\rm DPLO}}$ is also irreducible. Therefore (since an irreducible representation is already in completely reduced form), by \cref{lem:orbit-closed-span}, we have $\mathcal{H}_{p,G_{\rm DPLO}} = \fock$, which yields that $\mathcal{V}_{p,G_{\rm DPLO}}$ is infinite-dimensional.

Consider next $G_{\rm ALO}$. It is known that $\left.\phi\right|_{\mathrm{Mp}(2m,\mathbb{R})}$ is completely reducible (without multiplicities) via the following decomposition \cite[p. 194]{Folland-HarmonicAnalysis-1989}:
\begin{align}
\fock = \mathcal{H}_m^\mathrm{even} \oplus \mathcal{H}_m^\mathrm{odd}\,,
\end{align}
where $\mathcal{H}_m^\mathrm{even}$ and $\mathcal{H}_m^\mathrm{odd}$ denote the subspaces of $\fock$ of even and odd total photon number, respectively.
Since $G_{\rm ALO}$ enlarges $\mathrm{Mp}(2m,\mathbb{R})$ with new unitaries $\phi(g)$ differing from already present ones only by global phases, then $\left.\phi\right|_{G_{\rm ALO}}$ still decomposes into the same two even and odd irreducible components.
Therefore, by \cref{lem:orbit-closed-span}, we have
$\mathcal{H}_{p,G_{\rm ALO}} \ = \ $ $\mathcal{H}_m^\mathrm{even}, \mathcal{H}_m^\mathrm{odd}, \fock$ if $p$ has respectively support on only $\mathcal{H}_m^\mathrm{even}$, only $\mathcal{H}_m^\mathrm{odd}$, or both. Since $\mathcal{H}_m^\mathrm{even}$ and $\mathcal{H}_m^\mathrm{odd}$ are infinite dimensional, it follows in any case that $\mathcal{H}_{p,G_{\rm ALO}}$ is infinite dimensional, which yields that $\mathcal{V}_{p,G_{\rm ALO}}$ is infinite dimensional.

Lastly, consider $G_{\rm GO}$. Since, like $G_{\rm DPLO}$, it contains the Weyl-Heisenberg subgroup $G_{\rm WH}$, it follows likewise that $\left.\phi\right|_{G_{\rm GO}} = \phi$ is irreducible, and hence (from \cref{lem:orbit-closed-span}) that $\mathcal{V}_{p,G_{\rm GO}}$ is infinite-dimensional.
\end{proof}
\end{proposition}
One could also show that for $G=G_{\rm DPLO}, G_{\rm ALO}, G_{\rm GO}$ and nonzero $\rho \in B_2(\fock)$, the subspace $\spa_{\mathbb{C}}(\orb_{G}(\rho))$ is infinite-dimensional, but we do not give a proof of this here.

\subsubsection{Smooth and analytic vectors of $\phi$}
Next, we move to properties of smooth and analytic vectors of the extended metaplectic representation $\phi$. We make the following claims: (i) Schwartz states are all smooth vectors of $\phi$; (ii) exponential-decay states (in the Fock basis) are all analytic vectors of $\phi$. By exponential-decay states, we mean states $\ket{\psi} = \sum_{\bm{n} \in \mathbb{N}^m} c_{\bm{n}} \ket{\bm{n}}$ such that $|c_{\bm{n}}| \in \mathcal{O}(e^{ - c |\bm{n}|})$ for some $c>0$.

Let us first justify claim (i).
But clearly the action of any such finite degree polynomial operator on a Schwartz state is well-defined, and produces still a Schwartz state (as it can at worst slow down the decay of Fock basis coefficients by a polynomial factor). This shows that $[\phi'(X)]^n$ is well-defined on all Schwartz states for all $n$, and hence --- by the above-mentioned sufficient condition for smooth vectors --- all Schwartz states are smooth vectors of $\phi$.

For claim (ii), we first need to establish some lemmas:
\begin{lemma}\label{lem:norm-polyop-on-fockstate-bound}
If $P$ is a polynomial operator of degree $\leq d$ in the canonical operators $a_k,a_k^\dagger$, there exists a constant $C_P > 0$ such that for all $\bm{n} \in \mathbb{N}^m$:
\begin{equation}\label{eq:norm-polyop-on-fockstate-bound-main-claim}
\norm{P \ket{\bm{n}}} \leq C_P (|\bm{n}| + 1)^{d/2}\,.
\end{equation}
\begin{proof}
The normal ordering of $P$ gives us coefficients $c_{\bm{k},\bm{l}} \in \mathbb{C}$ such that
\begin{equation}\label{eq:polyop-degree-d--normal-ordering}
P=\sum_{\substack{\bm{k},\bm{l} \in \mathbb{N}^m\\|\bm{k}|+|\bm{l}|\leq d}} c_{\bm{k},\bm{l}} \ (\hat{\bm{a}}^\dagger)^{\bm{k}} \, \hat{\bm{a}}^{\bm{l}}\,.
\end{equation}
But by repeated applications of the defining relations
\begin{align*}
a^\dagger_j \ket{n_1,\dots,n_j,\dots,n_m}&:=\sqrt{n_j+1} \ket{n_1,\dots,n_j+1,\dots,n_m}\,,\\
a_j \ket{n_1,\dots,n_j,\dots,n_m}&:=\sqrt{n_j} \ket{n_1,\dots,n_j-1,\dots,n_m}\,,
\end{align*}
with the latter expression being defined as zero when $n_j = 0$, one obtains
\begin{align}
&(\hat{\bm{a}}^\dagger)^{\bm{k}} \, \hat{\bm{a}}^{\bm{l}} \ket{\bm{n}}\nonumber\\
&= \prod_{j=1}^m \sqrt{n_j(n_j - 1) \cdots (n_j - l_j + 1)} (\hat{\bm{a}}^\dagger)^{\bm{k}} \ket{\bm{n} - \bm{l}}\nonumber\\
&=\prod_{j=1}^m
\begin{aligned}[t]
&\sqrt{n_j(n_j - 1) \cdots (n_j - l_j + 1)}\\
&\quad \sqrt{(n_j - l_j + 1)(n_j - l_j + 2) \cdots (n_j - l_j + k_j)}\\
&\quad \ket{\bm{n} - \bm{l} + \bm{k}}
\end{aligned}\nonumber\\
&= \prod_{j=1}^m \sqrt{\frac{n_j!}{(n_j - l_j)!} \frac{(n_j - l_j + k_j)!}{(n_j - l_j)!}} \ket{\bm{n} - \bm{l} + \bm{k}}\,,
\end{align}
and hence
\begin{align}\label{eq:norm-polyop-on-fockstate-bound-1}
\norm{(\hat{\bm{a}}^\dagger)^{\bm{k}} \, \hat{\bm{a}}^{\bm{l}} \ket{\bm{n}}}
= \prod_{j=1}^m \sqrt{\frac{n_j!}{(n_j - l_j)!} \frac{(n_j - l_j + k_j)!}{(n_j - l_j)!}}\,,
\end{align}
where if $l_j > n_j$ for any $j$ the whole expression is zero.

We now use the bounds
\begin{align}
\frac{n_j!}{(n_j - l_j)!} &= n_j (n_j - 1) \cdots (n_j - l_j + 1)\nonumber\\
&\leq n_j^{l_j} \leq (n_j + 1)^{l_j}\nonumber\\[3pt]
&\leq (|\bm{n}| + 1)^{l_j}\,,
\end{align}
and
\begin{align}
&\frac{(n_j - l_j + k_j)!}{(n_j - l_j)!}\nonumber\\[5pt]
&= (n_j - l_j + 1)(n_j - l_j + 2) \cdots (n_j - l_j + k_j)\nonumber\\[3pt]
&\leq (n_j - l_j + k_j)^{k_j} \leq (n_j +k_j + 1)^{k_j}\nonumber\\[3pt]
&\leq (|\bm{n}| + |\bm{k}| + 1)^{k_j}\nonumber\\[3pt]
&\leq (|\bm{n}| + 1)^{k_j} (|\bm{k}| + 1)^{k_j}\,,
\end{align}
in \cref{eq:norm-polyop-on-fockstate-bound-1}, which yields
\begin{align}
\norm{(\hat{\bm{a}}^\dagger)^{\bm{k}} \, \hat{\bm{a}}^{\bm{l}} \ket{\bm{n}}}
&\leq \prod_{j=1}^m (|\bm{n}| + 1)^{(k_j + l_j)/2} (|\bm{k}| + 1)^{k_j/2}\nonumber\\
&= (|\bm{n}| + 1)^{(|\bm{k}|+|\bm{l}|)/2} (|\bm{k}| + 1)^{|\bm{k}|/2}\,.\label{eq:norm-polyop-on-fockstate-bound-2}
\end{align}
Thus, by taking the norm of \cref{eq:polyop-degree-d--normal-ordering}, and using the triangle inequality along with the obtained bound of \cref{eq:norm-polyop-on-fockstate-bound-2}, we obtain
\begin{align}
\norm{P \ket{\bm{n}}} &\leq \sum_{\substack{\bm{k},\bm{l} \in \mathbb{N}^m\\|\bm{k}|+|\bm{l}|\leq d}} |c_{\bm{k},\bm{l}}| \ \norm{(\hat{\bm{a}}^\dagger)^{\bm{k}} \, \hat{\bm{a}}^{\bm{l}} \ket{\bm{n}}} \nonumber\\
&\leq \sum_{\substack{\bm{k},\bm{l} \in \mathbb{N}^m\\|\bm{k}|+|\bm{l}|\leq d}} |c_{\bm{k},\bm{l}}| \ (|\bm{n}| + 1)^{(|\bm{k}|+|\bm{l}|)/2} (|\bm{k}| + 1)^{|\bm{k}|/2}\nonumber\\
&\leq \sum_{\substack{\bm{k},\bm{l} \in \mathbb{N}^m\\|\bm{k}|+|\bm{l}|\leq d}} |c_{\bm{k},\bm{l}}| \ (|\bm{n}| + 1)^{d/2} (d + 1)^{d/2}\,,
\end{align}
which establishes \cref{eq:norm-polyop-on-fockstate-bound-main-claim} by letting $$C_P := (d+ 1)^{d/2} \sum_{\substack{\bm{k},\bm{l} \in \mathbb{N}^m\\|\bm{k}|+|\bm{l}|\leq d}} |c_{\bm{k},\bm{l}}|\,.$$
\end{proof}
\end{lemma}

\begin{lemma}\label{lem:norm-polyop-on-finitesupfockstate-intermsofnumberoperator}
If $P$ is a polynomial operator of degree $\leq d$ in the canonical operators $a_k,a_k^\dagger$, there exists a constant $C'_P > 0$ such that for all states $\ket{\psi}$ of finite support in the Fock basis:
\begin{equation}\label{eq:norm-polyop-on-finitesupfockstate-intermsofnumberoperator-mainclaim}
\norm{P \ket{\psi}} \leq C'_P \norm{ (\hat{N}_{\rm tot} + \id)^{d/2} \ket{\psi} }\,.
\end{equation}
\begin{proof}
First, let us note that for a fixed $\bm{n} \in \mathbb{N}^m$, the vector $P \ket{\bm{n}}$ remains of finite support in the Fock basis (as $P$ is a finite-degree polynomial) and thus is written as
\begin{equation}\label{eq:norm-polyop-on-finitesupfockstate-intermsofnumberoperator--proof-eq1}
P \ket{\bm{n}} = \sum_{\bm{k} \in S_P^{\bm{n}}} b_{\bm{k}} \ket{\bm{k}}
\end{equation}
for some nonzero scalars $b_{\bm{k}} \in \mathbb{C}$ and some finite set $S_P^{\bm{n}} \subset \mathbb{N}^m$.
More precisely, since $P$ is of degree $\leq d$, the components $b_{\bm{k}}$ in \cref{eq:norm-polyop-on-finitesupfockstate-intermsofnumberoperator--proof-eq1} must be such that each $k_j$ is at most raised or lowered by $d$, i.e., $|n_j - k_j| \leq d$ for all $j$, which implies that $S_P^{\bm{n}} \subseteq \{\bm{k} \in \mathbb{Z}^m \,|\, n_j - d \leq k_j \leq n_j + d \}$.
Thus, by potentially setting extra coefficients $b_{\bm{k}}$ to zero, we may rewrite \cref{eq:norm-polyop-on-finitesupfockstate-intermsofnumberoperator--proof-eq1} as
\begin{equation}\label{eq:norm-polyop-on-finitesupfockstate-intermsofnumberoperator--proof-eq2}
P \ket{\bm{n}} = \sum_{\bm{\delta} \in S_P} b_{\bm{n} + \bm{\delta}} \ket{\bm{n} + \bm{\delta}}\,,
\end{equation}
where $S_P := \{ -d, \dots, d \}^m \subset \mathbb{Z}^m$, which is a finite set independent of $\bm{n}$.

Let now $\ket{\psi}$ be a finite-support state, which is written as
\begin{equation}
\ket{\psi} = \sum_{\bm{n} \in S_\psi} c_{\bm{n}} \ket{\bm{n}}
\end{equation}
for some scalars $c_{\bm{n}} \in \mathbb{C}$ and some finite set $S_\psi \subset \mathbb{N}^m$.
We have
\allowdisplaybreaks
\begin{align}
&\norm{P \ket{\psi}}^2\\
&= \norm{\sum_{\bm{n} \in S_\psi} c_{\bm{n}} P \ket{\bm{n}}}^2 \\
&= \norm{ \sum_{\bm{n} \in S_\psi} \sum_{\bm{\delta} \in S_P} c_{\bm{n}} b_{\bm{n} + \bm{\delta}} \ket{\bm{n} + \bm{\delta}} }^2 \\
&= \norm{ \sum_{\bm{\delta} \in S_P} \left(\sum_{\bm{n} \in S_\psi} c_{\bm{n}} b_{\bm{n} + \bm{\delta}} \ket{\bm{n} + \bm{\delta}}\right) }^2 \\
&\leq |S_P| \sum_{\bm{\delta} \in S_P} \norm{ \left(\sum_{\bm{n} \in S_\psi} c_{\bm{n}} b_{\bm{n} + \bm{\delta}} \ket{\bm{n} + \bm{\delta}}\right) }^2 \label{eq:norm-polyop-on-finitesupfockstate-intermsofnumberoperator--proof-chain-cite1}\\
&= |S_P| \sum_{\bm{\delta} \in S_P} \left(\sum_{\bm{n} \in S_\psi} |c_{\bm{n}} b_{\bm{n} + \bm{\delta}}|^2 \right) \label{eq:norm-polyop-on-finitesupfockstate-intermsofnumberoperator--proof-chain-cite2}\\
&= |S_P| \sum_{\bm{n} \in S_\psi} |c_{\bm{n}}|^2 \sum_{\bm{\delta} \in S_P} |b_{\bm{n} + \bm{\delta}}|^2 \\
&= |S_P| \sum_{\bm{n} \in S_\psi} |c_{\bm{n}}|^2 \ \norm{P \ket{\bm{n}}}^2 \label{eq:norm-polyop-on-finitesupfockstate-intermsofnumberoperator--proof-chain-cite3}\\
&\leq |S_P| \sum_{\bm{n} \in S_\psi} |c_{\bm{n}}|^2 \ \left(C_P (|\bm{n}| + 1)^{d/2}\right)^2 \label{eq:norm-polyop-on-finitesupfockstate-intermsofnumberoperator--proof-chain-cite4}\\
&= |S_P| (C_P)^2 \sum_{\bm{n} \in S_\psi} |c_{\bm{n}}|^2 \ (|\bm{n}| + 1)^{d} \\
&= |S_P| (C_P)^2 \sum_{\bm{n} \in S_\psi} |c_{\bm{n}}|^2 \ \norm{ (\hat{N}_{\rm tot} + \id)^{d/2} \ket{\bm{n}} }^2 \\
&= |S_P| (C_P)^2 \sum_{\bm{n} \in S_\psi} \ \norm{ (\hat{N}_{\rm tot} + \id)^{d/2} \, c_{\bm{n}} \ket{\bm{n}} }^2 \\
&= |S_P| (C_P)^2 \norm{ \sum_{\bm{n} \in S_\psi} (\hat{N}_{\rm tot} + \id)^{d/2} \, c_{\bm{n}} \ket{\bm{n}} }^2 \label{eq:norm-polyop-on-finitesupfockstate-intermsofnumberoperator--proof-chain-cite5}\\
&= |S_P| (C_P)^2 \norm{ (\hat{N}_{\rm tot} + \id)^{d/2} \sum_{\bm{n} \in S_\psi} c_{\bm{n}} \ket{\bm{n}} }^2 \\
&= |S_P| (C_P)^2 \norm{ (\hat{N}_{\rm tot} + \id)^{d/2} \ket{\psi} }^2\,,
\end{align}
which establishes \cref{eq:norm-polyop-on-finitesupfockstate-intermsofnumberoperator-mainclaim} by letting $C'_P := \sqrt{|S_P|} C_P\,.$

In the above steps,
in \cref{eq:norm-polyop-on-finitesupfockstate-intermsofnumberoperator--proof-chain-cite1} we used the triangle inequality for the norm followed by the norm inequality $\norm{\cdot}_1 \leq \sqrt{|S_P|} \norm{\cdot}_2$ on $\mathbb{R}^{|S_P|}$;
in \cref{eq:norm-polyop-on-finitesupfockstate-intermsofnumberoperator--proof-chain-cite2,eq:norm-polyop-on-finitesupfockstate-intermsofnumberoperator--proof-chain-cite3} we use Pythagoras' theorem;
in \cref{eq:norm-polyop-on-finitesupfockstate-intermsofnumberoperator--proof-chain-cite4} we applied \cref{lem:norm-polyop-on-fockstate-bound};
and in \cref{eq:norm-polyop-on-finitesupfockstate-intermsofnumberoperator--proof-chain-cite5} we used the fact that the operator $(\hat{N}_{\rm tot} + \id)^{d/2}$ is diagonal in the Fock basis, along with Pythagoras' theorem.
\end{proof}
\end{lemma}

\begin{lemma}\label{lem:norm-polyop-power-on-fockstate-bound}
If $P$ is a polynomial operator of degree $\leq d$ in the canonical operators $a_k,a_k^\dagger$, there exists a constant $C'_P > 0$ such that for all $k \in \mathbb{N}$ and all $\bm{n} \in \mathbb{N}^m$:
\begin{equation}\label{eq:norm-polyop-power-on-fockstate-bound-mainclaim}
\norm{P^k \ket{\bm{n}}} \leq (C'_P)^k \, \big(|\bm{n}| + k d + 1\big)^{kd/2}\,.
\end{equation}
\begin{proof}
We use the same constant $C'_P$ as the one provided by \cref{lem:norm-polyop-on-finitesupfockstate-intermsofnumberoperator}.
Fix $\bm{n} \in \mathbb{N}^m$, and let us prove \cref{eq:norm-polyop-power-on-fockstate-bound-mainclaim} by induction over $k$.
For $k=0$, the equation holds (as both sides are $1$).

Suppose now that \cref{eq:norm-polyop-power-on-fockstate-bound-mainclaim} is true for some $k \in \mathbb{N}$.
We have
\begin{align}
&\norm{ P^{k+1} \ket{\bm{n}} }\\
&= \norm{ P \left( P^k \ket{\bm{n}} \right) }\\
&\leq C'_P \norm{ (\hat{N}_{\rm tot} + \id)^{d/2} \left( P^k \ket{\bm{n}} \right) }\label{eq:norm-polyop-power-on-fockstate-bound--proof-chain-cite1}\\
&\leq C'_P \big(|\bm{n}| + k d + 1\big)^{d/2} \norm{ P^k \ket{\bm{n}} }\label{eq:norm-polyop-power-on-fockstate-bound--proof-chain-cite2}\\
&\leq C'_P \big(|\bm{n}| + k d + 1\big)^{d/2} \ (C'_P)^k \, \big(|\bm{n}| + k d + 1\big)^{kd/2}\label{eq:norm-polyop-power-on-fockstate-bound--proof-chain-cite3}\\
&= (C'_P)^{(k+1)} \, \big(|\bm{n}| + k d + 1\big)^{(k+1)d/2}\\
&\leq (C'_P)^{(k+1)} \, \big(|\bm{n}| + (k+1) d + 1\big)^{(k+1)d/2}\,,
\end{align}
which concludes the induction.
In the above steps:
in \cref{eq:norm-polyop-power-on-fockstate-bound--proof-chain-cite1} we applied \cref{lem:norm-polyop-on-finitesupfockstate-intermsofnumberoperator} to the finite-support state $P^k \ket{\bm{n}}$;
in \cref{eq:norm-polyop-power-on-fockstate-bound--proof-chain-cite2} we used the fact that $P^k$ is of degree $\leq k d$, along with the fact that the operator $(\hat{N}_{\rm tot} + \id)^{d/2}$ is diagonal in the Fock basis;
and in \cref{eq:norm-polyop-power-on-fockstate-bound--proof-chain-cite3} we used the induction hypothesis, i.e., \cref{eq:norm-polyop-power-on-fockstate-bound-mainclaim}.
\end{proof}
\end{lemma}

Having established \cref{lem:norm-polyop-power-on-fockstate-bound}, we can now prove our desired lemma about exponential-decay states:
\begin{lemma}\label{lem:exp-decay-states-satisfy-analytic-vector-condition}
If $\ket{\psi}$ is an exponential-decay state, then for all polynomial operators $P$ of degree $\leq 2$, there exists $t>0$ such that
\begin{equation}\label{eq:exp-decay-states-satisfy-analytic-vector-condition--proof-mainclaim}
\sum_{k=0}^\infty \frac{t^k}{k!} \norm{P^k \ket{\psi}} < +\infty\,.
\end{equation}
\begin{proof}
Let $\ket{\psi} = \sum_{\bm{n} \in \mathbb{N}^m} c_{\bm{n}} \ket{\bm{n}} \in \fock$, and let $P$ be a polynomial operator of degree $\leq d$.
Let us first show that if $d=2$, then for all $\beta>0$, there exists $t,A,C>0$ such that the inequality
\begin{equation}\label{eq:exp-decay-states-satisfy-analytic-vector-condition--proof-midclaim}
\sum_{k=0}^\infty \frac{t^k}{k!} \norm{P^k \ket{\psi}} \leq \sum_{\bm{n}\in\mathbb{N}^m} |c_{\bm{n}}| \big[A e^{\beta |\bm{n}|} + C\big]
\end{equation}
holds (with the terms being potentially infinite). Indeed, we have
\begin{align}
&\sum_{k=0}^\infty \frac{t^k}{k!} \norm{P^k \ket{\psi}}\\
&= \sum_{k=0}^\infty \frac{t^k}{k!} \norm{ \sum_{\bm{n}\in\mathbb{N}^m} c_{\bm{n}} P^k \ket{\bm{n}} }\\
&\leq \sum_{k=0}^\infty \frac{t^k}{k!} \sum_{\bm{n}\in\mathbb{N}^m} |c_{\bm{n}}| \, \norm{ P^k \ket{\bm{n}} }\label{eq:exp-decay-states-satisfy-analytic-vector-condition--proof-midclaim-chain-cite1}\\
&\leq \sum_{k=0}^\infty \frac{t^k}{k!} \sum_{\bm{n}\in\mathbb{N}^m} |c_{\bm{n}}| \, (C'_P)^k \, \big(|\bm{n}| + k d + 1\big)^{\frac{kd}{2}}\label{eq:exp-decay-states-satisfy-analytic-vector-condition--proof-midclaim-chain-cite2}\\
&= \sum_{\bm{n}\in\mathbb{N}^m} |c_{\bm{n}}| \, \sum_{k=0}^\infty \frac{(t C'_P)^k}{k!} \,\big(|\bm{n}| + k d + 1\big)^{\frac{kd}{2}}\,,
\end{align}
where in \cref{eq:exp-decay-states-satisfy-analytic-vector-condition--proof-midclaim-chain-cite1} we used the triangle inequality
and in \cref{eq:exp-decay-states-satisfy-analytic-vector-condition--proof-midclaim-chain-cite2} we applied \cref{lem:norm-polyop-power-on-fockstate-bound}.
But
\begin{align}
& \frac{(t C'_P)^k}{k!} \, \big(|\bm{n}| + k d + 1 \big)^{\frac{kd}{2}}\label{eq:exp-decay-states-satisfy-analytic-vector-condition--proof-midclaim-chain-cite0}\\
&\leq \frac{(t C'_P)^k}{k!} \, 2^{(\frac{kd}{2} -1)} \big[(|\bm{n}| + 1)^{\frac{kd}{2}} + (k d)^{\frac{kd}{2}} \big]\label{eq:exp-decay-states-satisfy-analytic-vector-condition--proof-midclaim-chain-cite3}\\
&\leq \frac{(t C'_P)^k}{k!} \, 2^{\frac{kd}{2}} \big[(|\bm{n}| + 1)^{\frac{kd}{2}} + (k d)^{\frac{kd}{2}} \big]\\
&= \frac{\{t C'_P [2(|\bm{n}|+1)]^{\frac{d}{2}}\}^k}{k!} \ + \ \big[t C'_P (2d)^{\frac{d}{2}}\big]^k \, \frac{(k^{\frac{d}{2}})^k}{k!}\\
&\leq \frac{\{t C'_P [2(|\bm{n}|+1)]^{\frac{d}{2}}\}^k}{k!} \ + \ \big[t C'_P (2d)^{\frac{d}{2}}\big]^k \, e^{(k^{d/2})}\label{eq:exp-decay-states-satisfy-analytic-vector-condition--proof-midclaim-chain-cite4}\\
&= \frac{[2 \, t \, C'_P \, (|\bm{n}|+1)]^k}{k!} \ + \ (4 \, t \, C'_P \, e)^k\,,\label{eq:exp-decay-states-satisfy-analytic-vector-condition--proof-midclaim-chain-cite5}
\end{align}
where, assuming $k \geq 1$, in \cref{eq:exp-decay-states-satisfy-analytic-vector-condition--proof-midclaim-chain-cite3} we used the inequality $(x+y)^\alpha \leq 2^{\alpha -1} (x^\alpha + y^\alpha)$ valid for all reals $x,y \geq 0$ and $\alpha \geq 1$, in \cref{eq:exp-decay-states-satisfy-analytic-vector-condition--proof-midclaim-chain-cite4} we used the inequality $x^k/k! \leq e^x$, and in \cref{eq:exp-decay-states-satisfy-analytic-vector-condition--proof-midclaim-chain-cite5} the specialization to $d=2$ enables further simplification. For $k=0$, this upper bound of \cref{eq:exp-decay-states-satisfy-analytic-vector-condition--proof-midclaim-chain-cite0} by \cref{eq:exp-decay-states-satisfy-analytic-vector-condition--proof-midclaim-chain-cite5} also holds, since it reduces to $1 \leq 2$.
Hence,
\begin{align}
&\sum_{k=0}^\infty \frac{t^k}{k!} \norm{P^k \ket{\psi}}\\
&\leq \sum_{\bm{n}\in\mathbb{N}^m} |c_{\bm{n}}| \left[ \sum_{k=0}^\infty \frac{(2 \, t \, C'_P \, (|\bm{n}|+1))^k}{k!} \ + \ \sum_{k=0}^\infty (4 \, t \, C'_P \, e)^k \right]\\
&= \sum_{\bm{n}\in\mathbb{N}^m} |c_{\bm{n}}| \left[ e^{(2 t C'_P (|\bm{n}|+1))} \ + \ \sum_{k=0}^\infty (4 \, t \, C'_P \, e)^k \right]\,.
\end{align}
Consider now an arbitrary $\beta > 0$.
Note that, for all $t < (4 \, C'_P \, e)^{-1}$, the innermost sum is finite,
and for all $t \leq \beta \, (2 \, C'_P)^{-1}$, we have $e^{[2 t C'_P (|\bm{n}|+1)]} \leq e^{\beta (|\bm{n}|+1)}$.

Thus, by picking any fixed $t>0$ satisfying $$t < \min((4 \, C'_P \, e)^{-1}, \beta \, (2 \, C'_P)^{-1})\,,$$ it follows that the inequality
\begin{equation}\label{eq:exp-decay-states-satisfy-analytic-vector-condition--proof-midclaim-versionwithplusone}
\sum_{k=0}^\infty \frac{t^k}{k!} \norm{P^k \ket{\psi}} \leq \sum_{\bm{n}\in\mathbb{N}^m} |c_{\bm{n}}| \big[e^{\beta (|\bm{n}|+1)} + C\big]
\end{equation}
holds, with $C := \sum_{k=0}^\infty (4 \, t \, C'_P \, e)^k < +\infty$. This establishes \cref{eq:exp-decay-states-satisfy-analytic-vector-condition--proof-midclaim} by letting $A := e^\beta$.

Now, if $\ket{\psi}$ is an exponential-decay state, it means that there exists $a > 0$ such that $|c_{\bm{n}}| \in \mathcal{O}(e^{- a |\bm{n}|})$, i.e., such that, for all $\bm{n} \in \mathbb{N}^m$, $|c_{\bm{n}}| < B e^{- a |\bm{n}|}$ for some constant $B > 0$.
Applying the claim of \cref{eq:exp-decay-states-satisfy-analytic-vector-condition--proof-midclaim} with $\beta := a/2 > 0$, we obtain that there exists $t,A,C>0$ such that
\begin{align}
&\sum_{k=0}^\infty \frac{t^k}{k!} \norm{P^k \ket{\psi}}\\
&\leq \sum_{\bm{n}\in\mathbb{N}^m} |c_{\bm{n}}| \big[A e^{\frac{a}{2} |\bm{n}|} + C\big] \\
&\leq \sum_{\bm{n}\in\mathbb{N}^m} B e^{- a |\bm{n}|} \big[A e^{\frac{a}{2} |\bm{n}|} + C\big] \\
&= A B \sum_{\bm{n}\in\mathbb{N}^m} e^{-\frac{a}{2} |\bm{n}|} \ + \ B C \sum_{\bm{n}\in\mathbb{N}^m} e^{- a |\bm{n}|} \,.
\end{align}
Since the values of both infinite sums are finite (since $a>0$), we have thus established \cref{eq:exp-decay-states-satisfy-analytic-vector-condition--proof-mainclaim}.
\end{proof}
\end{lemma}

Having established \cref{lem:exp-decay-states-satisfy-analytic-vector-condition}, our claim (ii) that all exponential-decay states are analytic vectors of $\phi$ now readily follows by application of the above criterion for analytic vectors of strongly continuous representations.

\subsubsection{Density-operator adapted representation $\Phi$, and its smooth and analytic vectors}
While the representation $\phi$ encodes the "ket picture" dynamics, it is useful to consider another representation (of the same extended metaplectic group $G_0$) induced from it, $\Phi$, which encodes the dynamics in the "density operator picture." It is simply the representation $\Phi:G_0 \to \uni(B_2(\fock))$ defined as the composition of $\phi$ with the adjoint representation of the full-unitary group on Fock space $\operatorname{Ad}:\uni(\fock) \to \uni(B_2(\fock))$; i.e., $\Phi := \operatorname{Ad} \circ \phi$, where $\operatorname{Ad}(U) \cdot A := U A U^\dagger$ for $A \in B_2(\fock)$.
The representation $\Phi$ is, by composition, strongly continuous [indeed one may verify that the map $\operatorname{Ad}$ is continuous with respect to the strong operator topologies on $\uni(\fock)$ and on $\uni(B_2(\fock))$ \cite[Cor. 2.29]{Hofmann-LieTheory-2007}]. Therefore, the listed properties of strongly continuous representations apply to it as well.
Analogously to how Schwartz states and finite-support states were shown to be, respectively, smooth and analytic vectors of $\phi$, it holds that Schwartz operators and finite-support operators [in the Fock basis of $B_2(\fock)$] are respectively smooth and analytic vectors of $\Phi$.
Indeed, with first the above sufficient condition for smooth vectors for $\Phi$, using exact analogous arguments as we did for $\phi$, gives this time that all Schwartz operators are smooth vectors of $\Phi$.
We hence know that for $X \in \mathfrak{g}$, $\Phi'(X) \cdot S$ is well-defined for all Schwartz operators $S$, i.e., that $[\Phi(e^{tX}) \cdot S - S]/t$ converges as $t \to 0$ to a certain point in $B_2(\fock)$, denoted as $\Phi'(X) \cdot S$.
We can then verify that it acts on $S$ by commutation with the original representation, i.e., we claim that
\begin{equation}\label{eq:Phiprime-action-by-commutation-claim}
\Phi'(X) \cdot S = [\phi'(X), S]
\end{equation}
for all Schwartz operators $S$.

To see this, denote for fixed $X \in \mathfrak{g}$, Schwartz operator $S$, and Schwartz state $\ket{\psi}$, the quantity
\begin{equation}
\beta(t) := \frac{[\Phi(e^{t X})\cdot S]\ket{\psi} - S\ket{\psi}}{t}\,.
\end{equation}
Let us denote $U_t := \phi(e^{t X})$.

We have
\begin{align}
\beta(t) &= \frac{U_{t} S U_{-t}\ket{\psi} - S\ket{\psi}}{t}\\[12pt]
&=\frac{U_{t} S U_{-t}\ket{\psi} - U_{t} S \ket{\psi}}{t}\nonumber\\[6pt]
&\qquad+ \frac{U_{t} S \ket{\psi} - S\ket{\psi}}{t}\\[12pt]
&:= \ \beta_1(t) \ \ +\ \ \beta_2(t)\label{eq:proof-adjoint-deriv-action--beta-decomp-beta1-beta2}\,.
\end{align}

To ease the notation, let us also denote $\alpha(t) := \frac{U_{-t}\ket{\psi} - \ket{\psi}}{t}$ and $\alpha_0 := \phi'(-X)\ket{\psi}$.
Consider now the first term $\beta_1(t) = U_{t} S \ \alpha(t)$. We claim that it tends to $S \alpha_0 = S \phi'(-X)\ket{\psi} = - S \phi'(X) \ket{\psi}$ when $t \to 0$.
Indeed, we have
\begin{align}
& \norm{\beta_1(t) - S \alpha_0}\\
&= \norm{ U_{t} S \, \alpha(t) - U_{t} S \alpha_0 + U_{t} S \alpha_0 - S \alpha_0 }\\
&= \norm{ U_{t} S [ \alpha(t) - \alpha_0] + (U_{t} - \id) S \alpha_0 }\\
&\leq \norm{ U_{t} S ( \alpha(t) - \alpha_0) } + \norm{ (U_{t} - \id) S \alpha_0 }\label{eq:proof-adjoint-deriv-action--beta1-cite1}\\
&= \norm{ S ( \alpha(t) - \alpha_0) } + \norm{ (U_{t} - \id) S \alpha_0 }\,,\label{eq:proof-adjoint-deriv-action--beta1-cite2}
\end{align}
where we used the triangle inequality in \cref{eq:proof-adjoint-deriv-action--beta1-cite1} and the unitarity of $U_t$ in \cref{eq:proof-adjoint-deriv-action--beta1-cite2}.
But since $\alpha(t) \to \alpha_0$ (as $\ket{\psi}$ is a Schwartz state and hence a smooth vector for $\phi$), and since the Schwartz operator $S$ is bounded, we get that $\norm{ S ( \alpha(t) - \alpha_0) } \to 0$;
and by strong continuity of $\phi$, we get that $\norm{ (U_{t} - \id) S \alpha_0 } \to 0$.
The claimed limit for $\beta_1(t)$ is thus established.

As for the second term $\beta_2(t)$, it tends to $\phi'(X) S \ket{\psi}$ when $t \to 0$ (since $S \ket{\psi}$ is a Schwartz state and hence a smooth vector for $\phi$).

Therefore, by taking the limit $t \to 0$ in \cref{eq:proof-adjoint-deriv-action--beta-decomp-beta1-beta2}, we obtain that
\begin{equation}
\beta(t) \xrightarrow[t \to 0]{} - S \phi'(X) \ket{\psi} + \phi'(X) S \ket{\psi} = [\phi'(X), S] \ket{\psi}\,,
\end{equation}
and since $\beta(t)$ also tends to $[\Phi'(X) \cdot S] \ket{\psi}$ when $t \to 0$, we have thus shown that the two maps $\ket{\psi} \mapsto [\Phi'(X) \cdot S] \ket{\psi}$ and $\ket{\psi} \mapsto [\phi'(X), S] \ket{\psi}$ coincide on Schwartz states.
Furthermore, we already know that $[\Phi'(X) \cdot S] \in B_2(\fock)$, and it also holds that $[\phi'(X), S]=\phi'(X) S - S \phi'(X) \in B_2(\fock)$ [since $S$ is a Schwartz operator and $\phi'(X)$ is a polynomial operator, see \cref{sec:SM-Schwartz-spaces-and-proof-of-lsc}].
Hence these two maps are defined and continuous on the whole Hilbert space $\fock$ [as operators in $B_2(\fock)$ are bounded] and agree on a dense subspace (Schwartz states). This implies that these two maps coincide everywhere, i.e., we have established that, at the level of operators in $B_2(\fock)$, the equality of \cref{eq:Phiprime-action-by-commutation-claim} holds for any $X \in \mathfrak{g}$ and Schwartz operator $S$.

Let us denote by $\operatorname{ad}_P$ the map realizing commutation with a polynomial operator $P$, i.e., $\operatorname{ad}_P(S) := [P,S]$, so that we have established that
\begin{equation}\label{eq:Phiprime-action-by-commutation-ad-notation}
\Phi'(X) \cdot S = \operatorname{ad}_{\phi'(X)}(S)
\end{equation}
for Schwartz operators $S$.

We can besides obtain the analogous version of \cref{lem:exp-decay-states-satisfy-analytic-vector-condition}, for exponential-decay \textit{operators} instead of states:
\begin{lemma}\label{lem:exp-decay-operators-satisfy-analytic-vector-condition}
If $S$ is an exponential-decay operator in $B_2(\fock)$, then for all polynomial operators $P$ of degree $\leq 2$, there exists $t>0$ such that
\begin{equation}\label{eq:exp-decay-operators-satisfy-analytic-vector-condition--proof-mainclaim}
\sum_{k=0}^\infty \frac{t^k}{k!} \norm{\operatorname{ad}_P^k(S)}_{B_2(\fock)} < +\infty\,,
\end{equation}
where the norm is the Hilbert-Schmidt norm on $B_2(\fock)$.
\begin{proof}
Using the isomorphism of Hilbert spaces $B_2(\fock) \simeq \fock \otimes \overline{\fock}$ given by the identification (see, e.g., Ref.~\cite[Prop. 2.6.9]{Kadison-FundamentalsTheory-1983})
\begin{equation}\label{eq:exp-decay-operators-satisfy-analytic-vector-condition--def-isom}
\ketbra{\bm{n}}{\bm{n}'} \mapsto \ket{\bm{n}} \otimes \ket{\overline{\bm{n}'}}\,,
\end{equation}
the Hilbert space of operators $B_2(\fock)$ may be viewed as a $2m$-mode Fock space. Denote this map by $\mathtt{v}: B_2(\fock) \to \fock \otimes \overline{\fock}$.
Furthermore, one can verify, either using properties of this isometry, or more directly using \cref{eq:norm-polyop-on-finitesupfockstate-intermsofnumberoperator--proof-eq2}, that the operator $\operatorname{ad}_P$ acts as $\mathtt{v}( \operatorname{ad}_P(S) ) = \widetilde{P}( \widetilde{S} )$, where
\begin{equation}\label{eq:exp-decay-operators-satisfy-analytic-vector-condition--def-Ptilde}
\widetilde{P} := P \otimes \id - \id \otimes \overline{P^\dagger}
\end{equation}
and
\begin{equation}
\widetilde{S} := \mathtt{v}(S)\,.
\end{equation}

It then follows that the corresponding action of $\operatorname{ad}_P^k$ on $S$ is given by $(\widetilde{P})^k (\widetilde{S})$.
Therefore, since $\mathtt{v}$ is an isometry, we have
\begin{equation}\label{eq:exp-decay-operators-satisfy-analytic-vector-condition--norm-equality}
\norm{\operatorname{ad}_P^k(S)}_{B_2(\fock)} = \norm{\mathtt{v}(\operatorname{ad}_P^k(S))} = \norm{(\widetilde{P})^k (\widetilde{S})}\,.
\end{equation}
But as $P$ is a polynomial operator on $\fock$ of degree $\leq 2$, so is the operator $\widetilde{P}$ [\cref{eq:exp-decay-operators-satisfy-analytic-vector-condition--def-Ptilde}].
Furthermore, as $S$ is an exponential-decay operator in $B_2(\fock)$, the vector $\widetilde{S}$ is an exponential-decay state in the $2m$-mode Fock space [as they have the same coefficients in their respective Fock bases, due to \cref{eq:exp-decay-operators-satisfy-analytic-vector-condition--def-isom}].

Therefore, \cref{lem:exp-decay-states-satisfy-analytic-vector-condition} readily applies to $\widetilde{P}$ and $\widetilde{S}$, giving that there exists $t>0$ such that
\begin{equation}
\sum_{k=0}^\infty \frac{t^k}{k!} \norm{(\widetilde{P})^k (\widetilde{S})} < +\infty\,,
\end{equation}
which along with \cref{eq:exp-decay-operators-satisfy-analytic-vector-condition--norm-equality} establishes \cref{eq:exp-decay-operators-satisfy-analytic-vector-condition--proof-mainclaim}.
\end{proof}
\end{lemma}

This \cref{lem:exp-decay-operators-satisfy-analytic-vector-condition}, along with \cref{eq:Phiprime-action-by-commutation-ad-notation}, now establishes that all exponential-decay operators are analytic vectors of $\Phi$ (by the above criterion for analytic vectors of strongly continuous representations).

\subsection{Proof of Theorem \ref{thm:orbit-dim-maintext}}
The main text's \cref{thm:orbit-dim-maintext} is now a direct consequence of \cref{thm:SC-repr-orbit-structure} and some of the other seen results:
\paragraph{Proof of \cref{thm:orbit-dim-maintext}.}
Let $G_{\fock}$ be one of the $m$-mode quantum optical unitary groups considered in \cref{tab:Lie-algebra-bases}, and recall that $G_0$ denotes the $m$-mode extended metaplectic group. By \cref{lem:optical-groups-are-images-of-closed-subgroups-under-EMR}, there exists a closed subgroup $G$ of $G_0$ such that $\phi(G) = G_{\fock}$.

Beware that $G_{\fock}$ is denoted as $G$ in the main text, while here $G$ denotes the associated closed subgroup of the extended metaplectic Lie group.

Since $G$ is a closed subgroup of the Lie group $G_0$, the closed subgroup theorem applies and yields the following two facts (as already used in the proof of \cref{thm:SC-repr-orbit-structure}).
First, $G$ is a Lie subgroup of $G_0$, and thus a Lie group itself and moreover an embedded submanifold of $G_0$, meaning in particular that its topology is the subset topology inherited from $G_0$. Therefore, the restricted representation $\left.\phi\right|_{G} :G\to \uni(\fock)$ is still strongly continuous.
Second, the Lie algebra of $G$ is given by
\begin{equation}
\mathfrak{g} = \{ X \in \mathfrak{g}_0 \,|\, \forall t \in \mathbb{R}\ e^{tX} \in G \}\,.
\end{equation}
Let now $\mathcal{B}_{\mathfrak{g}} := \{ X_1,\dots,X_d \}$ [$d:=\dim(\mathfrak{g})$] be the basis of $\mathfrak{g}$ given by $\phi'^{-1}(i H_k)$ where $i\{H_1,\dots,H_d\}$ is the Lie algebra basis of $\mathfrak{g}_{\fock}$ given in the main text's \cref{tab:Lie-algebra-bases} (that the $X_k$'s are unique and form a Lie algebra basis of $\mathfrak{g}$ follows from the fact that $\phi'$ is an injective Lie algebra homomorphism \cite{Folland-HarmonicAnalysis-1989}).

Consider now the ket picture.
For any Schwartz state $\ket{\psi}$, it is a smooth vector for the representation $\phi:G_0\to \uni(\fock)$, and thus also a smooth vector for its restriction $\left.\phi\right|_{G}$ (as $G$ is an embedded submanifold of $G_0$).
We can therefore apply \cref{thm:SC-repr-orbit-structure} to this representation $\left.\phi\right|_{G}$ and to the point $\ket{\psi}$, which establishes the manifold structure on $\orb_G(\ket{\psi}) = \{ U \ket{\psi} \,|\, U \in G_{\fock} \}$ and the formula
\begin{equation}
\dim(\orb_{G}(\ket{\psi})) = \rank_{\mathbb{R}}(\{ i H_1 \ket{\psi},\dots, i H_d \ket{\psi}\})\,.
\end{equation}
Since the map $\ket{\varphi} \mapsto i \ket{\varphi}$ is a linear isomorphism on $\fock$, the above rank is also equal to that in the claimed \cref{eq:concrete-rank-formula-ket}.

The density operator case is addressed analogously, by this time applying \cref{thm:SC-repr-orbit-structure} to the strongly continuous representation $\left.\Phi\right|_{G}$ and to a Schwartz operator $\rho$ (which is a smooth vector for this representation), yielding the manifold structure on $\orb_G(\rho) = \{ U \rho U^\dagger \,|\, U \in G_{\fock} \}$ and the formula
\begin{equation}
\dim(\orb_{G}(\rho)) = \rank_{\mathbb{R}}(\{ [i H_1,\rho],\dots,[i H_d,\rho]\})\,,
\end{equation}
which, since $[i H_k,\rho] = i[H_k,\rho]$ and since the map $A \mapsto i A$ is a linear isomorphism on $B_2(\fock)$, establishes \cref{eq:concrete-rank-formula-density} as well.$\qed$

\section{Formalization and proof of Theorem \ref{thm:VQC-orbit-dim-maintextinformal}}\label{subsec:proof-thm-2}

The goal of this section is to prove the technical \cref{prop:rank-of-analyticcomp-of-cstrank-and-nonsmooth}, which is our main needed result after which we can prove our formalized \cref{thm:VQC-bounds} (which generalizes \cref{thm:VQC-orbit-dim-maintextinformal} from the main text).
This way of proving our result through the route of \cref{prop:rank-of-analyticcomp-of-cstrank-and-nonsmooth} allows us to bypass the potential mismatch of topologies mentioned in \cref{eq:remark-unkown-if-topologies-coincide-or-not-except-PLO}, at the small cost of assuming that the states are exponential-decay states (so that the orbit map is analytic).

What \cref{prop:rank-of-analyticcomp-of-cstrank-and-nonsmooth} mainly relies on are the concepts of smooth maps, analytic maps, and maps of constant rank. Crucially though, since our statement of the desired \cref{thm:VQC-bounds} involves some maps taking values in a (potentially infinite-dimensional) Hilbert space, we need such concepts to be meaningful between infinite-dimensional spaces.
This is not an issue, as maps of this type can be defined on quite general classes of spaces. Notably, they can be readily defined on a class of spaces known as (smooth and analytic) \textit{Banach manifolds}, for which we refer the reader to Ref.~\cite{Abraham-ManifoldsTensor-1988} for an introductory treatment. Let us also note that the basic property (which we use) of the chain rule readily holds as well for smooth maps between Banach manifolds \cite[Theorem 3.3.7]{Abraham-ManifoldsTensor-1988} (as it does for smooth maps between finite-dimensional manifolds).

As stated, the statements for our purpose involve maps whose domains and codomains are either finite-dimensional manifolds or (potentially infinite-dimensional) Hilbert spaces.
The former are particular cases of the latter, and the latter are particular cases of not just Banach manifolds, but in fact more specifically of \textit{Hilbert manifolds}. These are manifolds that are "modeled on" Hilbert spaces (i.e., informally, they "locally look like" a given Hilbert space). Hence, it will suffice for us in this section to consider maps between Hilbert manifolds instead of Banach manifolds (that will also somewhat simplify our proof of the intermediate \cref{lem:bamber-analytic-lemma-generalized-to-analytic-Hilbert-space-codomain}).

\begin{lemma}\label{lem:rank-of-smoothcomp-of-cstrank-and-nonsmooth}
Let $M_1$ be a finite-dimensional smooth manifold, and let $M_2$ and $M_3$ be (possibly infinite-dimensional) smooth Hilbert manifolds.
Suppose that $f:M_1 \to M_2$ and $g:M_2 \to M_3$ are maps such that $f$ is smooth and of constant rank $r_f$, and the composition $g \circ f:M_1 \to M_3$ is smooth (but $g$ need not be smooth).
Then, for all $p \in M_1$:
\begin{equation}
\rank( D(g \circ f)(p) ) \leq r_f\,.
\end{equation}
\begin{proof}
Denote $r:=r_f$ and $d_k:=\dim(M_k)$ ($k=1,2,3$). Note that since the domain of $f$ is finite-dimensional, we have $r \leq d_1 < \infty$.
Denote by $\mathcal{H}$ and $\mathcal{K}$ the Hilbert spaces on which $M_2$ and $M_3$ are modeled.

Fix a $p \in M_1$.
Since $f$ is of constant rank $r$, a version of the constant rank theorem valid for general smooth maps between Banach manifolds \cite[Prop. 3.5.16]{Abraham-ManifoldsTensor-1988}, specialized to the case where the domain is finite-dimensional, gives us the following:
\begin{enumerate}[label=(\arabic*)]
\item there exist subspaces $E_a^{(1)},E_b^{(1)} \subseteq \mathbb{R}^{d_1}$ and $E_a^{(2)},E_b^{(2)} \subseteq \mathcal{H}$ such that $\mathbb{R}^{d_1} = E_a^{(1)} \oplus E_b^{(1)}$, $\mathcal{H} = E_a^{(2)} \oplus E_b^{(2)}$, and $\dim(E_a^{(1)}) = \dim(E_a^{(2)}) = r$;
\item there exists a local chart $(\mathcal{U}^{(1)},\alpha^{(1)})$ around $p$ and a local chart $(\mathcal{U}^{(2)},\alpha^{(2)})$ around $f(p)$, such that $\alpha^{(1)}(p) = 0$, $\alpha^{(2)}(f(p)) = 0$, $f(\mathcal{U}^{(1)}) \subseteq \mathcal{U}^{(2)}$, such that the open sets $U^{(1)} := \alpha^{(1)}(\mathcal{U}^{(1)})$ and $U^{(2)} := \alpha^{(2)}(\mathcal{U}^{(2)})$ are of the form $U^{(1)} = U_a^{(1)} \oplus U_b^{(1)}$ and $U^{(2)} = U_a^{(2)} \oplus U_b^{(2)}$ for some open sets $U_\gamma^{(j)} \subseteq E_\gamma^{(j)}$, $\gamma=a,b$ and $j=1,2$ (where we denote $W_a \oplus W_b := \{w_a + w_b \ |\ w_a \in W_a, w_b \in W_b\}$ for subsets $W_a,W_b$ of a common vector space);
\item there exists a linear isomorphism $J_a: E_a^{(1)} \to E_a^{(2)}$ such that $J_a(U_a^{(1)}) = U_a^{(2)}$;
\end{enumerate}
such that the map
$$\tilde{f} \,:=\, \alpha^{(2)} \circ f \circ (\alpha^{(1)})^{-1} \ \, :\, U_a^{(1)} \oplus U_b^{(1)} \to U_a^{(2)} \oplus U_b^{(2)}$$
acts as [for all $(x_a + x_b) \in U_a^{(1)} \oplus U_b^{(1)}$]:
\begin{equation}\label{eq:rank-of-smoothcomp-of-cstrank-and-nonsmooth-proof-ftilde-action}
\tilde{f}(x_a + x_b) = J_a(x_a) + 0\,.
\end{equation}

Since $M_3$ is a smooth Hilbert manifold, there exists a local chart $(\mathcal{U}^{(3)},\alpha^{(3)})$ around $g(f(p))$, which (by reducing the previous $\mathcal{U}^{(1)}$, using the continuity of $g \circ f$, while keeping $U^{(1)}$ of the above product form) can be assumed to be such that $(g \circ f)(\mathcal{U}^{(1)}) \subseteq \mathcal{U}^{(3)}$, and we let $\tilde{g} := \alpha^{(3)} \circ g \circ (\alpha^{(2)})^{-1}$, defined on $\tilde{f}(U^{(1)}) \subseteq U^{(2)}$, be the associated coordinate map of $g$.
Applying then $\tilde{g}$ to \cref{eq:rank-of-smoothcomp-of-cstrank-and-nonsmooth-proof-ftilde-action} gives
\begin{equation}\label{eq:rank-of-smoothcomp-of-cstrank-and-nonsmooth-proof-gtilde--ftilde-action}
(\tilde{g} \circ \tilde{f})(x_a + x_b) = \tilde{g}[J_a(x_a) + 0]\,.
\end{equation}
Letting $i:E_a^{(1)} \to \mathbb{R}^{d_1} = E_a^{(1)} \oplus E_b^{(1)}$ denote the inclusion map $x_a \mapsto (x_a + 0)$, %
consider the map $\tilde{G}: U_a^{(1)} \to \mathcal{K}$ defined as $\tilde{G} := (\tilde{g} \circ \tilde{f}) \circ \left.i\right|_{U_a^{(1)}}$.
[For concreteness, by \cref{eq:rank-of-smoothcomp-of-cstrank-and-nonsmooth-proof-gtilde--ftilde-action} this map acts as $\tilde{G}(x_a) = \tilde{g}(J_a(x_a) + 0)$.]

Since $g \circ f:M_1 \to M_3$ is smooth, the map $(\tilde{g} \circ \tilde{f})$ is also smooth (since it is a local representation of $g \circ f$ by charts on the two smooth Hilbert manifolds $M_1$ and $M_3$) and thus (since $i$ is smooth) the map $\tilde{G}$ is smooth (even though $\tilde{g}$ might not be).

Hence, since the derivative of a smooth map at any point has a rank that is upper-bounded by the dimension of the map's domain, we get for the map $\tilde{G}$ [whose domain is the open set $U_a^{(1)} \subseteq E_a^{(1)}$ and hence has $\dim(U_a^{(1)}) = \dim(E_a^{(1)}) = r$] that for all $x_a \in U_a^{(1)}$:
\begin{equation}\label{eq:rank-of-smoothcomp-of-cstrank-and-nonsmooth-proof-Gtilde-rank-inequality}
\rank( D(\tilde{G})(x_a) ) \leq r\,.
\end{equation}

But since $(\tilde{g} \circ \tilde{f}) = \tilde{G} \circ \pi$, with $\pi: U_a^{(1)} \oplus U_b^{(1)} \to U_a^{(1)}$ the projection map $(x_a + x_b) \mapsto x_a$, and because $\pi$ is a smooth submersion and composing on the right with a smooth submersion does not change the rank (by the chain rule), we have for all $q \in U^{(1)} = U_a^{(1)} \oplus U_b^{(1)}$:
\begin{equation}
\rank( D(\tilde{g} \circ \tilde{f})( q ) ) = \rank( D(\tilde{G})(\pi( q )) )\,,
\end{equation}
and hence in particular at the point $\tilde{p}:=\alpha^{(1)}(p) \in U^{(1)}$:
\begin{equation}
\rank( D(\tilde{g} \circ \tilde{f})( \tilde{p} ) ) = \rank( D(\tilde{G})(\pi(\tilde{p})) )\,.
\end{equation}
Combining the above equality with \cref{eq:rank-of-smoothcomp-of-cstrank-and-nonsmooth-proof-Gtilde-rank-inequality} gives
\begin{equation}\label{eq:rank-of-smoothcomp-of-cstrank-and-nonsmooth-proof-gtilde--ftilde-rank-inequality}
\rank( D(\tilde{g} \circ \tilde{f})( \tilde{p} ) ) \leq r\,,
\end{equation}
but since $(g \circ f) = (\alpha^{(3)})^{-1} \circ (\tilde{g} \circ \tilde{f}) \circ \alpha^{(1)}$ and composing with the charts again does not change the rank, we have $\rank( D(\tilde{g} \circ \tilde{f})( \tilde{p} ) ) = \rank( D(g \circ f)( p ) )$, and hence with \cref{eq:rank-of-smoothcomp-of-cstrank-and-nonsmooth-proof-gtilde--ftilde-rank-inequality} we have established that for all $p \in M_1$:
\begin{equation}
\rank( D(g \circ f)( p ) ) \leq r\,.
\end{equation}

\end{proof}
\end{lemma}

\begin{lemma}[{{\cite[Prop. B.4]{Bamber-HowMany-1985}}}]\label{lem:bamber-analytic-lemma}
Let $\mathcal{O}$ be a nonempty open-connected subset of $\mathbb{R}^a$, let $\varphi:\mathcal{O} \to \mathbb{R}^b$ be a smooth map, and denote $r_{\rm max} := \max_{y \in \mathcal{O}} \rank( D\varphi(y) )$ and
\begin{equation}
\mathcal{O}_{\rm gen} := \{ x \in \mathcal{O} \,|\, \rank( D\varphi(x) ) = r_{\rm \max} \}\,.
\end{equation}
If $\varphi$ is (real-)analytic, then the open subset $\mathcal{O}_{\rm gen}$ is of full Lebesgue measure inside $\mathcal{O}$.
\end{lemma}

\begin{lemma}\label{lem:bamber-analytic-lemma-generalized-to-analytic-Hilbert-space-codomain}
Let $\mathcal{O}$ be a nonempty open-connected subset of $\mathbb{R}^a$, let $\mathcal{H}$ be a (possibly infinite-dimensional) Hilbert space, let $\varphi:\mathcal{O} \to \mathcal{H}$ be a smooth map. Denote $r_{\rm max} := \max_{y \in \mathcal{O}} \rank( D\varphi(y) )$ and
\begin{equation}
\mathcal{O}_{\rm gen} := \{ x \in \mathcal{O} \,|\, \rank( D\varphi(x) ) = r_{\rm \max} \}\,.
\end{equation}
If $\varphi$ is (real-)analytic, then the open subset $\mathcal{O}_{\rm gen}$ is of full Lebesgue measure inside $\mathcal{O}$.
\begin{proof}
Fix an $x_0 \in \mathcal{O}_{\rm gen}$.
Let us consider the subspace $\mathcal{F} := \Im(D\varphi(x_0)) = \spa_\mathbb{R}(\{ D\varphi(x_0) e_1,\dots, D\varphi(x_0) e_a \})$, where $e_1,\dots,e_a$ is any basis of $\mathbb{R}^a$.
By definition of $x_0$, we have $\dim(\mathcal{F}) = r_{\rm max}$.
Since $\mathcal{F}$ is a finite-dimensional subspace of $\mathcal{H}$, it is closed and hence we can consider the orthogonal projection $P_\mathcal{F}:\mathcal{H} \to \mathcal{F}$ that projects onto it. This operator $P_\mathcal{F}$ is a bounded linear operator and hence is continuous.
Next, since $\{ D\varphi(x_0) e_j\}_{j=1}^a$ is a spanning family of $\mathcal{F}$, we can extract from it a basis of $\mathcal{F}$, which we denote $\{v_1,\dots,v_{r_{\rm max}}\} \subseteq \mathcal{F}$.
Consider now the linear map $L:\mathcal{F} \to \mathbb{R}^{r_{\rm max}}$ defined as $L(v_i) := (0\ 0\ \cdots\ 0 \ 1 \ 0 \ \cdots\ 0)^\intercal$ ($1$ in $i$th slot, $i=1,\dots,r_{\rm max}$), and extended by linearity.
Since $L$ is a linear map between finite-dimensional normed spaces, it is also continuous.
Lastly, define the map $f:\mathcal{O} \to \mathbb{R}^{r_{\rm max}}$ by $f := L \circ P_\mathcal{F} \circ \varphi$.
Since $\varphi$ is analytic, and since $P_\mathcal{F}$ and $L$ are so as well (being both continuous linear maps between vector spaces), the map $f$ is also analytic.

Introduce the notations
\begin{align*}
r_{\rm max}^{(\varphi)} &:= r_{\rm max}\,,\\
r_{\rm max}^{(f)} &:= \max_{y \in \mathcal{O}} \rank( D f(y) )\,,\\
\mathcal{O}_{\rm gen}^{(\varphi)} &:= \mathcal{O}_{\rm gen}\,,\\
\mathcal{O}_{\rm gen}^{(f)} &:= \{ x \in \mathcal{O} \,|\, \rank( D f(x) ) = r_{\rm max}^{(f)} \}\,.
\end{align*}

Note that we have
\begin{align}
D f (x_0) &=
D (L \circ P_\mathcal{F} \circ \varphi)\\
&= D(L \circ P_\mathcal{F}) \circ D\varphi(x_0)\\
&= (L \circ P_\mathcal{F}) \circ D\varphi(x_0)\\
&= L \circ D\varphi(x_0)\,,\label{eq:bamber-analytic-lemma-generalized-to-analytic-Hilbert-space-codomain--proof-cite1}
\end{align}
where the second equality is by the chain rule, the third is because $(L \circ P_\mathcal{F})$ is a linear map, and the last is because $P_\mathcal{F}$ is the projector onto $\Im(D\varphi(x_0))$.
Since $L$ is a linear isomorphism, it follows from \cref{eq:bamber-analytic-lemma-generalized-to-analytic-Hilbert-space-codomain--proof-cite1} that $\rank(D f (x_0)) = \rank(D\varphi(x_0)) = r_{\rm max}^{(\varphi)}$, and thus $r_{\rm max}^{(f)} \geq r_{\rm max}^{(\varphi)}$.
But for all $x \in \mathcal{O}$, we also have from the chain rule that
\begin{align}
\rank( D f (x) ) &=
\rank( D (L \circ P_\mathcal{F} \circ \varphi) ) \\
&= \rank( D(L \circ P_\mathcal{F}) \circ D\varphi(x) )\\
&\leq \rank( D\varphi(x) )\,.\label{eq:bamber-analytic-lemma-generalized-to-analytic-Hilbert-space-codomain--proof-cite2}
\end{align}
It follows from \cref{eq:bamber-analytic-lemma-generalized-to-analytic-Hilbert-space-codomain--proof-cite2} that $r_{\rm max}^{(f)} \leq r_{\rm max}^{(\varphi)}$, and hence we have established that
\begin{equation}\label{eq:bamber-analytic-lemma-generalized-to-analytic-Hilbert-space-codomain--proof-cite3}
r_{\rm max}^{(f)} = r_{\rm max}^{(\varphi)}\,.
\end{equation}
From \cref{eq:bamber-analytic-lemma-generalized-to-analytic-Hilbert-space-codomain--proof-cite2,eq:bamber-analytic-lemma-generalized-to-analytic-Hilbert-space-codomain--proof-cite3}, it also follows that
\begin{equation}\label{eq:bamber-analytic-lemma-generalized-to-analytic-Hilbert-space-codomain--proof-cite4}
\mathcal{O}_{\rm gen}^{(f)} \subseteq \mathcal{O}_{\rm gen}^{(\varphi)} \subseteq \mathcal{O}\,.
\end{equation}
By now applying \cref{lem:bamber-analytic-lemma} to the analytic map $f:\mathcal{O} \to \mathbb{R}^{r_{\rm max}}$, we obtain that $\mathcal{O}_{\rm gen}^{(f)}$ is of full Lebesgue measure inside $\mathcal{O}$. But by \cref{eq:bamber-analytic-lemma-generalized-to-analytic-Hilbert-space-codomain--proof-cite4}, this implies that $\mathcal{O}_{\rm gen}^{(\varphi)}$ is also of full Lebesgue measure inside $\mathcal{O}$, which establishes the lemma. (The claim that $\mathcal{O}_{\rm gen}^{(f)}$ is open is a consequence of the lower semicontinuity of the rank.)
\end{proof}
\end{lemma}

\begin{lemma}[Locally generic integer-valued function on a connected space is globally generic]\label{lem:local-generic-on-connected-implies-global-generic}
Let $X$ be a (nonempty) topological space.
Let $\mu$ be a strictly positive measure on $X$ (i.e., a measure defined on the Borel $\sigma$-algebra of $X$ and that assigns positive measure to all nonempty open sets of $X$).
Let $(A_i)_{i \in I}$ be an (at most) countable family of (nonempty) open sets of $X$ that cover $X$ ($X = \cup_{i \in I} A_i$).
Let $f : X \to \{0,1,\dots,N\}$ be a bounded (for some $N\geq0$) integer-valued function that is \textit{lower-semicontinuous}, and denote $r_{\rm max} := \max_{x \in X} f(x)$.
Let $(F_i)_{i \in I}$ be open sets of $X$ such that for all $i \in I$: (i) $F_i \subseteq A_i$, (ii) $F_i$ is of full measure inside $A_i$ [$\mu(A_i \setminus F_i) = 0$], and (iii) $\left.f\right|_{F_i}$ is constant.
Denote their union by $F := \cup_{i \in I} F_i$.
The following holds:
\begin{enumerate}[label=(\arabic*)]
    \item $F$ is an open set of full measure inside $X$.
    \item If $X$ is \textit{connected}, then $\left.f\right|_{F}$ is constant and equal to $r_{\rm max}$.
\end{enumerate}

\begin{proof}
First, $F$ is open, since it is a union of open sets. Let us show that it has full measure inside $X$, i.e., that
\begin{equation}\label{eq:countable-partition-and-full-measure-subsets-in-each-proof-goal}
\mu\Big( X \setminus \bigcup_{i \in I} F_i \Big) = 0\,.
\end{equation}
First, if $x \in X \setminus \bigcup_{i \in I} F_i$, there exists an $i \in I$ such that $x \in A_i$ (since the $A_i$'s cover $X$), and since $x \notin F_i$, we have $x \in A_i \setminus F_i$; thus we have the inclusion of sets
\begin{equation}\label{eq:countable-partition-and-full-measure-subsets-in-each-proof-setinclusion}
X \setminus \bigcup_{i \in I} F_i \subseteq \bigcup_{i \in I} (A_i \setminus F_i)\,.
\end{equation}
Hence, we get
\begin{align}
\mu\Big( X \setminus \bigcup_{i \in I} F_i \Big)
&\leq \mu\Big( \bigcup_{i \in I} (A_i \setminus F_i) \Big)\\
&\leq \sum_{i \in I} \mu( A_i \setminus F_i )\\
&= 0\,,
\end{align}
where the first inequality is by the set inclusion \cref{eq:countable-partition-and-full-measure-subsets-in-each-proof-setinclusion}, the second inequality is by countable subadditivity of measures, and the last equality is because, by assumption, each $F_i$ is of full measure inside $A_i$.
This establishes \cref{eq:countable-partition-and-full-measure-subsets-in-each-proof-goal} (since measures are non-negative).

Next, we move on to proving the second claim. For each $i \in I$, let us denote $r_i := \max_{x \in A_i} f(x)$. Note that the constant value that $f$ takes on $F_i$ (by assumption) must in fact be $r_i$. Indeed, for all $x \in F_i$, we have $f(x) \leq r_i$ (since the latter is a maximum over all $A_i$), and if we had $f(x) < r_i$ for all $x \in F_i$, then letting $G_i:= \{ x \in A_i \,|\, f(x) = r_i \}$ (which is open by lower-semicontinuity of $f$), we would have that $F_i$ and $G_i$ are two \textit{disjoint} nonempty subsets of $A_i$ with the first being full measure and the second being open, which is absurd because it would successively imply
\begin{align}
0
&< \mu( G_i )\\
&= \mu( G_i \setminus F_i )\\
&\leq \mu( A_i \setminus F_i )\\
&= 0\,.
\end{align}

Consider now two arbitrary indices $i_a,i_b \in I$. We claim that
\begin{equation}\label{local-generic-on-connected-implies-global-generic--intersection-implies-equal-maximum-claim}
A_{i_a} \cap A_{i_b} \neq \emptyset \implies r_{i_a} = r_{i_b}\,.
\end{equation}
To show this, let $S := A_{i_a} \cap A_{i_b}$ denote such a nonempty, open set, and suppose by contradiction that $r_{i_a} \neq r_{i_b}$, say $r_{i_a} < r_{i_b}$ (the other case would be treated symmetrically).
Note that we must have that $F_{i_b}$ intersects with $S$, since, as argued above, we cannot have two disjoint nonempty subsets (here of $A_{i_b}$) where one is full measure ($F_{i_b}$) and the other one is open ($S$).
Taking now an element $x \in F_{i_b} \cap S$, we obtain both that $f(x) = r_{i_b} > r_{i_a}$ (since $x \in F_{i_b}$), and $f(x) \leq r_{i_a}$ (since $x \in A_{i_a} \supseteq S$), which is a contradiction, thus establishing \cref{local-generic-on-connected-implies-global-generic--intersection-implies-equal-maximum-claim}.

Suppose now that $X$ is connected. We claim that in this case, all the values $r_i$ ($i \in I$) must be equal to a certain value $r$.
Indeed, let $i_a$ and $i_b$ be two arbitrary indices in $I$ ($A_{i_a}$ and $A_{i_b}$ may be disjoint). By the connectedness of $X$, the intersection graph of its cover $(A_i)_{i \in I}$ is connected, which implies that there exists a finite "path" of open sets $A_i$ going from $A_{i_a}$ to $A_{i_b}$ and maintaining nonempty intersections between consecutive sets; explicitly, there exists a finite number of indices $i_1,\dots,i_n \in I$ such that $i_1 = i_a$, $i_n = i_b$, and $A_{i_j} \cap A_{i_{j+1}} \neq \emptyset$ for all $j=1,\dots,n-1$. Applying \cref{local-generic-on-connected-implies-global-generic--intersection-implies-equal-maximum-claim} to these $n-1$ consecutive intersecting pairs $(A_{i_j},A_{i_{j+1}})$ yields
\begin{align}
r_{i_a} = r_{i_1} = r_{i_2} = \dots = r_{i_n} = r_{i_b}\,,
\end{align}
which shows that all the values $r_i$ are equal to a common value, which we denote by $r$. We have thus at this stage shown that $\left.f\right|_{F}$ is constant and equal to $r$.

It remains to show that $r = r_{\rm max}$. By contradiction, if we had $r < r_{\rm max}$, then analogously as above, by
letting $G:= \{ x \in X \,|\, f(x) = r_{\rm max} \}$, we would have $F$ be full measure inside $X$, $G$ be nonempty and open (by lower-semicontinuity of $f$), and $F$ and $G$ be disjoint, which (as we saw) is absurd.
This establishes the second claim of the lemma.

\end{proof}
\end{lemma}

\begin{lemma}\label{lem:analytic-maps-from-euclidean-domain-to-analytic-Hilbert-manifolds}
Let $\mathcal{O}$ be a nonempty open-connected subset of $\mathbb{R}^a$, and let $M$ be a (possibly infinite-dimensional) analytic Hilbert manifold. Let $\varphi:\mathcal{O} \to M$ be a smooth map. Denote $r_{\rm max} := \max_{y \in \mathcal{O}} \rank( D\varphi(y) )$ and
\begin{equation}
\mathcal{O}_{\rm gen} := \{ x \in \mathcal{O} \,|\, \rank( D\varphi(x) ) = r_{\rm \max} \}\,.
\end{equation}
If $\varphi$ is (real-)analytic, then $\mathcal{O}_{\rm gen}$ is an open subset of $\mathcal{O}$ of full Lebesgue measure.
\begin{proof}
Denote by $\mathcal{H}$ the Hilbert space on which $M$ is modeled.
Since $\varphi$ is an analytic map, it gives us the following property: for any $x\in \mathcal{O}$, there exists an open-connected neighborhood $U_x$ of $x$, and a local chart $(V_x,\alpha_x)$ around $\varphi(x)$ such that $\varphi(U_x) \subseteq V_x$ and such that the map $\alpha_x \circ \left.\varphi\right|_{U_x}: U_x \to \mathcal{H}$ is analytic. (The connectedness of $U_x$ can indeed always be assumed, by possibly taking smaller neighborhoods.)

Now consider the open cover $\mathcal{O} = \cup_{x \in \mathcal{O}} U_x$ of $\mathcal{O}$. Since in $\mathbb{R}^a$ every open cover admits an (at most) countable subcover (the Lindelöf property), there exists an (at most) countable subcover of the above cover of $\mathcal{O}$.
That is, there exists a countable family $(x_i)_{i \in I}$ of points of $\mathcal{O}$ such that $\mathcal{O} = \cup_{i \in I} U_{x_i}$.
For every $i$, the map $\alpha_{x_i} \circ \left.\varphi\right|_{U_{x_i}}: U_{x_i} \to \mathcal{H}$ satisfies the assumptions of \cref{lem:bamber-analytic-lemma-generalized-to-analytic-Hilbert-space-codomain}, which yields that $U_{x_i,\rm gen}$ is an open subset of full Lebesgue measure inside $U_{x_i}$.
We are now exactly in the setting of \cref{lem:local-generic-on-connected-implies-global-generic}, where
on the connected space $\mathcal{O}$ we have a function $f(x) := \rank( D\varphi(x) )$ that is lower-semicontinuous and constant on full-Lebesgue-measure subsets $U_{x_i,\rm gen}$ of every open set $U_{x_i}$ of an (at most) countable cover of $\mathcal{O}$ (and the Lebesgue measure from $\mathbb{R}^a$ is indeed strictly positive on $\mathbb{R}^a$, and so is its restriction to the open set $\mathcal{O}$). This lemma thus yields that $\mathcal{U}:=\cup_{i \in I} U_{x_i,\rm gen}$ is an open set of full Lebesgue measure inside $\mathcal{O}$, and that on it, the function $f$ is constant and equal to $r_{\rm max}$, which gives us that
\begin{equation}\label{eq:analytic-maps-from-euclidean-domain-to-analytic-Hilbert-manifolds--proof-obtained-set-inclusion}
\mathcal{U} \subseteq \mathcal{O}_{\rm gen}\,.
\end{equation}
Note that $\mathcal{O}_{\rm gen}$ is open due to the lower-semicontinuity of $f$. Since $\mathcal{U}$ is of full Lebesgue measure inside $\mathcal{O}$, it follows from \cref{eq:analytic-maps-from-euclidean-domain-to-analytic-Hilbert-manifolds--proof-obtained-set-inclusion} that $\mathcal{O}_{\rm gen}$ is also of full Lebesgue measure inside $\mathcal{O}$.

\end{proof}
\end{lemma}

We can finally establish the desired technical proposition of this section:
\begin{proposition}\label{prop:rank-of-analyticcomp-of-cstrank-and-nonsmooth}
Let $\mathcal{O}$ be a nonempty open-connected subset of $\mathbb{R}^a$, and let $M_2$ and $M_3$ be two (possibly infinite-dimensional) analytic Hilbert manifolds.
Let $f:\mathcal{O} \to M_2$ and $g:M_2 \to M_3$ be maps such that $f$ is smooth and the composition $g \circ f$ is smooth (but $g$ need not be smooth).
Denote
\begin{align}
r_{\rm max}^{(f)} &:= \max_{y \in \mathcal{O}} \rank( D f(y) )\,,\\
r_{\rm max}^{(g \circ f)} &:= \max_{y \in \mathcal{O}} \rank( D (g \circ f)(y) )\,,
\end{align}
and
\begin{equation}
\mathcal{O}_{\rm gen}^{(f)} := \{ x \in \mathcal{O} \,|\, \rank( D f(x) ) = r_{\rm \max}^{(f)} \}\,.
\end{equation}
The following holds:
\begin{itemize}
\item[(i)] If $f$ is (real-)analytic, then
\begin{equation}\label{eq:rank-of-analyticcomp-of-cstrank-and-nonsmooth-claim-i}
\forall x \in \mathcal{O}_{\rm gen}^{(f)}, \qquad \rank(D(g \circ f)(x)) \leq r_{\rm max}^{(f)}\,.
\end{equation}
\item[(ii)] If both $f$ and $(g \circ f)$ are (real-)analytic, then it even holds that
\begin{equation}\label{eq:rank-of-analyticcomp-of-cstrank-and-nonsmooth-claim-ii}
r_{\rm max}^{(g \circ f)} \leq r_{\rm max}^{(f)}\,.
\end{equation}
\end{itemize}
\begin{proof}
Since $f$ is analytic, by \cref{lem:analytic-maps-from-euclidean-domain-to-analytic-Hilbert-manifolds} the open subset $\mathcal{O}_{\rm gen}^{(f)}$ is of full Lebesgue measure inside $\mathcal{O}$, and so $\left.f\right|_{\mathcal{O}_{\rm gen}^{(f)}}: \mathcal{O}_{\rm gen}^{(f)} \to M_2$ is a smooth map of constant rank $r_{\rm max}^{(f)}$.
Since $g\circ f$ is smooth, its restriction to $\mathcal{O}_{\rm gen}^{(f)}$ is also smooth. Hence, we can apply \cref{lem:rank-of-smoothcomp-of-cstrank-and-nonsmooth} to the maps $\left.f\right|_{\mathcal{O}_{\rm gen}^{(f)}}$ and $g$, which exactly yields the claim of \cref{eq:rank-of-analyticcomp-of-cstrank-and-nonsmooth-claim-i}.

For (ii), if $g \circ f$ is also analytic, then we can now apply \cref{lem:analytic-maps-from-euclidean-domain-to-analytic-Hilbert-manifolds} again, this time to $g \circ f$. It yields that the open subset
\begin{equation}
\mathcal{O}_{\rm gen}^{(g \circ f)} := \{ x \in \mathcal{O} \,|\, \rank( D (g \circ f)(x) ) = r_{\rm \max}^{(g \circ f)} \}
\end{equation}
is of full Lebesgue measure inside $\mathcal{O}$, and so $\left.(g \circ f)\right|_{\mathcal{O}_{\rm gen}^{(g \circ f)}}: \mathcal{O}_{\rm gen}^{(g \circ f)} \to M_3$ is of constant rank $r_{\rm max}^{(g \circ f)}$.

As $\mathcal{O}_{\rm gen}^{(f)}$ and $\mathcal{O}_{\rm gen}^{(g \circ f)}$ are both full-measure subsets of $\mathcal{O}$, their intersection $\mathcal{S} := \mathcal{O}_{\rm gen}^{(f)} \cap \mathcal{O}_{\rm gen}^{(g \circ f)}$ is again a full-measure subset of $\mathcal{O}$.
Since $\mathcal{S} \subseteq \mathcal{O}_{\rm gen}^{(g \circ f)}$, the map $(g \circ f)$ is still of constant rank on $\mathcal{S}$; denote this rank by $r_0$. If $r_0$ were different from $r_{\rm max}^{(g \circ f)}$, then the subsets $\mathcal{O}_{\rm gen}^{(g \circ f)}$ and $\mathcal{S}$ would be disjoint, i.e., $\mathcal{S} \subseteq \mathcal{O} \setminus \mathcal{O}_{\rm gen}^{(g \circ f)}$, which would imply ($\mu$ denoting the Lebesgue measure on $\mathcal{O}$)
\begin{align}
\mu( \mathcal{O} )
&= \mu( \mathcal{S} ) + \mu( \mathcal{O} \setminus \mathcal{S} )\\
&\leq \mu\big( \mathcal{O} \setminus \mathcal{O}_{\rm gen}^{(g \circ f)} \big) + \mu( \mathcal{O} \setminus \mathcal{S} )\\
&= 0\,,
\end{align}
where the last equality is because $\mathcal{O}_{\rm gen}^{(g \circ f)}$ and $\mathcal{S}$ are both of full measure inside $\mathcal{O}$; this is a contradiction since $\mu( \mathcal{O} )>0$. Hence, $r_0 = r_{\rm max}^{(g \circ f)}$.

We can now conclude by letting $x$ be any point in the nonempty $\mathcal{S}$. Since on one hand $x \in \mathcal{O}_{\rm gen}^{(g \circ f)}$, we have
\begin{equation}
\rank(D(g \circ f)(x)) = r_{\rm max}^{(g \circ f)}\,,
\end{equation}
but since, on the other hand, $x \in \mathcal{O}_{\rm gen}^{(f)}$, we have by \cref{eq:rank-of-analyticcomp-of-cstrank-and-nonsmooth-claim-i} that
\begin{equation}
\rank(D(g \circ f)(x)) \leq r_{\rm max}^{(f)}\,.
\end{equation}
Combining the two gives \cref{eq:rank-of-analyticcomp-of-cstrank-and-nonsmooth-claim-ii}.
\end{proof}
\end{proposition}

\begin{theorem}[Number of directions explored by bosonic variational quantum circuits bounded by orbit dimension]\label{thm:VQC-bounds}
Consider an $m$-mode unitary variational quantum circuit (VQC) of the form
\begin{equation}\label{eq:VQC-bounds-unitary-VQC-ansatz-def}
U(\theta_1,\dots,\theta_p) = W_{p} e^{-i \theta_p H_{p}} W_{p-1} \cdots W_{1} e^{-i \theta_1 H_{1}} W_{0}\,,
\end{equation}
where $H_1,\dots,H_p$ are polynomial Hamiltonians of degree $\leq 2$ in the canonical operators, and $W_0,\dots,W_p \in G_{\rm GO}$ are fixed $m$-mode Gaussian unitaries.

Let $\Theta$ be a nonempty open-connected subset of $\mathbb{R}^p$.
Let $\ket{\psi}$ and $\rho$ be fixed initial states in $\fock$ and $B_2(\fock)$, respectively.

Consider the associated variational output state maps $\ket{\psi}_{\rm out}: \Theta \to \fock$ and $\rho_{\rm out}: \Theta \to B_2(\fock)$, i.e., the maps $\ket{\psi}_{\rm out}(\bm{\theta}) := U(\bm{\theta})\ket{\psi}$ and $\rho_{\rm out}(\bm{\theta}) := U(\bm{\theta}) \rho U(\bm{\theta})^\dagger$.
Denote
\begin{align}
r_{\rm max}^{(\ket{\psi}_{\rm out})} &:= \max_{\bm{\theta} \in \Theta} \rank( D \ket{\psi}_{\rm out}(\bm{\theta}) )\,,\\
r_{\rm max}^{(\rho_{\rm out})} &:= \max_{\bm{\theta} \in \Theta} \rank( D \rho_{\rm out}(\bm{\theta}) )\,.
\end{align}
Let $G_{\fock}$ be the subgroup of $\uni(\fock)$ generated by the unitaries $W_0,\dots,W_p$ and by the exponentials $e^{-i \theta H_k}$ for all $\theta \in \mathbb{R}$ and $k=1,\dots,p$, denote by $\phi$ the $m$-mode extended metaplectic representation, and let $G := \phi^{-1}(G_{\fock})$ be the associated subgroup of the $m$-mode extended metaplectic group.
The following holds:
\begin{enumerate}[label=(\arabic*)]
\item If $\ket{\psi}$ is an exponential-decay state, then
\begin{equation}\label{eq:VQC-bounds-goal-ket}
r_{\rm max}^{(\ket{\psi}_{\rm out})} \leq \dim(\orb_{\bar{G}}(\ket{\psi}))\,.
\end{equation}
\item If $\rho$ is an exponential-decay operator, then
\begin{equation}
r_{\rm max}^{(\rho_{\rm out})} \leq \dim(\orb_{\bar{G}}(\rho))\,.
\end{equation}
\end{enumerate}
In the above, ${\bar{G}}$ denotes the closure of $G$ (inside the extended metaplectic group), $\orb_{\bar{G}}(\ket{\psi})$ and $\orb_{\bar{G}}(\rho)$ denote the orbits of $\ket{\psi}$ and $\rho$ under the representations $\left.\phi\right|_{{\bar{G}}} :{\bar{G}}\to \uni(\fock)$ and $\operatorname{Ad} \circ \left.\phi\right|_{{\bar{G}}} :{\bar{G}}\to \uni(B_2(\fock))$ of ${\bar{G}}$, and their dimensions are in the sense of their smooth manifold structure provided by \cref{thm:SC-repr-orbit-structure}.

Lastly, for all four groups $G_{\fock}$ considered in the main text (cf. \cref{tab:Lie-algebra-bases}), we have $\bar{G} = G$ (by \cref{lem:optical-groups-are-images-of-closed-subgroups-under-EMR}).

\begin{proof}
First, let us be precise about the definition of the group $G_{\fock}$ that we introduced in the theorem's statement: it is the smallest subgroup of $\uni(\fock)$ that contains the unitaries $W_0,\dots,W_p$, and $e^{-i \theta H_k}$ for all $\theta \in \mathbb{R}$, $k=1,\dots,p$.
Recall that $G_0$ denotes the $m$-mode extended metaplectic Lie group, $\mathfrak{g}_0$ its Lie algebra, and $\phi: G_0 \to \uni(\fock)$ the $m$-mode extended metaplectic representation.
We then define, as in the statement of the theorem, the preimage:
\begin{equation}\label{eq:VQC-bounds--proof-cite0}
G := \phi^{-1}(G_{\fock})\,.
\end{equation}
From the fact that $\phi$ is a representation and that $G_{\fock}$ is a subgroup, it follows that $G$ is a subgroup of $G_0$. Note that the subgroup $G$ may not be closed in $G_0$, and hence may not be a Lie subgroup of $G_0$.
Consider now the closure $\bar{G}$ of $G$ (inside $G_0$). Note that since $G_0$ is a Lie group, hence in particular a topological group, and since the closure of a subgroup in a topological group is still a subgroup (see, e.g., Ref.~\cite[Chap. 3, Prop. 6]{Husain-IntroductionTopological-2018}), it holds that $\bar{G}$ is still a subgroup of $G_0$.
Therefore, $\bar{G}$ is a closed subgroup of the Lie group $G_0$, hence, by the closed subgroup theorem, it follows (as already used in the proof of \cref{thm:SC-repr-orbit-structure}) that $\bar{G}$ is a Lie subgroup of $G_0$ (and thus a Lie group itself) and that its Lie algebra can be identified as
\begin{equation}\label{eq:VQC-bounds--proof-cite1}
\mathfrak{g} = \{ X \in \mathfrak{g}_0 \,|\, \forall t \in \mathbb{R}\ e^{tX} \in \bar{G} \}\,.
\end{equation}

Consider the unitary ansatz $U(\bm{\theta})$ of \cref{eq:VQC-bounds-unitary-VQC-ansatz-def}.
First, note that as $\left.\phi\right|_{G} :G\to G_{\fock}$ is surjective [by \cref{eq:VQC-bounds--proof-cite0}], each $W_k \in G_{\fock}$ writes as $W_k = \phi(S_k)$ for some $S_k \in G \subseteq \bar{G}$. Second, since each $-i H_k$ is a skew-Hermitian polynomial operator of degree $\leq 2$, and since we know that the derived representation is surjective onto such operators (cf. \cref{sec:SM-EMRep-and-proof-of-orbit-structure}), we have
\begin{equation}\label{eq:VQC-bounds--proof-cite2}
-i H_k = \phi'(X_k)
\end{equation}
for some $X_k \in \mathfrak{g}_0$. But since $\phi(e^{\theta_k X_k}) = e^{\theta_k \phi'(X_k)} =e^{- i \theta_k H_k}$ [using \cref{eq:SC-repr-def-of-its-derivative-exponentiated-form,eq:VQC-bounds--proof-cite2}, respectively], and since the right-hand side is by definition an element of $G_{\fock}$, we get that $\phi(e^{\theta_k X_k}) \in G_{\fock}$, and hence [by \cref{eq:VQC-bounds--proof-cite0}] that $e^{\theta_k X_k} \in G \subseteq \bar{G}$ for all $\theta_k \in \mathbb{R}$. This implies [due to \cref{eq:VQC-bounds--proof-cite1}] that $X_k \in \mathfrak{g}$.
Therefore, we can rewrite the parametrized unitary as
\begin{align}
U(\bm{\theta})
&= W_{p} e^{-i \theta_p H_{p}} W_{p-1} \cdots W_{1} e^{-i \theta_1 H_{1}} W_{0}\nonumber\\
&= \phi(S_{p}) e^{\theta_p \phi'(X_{p})} \phi(S_{p-1}) \cdots \phi(S_{1}) e^{\theta_1 \phi'(X_{1})} \phi(S_{0})\nonumber\\
&= \phi(S_{p}) \phi(e^{\theta_p X_{p}}) \phi(S_{p-1}) \cdots \phi(S_{1}) \phi(e^{\theta_1 X_{1}}) \phi(S_{0})\nonumber\\
&= \phi( S_{p} e^{\theta_p X_{p}} S_{p-1} \cdots S_{1} e^{\theta_1 X_{1}} S_{0} )\nonumber\\
&= \phi( T(\bm{\theta}) )\,,
\end{align}
with
\begin{align}
T(\bm{\theta})
&:= S_{p} e^{\theta_p X_{p}} S_{p-1} \cdots S_{1} e^{\theta_1 X_{1}} S_{0}\,,
\end{align}
where the fourth equality is due to $\phi$ being a representation. And as the $X_k$'s and $S_k$'s in the above expression are all elements of $\mathfrak{g}$ and $\bar{G}$, respectively, the map $T$ takes values inside $\bar{G}$.
Because it is a composition of products and exponentials (which are both analytic operations on Lie groups), this map $T: \Theta \to \bar{G}$ is analytic.

Now, consider the case of the ket picture. Note again (as in the proof of \cref{thm:orbit-dim-maintext} at the end of \cref{sec:SM-EMRep-and-proof-of-orbit-structure}) that
the representation $\left.\phi\right|_{\bar{G}} :\bar{G}\to \uni(\fock)$ of the Lie group $\bar{G}$ is strongly continuous (since it is a restriction of the strongly continuous map $\phi$);
and that
since $\ket{\psi}$ is an exponential-decay state, it is in particular a Schwartz state, hence a smooth vector for the representation $\phi:G_0\to \uni(\fock)$, and thus also a smooth vector for its restriction to the embedded subgroup $\bar{G}$ of $G_0$.
We can therefore apply \cref{thm:SC-repr-orbit-structure} to this representation $\left.\phi\right|_{\bar{G}}$ and to the point $\ket{\psi} \in \fock$, which gives us a smooth manifold structure on the orbit $\orb_{\bar{G}}(\ket{\psi})$.
We recall, as explained in the proof of \cref{thm:SC-repr-orbit-structure}, this smooth manifold structure considered on $\orb_{\bar{G}}(\ket{\psi})$ is the unique one such that the map $F_{\ket{\psi}}: {\bar{G}}/{\bar{G}}_{\ket{\psi}} \to \orb_{\bar{G}}(\ket{\psi})$ is a diffeomorphism, and the map $\pi: {\bar{G}} \to {\bar{G}}/{\bar{G}}_{\ket{\psi}}$ is a smooth submersion. In fact, there exists also a unique analytic structure on the smooth manifold ${\bar{G}}/{\bar{G}}_{\ket{\psi}}$ that makes the map $\pi$ analytic (see, e.g., Ref.~\cite[Chap. 1]{Varadarajan-LieGroups-1984}).

We now decompose the map $\ket{\psi}_{\rm out}$ in two ways (recall also that $f_{\ket{\psi}}: \bar{G} \to \fock$ denotes the orbit map associated with $\ket{\psi}$, i.e., $f_{\ket{\psi}}(S) := \phi(S)\ket{\psi}$):
\begin{equation}\label{eq:VQC-bounds-proof-psiout-decomposition-1}
\ket{\psi}_{\rm out} = f_{\ket{\psi}} \circ T
\end{equation}
and
\begin{equation}\label{eq:VQC-bounds-proof-psiout-decomposition-2}
\ket{\psi}_{\rm out} = i \circ F_{\ket{\psi}} \circ \pi \circ T\,,
\end{equation}
where $i: \orb_{\bar{G}}(\ket{\psi}) \to \fock$ is the inclusion map,
seen as a map whose domain $\orb_{\bar{G}}(\ket{\psi})$ is considered with its smooth manifold structure considered above (from \cref{thm:SC-repr-orbit-structure}), and its codomain $\fock$ is considered with its Hilbert space structure.
In other words, $\ket{\psi}_{\rm out}$ can be decomposed as the map $T$, followed by either the top or the bottom route of the diagram \labelcref{diag:orbit-map-decomposition} of \cref{sec:SM-EMRep-and-proof-of-orbit-structure}.

Since $\ket{\psi}$ is an analytic vector for $\phi$ (because it is an exponential-decay state, cf. \cref{sec:SM-EMRep-and-proof-of-orbit-structure}), i.e., the map $f_{\ket{\psi}}$ is analytic, and since $T$ is analytic, it follows from \cref{eq:VQC-bounds-proof-psiout-decomposition-1} that the map $\ket{\psi}_{\rm out}$ is analytic (by composition).

Thus, using then \cref{eq:VQC-bounds-proof-psiout-decomposition-2}, $\ket{\psi}_{\rm out}$ is written in the form of an analytic map "$g \circ f$" with "$f$" the analytic map $\pi \circ T$, and "$g$" the not necessarily smooth map $i \circ F_{\ket{\psi}}$. \Cref{prop:rank-of-analyticcomp-of-cstrank-and-nonsmooth}, case (ii) hence applies and yields that
\begin{equation}\label{eq:VQC-bounds-proof-obtained-ineq}
r_{\rm max}^{(\ket{\psi}_{\rm out})} \leq r_{\rm max}^{(\pi \circ T)}\,.
\end{equation}

But the rank of the map $\pi \circ T$ at any point is upper-bounded by the dimension of its codomain, i.e., by $\dim({\bar{G}}/{\bar{G}}_{\ket{\psi}}) = \dim(\orb_{\bar{G}}(\ket{\psi}))$. Combining with \cref{eq:VQC-bounds-proof-obtained-ineq} gives the desired \cref{eq:VQC-bounds-goal-ket}.

The density operator case is treated analogously.

\end{proof}
\end{theorem}

\paragraph{Proof of \cref{thm:VQC-orbit-dim-maintextinformal}.}

The above \cref{thm:VQC-bounds} formalizes and generalizes the content of \cref{thm:VQC-orbit-dim-maintextinformal} in the main text.

It remains to address the point in \cref{thm:VQC-orbit-dim-maintextinformal} about the relation with the quantum Fisher information matrix (QFIM).
Indeed, for pure states, the rank of the QFIM associated with a parametrized state family $\ket{\psi(\bm{\theta})}$ is equal to the rank of the Jacobian of the map $\theta \mapsto \ketbra{\psi(\bm{\theta})}$ (see e.g., the proof of Ref.~\cite[Lemma D2]{Monbroussou-TrainabilityExpressivity-2025} for more background).$\qed$

\section{The Schwartz spaces and their topologies, and proof of \mbox{Proposition \ref{prop:lsc-orbit-dim}}}\label{sec:SM-Schwartz-spaces-and-proof-of-lsc}

Let $D(\fock) \subset B_1(\fock)$ denote the sets of density and trace-class operators, respectively.
Technically, not all normalized elements $\ket{\psi} \in \fock$ are physically meaningful, as some can have infinite average or higher moments of energy (total photon number). The space of physical states is therefore taken to be the \textit{Schwartz} space, the dense subspace $\mathcal{S}(\fock) \subset \fock$ of states $\ket{\psi}$ whose coefficients in the Fock basis $\braket{\bm{n}}{\psi}$ decay faster than any polynomial in $\bm{n}$ \cite{hall_quantum_2013}. The dense subspace of \textit{Schwartz operators} \cite{keyl_schwartz_2016} $\schop \subset B_1(\fock)$ provides the analogous concept at the operator level, and the space of Schwartz density operators $\mathcal{D}(\fock) \subset D(\fock)$ then coincides with the space of states $\rho$ whose Wigner functions decay faster than any polynomial in phase space \cite{hernandez_rapidly_2022}.
In these Schwartz spaces, with the topologies induced by the trace-norm, $\epsilon$-neighborhoods of Schwartz states still contain other Schwartz states of arbitrarily high energy. It can therefore be more appropriate (as in \cref{prop:lsc-orbit-dim}) to equip $\mathcal{S}(\fock)$ and $\mathcal{D}(\fock)$ with their \textit{Schwartz topologies}, for which this closeness does not hold.

In this section, we detail these concepts further, and then provide a proof for \cref{prop:lsc-orbit-dim}.

\subsection{Schwartz states}

For $\bm{\alpha} \in \NN^m$, we use the shorthand notations $(\bm{N}+\id)^{\bm{\alpha}} := N_1^{\alpha_1}\cdots N_m^{\alpha_m} + \id$ and $(\bm{n}+1)^{\bm{\alpha}} := n_1^{\alpha_1}\cdots n_m^{\alpha_m} + 1$.

The \textit{Schwartz space} $\mathcal{S}(\fock) \subset \fock$ can be defined as the space of all elements $\ket{\psi} \in \fock$ whose coefficients in the Fock basis $\braket{\bm{n}}{\psi}$ decay faster than any polynomial in $\bm{n}$.
Equivalently, one can see that it is the space of all elements $\ket{\psi} \in \fock$ such that the quantities $\norm{\ket{\psi}}_{\bm{\alpha}} := \norm{(\bm{N}+\id)^{\bm{\alpha}}\ket{\psi}}$ are finite for all $\bm{\alpha} \in \NN^m$, since
\begin{align}
\norm{\ket{\psi}}_{\bm{\alpha}} = \Bigg(\sum_{\bm{n} \in \NN^m} \big[(\bm{n}+1)^{\bm{\alpha}}\big]^2 \abs{\braket{\bm{n}}{\psi}}^2\Bigg)^{1/2}\,.
\end{align}
Note that the addition of $\id$ can just be thought of as a useful convention.
The Schwartz space is a vector space (a subspace of $\fock$), and it is in fact dense in $\fock$.
On Schwartz space, the quantities $\norm{\cdot}_{\bm{\alpha}}$ are \textit{seminorms}, i.e., they are non-negative, homogeneous, and they satisfy the triangle inequality (in fact, they are also norms as they are positive-definite, meaning that $\norm{\ket{\psi}}_{\bm{\alpha}} = 0$ implies $\ket{\psi}=0$, but this property will not be needed and so it is customary to still refer to them as seminorms).

Given a vector space $X$, a family of seminorms $(\norm{\cdot}_{a})_{a \in A}$ on it induces a topology on $X$, defined as the coarsest topology that makes all the seminorms continuous.
In other words, a nonempty subset $U \subseteq X$ is open in this topology if and only if for all $x\in U$, $\exists \epsilon > 0,\ n \in \NN,\ a_1,\dots,a_n \in A$ s.t. $\cap_{i=1}^n\{ x' \in X \,|\, \norm{x' - x}_{a_i} < \epsilon \} \subseteq U$.
Equipped with such a topology, $X$ is called a \textit{locally convex topological vector space} (LCTVS). Note that a normed vector space (with its standard topology) is a special case of this, where the family of seminorms consists of a single norm.
Let us summarize how several notions of convergence are characterized in such an LCTVS. A sequence $(x_n)_{n \in \NN}$ in $X$ converges to a point $x \in X$ if and only if for all $n \in \NN,\ a_1,\dots,a_n \in A,\ \epsilon_1,\dots,\epsilon_n > 0,\ \exists\,N\in\NN\text{ s.t. } \forall n\geq N\ \norm{x_n - x}_{a_i} < \epsilon_i \text{ for all }i=1\dots,n$.
Furthermore, given now two LCTVSs $X$ and $Y$ whose topologies are induced from the families of seminorms $(\norm{\cdot}_{a})_{a \in A}$ and $(\norm{\cdot}_{b})_{b \in B}$, respectively, a \textit{linear} map $f:X \to Y$ is continuous if and only if (see, e.g., Ref.~\cite[Statement III.1.1]{Schaefer-TopologicalVector-1999})
\begin{align}
\begin{aligned}\label{eq:linear-map-between-LCTVSs-continuity-criterion}
\forall b \in B \ \exists\, n \in \NN,\ a_1,\dots,a_n \in A,\ C>0\\
\text{ s.t. } \forall x \in X \ \ \, \norm{f(x)}_{b} \leq C \max_{i=1,\dots,n}\norm{x}_{a_i} \,.
\end{aligned}
\end{align}

The \textit{Schwartz topology} on Schwartz space can be defined in the above way, as the topology induced by the seminorms $\norm{\cdot}_{\bm{\alpha}}$ (e.g., \cite[p.~6]{keyl_schwartz_2016}). This topology is \textit{finer} than the standard topology induced by the Hilbert-space norm, meaning that it is harder for sequences of Schwartz states to converge under the Schwartz topology than under the standard topology. To illustrate, consider the sequence of states $\ket{\psi_n} = \sqrt{1 - \frac{1}{n}}\ket{0} + \frac{1}{\sqrt{n}}\ket{n}$ in the single-mode Fock space. This is a sequence of Schwartz states (each $\ket{\psi_n}$ has finite support over the Fock basis) that converges to $\ket{0}$ (still a Schwartz state) in the standard topology (since $\norm{\ket{\psi_n} - \ket{0}} \underset{n\to\infty}{\longrightarrow} 0$); however, it does not converge to $\ket{0}$ in the Schwartz topology (because $\norm{\ket{\psi_n} - \ket{0}}_{\alpha}$ diverges for any $\alpha\geq2$).

Note that in the hard-energy-cutoff subspace $\mathcal{H}_{m}^{\leq N} \subset \fock$ for a fixed $N\geq0$, all states are Schwartz states, and the Schwartz topology coincides with the standard one. This is because in this case,
we can relate seminorms to the Hilbert space norm up to constants, i.e., we have $c\norm{\cdot} \leq \norm{\cdot}_{\bm{\alpha}} \leq C \norm{\cdot}$ [with $c=1$, $C=N^{|\bm{\alpha}|} + 1$].

It is clear that the canonical operators $a_k,a^\dagger_k$ ($k=1,\dots,m$) are well-defined on the Schwartz space $\mathcal{S}(\fock)$ and that they stabilize it (i.e., their action on a Schwartz state is again a Schwartz state). It follows that the same is still true for any linear operator $P$ that is a polynomial in the canonical operators, i.e., a complex polynomial (of finite degree) in the variables $\{a_k,a^\dagger_k \,|\,k=1,\dots,m \}$. Therefore, all Schwartz states $\ket{\psi}$ admit finite seminorms $\norm{\ket{\psi}}_{P} := \norm{P\ket{\psi}}$ for all such polynomials $P$. In fact, these $P$ seminorms induce the same topology as the $\bm{\alpha}$ seminorms do on Schwartz space, namely, the Schwartz topology (e.g., \cite[p.~6]{keyl_schwartz_2016}).

Now, let $P$ be any fixed polynomial in the canonical operators, and let us show that, as a linear map $P:\mathcal{S}(\fock) \to \mathcal{S}(\fock)$, it is continuous (with respect to the Schwartz topology on both the domain and codomain).
For this, we aim to use the criterion of \cref{eq:linear-map-between-LCTVSs-continuity-criterion}, with the $P$ seminorms on both the domain and codomain.
For any polynomial $P_b$ in the canonical operators, notice that we have for any $\ket{\psi} \in \mathcal{S}(\fock)$, $\norm{P \ket{\psi}}_{P_b} = \norm{P_b P \ket{\psi}} = \norm{\ket{\psi}}_{P_b P}$, and the last expression is a $P_{b'}$ seminorm of $\ket{\psi}$, with $P_{b'}:=P_b P$. Hence, this establishes \cref{eq:linear-map-between-LCTVSs-continuity-criterion} for the map $P$ (with $n=1$ and $C=1$).

Note that if we had just considered the standard topology on $\mathcal{S}(\fock)$, the map $P:\mathcal{S}(\fock) \to \mathcal{S}(\fock)$ would not have been continuous in general. To illustrate why, the number operator $P:=a_1^\dagger a_1$ for a single-mode system provides a counterexample. Indeed, consider again the sequence of states $\ket{\psi_n}$ defined above. In the standard topology, we have at the same time that $\ket{\psi_n}$ converges to $\ket{0}$ but $P\ket{\psi_n}$ does not converge to $P\ket{0} = 0$ (since $\norm{P\ket{\psi_n} - P\ket{0}} = \sqrt{n}$ does not converge to zero as $n\to\infty$), hence $P$ is not continuous (in the standard topology) at the point $\ket{0}$.

Consider a real Hilbert space $\mathcal{V}$ with its inner product denoted as $\ip{\cdot}{\cdot}$. For a fixed number $k\geq1$, consider the map $\mathcal{R}_k$ that outputs the rank of a given list of $k$ vectors:
\begin{align}
\begin{aligned}
\mathcal{R}_k \colon \mathcal{V} \times \cdots \times \mathcal{V} &\longrightarrow \mathbb{R},\\
(v_1,\dots,v_k) &\longmapsto \rank(\{ v_1,\dots, v_k\})\,.
\end{aligned}
\end{align}
It holds that when each copy of $\mathcal{V}$ is given its standard Hilbert-space topology, the map $\mathcal{R}_k$ is lower semicontinuous.

Indeed, as we explained in the main text, this rank can be recast in terms of that of the associated Gram matrix, i.e., we have $\mathcal{R}_k(v_1,\dots,v_k) = \rank(\gram(v_1,\dots,v_k))$, where $\gram(v_1,\dots,v_k) \in \mathbb{R}^{k \times k}$ is defined as having entries $[\gram(v_1,\dots,v_k)]_{i,j}:=\ip{v_i}{v_j}$.
But each entry $\ip{v_i}{v_j}$ is a continuous function of the vectors $v_i$ and $v_j$, from which it follows that the whole matrix $\gram(v_1,\dots,v_k)$ is a continuous function of $(v_1,\dots,v_k)$.
It is a fact from linear algebra that the rank of a matrix, as a function of its coefficients, is lower semicontinuous. The reason being that if an $r \times r$ submatrix of a $k \times k$ matrix $A$ is invertible, then (due to the continuity of the determinant) there exists an $\epsilon>0$ such that changing any coefficients in this submatrix by no more than $\epsilon$ keeps it invertible and hence guarantees that the rank of $A$ is at least $r$.
Thus the map $\mathcal{R}_k$, being a composition of a continuous map $(v_1,\dots,v_k) \mapsto \gram(v_1,\dots,v_k)$ with a lower semicontinuous map, is lower semicontinuous.

\paragraph{Proof of \cref{prop:lsc-orbit-dim} for the ket picture.}

The proof for the case of the ket picture is now obtained by composition of continuous maps with lower semicontinuous maps.$\qed$

We now turn to the density-operator picture.
\subsection{Schwartz operators}
We briefly introduce the notions of Hilbert-Schmidt and trace-class operators on $\fock$, but for more details we refer to, e.g., Ref.~\cite{hall_quantum_2013}.

The set $B_2(\fock)$ of \textit{Hilbert-Schmidt} operators on $\fock$ is the space of linear maps $A:\fock \to \fock$ for which $\norm{A}_2^2 := \Tr[A^\dagger A] < \infty$.
The set $B_1(\fock)$ of \textit{trace-class} operators on $\fock$ is the space of linear maps $A:\fock \to \fock$ for which $\norm{A}_1 := \Tr[|A|] < \infty$.
The quantities $\norm{\cdot}_1$ and $\norm{\cdot}_2$ define norms on $B_1(\fock)$ and $B_2(\fock)$, called the \textit{trace norm} and the \textit{Hilbert-Schmidt norm}, respectively.
Moreover, the quantity $\ip{A}{B} := \Tr[A^\dagger B]$ is well-defined for $A,B \in B_2(\fock)$, and $\ip{\cdot}{\cdot}$ provides an inner product on $B_2(\fock)$, which turns it into a Hilbert space. The (operator) Fock basis $\{\ketbra{\bm{n}}{\bm{n}'} \ |\ \bm{n},\bm{n}' \in \NN^m\}$ is then an orthonormal basis of $B_2(\fock)$. Note also that in finite dimensions, $B_1(\fock)$ and $B_2(\fock)$ coincide and become merely the space of linear maps on $\fock$; however, in infinite dimensions there is the strict inclusion $B_1(\fock) \subset B_2(\fock)$.

The space $\schop$ of \textit{Schwartz operators} on $\fock$ can be defined as the space of all $A \in B_2(\fock)$ such that its Fock basis coefficients $\braket{\bm{n}}{A|\bm{n}'}$ decay faster than any polynomial in $\bm{n},\bm{n}'$ \cite[Prop.~3.7]{keyl_schwartz_2016}.
In fact, Schwartz operators are always in $B_1(\fock)$ \cite[Lemma~3.6]{keyl_schwartz_2016}.

Now, consider any fixed polynomial $P$ in the canonical operators (as we did when discussing Schwartz states).

The \textit{Schwartz topology} on Schwartz operators can be defined as the topology induced by the seminorms $\norm{A}_{P_L,P_R} := \norm{P_L A P_R}_1$, where $P_L$ runs over all polynomials of the form $P_L = \bm{q}^{\bm{\alpha}} \bm{p}^{\bm{\beta}}$ and $P_R$ runs over all polynomials of the form $P_R = \bm{p}^{\bm{\alpha'}} \bm{q}^{\bm{\beta'}}$ \cite[p.16]{keyl_schwartz_2016}. Here, we used shorthand notations $\bm{q}^{\bm{\alpha}} := q_1^{\alpha_1}\cdots q_m^{\alpha_m}$ and $\bm{p}^{\bm{\beta}} := p_1^{\beta_1}\cdots p_m^{\beta_m}$, with $\bm{\alpha},\bm{\beta},\bm{\alpha'},\bm{\beta'} \in \NN^m$, and $q_k,p_k$ are the position and momentum operators of \cref{eq:def-generator-q,eq:def-generator-p}.
Note that for any two such polynomials $P_L,P_R$ in the canonical operators, and any Schwartz operator $A$, the operator $P_L A P_R$ is clearly still a Schwartz operator; thus it is still in $B_1(\fock)$. This justifies that the seminorms $\norm{\cdot}_{P_L,P_R}$ are well-defined on Schwartz operators.

\paragraph{Proof of \cref{prop:lsc-orbit-dim} for the density operator picture.}

Now, for fixed $P_L,P_R$, let us show that the two linear maps $f_L,f_R:\schop \to B_1(\fock)$ defined by $f_L(A):=P_L A$ and $f_R(A):=A P_R$, are continuous (with respect to the Schwartz topology on the domain, and the trace norm topology on the codomain).
For this, we again aim to use the criterion of \cref{eq:linear-map-between-LCTVSs-continuity-criterion}, with the above seminorms on the domain and with the trace norm on the codomain. We have $\norm{f_L(A)}_1 = \norm{P_L A}_1 = \norm{A}_{P_L,I}$, which establishes \cref{eq:linear-map-between-LCTVSs-continuity-criterion} for the map $f_L$ (with $n=1$ and $C=1$).
Similarly, one obtains the desired continuity of the map $f_R$.
Now, given any complex polynomial $P$ in the canonical operators $\{a_k,a^\dagger_k \,|\,k=1,\dots,m \}$, $P$ can be rewritten (by inverting relations \cref{eq:def-generator-q,eq:def-generator-p} and repeatedly using the commutation relations of \cref{eq:CCR} --- see also the first paragraph of \cref{sec:measuring-pure-gram-matrix-entries} for more background on orderings of polynomial operators) as a complex linear combination of polynomials entirely of the form of the $P_L$'s, or also of the form of the $P_R$'s. Hence, the maps $f_L,f_R$ above are also continuous, as linear combinations of continuous maps. Likewise, the map $F:\schop \to B_1(\fock)$ defined by $F(A):=[P,A] = PA - AP$ is also continuous (as a sum of two continuous maps).$\qed$

\section{Measurement of Gram matrix entries for pure states (including proof of Proposition \ref{prop:homodyne-measurements-maintext})}\label{sec:measuring-pure-gram-matrix-entries}

\subsection{Polynomial operators and their orderings}

Within the following paragraphs, we employ the bold shorthand notations $\bm{q} = (q_1,\dots,q_m),\bm{p} = (p_1,\dots,p_m), \bm{a} = (a_1,\dots,a_m),\bm{a}^\dagger = (a_1^\dagger,\dots,a_m^\dagger)$ for $m$ modes, and we denote the canonical operators with hats (e.g., $\hat{q}_k$) to distinguish them from abstract symbols ($q_k$). We denote the space of formal complex polynomials in the $2m$ noncommuting variables $\bm{q},\bm{p}$ by $\mathbb{C}\langle\bm{q},\bm{p}\rangle$. Given a polynomial $A=A(\bm{q},\bm{p}) \in \mathbb{C}\langle\bm{q},\bm{p}\rangle$, we denote by $\hat{A}:=A(\hat{\bm{q}},\hat{\bm{p}})$ the corresponding \textit{polynomial operator} on Fock space (obtained by substituting the formal variables $q_k,p_k$ by the operators $\hat{q}_k,\hat{p}_k$). The purpose of these precisions in notation is that the canonical commutation relations (CCR)
\begin{align}\label{eq:CCR-q-p}
    [\hat{q}_k,\hat{p}_l] = i \delta_{kl}
\end{align}
are valid at the level of operators, but not at the level of formal polynomials, and hence different polynomials $A(\bm{q},\bm{p}),B(\bm{q},\bm{p})$ can become equal as operators $\hat{A}=\hat{B}$.
One can also always go back and forth between operators $\hat{\bm{q}},\hat{\bm{p}}$ and operators $\hat{\bm{a}},\hat{\bm{a}}^\dagger$, via the relations \cref{eq:def-generator-q,eq:def-generator-p}. Similarly as above, a polynomial $A'=A'(\bm{a},\bm{a}^\dagger) \in \mathbb{C}\langle\bm{a},\bm{a}^\dagger\rangle$ induces the operator $\hat{A'}:=A'(\hat{\bm{a}},\hat{\bm{a}}^\dagger)$, and we recall the form of the CCR with those operators:
\begin{align}
    [\hat{a}_k,\hat{a}^\dagger_l] = \delta_{kl}\,.
\end{align}
We denote by $\operatorname{Poly}(\mathcal{H}_m)$ the space of polynomial operators on $m$-mode Fock space.

Having discussed how the CCR enables different writings of the same operator, we are led to the concept of \textit{orderings}. Let $A=A(\bm{q},\bm{p}) \in \mathbb{C}\langle\bm{q},\bm{p}\rangle$, and $\hat{A}$ its corresponding operator. By repeated use of the CCR, one can always rewrite $\hat{A}$ in a form where all the $q$'s are to the left of all the $p$'s (and in increasing order of modes). This is called the \textit{standard ordering} of $\hat{A}$, which we denote as $\hat{A} = A_{\rm s}(\hat{\bm{q}},\hat{\bm{p}})$, and it can be shown to be unique (see, e.g., Ref.~\cite[Prop. 1.2.1]{Coutinho-PrimerAlgebraic-1995}).
Likewise, other orderings can be considered \cite{Lee-TheoryApplication-1995}. The \textit{Weyl ordering} $\hat{A} = A_{\rm W}(\hat{\bm{q}},\hat{\bm{p}})$ of $\hat{A}$ is its rewriting into a linear combination of terms that are all products (in increasing order of modes) of totally symmetric expressions in $q_k,p_k$. We also need to mention the \textit{antinormal ordering} $\hat{A} = A_{\rm an}(\hat{\bm{a}},\hat{\bm{a}}^\dagger)$, which is the rewriting of $\hat{A}$ in terms of $\hat{\bm{a}},\hat{\bm{a}}^\dagger$ and in the form where all the $a$'s are to the left of all the $a^\dagger$'s (and in increasing order of modes).

We introduce the Weyl symmetrization operator $\operatorname{W}$, which acts on polynomial operators. On single-mode polynomial operators, it is defined via
\begin{align}\label{eq:def-Weyl-symmetrization-singlemode}
\operatorname{W}(\hat{q}^{j} \hat{p}^{N-j}) := \frac{1}{N!} \sum_{\sigma \in S_N} \hat{b}_{\sigma(1)}\cdots \hat{b}_{\sigma(N)}\,,
\end{align}
with $\hat{b}_1 = \cdots = \hat{b}_j = \hat{q}$, $\hat{b}_{j+1} = \cdots = \hat{b}_N = \hat{p}$, and $S_N$ the permutation group on $\{1,\dots,N\}$.
It is then defined on all multimode polynomial operators via
\begin{align}\label{eq:def-Weyl-symmetrization-multimode}
\operatorname{W}(\hat{\bm{q}}^{\bm{j}} \hat{\bm{p}}^{\bm{N}-\bm{j}})
:= \operatorname{W}(\hat{q}_1^{j_1} \hat{p}_1^{N_1 - j_1}) \cdots \operatorname{W}(\hat{q}_m^{j_m} \hat{p}_m^{N_m - j_m})\,,
\end{align}
and extension by linearity. We use here the bold notation for exponents, $\hat{\bm{q}}^{\bm{j}} \hat{\bm{p}}^{\bm{N}-\bm{j}} := \hat{q}_1^{j_1} \cdots \hat{q}_m^{j_m} \hat{p}_1^{N_1 - j_1} \cdots \hat{p}_m^{N_m - j_m}$.

Note that another way to say that any polynomial operator $\hat{A}$ can be uniquely written in standard ordering, or in Weyl ordering, is to say that the sets of standard-ordered monomials
$\{ \hat{\bm{q}}^{\bm{j}} \hat{\bm{p}}^{\bm{N}-\bm{j}} \,|\, \bm{N},\bm{j}\in\mathbb{N}^m, 0 \leq j_k \leq N_k\}$
and Weyl-ordered monomials
$\{ \operatorname{W}( \hat{\bm{q}}^{\bm{j}} \hat{\bm{p}}^{\bm{N}-\bm{j}} ) \,|\, \bm{N},\bm{j}\in\mathbb{N}^m, 0 \leq j_k \leq N_k\}$
are bases of the space of all polynomial operators.
As the apparent \textit{degree} of a polynomial operator $\hat{A}$ can vary between different rewritings of it via the CCR, one may define its degree $\deg(\hat{A})$ with respect to a reference ordering. Let us define $\deg(\hat{A}) := \deg(A_{\rm s})$.
By going through the process of rewriting a polynomial operator from standard ordering to Weyl ordering, one sees that $\deg(A_{\rm s}) = \deg(A_{\rm W})$. Therefore, our definition of $\deg(\hat{A})$ is the degree of either its standard or Weyl ordering.
We denote by $\operatorname{Poly}_{\leq N}(\mathcal{H}_m)$ the space of all polynomial operators over $m$ modes of degree at most $N$.
It now follows that Weyl-ordered monomials of degree at most $N$ form a basis of this space, and hence in the single-mode case:
\begin{align}\label{eq:Weyl-basis-deg-leq-N-spans-all-polys-deg-leq-N}
\begin{aligned}
&\spa_{\mathbb{C}}\big(\{ \operatorname{W}(\hat{q}^{j} \hat{p}^{n-j}) \,|\, 0 \leq n\leq N,\ 0 \leq j \leq n \}\big)\\[3pt]
&\qquad=
\operatorname{Poly}_{\leq N}(\mathcal{H}_1)\,.
\end{aligned}
\end{align}

\subsection{General observable estimation with heterodyne measurements}
A \textit{heterodyne} measurement setup gives the ability to sample the \textit{Husimi Q function} $Q_\rho(\bm{\alpha}) := \frac{1}{\pi^m} \braket{\bm{\alpha}}{\rho|\bm{\alpha}}$ of a state $\rho$ \cite{Serafini-QuantumContinuous-2023};
where $\ket{\bm{\alpha}}$ is the coherent state \cite{Glauber-CoherentIncoherent-1963} $\ket{\bm{\alpha}} = \otimes_{k=1}^m \ket{\alpha_k}$ with
\begin{equation}\label{eq:coherent-state-Fock-repr}
\ket{\alpha_k} := e^{-{|\alpha_k |^{2} \over 2}}\sum _{n=0}^{\infty }\frac{\alpha_k^{n}}{\sqrt {n!}}\ket{n}
\end{equation}
and $\alpha_k \in \mathbb{C}$. That is, each shot of the (multimode) heterodyne measurement produces a random value $\bm{\alpha} \in \mathbb{C}^m$ distributed according to the probability density function $Q_\rho$ on $\mathbb{C}^m$.

In turn, having access to samples drawn from the Husimi Q function enables the estimation of expectation values of arbitrary polynomial operators on Schwartz operators $\rho$. Indeed, by the so-called \textit{optical equivalence theorem}, it holds (see, e.g., Ref.~\cite[Sec. 2.3]{Lee-TheoryApplication-1995}) that
\begin{align}\label{eq:opt-equiv-thm}
    \Tr[ \rho \hat{A} ] = \int_{\mathbb{C}^m} Q_\rho(\bm{\alpha}) A_{\rm an}(\bm{\alpha}) \, d^2\bm{\alpha}\,,
\end{align}
where $A_{\rm an}(\bm{\alpha}):=A_{\rm an}(\bm{\alpha}, \bar{\bm{\alpha}})$ with $A_{\rm an}(\bm{a},\bm{a}^\dagger) \in \mathbb{C}\langle\bm{a},\bm{a}^\dagger\rangle$ the polynomial expressing $\hat{A}$ in the antinormal ordering.
Note that both sides of \cref{eq:opt-equiv-thm} are indeed well-defined, since $\rho \hat{A} \in B_1(\fock)$ (being the product of a Schwartz operator and a polynomial operator, cf. \cref{sec:SM-Schwartz-spaces-and-proof-of-lsc}), and since $Q_\rho(\cdot) A_{\rm an}(\cdot)$ is a Schwartz function (being the product of a Schwartz and a polynomial function) and thus is integrable.

Since $Q_\rho$ is a genuine probability distribution function, \cref{eq:opt-equiv-thm} can be more succinctly written as the expectation value of the random variable $\bm{\alpha}$ postprocessed with $A_{\rm an}$:
\begin{align}\label{eq:opt-equiv-thm-succint}
    \Tr[ \rho \hat{A} ] = \mathbb{E}_{\bm{\alpha}\sim Q_\rho}[A_{\rm an}(\bm{\alpha})]\,.
\end{align}
\Cref{eq:opt-equiv-thm-succint} establishes that a multimode heterodyne measurement on $\rho$ postprocessed with the classical function $\bm{\alpha}\mapsto A_{\rm an}(\bm{\alpha})$ provides an unbiased estimator of the quantity $\Tr[\rho \hat{A}]$.

\subsection{General observable estimation with homodyne measurements}
Let us turn our focus to homodyne measurements.
The motivation and purpose of this paragraph is to formalize and generalize the following standard observation or "trick": if one can experimentally only measure the observables $\hat{q}$, $\hat{p}$, and $(\hat{q}+\hat{p})$ (which corresponds up to prefactors to homodyne measurement at angles $\theta=0$, $\pi/2$, and $\pi/4$, respectively), then because $(\hat{q}+\hat{p})^2 = \hat{q}^2 + 2\{\hat{q}, \hat{p}\} + \hat{p}^2$ [with the symmetrized product notation $\{A,B\} := (AB + BA)/2$], one can estimate the expectation value of the observable $\{\hat{q}, \hat{p}\}$ as well "for free" by taking shots of the three previous measurements and then postprocessing them by squaring the outcomes and combining them via the relation $\{\hat{q}, \hat{p}\} = [(\hat{q}+\hat{p})^2 - \hat{q}^2 - \hat{p}^2]/2$.

To understand this idea in a more general context, we first need to recall general facts about PVMs, observables, and spaces of operators which can be estimated from a given set of observables. These needed facts are part of standard theory of quantum measurement (in both finite and infinite dimensions), but we precisely recall them as it is challenging to find them explicitly stated in the literature.

Our goal is ultimately to show that five angles of homodyne measurements per mode (measured four modes at a time, for all combinations of these angle settings, and for all four-mode subsets of the $m$ modes) suffice to build unbiased estimators for any Gram matrix entries.

We recall that contrary to the heterodyne measurement, which is mathematically merely a positive-operator-valued measure (POVM) (see, e.g., Ref.~\cite{Davies-QuantumTheory-1976}), homodyne measurements are projective (PVMs), and such projective measurements can be equivalently described by self-adjoint observables. The outcome space of PVMs can also be defined to be $\mathbb{C}$ instead of just $\mathbb{R}$, in which case the PVM is equivalently described by a \textit{normal} operator (via the spectral theorem, see Ref.~\cite[Sec. 5.5.1]{Schmudgen-UnboundedSelfadjoint-2012}, or see also Ref.~\cite{Busch-QuantumMeasurement-2016}). An operator is said to be normal if it commutes with its adjoint.

Suppose $\mathcal{O}=\{O_1,O_2,\dots\}$ is a finite set of "implementable" observables, i.e., for each $O \in \mathcal{O}$ it is assumed that the PVM associated with $O$ is accessible in the laboratory.
Let $A$ be a normal operator. Given access to the measurement devices of $\mathcal{O}$, we say (loosely) that $A$ is \textit{samplable} if using only measurements from $\mathcal{O}$ and some classical postprocessing, one is able to "simulate" a PVM whose associated normal operator is $A$.
Similarly, we say that an operator $A$ is \textit{estimatable} if using only measurements from $\mathcal{O}$ and some classical postprocessing, one is able to produce an unbiased estimator of the quantity $\Tr[\rho A]$ on the arbitrary state $\rho$. (We do not consider notions of computational complexity here.)
Notice that, if $A$ is samplable, then it is estimatable (since a shot of the simulated PVM associated with $A$ gives a random variable whose expectation value is $\Tr[\rho A]$).

We introduce the following sets of operators:
\begin{align}
\text{SAMPL}_{\mathcal{O}} &:= \left\{ f(O_{1},\dots,O_{k}) \,\Bigg|\,
\begin{aligned}[c]
&k\geq 1,\ f \in \mathcal{B}(\mathbb{R}^{\times k},\mathbb{C}),\\
&O_{j} \in \mathcal{O},\ [O_{j},O_{j'}]=0
\end{aligned}\right\}\,,\label{eq:SAMPL-set}\\
\text{ESTIM}_{\mathcal{O}} &:= \text{span}_{\mathbb{C}}(\text{SAMPL}_{\mathcal{O}})\,.\label{eq:ESTIM-set}
\end{align}
In the above \cref{eq:SAMPL-set}, $\mathcal{B}(\mathbb{R}^{\times k},\mathbb{C})$ denotes Borel measurable functions from $\mathbb{R}^{\times k}$ to $\mathbb{C}$ (this includes continuous functions), and the commutativity condition $[O_{j},O_{j'}]=0$ should actually be understood as \textit{strong commutativity} (i.e., the associated PVMs commute), but note that this is fulfilled anyways for observables $O_{j},O_{j'}$ acting on separate modes. Furthermore, $f(O_{1},\dots,O_{k})$ is the normal operator defined by the \textit{functional calculus} of the tuple of pairwise commuting operators $(O_{1},\dots,O_{k})$, see Ref.~\cite[Sec. 5.5.2]{Schmudgen-UnboundedSelfadjoint-2012}; but note that when $f$ is a polynomial, this operator $f(O_{1},\dots,O_{k})$ coincides with the usual way of applying the polynomial $f$ to the operators $O_{1},\dots,O_{k}$ (cf. \cite[Theorem 5.9]{Schmudgen-UnboundedSelfadjoint-2012}).

We now make the following claims:
\begin{lemma}\label{claims:ESTIM-SAMPL-sets}\mbox{}\\[-12pt]
\begin{enumerate}[label=(\arabic*)]
    \item All operators $A \in \text{SAMPL}_{\mathcal{O}}$ are samplable.
    Explicitly, one can sample an operator $A=f(O_1,\dots,O_k) \in \text{SAMPL}_{\mathcal{O}}$ by first performing a joint measurement of the observables $O_1,\dots,O_k$, and then postprocessing the outcomes $(o_1,\dots,o_k)$ via $(o_1,\dots,o_k)\mapsto f(o_1,\dots,o_k)$.
    \item All operators $A \in \text{ESTIM}_{\mathcal{O}}$ are estimatable.
    Explicitly, if $A=c_1 A_1 + \cdots + c_l A_l$ for some $l\geq1$, $c_1,\dots,c_l \in \mathbb{C}$, $A_1,\dots,A_l \in \text{SAMPL}_{\mathcal{O}}$, then one has as an unbiased estimator of $\Tr[\rho A]$ (given any state $\rho$) the random variable $c_1 a_1 + \cdots + c_l a_l$, where each $a_j$ is the sample associated with operator $A_j$ (which can be obtained through the previous point).
\end{enumerate}
\end{lemma}
\begin{proof}
The second claim is immediate by linearity of expectation values, and the first claim could be rigorously verified by checking that the PVM uniquely associated with $f(O_1,\dots,O_k)$ (via the spectral theorem) is none other than the one describing the joint measurement $(O_1,\dots,O_k)$ postprocessed with $f$ (cf. \cite[Sec. 5.7]{Busch-QuantumMeasurement-2016}).
\end{proof}

The observable describing a single-mode homodyne measurement with phase parameter $\theta$ is $Q_{\theta} := \cos(\theta) \hat{q} + \sin(\theta) \hat{p}$ \cite{Lvovsky-ContinuousvariableOptical-2009}.

Notice that the operator $(\hat{q} + \hat{p})^N = \hat{q}\cdots\hat{q} + \hat{q}\cdots\hat{q}\hat{p} + \cdots + \hat{p}\cdots\hat{p}$ expands to give a sum of all possible words of the form $\hat{b}_1 \cdots \hat{b}_N$ (with each $\hat{b}_j$ being either $\hat{q}$ or $\hat{p}$). Regrouping the common words yields the following formula (a generalization of Newton's binomial formula to the noncommutative setting):
\begin{align}
(\hat{q} + \hat{p})^N &= \sum_{j=0}^{N} \binom{N}{j} \operatorname{W}(\hat{q}^j \hat{p}^{N-j})\,.
\end{align}

Proceeding likewise but expanding this time the operator $Q_\theta^N = [\cos(\theta) \hat{q} + \sin(\theta) \hat{p}]^N$ yields more generally:
\begin{align}\label{eq:Qtheta-power-N-expansion}
Q_\theta^N &= \sum_{j=0}^{N} \binom{N}{j} \cos(\theta)^j \sin(\theta)^{N-j} \operatorname{W}(\hat{q}^j \hat{p}^{N-j})\,.
\end{align}

Consider now $N+1$ phase parameters $\theta_0,\theta_1,\dots,\theta_N \in \mathbb{R}$. \Cref{eq:Qtheta-power-N-expansion} applied to each $\theta_k$ provides a linear system relating the $N+1$ operators $Q_{\theta_k}^N$ to the $N+1$ operators $\operatorname{W}(\hat{q}^k \hat{p}^{N-k})$ ($k=0,\dots,N$). Explicitly,
\begin{align}\label{eq:Qtheta-power-N-matrix-M-relation}
\begin{pmatrix}
    Q_{\theta_0}^N \\
    \vdots \\
    Q_{\theta_k}^N \\
    \vdots \\
    Q_{\theta_N}^N
\end{pmatrix}
= M
\begin{pmatrix}
    \operatorname{W}(\hat{q}^0 \hat{p}^{N}) \\
    \vdots \\
    \operatorname{W}(\hat{q}^k \hat{p}^{N-k}) \\
    \vdots \\
    \operatorname{W}(\hat{q}^N \hat{p}^{0})
\end{pmatrix}\,,
\end{align}
where the matrix $M=M(\theta_0,\dots,\theta_N) \in \mathbb{R}^{(N+1) \times (N+1)}$ is given by
$[M]_{kl} := \binom{N}{l} \cos(\theta_k)^{l} \sin(\theta_k)^{N-l}$.

Let us first assume that $\sin(\theta_k) \neq 0$ for all $k$. The matrix $M$ may then be seen as a rescaling of a certain Vandermonde matrix $V$:
\begin{align}\label{eq:Qtheta-power-N-expansion-matrix-M-decomposition}
    M = D_a V D_b \,,
\end{align}
with $D_a := \operatorname{diag}(\sin(\theta_0)^N,\dots,\sin(\theta_N)^N)$, $D_b := \operatorname{diag}(\binom{N}{0},\dots,\binom{N}{N})$ and $V \in \mathbb{R}^{(N+1) \times (N+1)}$ given by $[V]_{kl} := \cot(\theta_k)^l$, where $\cot(\theta):=\cos(\theta)/\sin(\theta)$.
Recalling that the Vandermonde matrix $V$ has determinant $\det(V) = \prod_{k<l} [\cot(\theta_l) - \cot(\theta_k)]$, \cref{eq:Qtheta-power-N-expansion-matrix-M-decomposition} yields after simplifications
\begin{align}\label{eq:Qtheta-power-N-detM-formula}
    \det(M) = \prod_{j} \binom{N}{j} \ \prod_{k<l} \sin(\theta_k - \theta_l)\,.
\end{align}
By continuity of both sides of \cref{eq:Qtheta-power-N-detM-formula}, this determinant formula stays valid even if some $\sin(\theta_k)$ are zero. We thus conclude from it that $M$ is invertible if and only if the phase parameters $(\theta_0,\dots,\theta_N)$ are such that no two phases $\theta_k,\theta_l$ are equal modulo $\pi$.

Let us henceforth fix $N$ and $(\theta_0,\dots,\theta_N)$, and assume that this condition is fulfilled. Then, the invertibility of $M$ implies [by \cref{eq:Qtheta-power-N-matrix-M-relation}] that
\begin{align}\label{eq:Qtheta-power-N-invertibility-consequence}
\{ \operatorname{W}(\hat{q}^{j} \hat{p}^{N-j}) \,|\, 0 \leq j \leq N \} \subset \spa_{\mathbb{R}}\left(\{ Q_{\theta_k}^N \,|\, 0 \leq k \leq N \}\right)\,.
\end{align}
By considering again the result of \cref{eq:Qtheta-power-N-invertibility-consequence} but for only the first $n+1$ phases $(\theta_0,\dots,\theta_n)$, for every $n=0,\dots,N$, one obtains
\begin{align}
\begin{aligned}
\{ \operatorname{W}(\hat{q}^{j} \hat{p}^{n-j}) &\,|\, 0 \leq n\leq N,\ 0 \leq j \leq n \}\\[3pt]
&\subset
\spa_{\mathbb{R}}\left(\left\{ Q_{\theta_k}^n \,\Bigg|\,
\begin{aligned}[c]
&k=0,\dots,N \,,\\
&n=k,\dots,N
\end{aligned}\right\}\right)\,.
\end{aligned}
\end{align}
Now taking the complex span on both sides and using \cref{eq:Weyl-basis-deg-leq-N-spans-all-polys-deg-leq-N}, we conclude that
\begin{align}\label{eq:homodyne-single-mode-claim}
\operatorname{Poly}_{\leq N}(\mathcal{H}_1) =
\spa_{\mathbb{C}}\left(\left\{ Q_{\theta_k}^n \,\Bigg|\,
\begin{aligned}[c]
&k=0,\dots,N \,,\\
&n=k,\dots,N
\end{aligned}\right\}\right)\,.
\end{align}
To get to the multimode setting, we take the tensor product of the vector space of \cref{eq:homodyne-single-mode-claim} with itself $m$ times, with $N=N_1,\dots,N_m$, which yields
\begin{align}\label{eq:homodyne-multimode-first-step}
\begin{aligned}
&\bigotimes_{k=1}^m \operatorname{Poly}_{\leq N_k}(\mathcal{H}_1)\\[5pt]
&\,= \spa_{\mathbb{C}}\left(\left\{ Q_{\theta_{k_1}}^{n_1} \otimes \cdots \otimes Q_{\theta_{k_m}}^{n_m} \,\Bigg|\, 0 \leq k_j \leq n_j \leq N_j \right\}\right)\,.
\end{aligned}
\end{align}
We now specialize the values of $N_1,\dots,N_m$ in \cref{eq:homodyne-multimode-first-step} to either $N_k = N$ for some of the modes or $N_k=0$ for the other modes, which finally establishes our main result of this section:
\begin{theorem}\label{thm:homodyne-claim}
Fix integers $m\geq1$, $N\geq0$, $K\geq1$ and $\{i_1,\dots,i_K\} \subseteq \{1,\dots,m\}$.
Fix any $N+1$ phase angles $(\theta_0,\dots,\theta_N) \subset \mathbb{R}$ such that no two angles are equal modulo $\pi$.
Denote by $\operatorname{Poly}_{\!\leq (N,\dots,N)}^{(i_1,\dots,i_K)}(\mathcal{H}_m)$ the space of polynomial operators on $m$-mode Fock space of degree at most $N$ in each of the $K$ modes $i_1,\dots,i_K$ and that act trivially on the other modes. It holds that
\begin{align}
\begin{aligned}
&\operatorname{Poly}_{\!\leq (N,\dots,N)}^{(i_1,\dots,i_K)}(\mathcal{H}_m)\\[5pt]
&\,= \spa_{\mathbb{C}}\left(\left\{ Q_{i_1,\theta_{k_1}}^{n_1} \cdots \, Q_{i_K,\theta_{k_K}}^{n_K} \,\Bigg|\,
\begin{aligned}[c]
&0 \leq k_l \leq n_l \leq N\\
&\forall l \in \{1,\dots,K\}
\end{aligned}
\right\}\right)\,,
\end{aligned}
\end{align}
where $Q_{i,\theta}$ denotes the homodyne observable $Q_{\theta}$ of mode $i$ with phase parameter $\theta$.
\end{theorem}

\subsection{Gram matrix estimation for pure states}

We can finally detail the steps to follow for the experimental estimation of the orbit Gram matrix entries in the ket and ketbra pictures, for the case of pure states [\cref{eq:concrete-gram-formula-ket-expression,eq:concrete-gram-formula-ketbra-expression}] --- and in particular detail and prove the protocol of \cref{prop:homodyne-measurements-maintext} from the main text. Explicitly, we recall that the entries to estimate are respectively of the form
\begin{align}
[\gram_G(\ket{\psi})]_{I,J} &= \mathbb{E}_{\psi}\left(\frac{H_{I} H_{J} + H_{J} H_{I}}{2}\right)\,,\label{eq:concrete-gram-formula-ket-expression-explicit}\\[6pt]
[\gram_G(\ketbra{\psi})]_{I,J} &=
\begin{aligned}[t]
&2\,[\gram_G(\ket{\psi})]_{I,J}\\
- &2\, \mathbb{E}_{\psi}\left(H_{I}\right) \mathbb{E}_{\psi}\left(H_{J}\right)\,,
\end{aligned}\label{eq:concrete-gram-formula-ketbra-expression-explicit}
\end{align}
where $H_I,H_J \in \mathcal{B}_{\mathfrak{g}}$, with $\mathcal{B}_{\mathfrak{g}}$ the Lie algebra basis of consideration among those listed in \cref{tab:Lie-algebra-bases}. We consider the case of Gaussian optics ($G = G_{\rm GO}$), since its Lie algebra basis contains all other cases as subsets anyways.

First, after making the substitutions
\begin{align}
a_k^\dagger &= \frac{1}{\sqrt{2}}(q_k - i p_k)\,,\\
a_k &= \frac{1}{\sqrt{2}}(q_k + i p_k)\,
\end{align}
in each basis Hamiltonian $H_I \in \mathcal{B}_{\mathfrak{g}}$ [whose forms are detailed in \cref{eq:def-generator-e,eq:def-generator-E,eq:def-generator-r,eq:def-generator-R,eq:def-generator-N,eq:def-generator-s,eq:def-generator-S,eq:def-generator-q,eq:def-generator-p}], one obtains writings of them in standard ordering that are all of degree at most $2$. This establishes that $H_I \in \operatorname{Poly}_{\leq 2}(\mathcal{H}_m)$ for all $H_I \in \mathcal{B}_{\mathfrak{g}}$. Indeed, here are for completeness the standard ordering forms of all these basis elements:
\begin{align}
e_{kl} &= \frac{1}{2}\bigl(q_k q_l + p_k p_l\bigr)\,,\label{eq:def-generator-e-qpform}\\
E_{kl} &= \frac{1}{2}\bigl(q_l p_k - q_k p_l\bigr)\,,\label{eq:def-generator-E-qpform}\\
r_{kl} &= \frac{1}{2}\bigl(q_k q_l - p_k p_l\bigr)\,,\label{eq:def-generator-r-qpform}\\
R_{kl} &= \frac{1}{2}\bigl(q_k p_l + q_l p_k\bigr)\,,\label{eq:def-generator-R-qpform}\\
N_{k} &= \frac{1}{2}\bigl(q_k^{2} + p_k^{2}\bigr) - \frac{1}{2} \id\,,\label{eq:def-generator-N-qpform}\\
s_{k} &= \frac{1}{2}\bigl(q_k^{2} - p_k^{2}\bigr)\,,\label{eq:def-generator-s-qpform}\\
S_{k} &= q_k p_k - \frac{i}{2} \id\,,\label{eq:def-generator-S-qpform}\\
q_{k} &= q_k\,,\label{eq:def-generator-q-qpform}\\
p_{k} &= p_k\,,\label{eq:def-generator-p-qpform}\\
\id &= \id\,.\label{eq:def-generator-id-qpform}
\end{align}
Second, one then expands the symmetrized products $(H_I H_J + H_J H_I)/2$ for all pairs of basis Hamiltonians $H_I,H_J \in \mathcal{B}_{\mathfrak{g}}$ and rewrites them in standard order. One could verify that this will always yield polynomials of degree at most $4$, and more precisely of degree \textit{per mode} at most $N:=4$ and involving at most $K:=4$ different modes.%

Therefore, to estimate all Gram matrix entries $I,J$, one can, according to \cref{eq:concrete-gram-formula-ket-expression-explicit,eq:concrete-gram-formula-ketbra-expression-explicit}, estimate expectation values of all observables of form $(H_I H_J + H_J H_I)/2$ over all pairs $H_I,H_J\in \mathcal{B}_\mathfrak{g}$ [as well as expectation values of all the observables $H_I$ themselves, to get orbit dimensions also in the ketbra picture using then \cref{eq:concrete-gram-formula-ketbra-expression-explicit}]. This amounts to estimating, on the probed pure state $\psi$, expectation values of a list $\mathcal{L}$ of observables, of size $|\mathcal{L}| = \mathcal{O}(\dim(\mathfrak{g})^2) = \mathcal{O}(m^4)$, with
\begin{align}
\mathcal{L} \subset \bigcup_{1\leq i_1 < \dots < i_K \leq m} \operatorname{Poly}_{\!\leq (N,\dots,N)}^{(i_1,\dots,i_K)}(\mathcal{H}_m)
\end{align}
for $N=4$ and $K=4$.

By \cref{thm:homodyne-claim}, each of these observables $O \in \mathcal{L}$ can be written as a complex linear combination of certain products of powers of homodyne observables:
\begin{align}\label{eq:homodyne-decomposition-of-O}
O = \sum_{\substack{0 \leq k_1 \leq n_1 \leq 4, \\[-3pt] \vdots \\[1pt] 0 \leq k_4 \leq n_4 \leq 4}} c_{(k_1,n_1,\dots,k_4,n_4)} \ Q_{i_1,\theta_{k_1}}^{n_1} \cdots\, Q_{i_4,\theta_{k_4}}^{n_4}\!,\,\\[-15pt]\nonumber
\end{align}
for some modes $(i_1,\dots,i_4)$, where $(\theta_0,\dots,\theta_4)$ are any five homodyne phase angles such that no two are equal modulo $\pi$.

While proving \cref{thm:homodyne-claim}, a method to obtain explicitly the complex coefficients $c_{(\cdots)}$ in \cref{eq:homodyne-decomposition-of-O} was obtained along the way, which we now summarize. One first rewrites [using the CCR \cref{eq:CCR-q-p}] the polynomial operator $O$ of order at most $4$ from standard ordering to Weyl ordering, i.e., one rewrites $O$ as a linear combination of Weyl-ordered monomials $\operatorname{W}( \hat{\bm{q}}^{\bm{j}} \hat{\bm{p}}^{\bm{N}-\bm{j}} )$ of order at most $4$ (we spell out an example in the last paragraph). For each such Weyl-ordered monomial, which is a product of single-mode Weyl-ordered monomials $\operatorname{W}(\hat{q}_k^{j_k} \hat{p}_k^{N_k - j_k})$ [\cref{eq:def-Weyl-symmetrization-multimode}] of degree at most $4$, one reexpresses each of these single-mode Weyl-ordered monomials $\operatorname{W}(\hat{q}_k^{j_k} \hat{p}_k^{N_k - j_k})$ as a (real) linear combination of single-mode homodynes (of different phases) raised to the power of the degree $N_k=\deg(\operatorname{W}(\hat{q}_k^{j_k} \hat{p}_k^{N_k - j_k})) \in \{0,\dots,4\}$ of this single-mode monomial, by inverting the $(N_k + 1) \times (N_k + 1)$ matrix $M$ in the linear system of \cref{eq:Qtheta-power-N-matrix-M-relation}.
By doing so, one arrives at the desired form of $O$ of \cref{eq:homodyne-decomposition-of-O} (with explicit coefficients $c_{(\cdots)}$). Note that in total one only needs to invert five different matrices $M$ (for $N=0,1,2,3,4$), and so these five inversions (which may be done with symbolic computation or numerically) can be performed once and for all as a precomputation, as soon as the homodyne angles $(\theta_0,\dots,\theta_4)$ have been chosen.

Since each observable in the linear combination of \cref{eq:homodyne-decomposition-of-O} is of the form $O':= Q_{i_1,\theta_{k_1}}^{n_1} \cdots\, Q_{i_4,\theta_{k_4}}^{n_4} = f(Q_{i_1,\theta_{k_1}},\dots,Q_{i_4,\theta_{k_4}})$ with the continuous map $f(o_1,\dots,o_4) := o_1^{n_1} \cdots\, o_4^{n_4}$ and with four locally disjoint observables $Q_{i_1,\theta_{k_1}},\dots,Q_{i_4,\theta_{k_4}}$, each such observable $O' \in \text{ESTIM}_{\mathcal{O}}$ [\cref{eq:ESTIM-set}], with $\mathcal{O} := \{ Q_{i_1,\theta_{k_1}},\dots,Q_{i_4,\theta_{k_4}} \,|\, (k_1,\dots,k_4) \in \{0,\dots,4\}^4 \}$.
Hence, by the second claim of \cref{claims:ESTIM-SAMPL-sets}, this observable $O'$ is estimatable: explicitly, one can estimate the expectation value of $O'$ by first performing a joint measurement of the four commuting observables $Q_{i_1,\theta_{k_1}},\dots,Q_{i_4,\theta_{k_4}}$, and then postprocessing the outcomes $(o_1,\dots,o_4)$ via $(o_1,\dots,o_4)\mapsto o_1^{n_1} \cdots\, o_4^{n_4}$ to obtain an unbiased estimator of the expectation value.

Therefore, $O$ itself is estimatable ($O \in \text{ESTIM}_{\mathcal{O}}$), with an unbiased estimator for its expectation value given by the linear combination [using the same complex coefficients $c_{(\cdots)}$ as those found for the decomposition of $O$ of \cref{eq:homodyne-decomposition-of-O}] of all these postprocessed outcomes $o_1^{n_1} \cdots\, o_4^{n_4}$.

The procedure just described shows that making shots of all combinations of the joint homodyne measurements $\{ (Q_{i_1,\theta_{k_1}},\dots,Q_{i_4,\theta_{k_4}}) \,|\, (k_1,\dots,k_4) \in \{0,\dots,4\}^4 \}$ suffices to estimate (via suitable classical postprocessing) the expectation values of all observables in $\operatorname{Poly}_{\!\leq (4,\dots,4)}^{(i_1,\dots,i_4)}(\mathcal{H}_m)$. That is, $5^4=625$ joint homodyne measurement settings on the four modes $(i_1,i_2,i_3,i_4)$ suffice to estimate expectation values of all the desired observables $O \in \mathcal{L}$ that involve these four modes.
Since there are $\binom{m}{4}$ possible choices of four modes $(i_1,i_2,i_3,i_4)$ in total, it hence suffices to make altogether $\binom{m}{4} 5^4 \in \mathcal{O}(m^4)$ different types of four-mode joint homodyne measurements to estimate all the desired expectation values.
Note that this is just the counting for one sufficient method, but further optimization for lower constants may be possible.
This concludes the proof and explanation of \cref{prop:homodyne-measurements-maintext} of the main text.$\qed$

As for the heterodyne method, it was explained earlier [\cref{eq:opt-equiv-thm-succint}] that for any operator $O \in \operatorname{Poly}(\mathcal{H}_m)$, shots of the ($m$-mode) heterodyne measurement postprocessed with the classical function $\bm{\alpha}\mapsto O_{\rm an}(\bm{\alpha})$ provide an unbiased estimator for the expectation value of $O$. Therefore, the only task required for this method is to cast all the desired observables $O \in \mathcal{L}$ in antinormal form.

Note that such heterodyne measurements come with an increased initial cost for the variance of obtained estimators; however, deciding which of the two unbiased estimator approaches (homodyne or heterodyne) for Gram matrix entries ends up requiring the lowest amount of shots overall would involve a study of both variances, which may depend on properties of the specific state at play, such as its average energy \cite{MauroDAriano-QuantumTomography-2003}.

\subsection{Example of rewriting of an operator into different orderings}
Consider the Gram matrix entries (ket or ketbra pictures) associated with $H_I := E_{12}$ and $H_J := N_2$. That is, consider the specific observable $O := \{E_{12},N_{2}\} = (E_{12} N_{2} + N_{2} E_{12})/2$. Let us now drop the hats from canonical operators to ease the notation.

One first expands this symmetrized product by substituting the standard-ordered forms of the basis Hamiltonians $E_{12}$ and $N_2$ [\cref{eq:def-generator-E-qpform,eq:def-generator-N-qpform}], followed then, in the obtained linear combination of monomials, with some swapping of orders between the different $b_k^{n_k},b_l^{m_l}$'s that belong to different modes ($b = q$ or $p$) to get individual operators ordered by mode and operators $q_k^{n_k},p_k^{m_k}$'s of the same mode next to each other. This yields
\begin{equation}
\begin{aligned}
O &= \frac{1}{4} p_1 q_2^3
+ \frac{1}{8} p_1 q_2 p_2^2
+ \frac{1}{8} p_1 p_2^2 q_2
+ \frac{1}{4} q_1 p_2\\
&- \frac{1}{4} p_1 q_2
- \frac{1}{8} q_1 p_2 q_2^2
- \frac{1}{8} q_1 q_2^2 p_2
- \frac{1}{4} q_1 p_2^3
\,.
\end{aligned}
\end{equation}

To bring $O$ in standard order, the $q_k^{n_k}$'s are placed to the left of their respective $p_k^{m_k}$'s, with an addition of the appropriate commutator term to make the equality still hold:
\begin{equation}
\begin{aligned}
O &= \frac{1}{4} p_1 q_2^3
+ \frac{1}{8} p_1 q_2 p_2^2
+ \frac{1}{8} p_1 (q_2 p_2^2 + [p_2^2, q_2])
+ \frac{1}{4} q_1 p_2\\
&- \frac{1}{4} p_1 q_2
- \frac{1}{8} q_1 (q_2^2 p_2 + [p_2, q_2^2])
- \frac{1}{8} q_1 q_2^2 p_2
- \frac{1}{4} q_1 p_2^3
\,.
\end{aligned}
\end{equation}
The commutators introduced above are then evaluated using (repeated uses of) the CCR [\cref{eq:CCR-q-p}]. In fact, the following useful generalizations of the single-mode CCR formula can be established by induction:
\begin{align}
[q^n, p] &= i n q^{n-1}\,,\label{eq:generalized-CCRs-qpform-1}\\
[q, p^m] &= i m p^{m-1}\,.\label{eq:generalized-CCRs-qpform-2}
\end{align}
Thanks to these formulas \cref{eq:generalized-CCRs-qpform-1,eq:generalized-CCRs-qpform-2}, one simplifies the above into
\begin{equation}
\begin{aligned}
O &= \frac{1}{4} p_1 q_2^3
+ \frac{1}{8} p_1 q_2 p_2^2
+ \frac{1}{8} p_1 (q_2 p_2^2 - 2 i p_2)
+ \frac{1}{4} q_1 p_2\\
&- \frac{1}{4} p_1 q_2
- \frac{1}{8} q_1 (q_2^2 p_2 - 2 i q_2)
- \frac{1}{8} q_1 q_2^2 p_2
- \frac{1}{4} q_1 p_2^3
\,,
\end{aligned}
\end{equation}
which, by expanding it and regrouping terms, yields
\begin{equation}
\begin{aligned}
O &= \frac{1}{4} p_1 q_2^3
+ \frac{1}{4} p_1 q_2 p_2^2
- \frac{1}{4} i p_1 p_2
+ \frac{1}{4} q_1 p_2\\
& - \frac{1}{4} p_1 q_2
- \frac{1}{4} q_1 q_2^2 p_2
+ \frac{1}{4} i q_1 q_2
- \frac{1}{4} q_1 p_2^3
\,.
\end{aligned}
\end{equation}
Actually, since we previously defined the standard order to be such that all the $q_k$'s are to the left of all the $p_k$'s (in each monomial),
we now again swap the orders between the operators of different modes, to finally obtain the standard order form of $O$:
\begin{equation}\label{eq:example-operator-O-standard-order}
\begin{aligned}
O &= \frac{1}{4} q_2^3 p_1
+ \frac{1}{4} q_2 p_1 p_2^2
- \frac{1}{4} i p_1 p_2
+ \frac{1}{4} q_1 p_2\\
& - \frac{1}{4} q_2 p_1
- \frac{1}{4} q_1 q_2^2 p_2
+ \frac{1}{4} i q_1 q_2
- \frac{1}{4} q_1 p_2^3
\,.
\end{aligned}
\end{equation}

Next, we turn to the conversion of $O$ from its standard ordering to its Weyl ordering. First, one notices in this case that among the eight monomials in the standard order form of \cref{eq:example-operator-O-standard-order}, six of them are already Weyl-ordered monomials (those with only one kind of operator per mode), i.e., we have:
\begin{equation}\label{eq:example-operator-O-partial-Weyl-order}
\begin{aligned}
O &= \frac{1}{4} \operatorname{W}(q_2^3 p_1)
+ \frac{1}{4} q_2 p_1 p_2^2
- \frac{1}{4} i \operatorname{W}(p_1 p_2)
+ \frac{1}{4} \operatorname{W}(q_1 p_2)\\
& - \frac{1}{4} \operatorname{W}(q_2 p_1)
- \frac{1}{4} q_1 q_2^2 p_2
+ \frac{1}{4} i \operatorname{W}(q_1 q_2)
- \frac{1}{4} \operatorname{W}(q_1 p_2^3)
\,.
\end{aligned}
\end{equation}
Therefore, it is only left to turn each of the two remaining monomials $q_2 p_1 p_2^2$ and $q_1 q_2^2 p_2$ into Weyl-ordered form. Let us focus first on the former. Since its mode-$1$ part ($p_1$) is already Weyl ordered [$p_1 = \operatorname{W}(p_1)$], one only needs [cf. \cref{eq:def-Weyl-symmetrization-multimode}] to focus on its mode-$2$ part, $q_2 p_2^2$. It can be turned into Weyl-ordered form, by iteratively making all successive permutations of monomials appear in its decomposition, as follows. Denote all (three) possible permutations of $q p^2$ (we drop the mode-$2$ index to ease notation) as $A := q p^2$, $B := p q p$, and $C := p^2 q$.

First, one relates $A$ to $B$ by writing
\begin{align}
A = q p p = (p q + [q, p]) p = p q p + i p\,,
\end{align}
where the CCR was used in the last equality.
Hence $A = B + b$, with $b := i p$.
Likewise, one relates $B$ to $C$ by writing
\begin{align}
B = p q p = p (p q + [q, p]) = p^2 q + i p\,,
\end{align}
which gives $B = C + c$, with $c := i p$.
These two relations imply that
\begin{align}
3A
&= A + A + A\\
&= A + (B + b) + (B + b)\\
&= A + (B + b) + [(C +c) + b]\\
&= (A + B + C) + (2b + c)\,,
\end{align}
Hence,
\begin{align}
A = \frac{1}{3}(A + B + C) \,+\, \frac{1}{3}(2b + c)\,.
\end{align}
But since $(A + B + C)/3 = \operatorname{W}(q p^2)$ [cf. \cref{eq:def-Weyl-symmetrization-singlemode}], and $(2b + c)/3 = i p$, we have obtained
\begin{align}
q p^2 = \operatorname{W}(q p^2) + i p\,.
\end{align}
Through this process, the standard-ordered monomial $qp^2$ has been rewritten as its Weyl symmetrization, plus a remainder term \textit{of lower degree}. Here, this remainder term $i p$ is already Weyl ordered, i.e.,
\begin{align}
q p^2 = \operatorname{W}(q p^2) + i \operatorname{W}(p)\,,
\end{align}
so the process stops here (but in general, it is clear how one could iterate the same process further on the lower-degree remainder terms, to turn those into Weyl-ordered form as well).

We have thus obtained that
\begin{align}
q_2 p_1 p_2^2 &= \operatorname{W}(p_1) \operatorname{W}(q_2 p_2^2) + i \operatorname{W}(p_1) \operatorname{W}(p_2)\\
&= \operatorname{W}(q_2 p_1 p_2^2) + i \operatorname{W}(p_1 p_2)\,.\label{eq:example-operator-O-partial-Weyl-order-term1}
\end{align}
After applying the same process to the other term $q_1 q_2^2 p_2$, one finds likewise
\begin{align}
q_1 q_2^2 p_2 &= \operatorname{W}(q_1) \operatorname{W}(q_2^2 p_2) + i \operatorname{W}(q_1) \operatorname{W}(q_2)\\
&= \operatorname{W}(q_1 q_2^2 p_2) + i \operatorname{W}(q_1 q_2)\,.\label{eq:example-operator-O-partial-Weyl-order-term2}
\end{align}

One may now substitute \cref{eq:example-operator-O-partial-Weyl-order-term1,eq:example-operator-O-partial-Weyl-order-term2} into the current expression of $O$ of \cref{eq:example-operator-O-partial-Weyl-order}, to finally obtain the Weyl-ordered form of $O$:
\begin{equation}\label{eq:example-operator-O-full-Weyl-order}
\begin{aligned}
O &= \frac{1}{4} \operatorname{W}(q_2^3 p_1)
+ \frac{1}{4} \operatorname{W}(q_2 p_1 p_2^2) + \frac{1}{4} i \operatorname{W}(p_1 p_2)
\\
& - \frac{1}{4} i \operatorname{W}(p_1 p_2)
+ \frac{1}{4} \operatorname{W}(q_1 p_2)
- \frac{1}{4} \operatorname{W}(q_2 p_1)\\
&- \frac{1}{4} \operatorname{W}(q_1 q_2^2 p_2) - \frac{1}{4} i \operatorname{W}(q_1 q_2)
+ \frac{1}{4} i \operatorname{W}(q_1 q_2)\\
&- \frac{1}{4} \operatorname{W}(q_1 p_2^3)
\,.
\end{aligned}
\end{equation}

\section{Sufficient condition for correctness of numerical rank of noisy Gram matrix}\label{sec:SM-numerical-rank-noisy-Gram-matrix}
In this appendix, we spell out the proof of the claim made in the main text (below \cref{prop:homodyne-measurements-maintext}) stating that,
if $\gram_{G,\mathrm{expt}}(\rho_{\mathrm{expt}})$ and $\gram_G(\rho)$ are $\epsilon$-close in each entry with $\epsilon< \lambda_{\mathrm{min}}^+(\gram_G(\rho))/[2\dim(\mathfrak{g})]$, then the numerical rank evaluation on the experimentally built matrix is guaranteed to yield $\dim(\orb_{G}(\rho))$ when $\tau \in (\dim(\mathfrak{g}) \epsilon,\ \lambda_{\mathrm{min}}^+(\gram_G(\rho)) - \dim(\mathfrak{g}) \epsilon)$.
Here, recall that $\lambda_{\mathrm{min}}^+(\cdot)$ denotes the smallest positive eigenvalue.

The hypothesis that $\gram_{G,\mathrm{expt}}(\rho_{\mathrm{expt}})$ and $\gram_G(\rho)$ are $\epsilon$-close in each entry writes as
\begin{align}\label{eq:entrywise-epsilon-closeness-of-Grams}
\norm{ \gram_{G,\mathrm{expt}}(\rho_{\mathrm{expt}}) - \gram_G(\rho) }_{\mathrm{ew},\infty} &\leq \epsilon\,,
\end{align}
where $\norm{\cdot}_{\mathrm{ew},\infty}$ denotes the entrywise infinity norm (maximum absolute value of entries). Since these matrices are of size $\dim(\mathfrak{g}) \times \dim(\mathfrak{g})$, \cref{eq:entrywise-epsilon-closeness-of-Grams} implies by a standard norm inequality that
\begin{align}
\norm{ \gram_{G,\mathrm{expt}}(\rho_{\mathrm{expt}}) - \gram_G(\rho) }_{\infty} &\leq \dim(\mathfrak{g}) \epsilon\,,
\end{align}
where the norm $\norm{\cdot}_{\infty}$ denotes this time the Schatten infinity norm, i.e., the maximal singular value.
But since these two Gram matrices are real symmetric, Weyl's inequality \cite[Theorem 4.3.1]{horn_matrix_2013} gives
\begin{equation}
\begin{aligned}
&|\lambda_k(\gram_{G,\mathrm{expt}}(\rho_{\mathrm{expt}})) - \lambda_k(\gram_G(\rho))|\\[5pt]
&\leq \ \norm{ \gram_{G,\mathrm{expt}}(\rho_{\mathrm{expt}}) - \gram_G(\rho) }_{\infty}
\end{aligned}
\end{equation}
for all $k=1,\dots,\dim(\mathfrak{g})$, where $\lambda_1(\cdot) \leq \dots \leq \lambda_{\dim(\mathfrak{g})}(\cdot)$ denote the ordered eigenvalues.
Combining the two previous inequalities yields for all $k$:
\begin{align}\label{eq:Weyl-inequality-consequence-Grams}
|\lambda_k(\gram_{G,\mathrm{expt}}(\rho_{\mathrm{expt}})) - \lambda_k(\gram_G(\rho))| &\leq \dim(\mathfrak{g}) \epsilon\,.
\end{align}
Now, consider the numerical rank evaluation on $\gram_{G,\mathrm{expt}}(\rho_{\mathrm{expt}})$ with a cutoff threshold $\tau>0$, which will return the number of its eigenvalues strictly greater than $\tau$ (note that eigenvalues of Gram matrices are always non-negative since they are positive semidefinite).
A sufficient condition for this evaluation to yield the true rank of $\gram_G(\rho)$ is that (i) the zero eigenvalues of $\gram_G(\rho)$ remain smaller than $\tau$ after the perturbation, and (ii) the positive eigenvalues of $\gram_G(\rho)$ are greater than $\tau$ and remain so after the perturbation.
[Here, "perturbation" means going from the matrix $\gram_G(\rho)$ to the matrix $\gram_{G,\mathrm{expt}}(\rho_{\mathrm{expt}})$.]
Since by \cref{eq:Weyl-inequality-consequence-Grams}, each eigenvalue can be shifted by at most $\dim(\mathfrak{g}) \epsilon$ after the perturbation, a further sufficient condition for (i) and (ii), respectively, is (i') $\dim(\mathfrak{g}) \epsilon \leq \tau$ and (ii') $\lambda_{\mathrm{min}}^+(\gram_G(\rho)) - \dim(\mathfrak{g}) \epsilon > \tau$.
Combining (i') and (ii') gives the claimed sufficient condition
\begin{equation}\label{eq:sufficient-condition-for-numerical-rank-evaluation-correctness-interval}
\dim(\mathfrak{g}) \epsilon < \tau < \lambda_{\mathrm{min}}^+(\gram_G(\rho)) - \dim(\mathfrak{g}) \epsilon\,,
\end{equation}
and there exists $\tau>0$ satisfying \cref{eq:sufficient-condition-for-numerical-rank-evaluation-correctness-interval} if and only if $\lambda_{\mathrm{min}}^+(\gram_G(\rho)) > 2 \dim(\mathfrak{g}) \epsilon$, i.e.,
\begin{equation}
\epsilon< \lambda_{\mathrm{min}}^+(\gram_G(\rho))/[2\dim(\mathfrak{g})]\,.\qed
\end{equation}

\section{Proof of Gram-matrix expression in terms of second derivatives of SWAP tests [Eq.~(\ref{eq:gram-matrix-second-order-derivative-expression})]}\label{sec:SM-swap-test-second-derivatives}

If $\phi:G\to \uni(\mathcal{H})$ is a strongly continuous unitary representation, $p \in \mathcal{H}^\infty$, and $X \in \mathfrak{g}$, then [as we had established during the proof of \cref{thm:SC-repr-orbit-structure}, \cref{eq:SC-repr-orbit-structure-proof-gammatdot-at-t-and-at-0-relation}] it holds that, for all $t \in \mathbb{R}$,
\begin{equation}\label{eq:justification-second-derivatives--cite1}
\frac{d}{dt} \Big( \phi(e^{tX}) \cdot p \Big) = \phi(e^{tX}) \cdot \big( \phi'(X) \cdot p \big)\,.
\end{equation}
As both $p$ and $(\phi'(X) \cdot p)$ are in $\mathcal{H}^\infty$ (cf. \cref{sec:SM-EMRep-and-proof-of-orbit-structure}), we can differentiate both sides of \cref{eq:justification-second-derivatives--cite1} again with respect to $t$, and evaluate at $t=0$, which gives
\begin{equation}
\evalat[\Big]{\frac{d^2}{dt^2}}{t=0} \Big( \phi(e^{tX}) \cdot p \Big) = \phi'(X) \cdot \big( \phi'(X) \cdot p \big)\,,
\end{equation}
i.e.,
\begin{equation}\label{eq:justification-second-derivatives--cite2}
\evalat[\Big]{\frac{d^2}{dt^2}}{t=0} \Big( \phi(e^{tX}) \cdot p \Big) = \phi'(X)^2 \cdot p\,.
\end{equation}

We can hence apply \cref{eq:justification-second-derivatives--cite2} to our (density operator picture) representation $\Phi:G_0 \to \uni(B_2(\fock))$ of the extended metaplectic group (cf. \cref{sec:SM-EMRep-and-proof-of-orbit-structure}), and to a Schwartz operator $\rho$ (which is a smooth vector for $\Phi$, cf. \cref{sec:SM-EMRep-and-proof-of-orbit-structure}). As we have seen that $\Phi'(X)$ acts as $\operatorname{ad}_{\phi'(X)}$ on Schwartz operators [\cref{eq:Phiprime-action-by-commutation-ad-notation}], we get
\begin{equation}
\evalat[\Big]{\frac{d^2}{dt^2}}{t=0} \Big( \Phi(e^{tX}) \cdot \rho \Big) = \Phi'(X)^2 \cdot \rho\,,
\end{equation}
i.e.,
\begin{equation}\label{eq:justification-second-derivatives--cite3}
\evalat[\Big]{\frac{d^2}{dt^2}}{t=0} \Big( \Phi(e^{tX}) \cdot \rho \Big) = \big[\phi'(X),\ \left[\phi'(X), \rho\right]\big]\,.
\end{equation}

Now, for a fixed Schwartz operator $\rho$, let us define
\begin{equation}\label{eq:justification-second-derivatives--defbetaXIXJ}
\beta_{X_I, X_J}(t) := \ipbig{ \Phi(e^{t X_I}) \cdot \rho }{\ \Phi(e^{t X_J}) \cdot \rho }
\end{equation}
for $X_I, X_J \in \mathfrak{g}$ and $t \in \mathbb{R}$ [recall that the inner product here is, as in the main text, the Hilbert-Schmidt inner product $\ip{A}{B}:=\Tr[A^\dagger B]$ on $B_2(\fock)$].
By differentiating \cref{eq:justification-second-derivatives--defbetaXIXJ} twice with respect to $t$, and evaluating at $t=0$, we obtain, by bilinearity of the inner product when restricted to Hermitian inputs (Leibniz product rule),
\begin{equation}\label{eq:justification-second-derivatives--betaXIXJdotdotzero}
\evalat[\Big]{\frac{d^2}{dt^2}}{t=0} \beta_{X_I, X_J}(t) =
\begin{aligned}[t]
&\ipbig{ \Phi'(X_I)^2 \cdot \rho }{\ \rho }\\
+\ 2\,&\ipbig{ \Phi'(X_I) \cdot \rho }{\ \Phi'(X_J) \cdot \rho }\\
+\ \ \ &\ipbig{ \rho }{\ \Phi'(X_J)^2 \cdot \rho }\,.
\end{aligned}
\end{equation}
Specializing \cref{eq:justification-second-derivatives--betaXIXJdotdotzero} to the cases $X_J=0$ and $X_I=0$, we also get
\begin{align}
\evalat[\Big]{\frac{d^2}{dt^2}}{t=0} \beta_{X_I, 0}(t) &= \ipbig{ \Phi'(X_I)^2 \cdot \rho }{\ \rho }\,,\label{eq:justification-second-derivatives--betaXI0dotdotzero}\\[3pt]
\evalat[\Big]{\frac{d^2}{dt^2}}{t=0} \beta_{0, X_J}(t) &= \ipbig{ \rho }{\ \Phi'(X_J)^2 \cdot \rho }\,.\label{eq:justification-second-derivatives--beta0XJdotdotzero}
\end{align}
Hence, combining \cref{eq:justification-second-derivatives--defbetaXIXJ,eq:justification-second-derivatives--betaXIXJdotdotzero,eq:justification-second-derivatives--betaXI0dotdotzero,eq:justification-second-derivatives--beta0XJdotdotzero}, directly yields
\begin{equation}\label{eq:justification-second-derivatives--finalexpression}
\ipbig{ \Phi'(X_I) \cdot \rho }{\ \Phi'(X_J) \cdot \rho } =
\begin{aligned}[t]
\frac{1}{2}\Big(&\evalat[\Big]{\frac{d^2}{dt^2}}{t=0} \beta_{X_I, X_J}(t)\\[4pt]
-\ &\evalat[\Big]{\frac{d^2}{dt^2}}{t=0} \beta_{X_I, 0}(t)\\[4pt]
-\ &\evalat[\Big]{\frac{d^2}{dt^2}}{t=0} \beta_{0, X_J}(t)
\Big)\,.
\end{aligned}
\end{equation}

By now picking the basis $\mathcal{B}_{\mathfrak{g}}=\{ X_1,\dots,X_d \}$ of $\mathfrak{g}$ given by $\phi'^{-1}(i H_k)$ with $\{H_1,\dots,H_d\}$ the basis of quadratic Hamiltonians of the main text (cf.\ proof of \cref{thm:orbit-dim-maintext} in \cref{sec:SM-EMRep-and-proof-of-orbit-structure}), \cref{eq:justification-second-derivatives--finalexpression} applied to $X_I, X_J \in \mathcal{B}_{\mathfrak{g}}$ establishes \cref{eq:gram-matrix-second-order-derivative-expression} of the main text, since $[\gram_G(\rho)]_{I,J}=\ip{[H_{I},\rho]}{[H_{J},\rho]}$.$\qed$

\section{Genericity of orbit dimensions, and proof of \mbox{Proposition \ref{prop:generic-orbit-dim}}}\label{sec:SM-genericity}

\begin{lemma}\label{lem:generic-orbit-dim}
Let $\mathcal{H}_{\rm fin} \subset \mathcal{S}(\fock)$ be a finite-dimensional subspace of Schwartz states, and let a pure state $\ket{\psi}$ be drawn at random from the complex unit sphere $S_{\mathcal{H}_{\rm fin}}$ (according to the uniform measure).
Then, with probability one,
\begin{align}\label{eq:generic-orbit-dim-lem-equality}
\dim(\orb_{G}(\ket{\psi})) = d_{G,\rm max}(\mathcal{H}_{\rm fin})\,,
\end{align}
with
\begin{align}
d_{G,\rm max}(\mathcal{H}_{\rm fin}) := \max_{\ket{\psi} \in S_{\mathcal{H}_{\rm fin}}}\!\!\big( \dim(\orb_{G}(\ket{\psi})) \,\big)\,.
\end{align}
As a consequence, it also holds that the set of states satisfying \cref{eq:generic-orbit-dim-lem-equality} is dense in $S_{\mathcal{H}_{\rm fin}}$.
\end{lemma}
\begin{proof}
Consider the complex sphere $S_{\mathcal{H}_{\rm fin}}$. By choosing an orthonormal basis of $\mathcal{H}_{\rm fin}$, this sphere is identified to the unit complex sphere in $\mathbb{C}^{\dim(\mathcal{H}_{\rm fin})}$, and hence as well (by splitting real and imaginary parts) to the unit real sphere in $\mathbb{R}^{2\dim(\mathcal{H}_{\rm fin})}$, which we may also denote by $S_{\mathbb{R}}^{b}$ with $b:=2\dim(\mathcal{H}_{\rm fin})-1$. It is a compact, connected manifold of (real) dimension $b$. In fact, it is a smooth and even analytic manifold, which we recall, means that it possesses atlases of local coordinate charts from which it is possible to define the notion of analytic maps $F:S_{\mathcal{H}_{\rm fin}}\to \mathbb{R}$ (they are maps for which their local coordinate representations, from open subsets of $\mathbb{R}^b$ to $\mathbb{R}$, are \textit{real-analytic}, i.e., are smooth and agree locally with their Taylor series). Let $U_1,\dots,U_k$ be such a (finite) atlas of $S_{\mathcal{H}_{\rm fin}}$, i.e., open connected sets $U_i$ of the sphere whose union is the whole sphere, along with their associated local coordinate charts $x_i$, which are homeomorphisms $x_i:U_i\to V_i\subseteq\mathbb{R}^b$ between $U_i \subseteq S_{\mathcal{H}_{\rm fin}}$ and an open subset $V_i\subseteq\mathbb{R}^b$.

Let us fix an $i \in \{1,\dots,k\}$ and focus only on the domain $U_i$ of the sphere, along with its local coordinate chart $x_i:U_i\to V_i\subseteq\mathbb{R}^b$.
Let $\ket{\psi_i} \in U_i$ be a point in this domain whose orbit dimension (under $G$) achieves the maximal value reached in this domain:
\begin{align}
\dim(\orb_{G}(\ket{\psi_i})) = d_i\,,
\end{align}
where we denoted
\begin{align}\label{eq:def-local-max-orbit-dim-di}
d_i := \max_{\ket{\psi} \in U_i}\!\!\big( \dim(\orb_{G}(\ket{\psi})) \,\big)\,.
\end{align}
Consider the $\dim(G) \times \dim(G)$ Gram matrix $\gram_G(\ket{\psi_i})$ associated with $\ket{\psi_i}$. We recall (see \cref{sec:orbs-of-q-optics}) that it is the matrix defined by the entries
\begin{align}
[\gram_G(\ket{\psi_i})]_{I,J} := \Re \braket{H_I \psi_i}{H_J \psi_i}\,,
\end{align}
with $\mathcal{B}_\mathfrak{g}=i\{H_1,\dots,H_{\dim(G)}\}$ (\cref{tab:Lie-algebra-bases}).
By assumption, this matrix has rank $d_i \leq \dim(G)$. It therefore admits a $d_i \times d_i$ submatrix that is still of rank $d_i$ and hence invertible. Let
$1 \leq I_1 < \cdots < I_{d_i} \leq \dim(G)$ and $1 \leq J_1 < \cdots < J_{d_i} \leq \dim(G)$ be respectively row and column indices such that the associated submatrix of $\gram_G(\ket{\psi_i})$ with these rows and columns is invertible. Let us denote $\bm{I} := (I_1,\dots,I_{d_i})$ and $\bm{J} := (J_1,\dots,J_{d_i})$, and use the notation $A^{(\bm{I},\bm{J})}$ for the submatrix of a matrix $A \in \mathbb{R}^{\dim(G) \times \dim(G)}$ consisting of rows $\bm{I}$ and columns $\bm{J}$.
Now, consider the maps
\begin{align}
\begin{aligned}
M \colon V_i &\longrightarrow \mathbb{R}^{d_i \times d_i}\,,\\
q &\longmapsto \big(\gram_G(x_i^{-1}(q))\,\big)^{(\bm{I},\bm{J})}\,,
\end{aligned}
\end{align}
\begin{align}
\begin{aligned}
D \colon V_i &\longrightarrow \mathbb{R}\,,\\
q &\longmapsto \det(M(q))\,,
\end{aligned}
\end{align}
and for all entry indices $1 \leq r,s \leq d_i$ the maps
\begin{align}
\begin{aligned}
M_{rs} \colon V_i &\longrightarrow \mathbb{R}\,,\\
q &\longmapsto [M(q)]_{rs}\,.
\end{aligned}
\end{align}
First, we claim that the atlas on the sphere $S_{\mathbb{R}}^{b} \subset \mathbb{R}^{b+1}$ can be chosen in such a way that the maps $M_{rs}$ are real-analytic.
Indeed, we choose the standard "stereographic projection" atlas (which is an open cover of size $k=2$), for which one observes (e.g., \cite[p.30]{lee_introduction_2012}) that the charts $x_i$ have all the "entries in the ambient space $\mathbb{R}^{b+1}$" maps $q \mapsto (\tilde{x}_i^{-1}(q))_a := (\iota(x_i^{-1}(q)))_a$ be analytic (for all $a=1,\dots,b+1$, where $\iota: S_{\mathbb{R}}^{b} \to \mathbb{R}^{b+1}$ is the inclusion map).

Then, each map $M_{rs}$ takes the form (for some basis Hamiltonians labels $I_r,J_s$)
\begin{align}\label{eq:proof-prop1-Mrs-form}
M_{rs}(q) = \Re \braket{H_{I_r} \mathcal{I}(\tilde{x}_i^{-1}(q))}{H_{J_s} \mathcal{I}(\tilde{x}_i^{-1}(q))}\,,
\end{align}
where $\mathcal{I}:\mathbb{R}^{b+1} \to \mathcal{H}_{\rm fin}$ is the isomorphism achieving the identification mentioned above.
Recall that polynomial operators are continuous when seen as operators from Schwartz space to itself (see \cref{sec:SM-Schwartz-spaces-and-proof-of-lsc}). Thus the Hamiltonians $H_I$, considered as operators $H_I:\mathcal{S}(\fock) \to \mathcal{S}(\fock)$, are linear and continuous on the topological vector space $\mathcal{S}(\fock)$. But clearly a linear and continuous map on such a space is differentiable with its derivative being constant and equal to itself; and thus it is smooth (the derivatives of higher orders are all zero), and hence even analytic (it is equal to its Taylor series at any point).
Likewise, the identification map $\mathcal{I}$ is linear and continuous, and hence analytic, and the real part of the inner product map $\Re \braket{\cdot}{\cdot \cdot}: \fock \times \fock \to \mathbb{R}$ being continuous and bilinear, is also an analytic map for analogous reasons.
Lastly, as the Schwartz topology is finer than the norm topology from $\fock$ (cf. \cref{sec:SM-Schwartz-spaces-and-proof-of-lsc}), the inclusion map $\mathcal{S}(\fock) \to \fock$ is continuous, and hence analytic as well.
Thus, the map $M_{rs}$ is analytic, by being a composition [\cref{eq:proof-prop1-Mrs-form}] of analytic maps.

Since the determinant of a matrix is an analytic function of its entries, the analyticity of the maps $M_{rs}$ implies by composition that the map $D$ is also analytic.
Because $D$ is a real-analytic function defined on an open connected domain of $V_i \subseteq \mathbb{R}^b$ over which it is not identically zero (it is nonzero at the point $x_i(\ket{\psi_i}) \in V_i$), its zero set $D^{-1}(\{0\})$ must have zero Lebesgue measure in $\mathbb{R}^{b}$ \cite{Mityagin-ZeroSet-2020}.

We now justify why this fact carries back over to the sphere $S_{\mathcal{H}_{\rm fin}}$, i.e., to the fact that the set $Z_i := x_i^{-1}( D^{-1}(\{0\}) )$ has measure zero in $S_{\mathcal{H}_{\rm fin}}$ for its uniform measure.

The uniform measure on the sphere $S_{\mathbb{R}}^{b} \subset \mathbb{R}^{b+1}$ is the unique rotationally invariant measure on the sphere which assigns measure $1$ to the whole sphere. But it can also be defined in a geometrical way, as the measure $\mu^{(g)}$ induced by the Riemannian metric $g$ on the sphere (the "round metric"), itself induced by the Euclidean metric on $\mathbb{R}^{b+1}$. Indeed, the metric $g$ induces a measure $\mu^{(g)}$ on the sphere, which is defined by pulling back the Lebesgue measure $\nu_{\mathbb{R}^{b}}$ of $\mathbb{R}^{b}$ along the local charts $x_i$ and reweighting them using $g$. Explicitly, for any set $W \subseteq U_i$ [such that $x_i(W)$ is $\nu_{\mathbb{R}^{b}}$-measurable]:
\begin{align}\label{eq:measure-induced-by-metric}
\mu^{(g)}(W) := \int_{x_i(W)} \sqrt{|\mathcal{G}_i(x_i^{-1}(q))|} \,d\nu_{\mathbb{R}^{b}}(q)\,,
\end{align}
where $\mathcal{G}_i(p)$ denotes the determinant of the matrix representation of the metric tensor $g$ in the local coordinates $x_i$, evaluated at point $p\in U_i$. Note that this reweighting by the metric is responsible for making the local definition of the measure $\mu^{(g)}$ on the sphere \cref{{eq:measure-induced-by-metric}} actually independent of the local charts chosen.
We refer to Ref.~\cite[Sec. XII.1]{Amann-AnalysisIII-2009} for a clear exposition.
As one can check (e.g., \cite[p.61]{Lee-IntroductionRiemannian-2018}), for the stereographic atlas of the sphere $S_{\mathbb{R}}^{b}$ both functions $\mathcal{G}_i$ are bounded on their domain $U_i$, i.e., $0 < \mathcal{G}_i \leq C^2,\ i=1,2$ (with $C=2^b$). It thus follows from \cref{eq:measure-induced-by-metric} that
\begin{align}
\mu^{(g)}(W) \leq C \,\, \nu_{\mathbb{R}^{b}}(x_i(W))\,.
\end{align}
Applying this inequality to the case $W=Z_i$ gives a zero right-hand side, and hence that $Z_i$ has measure zero in $S_{\mathcal{H}_{\rm fin}}$ for its uniform measure.

Letting go of a fixed index $i$ and a particular local domain $U_i$ of the sphere, we are now exactly in the setting of \cref{lem:local-generic-on-connected-implies-global-generic} (which we had already used in \cref{subsec:proof-thm-2}), where
on the connected space $S_{\mathcal{H}_{\rm fin}}$ we have a function $f(\ket{\psi}) := \dim(\orb_{G}(\ket{\psi}))$ that is lower-semicontinuous (cf. \cref{prop:lsc-orbit-dim}, noting that when restricting to a finite-dimensional setting as is done here, the Schwartz topology coincides back to the standard topology of $S_{\mathcal{H}_{\rm fin}}$) and constant on full-$\mu^{(g)}$-measure subsets $U_i \setminus Z_i$ of every open set $U_i$ of a (finite) cover of $S_{\mathcal{H}_{\rm fin}}$ (and a measure $\mu^{(g)}$ induced from a metric $g$ is indeed always strictly positive \cite[Theorem XIII.1.7]{Amann-AnalysisIII-2009}).

This lemma thus yields that $\mathcal{U}:=\cup_{i=1,\dots,k} (U_i \setminus Z_i)$ is an open set of full $\mu^{(g)}$ measure inside $S_{\mathcal{H}_{\rm fin}}$, and that on it, the function $f$ is constant and equal to $d_{G,\rm max}(\mathcal{H}_{\rm fin})$. This establishes the main claim of the lemma (as when we "draw a point $\ket{\psi}$ at random from $S_{\mathcal{H}_{\rm fin}}$ according to the uniform measure", we land inside this full-measure subset with probability $1$).

As a consequence of $\mathcal{U}$ being of full $\mu^{(g)}$ measure, because this measure is strictly positive, $S_{\mathcal{H}_{\rm fin}} \setminus \mathcal{U}$ cannot contain any open set, i.e., it has empty interior. Equivalently, this means that $\mathcal{U}$ is dense in $S_{\mathcal{H}_{\rm fin}}$.

\end{proof}
We remark that a version of \cref{lem:generic-orbit-dim} can be stated and proved in full analogy for the ketbra case instead of the ket case.
Indeed, one proceeds with the same proof, still on the unit ket sphere on $\mathcal{H}_{\rm fin}$, but given a $\ket{\psi}$ on this sphere, the orbit dimension of the ketbra $\ketbra{\psi}{\psi}$ is considered instead, i.e., the Gram matrix with entries \cref{eq:concrete-gram-formula-ketbra-expression} is considered instead of \cref{eq:concrete-gram-formula-ket-expression}. The proof proceeds likewise, and yields a new full $\mu^{(g)}$ measure inside $S_{\mathcal{H}_{\rm fin}}$, on which the ketbra orbit-dimension function $\dim(\orb_{G}(\ketbra{\psi}{\psi}))$ is constant and equal to the associated maximal value. Since the Haar measure of pure density matrices on $\mathcal{H}_{\rm fin}$ is realized by the pushforward of the uniform measure on the unit sphere $S_{\mathcal{H}_{\rm fin}}$ via the projection map (the map that sends a $\ket{\psi}$ to $\ketbra{\psi}{\psi}$), this readily gives that if a pure density matrix $\ketbra{\psi}{\psi}$ on $\mathcal{H}_{\rm fin}$ is drawn at random from the Haar measure, then with probability one its $\dim(\orb_{G}(\ketbra{\psi}{\psi}))$ value is the maximal value attained over all such pure density matrices.

We are now able to prove the main text's \cref{prop:generic-orbit-dim}:

\paragraph{Proof of \cref{prop:generic-orbit-dim}.}
Due to \cref{lem:generic-orbit-dim}, it remains
to give a justification for the expression of the quantity $d_{G,\rm max}(\mathcal{H}_{\rm fin})$ in the case of the finite-dimensional subspaces
$\mathcal{H}_{\rm fin} := \mathcal{H}_{m}^{\leq N}$ of total photon number at most $N$.
We claim that
\begin{align}\label{eq:generic-orbit-dim-eq-SM}
d_{G,\rm max}(\mathcal{H}_{m}^{\leq N}) = \dim(G) - \delta_{N=0} m^2 - \delta_{N=1} (m-1)^2\,.
\end{align}
We first prove that for all $\ket{\psi} \in \mathcal{H}_{m}^{\leq N}$
\begin{equation}\label{eq:generic-orbit-dim-ineqfirst-SM}
\dim(\orb_{G}(\ket{\psi})) \leq \dim(G) - \delta_{N=0} m^2 - \delta_{N=1} (m-1)^2\,.
\end{equation}
Note that here we again use the orbit-stabilizer dimension formula [cf. \cref{eq:orb-stab-dim-formula-generic}].
For $N\geq2$, the bound $d_{G,\rm max}(\mathcal{H}_{m}^{\leq N})\leq \dim(G)$ is immediate from \cref{eq:concrete-rank-formula-ket}. For $N=0$, as the vacuum is annihilated by the whole $m^2$-dimensional PLO Lie algebra (since they preserve photon number), and as this PLO algebra is contained inside all four Lie algebras of \cref{tab:Lie-algebra-bases}, the vacuum's orbit dimension is at most $\dim(G) - m^2$. For $N=1$, let $\ket{\psi} = \alpha \ket{\bm{0}} + \beta\ket{\psi^{1}} \in \mathcal{H}_{m}^{\leq1}$, with $\ket{\psi^{1}}$ in the single-photon subspace $\hilone^1$. If $\beta=0$, the above $N=0$ bound applies, which implies that the (looser) $N=1$ bound also holds. Suppose now that $\beta \neq 0$. Since $G_{\rm PLO}$ is universal on $\hilone^1$, there always exists a PLO unitary $U$ that maps $\ket{1,0,\dots,0}$ to $\ket{\psi^{1}}$, and hence such that $\ket{\psi} = \alpha \ket{\bm{0}} + \beta U\ket{1,0,\dots,0} = U\left( \alpha \ket{\bm{0}} + \beta \ket{1,0,\dots,0} \right) := U\ket{\psi'}$. Consequently, $\ket{\psi}$ and $\ket{\psi'}$ have the same $G$-orbit dimension, but the latter is annihilated by all PLO Hamiltonians acting only on the modes $2,\dots,m$, which form a subalgebra (of all $m$-mode PLO Hamiltonians) of dimension $(m-1)^2$; hence $\dim(\orb_{G}(\ket{\psi})) \leq \dim(G) - (m-1)^2$.

Next, let us prove that these upper bounds on orbit dimensions are attained, by exhibiting for each setting $(m,N)$ explicit example states $\ket{\psi_{m,N}}$ for which \cref{eq:generic-orbit-dim-ineqfirst-SM} becomes an equality.
Recall that Hamiltonians in the Lie algebra of ($m$-mode) $G_{\rm GO}$ are exactly (in the normal-ordered form) all Hermitian polynomial operators of degree $\leq 2$. Those can be exactly written as the Hamiltonians of the form
\begin{align}\label{eq:generic-orbit-dim-general-GO-Hamiltonian-normal-ordered}
H
=&\ c\id+\sum_{j,k=1}^m A_{jk}a_j^\dagger a_k
+\frac{1}{2}\sum_{j,k=1}^m\left(B_{jk}a_j^\dagger a_k^\dagger+\overline{B_{jk}}a_j a_k\right)\nonumber\\
&+\sum_{k=1}^m\left(u_k a_k^\dagger+\overline{u_k}a_k\right)\,,
\end{align}
for $c\in\mathbb{R}$, $A=A^\dagger\in\mathbb{C}^{m\times m}$, $B=B^\top\in\mathbb{C}^{m\times m}$, and $\bm{u}=(u_1,\dots,u_m)\in\mathbb{C}^m$. Furthermore, note that the Hamiltonians in the algebras of the other three groups are still of this form but with additional constraints on these parameters.

For the $N=0$ case, the only state choice is the vacuum $\ket{\psi_{m,0}}:=\ket{\bm0}$. We could appeal to the results given in \cref{tab:Lie-algebra-bases}, but to be explicit, we have, for an $H$ of the above form:
\begin{align}
H\ket{\bm0}
=c\ket{\bm0}+\sum_{k=1}^m u_k a_k^\dagger\ket{\bm0}
+\frac{1}{2}\sum_{j,k=1}^m B_{jk}a_j^\dagger a_k^\dagger\ket{\bm0}\,.
\end{align}
These three terms lie in pairwise orthogonal photon-number sectors; hence, $H\ket{\bm0}=0$ implies that $c=0$, $\bm{u}=0$, and $B=0$ ($A$ is left arbitrary). This implies that $H$ must be a Hamiltonian in the algebra of $G_{PLO}$. Since conversely those all annihilate the vacuum, this establishes that $\dim(\orb_{G}(\ket{\bm0})) = \dim(G) - m^2$ for all four considered groups $G$.

For the $N=1$ case, let us take $\ket{\psi_{m,1}} := \alpha \ket{\bm{0}} + \beta \ket{1,0,\dots,0} = \alpha \ket{\bm{0}} + \beta a_1^\dagger\ket{\bm{0}}$, with $\alpha,\beta \in \mathbb{C}\setminus\{0\}$. Suppose that $H\ket{\psi_{m,1}}=0$.
Since the three-photon component of $H\ket{\psi_{m,1}}$ is
\begin{align}
\frac{\beta}{2}\sum_{j,k=1}^m B_{jk}a_j^\dagger a_k^\dagger a_1^\dagger\ket{\bm0}\,,
\end{align}
this equality implies that $B=0$. The two-photon component of $H\ket{\psi_{m,1}}$ thus reduces to
\begin{align}
\beta \sum_{k=1}^m u_k a_k^\dagger a_1^\dagger\ket{\bm0}\,,
\end{align}
so the equality implies in turn that $\bm{u}=0$. With $B=0$ and $\bm{u}=0$, only the constant and passive terms in $H$ remain. Since the passive part annihilates the vacuum component, and since furthermore
\begin{align}
a_k a_1^\dagger\ket{\bm0}
= [a_k,a_1^\dagger]\ket{\bm0}+a_1^\dagger a_k\ket{\bm0}
=\delta_{k1}\ket{\bm0}\,,
\end{align}
we get
\begin{align}
H\ket{\psi_{m,1}}
=c \alpha \ket{\bm0}
+\beta\left(c a_1^\dagger+\sum_{j=1}^m A_{j1}a_j^\dagger\right)\ket{\bm0}\,,
\end{align}
which implies that $c=0$ and that the first column of $A$ vanishes. Since $A$ is Hermitian, it means that both its first row and column vanish, leaving its remaining block to be an arbitrary $(m - 1) \times (m-1)$ Hermitian matrix.
This means that $H$ must be a Hamiltonian in the subalgebra of $G_{PLO}$ of passive Hamiltonians over the modes $2,\dots,m$. Since conversely these annihilate $\ket{\psi_{m,1}}$, this establishes that $\dim(\orb_{G}(\ket{\psi_{m,1}})) = \dim(G) - (m-1)^2$ for all four considered groups $G$.

Finally, for the $N\geq2$ case, let us take the state $\ket{\psi_{m,2}} \in \mathcal{H}_{m}^{\leq2}\subseteq\mathcal{H}_{m}^{\leq N}$ given by
\begin{align}
\ket{\psi_{m,2}}:=\alpha \ket{\bm{0}} +\beta\, Q(\bm{a}^\dagger)\ket{\bm0}\,,\quad
Q(\bm{a}^\dagger):=\frac{1}{2}\sum_{k=1}^m \lambda_k (a_k^\dagger)^2\,,
\end{align}
for any pairwise distinct positive reals $\lambda_1,\dots,\lambda_m$ and complex scalars $\alpha,\beta \in \mathbb{C}\setminus\{0\}$.
If $H\ket{\psi_{m,2}}=0$, the four-photon component is
\begin{align}
\frac{\beta}{2}\sum_{j,k=1}^m B_{jk}a_j^\dagger a_k^\dagger Q(\bm{a}^\dagger)\ket{\bm0}\,,
\end{align}
and hence $B=0$. The three-photon component then gives
\begin{align}
\beta \sum_{k=1}^m u_k a_k^\dagger Q(\bm{a}^\dagger)\ket{\bm0}=0\,,
\end{align}
implying that $\bm{u}=0$. With $B=0$ and $\bm{u}=0$, the equality on the zero-photon component gives $c=0$. The remaining two-photon component is
\begin{align}
\sum_{j,k=1}^m A_{jk}\lambda_k a_j^\dagger a_k^\dagger\ket{\bm0}\,.
\end{align}
The coefficient of $(a_k^\dagger)^2\ket{\bm0}$ is $A_{kk}\lambda_k$, implying that $A_{kk}=0$. For $j<k$, the coefficient of $a_j^\dagger a_k^\dagger\ket{\bm0}$ is
\begin{align}
A_{jk}\lambda_k+\overline{A_{jk}}\lambda_j\,,
\end{align}
hence
\begin{align}
A_{jk}\lambda_k = - \overline{A_{jk}}\lambda_j\,,
\end{align}
which implies, by taking the modulus and rearranging,
\begin{align}
|A_{jk}|(\lambda_k - \lambda_j) = 0\,,
\end{align}
implying that $A_{jk}=0$.
Hence $A=0$, and so $H=0$, establishing that among all Hamiltonians of degree $\leq 2$, only the zero Hamiltonian annihilates $\ket{\psi_{m,2}}$, thus $\dim(\operatorname{Stab}_{G}(\ket{\psi_{m,2}})) = 0$, and hence $\dim(\orb_{G}(\ket{\psi_{m,2}})) = \dim(G)$ for all four considered groups $G$.

Therefore, the upper bound $d_{G,\rm max}(\mathcal{H}_{m}^{\leq N})\leq \dim(G)$ is achieved for every $N\geq2$ and for all four groups.
We have hence shown that
\begin{equation}
\dim(\orb_{G}(\ket{\psi_{m,N}})) = \dim(G) - \delta_{N=0} m^2 - \delta_{N=1} (m-1)^2\,,
\end{equation}
which concludes the proof.$\qed$

For the ketbra picture, one could show (by an analogous treatment) the analogous claims, i.e., that the value that a Haar random pure density matrix on $\mathcal{H}_{\rm fin}$ achieves with probability one for its orbit dimension is
\begin{align}
\dim(\orb_{G}(\ketbra{\psi}{\psi})) =
\begin{aligned}[t]
&d_{G,\rm max}(\mathcal{H}_{m}^{\leq N})\\
&-\ \delta_{G \in (G_{\rm DPLO},G_{\rm ALO},G_{\rm GO})}\,,
\end{aligned}
\end{align}
where $d_{G,\rm max}$ is the ket picture maximal value given in \cref{eq:generic-orbit-dim-eq-SM}. The subtraction of $1$ in the cases of $G=G_{\rm DPLO},G_{\rm ALO},G_{\rm GO}$ reflects the fact that their Lie algebra bases given in \cref{tab:Lie-algebra-bases} contain the identity operator, which always contributes to a zero vector ($[\id,\rho]=0$) in the list of vectors of \cref{eq:concrete-rank-formula-density}.

In earlier numerical explorations, we also observed that the following different family of states $\ket{\psi_{m,N}}$ achieves the generic orbit dimension value: it is the uniform superposition of all Fock basis states of total photon number at most $\min(2,N)$, with each term in the superposition having a different phase factor, with equal spacing of phases.

While we do not prove its genericity, we still include it as a conjecture for completeness.
Explicitly, define
\begin{align}\label{eq:generic-orbit-dim-uniform-superp-2photons-Nmodes-linearstate}
\ket{\psi_{m,N}} := \frac{1}{\sqrt{J_m}} \sum_{j=1}^{J_m} e^{2\pi i j/J_m} \ket{\bm{n}_{j}}\,,
\end{align}
where the Fock basis of zero to $\min(2,N)$ photons $(\ket{\bm{n}_{1}},\ket{\bm{n}_{2}},\dots,\ket{\bm{n}_{J_m}})$ is ordered lexicographically, and contains $J_m = 1, (m+1), (m+1)(m+2)/2$ elements for $N=0$, $N=1$, and $N\geq2$, respectively. Then, we have the following conjecture:
\begin{conjecture}\label{conj:universal-maximal-orbit-dim-family}
For all $N\geq0, m\geq1$, and for all four groups $G$ considered (\cref{tab:Lie-algebra-bases}), the state $\ket{\psi_{m,N}}$ of \cref{eq:generic-orbit-dim-uniform-superp-2photons-Nmodes-linearstate} achieves the maximal orbit dimension value [right-hand side of \cref{eq:generic-orbit-dim-eq-SM}].
\end{conjecture}

\section{Minimal orbit dimensions and some relations with non-Gaussianity, stellar rank, and symplectic rank}\label{sec:SM-min-orb-dims}

\subsection{Affine symplectic representation of Gaussian optics for mixed states, and symplectic spectrum}\label{subsecs:affine-q-optics-introduction}

For clarity, in this appendix we denote canonical operators with hats.
Let $\hat{\bm{r}} := (\hat{q}_1,\hat{p}_1,\dots,\hat{q}_m,\hat{p}_m)^\intercal$ be the vector of canonical displacement operators, and denote by $\Omega$ the symplectic form associated with the "$qpqp\dots$" convention \cite{Adesso-ContinuousVariable-2014}, i.e., the $2m \times 2m$ matrix
\begin{equation}
\Omega := \bigoplus_{j=1}^m \Omega_1\,,\quad\text{ with }\quad\Omega_1 := \begin{pmatrix} 0 & 1\\ -1 & 0 \end{pmatrix}\,.
\end{equation}
Denote also by $\mathrm{Sp}(2m,\mathbb{R})$ the symplectic group of matrices represented in this convention:
\begin{equation}
\mathrm{Sp}(2m,\mathbb{R}) := \{ S \in \mathrm{GL}(2m,\mathbb{R}) \,|\, S \,\Omega\, S^\intercal = \Omega \}\,.
\end{equation}

For a Schwartz operator $\rho$ and a polynomial operator $O$, denote the associated expectation value of $O$ by
$\left\langle O \right\rangle_\rho := \Tr[ \rho \, O ]$.

The \textit{first moment} of $\rho$ is the vector $\bm{r}_\rho \in \mathbb{R}^{2m}$ defined by $(\bm{r}_\rho)_{j} := \left\langle \hat{r}_j \right\rangle_\rho$.
The \textit{covariance matrix} of $\rho$ is the matrix $V_\rho \in \mathbb{R}^{2m \times 2m}$ with entries
\begin{equation}
[V_\rho]_{kl} := \operatorname{Cov}_{\rho}(\hat{r}_k,\hat{r}_l)\,,
\end{equation}
where (like in the main text) we use the notation $\operatorname{Cov}_{\rho}(O,O'):=\left\langle \{O,O'\} \right\rangle_\rho - \left\langle O \right\rangle_\rho \left\langle O' \right\rangle_\rho$, with $\{A,B\}:=\frac{1}{2}(AB + BA)$ \cite{Adesso-ContinuousVariable-2014}.

Previously, in \cref{sec:SM-EMRep-and-proof-of-orbit-structure} and onwards, unitaries of Gaussian optics were modeled as the image of the extended metaplectic representation $\phi:G_0 \to \uni(\fock)$.
To recall in more detail (cf. proof of \cref{lem:optical-groups-are-images-of-closed-subgroups-under-EMR}), here $G_0$ is the extended metaplectic Lie group given by [cf. \cref{eq:EMG-explicit-semidirect-def}]
\begin{equation}
G_0 := \mathrm{H}_{2m+1} \rtimes \mathrm{Mp}(2m,\mathbb{R})\,,
\end{equation}
with $\mathrm{H}_{2m+1}$ and $\mathrm{Mp}(2m,\mathbb{R})$ the so-called Heisenberg and metaplectic groups, respectively; and where $\mathrm{Mp}(2m,\mathbb{R})$ is related to the symplectic group $\mathrm{Sp}(2m,\mathbb{R})$ through a 2-to-1 Lie homomorphism $\Pi : \mathrm{Mp}(2m,\mathbb{R}) \to \mathrm{Sp}(2m,\mathbb{R})$.

In practice, however, a Gaussian quantum optics unitary is legitimately specified via a pair $(\bm{d},S)$ of a displacement vector $\bm{d} \in \mathbb{R}^{2m}$ and a symplectic matrix $S \in \mathrm{Sp}(2m,\mathbb{R})$.
This is because when acting on density operators, the technical issues associated with global phase on the unitary operator (which were mentioned in \cref{subsec:background-on-optics-through-a-repr}) disappear due to the adjoint action. In other words, the pair $(\bm{d},S)$ \textit{does} specify a Gaussian unitary \textit{channel} uniquely (and there is in fact a one-to-one correspondence) \cite{weedbrook_gaussian_2012,arvind_real_1995}, and furthermore this mapping is a valid group homomorphism $\tilde{\Phi}: \tilde{G_0} \to \uni(B_2(\fock))$,
where $\tilde{G_0} := \mathbb{R}^{2m} \rtimes \mathrm{Sp}(2m,\mathbb{R})$ is called the affine symplectic group.

To ease the notation, from now on in this section we drop the tildes and simply write $G_0$ for $\tilde{G_0}$ and $\Phi$ for $\tilde{\Phi}$ --- so that contrary to the previous Appendix sections, $G_0$ is here the affine symplectic group and $\Phi$ is its representation (that acts on density operators).
Note in particular that
\begin{equation}\label{eq:dim-G0-affine-symplectic}
\dim(G_0) = 2m^2 + 3m
\end{equation}
is one less than the main text's Gaussian optics unitary group dimension $\dim(G_{\rm GO})$ (\cref{tab:Lie-algebra-bases}),
corresponding to the fact that the global phase unitaries are not part of $G_0$ anymore.
Likewise, as the Lie algebra bases for the groups $G_{\rm DPLO}$ and $G_{\rm ALO}$ given in the main text's \cref{tab:Lie-algebra-bases} contain the identity, their corresponding Lie subgroups $G$ of $G_0$ in this section will also be one dimension smaller. This does not change density-operator orbits (as the global phases disappear through the adjoint action).
The four quantum optical settings considered in the main text correspond here to the Lie subgroups of $G_0$ whose underlying sets are
\begin{align}
G_{\rm PLO} =&\ \{ \bm{0} \} \,\times\, \mathrm{Sp}(2m,\mathbb{R}) \cap \mathrm{O}(2m)\,,\label{eq:G0-affine-symplectic-subgroup-def-PLO}\\
G_{\rm DPLO} =&\ \mathbb{R}^{2m} \,\times\, \mathrm{Sp}(2m,\mathbb{R}) \cap \mathrm{O}(2m)\,,\\
G_{\rm ALO} =&\ \{ \bm{0} \} \,\times\, \mathrm{Sp}(2m,\mathbb{R})\,,\\
G_{\rm GO} =&\ \mathbb{R}^{2m} \,\times\, \mathrm{Sp}(2m,\mathbb{R}) = G_0\,.
\end{align}

For a Lie subgroup $G \subset G_0$, denote by $\operatorname{Stab}_{G}(\rho)$ the stabilizer subgroup of $\rho$ under the action of the representation $\Phi$ restricted to $G$, i.e.,
\begin{equation}
\operatorname{Stab}_{G}(\rho) := \{ g=(\bm{d},S) \in G \,|\, \Phi(g) \cdot \rho = \rho \}\,.
\end{equation}
Likewise, we denote the corresponding orbits in Fock space by
\begin{equation}
\orb_{G}(\rho) := \{ \Phi(g) \cdot \rho \,|\, g \in G \}\,.
\end{equation}

We recall the action of Gaussian unitary channels on a state's first moment and covariance matrix \cite{weedbrook_gaussian_2012,Adesso-ContinuousVariable-2014,arvind_real_1995,Serafini-QuantumContinuous-2023}: for any (Schwartz) state $\rho$ and $(\bm{d},S) \in G_0$, the corresponding evolved state $\rho' := \Phi(\bm{d},S) \cdot \rho$ satisfies
\begin{align}
\bm{r}_{\rho'} &= S \, \bm{r}_\rho + \bm{d}\,,\label{eq:affine-symplectic-action-on-first-moment}\\
V_{\rho'} &= S \,V_\rho\, S^\intercal\,.\label{eq:affine-symplectic-action-on-covariance-matrix}
\end{align}

Next, we recall concepts relative to the symplectic spectrum and the symplectic rank of a state.
Williamson's theorem \cite{Adesso-ContinuousVariable-2014,Nicacio-WilliamsonTheorem-2021} asserts that for any matrix $V \in \mathbb{R}^{2m \times 2m}$ that is symmetric and positive definite, there exists a unique diagonal matrix $D=\operatorname{diag}(\sigma_1,\sigma_1,\sigma_2,\sigma_2,\dots,\sigma_m,\sigma_m)$ with $0 \leq \sigma_1 \leq \cdots \leq \sigma_m$ such that
\begin{equation}\label{eq:Williamson-dg-claim}
S \,V\, S^\intercal = D
\end{equation}
for some symplectic matrix $S \in \mathrm{Sp}(2m,\mathbb{R})$.
As the covariance matrix $V_\rho$ of a (Schwartz) state $\rho$ is always (symmetric and) positive definite \cite{Adesso-ContinuousVariable-2014}, this theorem is always applicable to it. The $m$-tuple $\bm{\sigma}:=(\sigma_1,\sigma_2,\dots,\sigma_m)$ is then called the \textit{symplectic spectrum} of $V$ (resp. $\rho$), and its elements $\sigma_j$ the \textit{symplectic eigenvalues}. In fact, such covariance matrices $V_\rho$ of states $\rho$ always satisfy $1/2 \leq \sigma_1$ \cite{Adesso-ContinuousVariable-2014} (beware that in several references such as Ref.~\cite{Adesso-ContinuousVariable-2014}, different convention choices such as a covariance matrix doubled compared with ours, lead to $1 \leq \sigma_1$ instead).
Throughout this section, we denote by $\aleph(\rho)$ the number of symplectic eigenvalues $\sigma_j$ in $\bm{\sigma}$ that are strictly greater than $1/2$; this quantity satisfies $0 \leq \aleph(\rho) \leq m$ and is called the \textit{symplectic rank} of $\rho$ \cite{Adesso-ContinuousVariable-2014,Serafini-QuantumContinuous-2023}. [Beware that $\aleph(\rho)$ is called \textit{ordinary symplectic rank} in the main text, and for impure states it differs from the notion of symplectic rank that is defined in Ref.~\cite{Mele-SymplecticRank-2026}, cf. \cref{subsec:relations-with-stellar-or-symplectic-ranks}.]
We also denote by $r=r_V$ (resp. $r_{\rho}$) the number of \textit{distinct} symplectic eigenvalues of $V$ (resp. $\rho$), by $\eta_1 < \cdots < \eta_{r}$ the distinct symplectic eigenvalues, and by $k_1,\dots,k_{r}$ their multiplicities, i.e., $k_a$ is the number of times that $\eta_a$ appears in $\bm{\sigma}$.

\subsection{Technical lemma: Symplectic stabilizers of covariance matrices}

\begin{lemma}\label{lem:stabilizer-of-V-under-Sp}
Let $V \in \mathbb{R}^{2m \times 2m}$ be a symmetric positive-definite matrix.
Then, the stabilizer subgroup of $V$ under the congruence action of $\mathrm{Sp}(2m,\mathbb{R})$ is given by
\begin{equation}\label{eq:stabilizer-of-V-under-Sp-claim}
\operatorname{Stab}_{\mathrm{Sp}(2m,\mathbb{R})}(V)= S^{-1} \left(\bigtimes_{a=1}^{r} \mathrm{Sp}(2 k_a,\mathbb{R}) \cap \mathrm{O}(2 k_a) \right) S\,,
\end{equation}
where $S \in \mathrm{Sp}(2m,\mathbb{R})$ is a symplectic matrix that achieves a Williamson decomposition of $V$ [i.e., such that \cref{eq:Williamson-dg-claim} holds].
Here, $k_1,\dots,k_r$ are the multiplicities in the symplectic spectrum of $V$, and the Cartesian product in parentheses denotes the subset of $\mathrm{Sp}(2m,\mathbb{R})$ consisting of block-diagonal matrices $S_1 \oplus \cdots \oplus S_r$ with $S_a \in \mathrm{Sp}(2 k_a,\mathbb{R}) \cap \mathrm{O}(2 k_a)$ for all $a=1,\dots,r$.

\begin{proof}
Let $S \in \mathrm{Sp}(2m,\mathbb{R})$ be as in the statement, i.e., such that $S V S^\intercal = D$ is a Williamson decomposition of $V$ [\cref{eq:Williamson-dg-claim}], or, equivalently,
\begin{equation}\label{eq:stabilizer-of-V-under-Sp-Williamson-decomposition}
V = S^{-1} D (S^{-1})^\intercal\,.
\end{equation}
Note that we have
\begin{equation}\label{eq:stabilizer-of-V-under-Sp-D-block-form}
D = \oplus_{a=1}^r \eta_a \id_{2 k_a}
\end{equation}
and
\begin{equation}\label{eq:stabilizer-of-V-under-Sp-Omega-block-form}
\Omega = \oplus_{a=1}^r \Omega_{k_a}\,,
\end{equation}
where $\id_{2 k_a}$ is the $2 k_a \times 2 k_a$ identity matrix, and $\Omega_{k_a}$ is the $2 k_a \times 2 k_a$ matrix given by $\Omega_{k_a} := \bigoplus_{j=1}^{k_a} \Omega_1$.

The "$\supseteq$" inclusion in \cref{eq:stabilizer-of-V-under-Sp-claim} is straightforward to check: if $\mathcal{S} := S^{-1} \left( \oplus_{a=1}^r S_a \right) S$ with each $S_a \in \mathrm{Sp}(2 k_a,\mathbb{R}) \cap \mathrm{O}(2 k_a)$, then
\begin{align}
\mathcal{S} V \mathcal{S}^\intercal
&= S^{-1} \left( \oplus_{a=1}^r S_a \right) S \,\,V\,\, S^{\intercal} \left( \oplus_{a=1}^r S_a^\intercal \right) (S^{-1})^\intercal\\
&= S^{-1} \left( \oplus_{a=1}^r S_a \right) S \,\,[S^{-1} D (S^{-1})^\intercal]\,\, S^{\intercal} \left( \oplus_{a=1}^r S_a^\intercal \right) (S^{-1})^\intercal\\
&= S^{-1} \left( \oplus_{a=1}^r S_a \right) \,\, D \,\, \left( \oplus_{a=1}^r S_a^\intercal \right) (S^{-1})^\intercal\\
&= S^{-1} \left( \oplus_{a=1}^r S_a \right) \,\, \left(\oplus_{a=1}^r \eta_a \id_{2 k_a}\right) \,\, \left( \oplus_{a=1}^r S_a^\intercal \right) (S^{-1})^\intercal\\
&= S^{-1} \left( \oplus_{a=1}^r \eta_a S_a S_a^\intercal \right) (S^{-1})^\intercal\\
&= S^{-1} \left( \oplus_{a=1}^r \eta_a \id_{2 k_a} \right) (S^{-1})^\intercal\\
&= S^{-1} D (S^{-1})^\intercal\\
&= V\,,
\end{align}
where \cref{eq:stabilizer-of-V-under-Sp-Williamson-decomposition} was used in the second and the last equality, \cref{eq:stabilizer-of-V-under-Sp-D-block-form} was used in the fourth equality, and $S_a \in \mathrm{O}(2 k_a)$ was used in the sixth equality.

We now show the "$\subseteq$" inclusion.
Let $\mathcal{S} \in \operatorname{Stab}_{\mathrm{Sp}(2m,\mathbb{R})}(V)$.
Let us define
\begin{align}
\tilde{\mathcal{S}} &:= S \,\mathcal{S}\, S^{-1}\label{eq:stabilizer-of-V-under-Sp-tilde-calS-def}
\intertext{as well as}
A &:= S \,V\Omega\, S^{-1}\,.
\end{align}
We claim that
\begin{equation}\label{eq:stabilizer-of-V-under-Sp-eq-claim-on-A}
\tilde{\mathcal{S}} \, A \, \tilde{\mathcal{S}}^{-1} = A\,.
\end{equation}
Indeed, we have
\begin{align}
\tilde{\mathcal{S}} \, A \, \tilde{\mathcal{S}}^{-1}
&= (S \mathcal{S} S^{-1}) (S V \Omega S^{-1}) (S \mathcal{S}^{-1} S^{-1})\\
&= S \mathcal{S} \, V \,\Omega\, \, \mathcal{S}^{-1} S^{-1}\\
&= S \mathcal{S} \, V \,(\mathcal{S}^\intercal \Omega \mathcal{S})\, \, \mathcal{S}^{-1} S^{-1}\\
&= S \,V\Omega\, S^{-1}\\
&= A\,,
\end{align}
where in the third equality we used $\mathcal{S}^\intercal \,\Omega\, \mathcal{S} = \Omega$, which follows from $\mathcal{S}\in\mathrm{Sp}(2m,\mathbb{R})$ and the fact that the latter set is stable under transposition \cite{Nicacio-WilliamsonTheorem-2021}, and in the fourth equality we used $\mathcal{S} \in \operatorname{Stab}_{\mathrm{Sp}(2m,\mathbb{R})}(V)$.
Thus, $\tilde{\mathcal{S}}$ commutes with $A$, which implies that it also commutes with $A^2$. But
\begin{align}
A &= S \,V\Omega\, S^{-1}
= S [S^{-1} D (S^{-1})^\intercal] \Omega S^{-1}
= D \, [(S^{-1})^\intercal \Omega S^{-1}]\\
&= D \Omega
= \left(\oplus_{a=1}^r \eta_a \id_{2 k_a} \right) \left(\oplus_{a=1}^r \Omega_{k_a} \right)\\
&= \oplus_{a=1}^r \eta_a \Omega_{k_a}\,,
\end{align}
where in the fourth equality we used $(S^{-1})^\intercal \,\Omega\, S^{-1} = \Omega$, which again is a consequence of $S^{-1}\in\mathrm{Sp}(2m,\mathbb{R})$, and hence
\begin{align}\label{eq:stabilizer-of-V-under-Sp-Asquared-block-form}
A^2 &= \oplus_{a=1}^r \eta_a^2 \Omega_{k_a}^2
= - \left(\oplus_{a=1}^r \eta_a^2 \id_{2 k_a}\right)\,,
\end{align}
where we used $\Omega_{k_a}^2 = -\id_{2 k_a}$ in the second equality.
Since $\tilde{\mathcal{S}}$ commutes with the above matrix, it must stabilize each of its eigenspaces. But evidently from the diagonal form of \cref{eq:stabilizer-of-V-under-Sp-Asquared-block-form}, this means exactly that $\tilde{\mathcal{S}}$ is block diagonal with blocks of size $2 k_a \times 2 k_a$, i.e., it must be of the form
\begin{equation}\label{eq:stabilizer-of-V-under-Sp-tilde-calS-block-form}
\tilde{\mathcal{S}} = \oplus_{a=1}^r S_a\,,
\end{equation}
with each $S_a \in \mathbb{R}^{2 k_a \times 2 k_a}$.

Now, first, the condition $\tilde{\mathcal{S}} \in \mathrm{Sp}(2m,\mathbb{R})$ ($\tilde{\mathcal{S}} \Omega \tilde{\mathcal{S}}^\intercal = \Omega$) becomes an equation of block matrices, yielding inside each block the condition $S_a \Omega_{k_a} S_a^\intercal = \Omega_{k_a}$, i.e., that $S_a \in \mathrm{Sp}(2 k_a,\mathbb{R})$ for all $a=1,\dots,r$.
Second, the condition $\tilde{\mathcal{S}} \in \operatorname{Stab}_{\mathrm{Sp}(2m,\mathbb{R})}(D)$ ($\tilde{\mathcal{S}} D \tilde{\mathcal{S}}^\intercal = D$, which follows from $\mathcal{S} V \mathcal{S}^\intercal = V$ and $S V S^\intercal = D$) becomes also an equation of block matrices, yielding inside each block
\begin{equation}
S_a \, (\eta_a \id_{2 k_a}) \, S_a^\intercal = \eta_a \id_{2 k_a}\,,
\end{equation}
which, by dividing both sides by $\eta_a$ (which is nonzero since $V$ is positive definite), becomes $S_a S_a^\intercal = \id_{2 k_a}$, i.e., $S_a \in \mathrm{O}(2 k_a)$ for all $a=1,\dots,r$. Putting the two together, we have shown that each $S_a$ must be in $\mathrm{Sp}(2 k_a,\mathbb{R}) \cap \mathrm{O}(2 k_a)$, which along with \cref{eq:stabilizer-of-V-under-Sp-tilde-calS-def,eq:stabilizer-of-V-under-Sp-tilde-calS-block-form} establishes the desired form of $\mathcal{S}$.
\end{proof}
\end{lemma}

\subsection{Lower bounds on minimal orbit dimensions and cases of minimality}\label{subsec:SM-minimal-orbit-dimensions}

We may now prove the following theorem, which is exactly the main text's \cref{thm:minimal-orbdim-statements-main} [here "$\dim(G)$" appears instead of the main text's conditional subtraction by 1, "$\dim(G)^{\scriptscriptstyle\!-}$", because recall that in this entire \cref{sec:SM-min-orb-dims}, the groups considered have no global-phase component at all, as per \cref{subsecs:affine-q-optics-introduction}].

\begin{theorem}[Characterization of minimal orbit dimensions]
\label{thm:minimal-orbdim-statements}
Let $\rho$ be a Schwartz state, and consider one of the four Lie groups $G \in \{G_{\rm PLO},G_{\rm DPLO},G_{\rm ALO},G_{\rm GO}\}$.
Then,
\begin{enumerate}[label=(\arabic*)]
\item the following lower bounds hold:
\begin{equation}\label{eq:minimal-centered-orbdim-general-lowerbounds}
\dim(\orb_{G}(\rho)) \geq \dim(G) - K(\rho) \geq \dim(G) - m^2\,,
\end{equation}
where $K(\rho) := k_1^2 + \cdots + k_r^2$ denotes the sum of the squares of the multiplicities in the symplectic spectrum of $\rho$;

\item
and the following equivalence holds:
\begin{equation}\label{eq:minimal-orbdim-purestates-characterization-proof-GO-equiv}
\begin{gathered}
\dim(\orb_G(\rho)) = \dim(G) - m^2\\
\iff\\
\rho \sim_G \sigma,\\
\end{gathered}
\end{equation}
where $\sigma$ is a state that is stabilized by all of $G_{\rm PLO}$, i.e.,
$\sigma = \oplus_{n \in \mathbb{N}} \, p_n \id_{\fockn}$ for some $p_n \in \mathbb{R}\,.$
\end{enumerate}

\begin{proof}
First, in the cases $G \in \{G_{\rm DPLO},G_{\rm GO}\}$, we will, without loss of generality, suppose that the state $\rho$ is centered in phase space, i.e., $\bm{r}_\rho = \bm{0}$ (indeed, the validity for a state $\rho$ of all the statements of the present theorem implies their validity for all $\rho' \in \orb_G(\rho)$, and in these cases where $G$ contains displacement unitaries, orbits always contain a centered state).

Let $g = (\bm{d},S) \in G$, denote $\rho' := \Phi(g) \cdot \rho$, and suppose that $g \in \operatorname{Stab}_G(\rho)$.
If $G \in \{G_{\rm PLO},G_{\rm ALO}\}$, we have $\bm{d} = \bm{0}$ by definition of these groups $G$, but if $G \in \{G_{\rm DPLO},G_{\rm GO}\}$, we also have $\bm{d} = \bm{0}$, since the fact that $\rho' = \rho$ implies \textit{a fortiori} that $\bm{r}_{\rho'} = \bm{r}_{\rho}$ ($=\bm{0}$ as assumed above), which by \cref{eq:affine-symplectic-action-on-first-moment} yields $\bm{d} = \bm{0}$. The fact that the second moments also match, i.e., $V_{\rho'} = V_\rho$, in turn yields (in all four cases of $G$), via \cref{eq:affine-symplectic-action-on-covariance-matrix}, that $S \in \operatorname{Stab}_{\mathrm{Sp}(2m,\mathbb{R})}(V_\rho)$. Here,
\begin{equation}
\operatorname{Stab}_{\mathrm{Sp}(2m,\mathbb{R})}(V_\rho) := \{ S \in \mathrm{Sp}(2m,\mathbb{R})\,|\, S \,V_\rho\, S^\intercal = V_\rho \}
\end{equation}
is the stabilizer subgroup of $V_\rho$ under the congruence action of $\mathrm{Sp}(2m,\mathbb{R})$.
We have hence obtained the following inclusion of Lie subgroups of $G_0$:
\begin{equation}\label{eq:stab-G-subseteq-stab-Sp}
\operatorname{Stab}_{G}(\rho) \,\subseteq\, \{\bm{0}\} \,\times\, \operatorname{Stab}_{\mathrm{Sp}(2m,\mathbb{R})}(V_\rho)\,,
\end{equation}
which implies that
\begin{equation}\label{eq:stab-G-dim-leq-stab-Sp-dim}
\dim(\operatorname{Stab}_{G}(\rho)) \leq \dim(\operatorname{Stab}_{\mathrm{Sp}(2m,\mathbb{R})}(V_\rho))\,.
\end{equation}
By the orbit-stabilizer dimension formula [\cref{eq:orb-stab-dim-formula-generic}, applied to the representation $\Phi$ restricted to $G$], we have
\begin{equation}\label{eq:orbstabdim-formula-G}
\dim(\orb_{G}(\rho)) = \dim(G) - \dim(\operatorname{Stab}_{G}(\rho))\,.
\end{equation}
But it follows from \cref{lem:stabilizer-of-V-under-Sp} that
\begin{equation}\label{eq:stab-Sp-dim-consequence}
\dim(\operatorname{Stab}_{\mathrm{Sp}(2m,\mathbb{R})}(V_\rho)) = K(\rho)\,,
\end{equation}
with
\begin{equation}\label{eq:K-of-rho-def}
K(\rho) := k_1^2 + \cdots + k_r^2\,.
\end{equation}
Indeed, the Lie group $\mathrm{Sp}(2 k,\mathbb{R}) \cap \mathrm{O}(2 k)$ is isomorphic to the unitary group $\uni(k)$ \cite{arvind_real_1995}, and hence $\dim( \mathrm{Sp}(2 k,\mathbb{R}) \cap \mathrm{O}(2 k) ) = \dim( \uni(k) ) = k^2$.
Notice from \cref{eq:K-of-rho-def} that
\begin{equation}\label{eq:K-of-rho-general-bounds}
m \leq K(\rho) \leq m^2\,,
\end{equation}
with (i) $K(\rho) = m$ if and only if the symplectic eigenvalues of $\rho$ are all distinct, and (ii) $K(\rho) = m^2$ if and only if they are all equal.
Combining \cref{eq:stab-G-dim-leq-stab-Sp-dim,eq:orbstabdim-formula-G,eq:stab-Sp-dim-consequence,eq:K-of-rho-general-bounds} now establishes the desired lower bounds of \cref{eq:minimal-centered-orbdim-general-lowerbounds}.

Next, we prove the equivalence of \cref{eq:minimal-orbdim-purestates-characterization-proof-GO-equiv}.
The direction "$\impliedby$" is straightforward: if $\rho \sim_G \sigma$ with $G_{\rm PLO} \subseteq \operatorname{Stab}_G(\sigma)$, then $\dim(\operatorname{Stab}_G(\sigma)) \geq \dim(G_{\rm PLO}) = m^2$, hence by the orbit stabilizer dimension formula,
\begin{align}
\dim(\orb_G(\rho)) &= \dim(\orb_G(\sigma))\\
&= \dim(G) - \dim(\operatorname{Stab}_G(\sigma))\\
&\leq \dim(G) - m^2\,,
\end{align}
and this must be an equality since the corresponding lower bound of \cref{eq:minimal-centered-orbdim-general-lowerbounds} is already established.

Turning to proving the "$\implies$" direction, suppose that
\begin{equation}\label{eq:minimal-centered-orbdim-general-lowerbounds-proof-impl-forward-hyp}
\dim(\orb_G(\rho)) = \dim(G) - m^2\,.
\end{equation}
Let us consider the possible groups $G$ separately.
\begin{enumerate}[label=(\arabic*)]
\item For $G=G_{\rm PLO}$,
the right-hand side of \cref{eq:minimal-centered-orbdim-general-lowerbounds-proof-impl-forward-hyp} is zero, hence (since $G_{\rm PLO}$ is connected) $\orb_{G_{\rm PLO}}(\rho)$ consists of just a single point ($\rho$), meaning that $\operatorname{Stab}_{G_{\rm PLO}}(\rho) = G_{\rm PLO}$.

\item For $G=G_{\rm DPLO}$,
since we saw that elements in $\operatorname{Stab}_{G}(\rho)$ have no displacement part [as $\rho$ is assumed centered for this case, cf. \cref{eq:stab-G-subseteq-stab-Sp}], we have
\begin{equation}\label{eq:minimal-centered-orbdim-general-lowerbounds-proof-impl-DPLO-inclusion}
\operatorname{Stab}_{G_{\rm DPLO}}(\rho) \subseteq G_{\rm PLO}\,.
\end{equation}

But using again the orbit stabilizer dimension formula, along with \cref{eq:minimal-centered-orbdim-general-lowerbounds-proof-impl-forward-hyp}, we have
\begin{align}
&\dim(\operatorname{Stab}_{G_{\rm DPLO}}(\rho))\nonumber\\
&= \dim(G_{\rm DPLO}) - \dim(\orb_{G_{\rm DPLO}}(\rho))\\
&= m^2\,,
\end{align}
and thus \cref{eq:minimal-centered-orbdim-general-lowerbounds-proof-impl-DPLO-inclusion} is an inclusion $A \subseteq B$ of two Lie subgroups of the same dimension. But if a Lie subgroup $A$ of a Lie group $B$ is such that $\dim(A)=\dim(B)$, then their Lie algebras must be equal, from which it follows that $A^0 = B^0$, where $A^0$ (resp. $B^0$) denotes the connected component of $A$ (resp. $B$) containing the identity. Here, since $B = G_{\rm PLO}$ is already connected [cf. \cref{eq:G0-affine-symplectic-subgroup-def-PLO}], we have $B^0 = B$ and therefore $B = B^0 = A^0 \subseteq A$, which establishes that $A = B$, i.e.,
\begin{equation}
\operatorname{Stab}_{G_{\rm DPLO}}(\rho) = G_{\rm PLO}\,.
\end{equation}

\item For $G \in \{G_{\rm ALO},G_{\rm GO}\}$,
by \cref{eq:minimal-centered-orbdim-general-lowerbounds-proof-impl-forward-hyp}, the two lower bounds of \cref{eq:minimal-centered-orbdim-general-lowerbounds} must be equalities, that is
\begin{equation}\label{eq:minimal-centered-orbdim-general-lowerbounds-proof-imp1-equality}
\qquad \dim(\orb_{G}(\rho)) = \dim(G) - K(\rho) = \dim(G) - m^2\,,
\end{equation}
which in turn implies that $K(\rho) = m^2$, i.e., [as mentioned below \cref{eq:K-of-rho-general-bounds}], all the symplectic eigenvalues of $\rho$ are equal. \Cref{lem:stabilizer-of-V-under-Sp} thus yields that there exists a state $\sigma \in \orb_{G_{\rm ALO}}(\rho)$ such that
\begin{equation}\label{eq:minimal-centered-orbdim-general-lowerbounds-proof-imp1-max-stab-Sp}
\qquad \operatorname{Stab}_{\mathrm{Sp}(2m,\mathbb{R})}(V_{\sigma}) = \mathrm{Sp}(2m,\mathbb{R}) \cap \mathrm{O}(2m)\,.
\end{equation}
Thus, applying \cref{eq:stab-G-subseteq-stab-Sp} to $\sigma$, and combining with \cref{eq:minimal-centered-orbdim-general-lowerbounds-proof-imp1-max-stab-Sp}, gives that
\begin{equation}\label{eq:minimal-centered-orbdim-general-lowerbounds-proof-imp1-stab-G-subseteq-SpcapO}
\operatorname{Stab}_{G}(\sigma) \,\subseteq\, \{\bm{0}\} \,\times\, \mathrm{Sp}(2m,\mathbb{R}) \cap \mathrm{O}(2m)\,,
\end{equation}
or in other words [cf. \cref{eq:G0-affine-symplectic-subgroup-def-PLO}]:
\begin{equation}\label{eq:minimal-centered-orbdim-general-lowerbounds-proof-imp1-stab-G-subseteq-PLO}
\operatorname{Stab}_{G}(\sigma) \subseteq G_{\rm PLO}\,.
\end{equation}

Using again the orbit stabilizer dimension formula along with \cref{eq:minimal-centered-orbdim-general-lowerbounds-proof-impl-forward-hyp}, we have
\begin{align}
&\dim(\operatorname{Stab}_{G}(\sigma))\nonumber\\
&= \dim(G) - \dim(\orb_{G}(\sigma))\\
&= \dim(G) - \dim(\orb_{G}(\rho))\\
&= m^2\,,
\end{align}
and thus \cref{eq:minimal-centered-orbdim-general-lowerbounds-proof-imp1-stab-G-subseteq-PLO} is an inclusion $A \subseteq B$ of two Lie subgroups of the same dimension (with $B$ connected). As explained above, this implies that $A = B$, i.e.,
\begin{equation}
\operatorname{Stab}_{G}(\sigma) = G_{\rm PLO}\,.
\end{equation}
\end{enumerate}

Lastly, we justify the last "i.e." in the Theorem's statement, which claims that a state $\sigma$ is stabilized by all of $G_{\rm PLO}$ if and only if it is of the form
\begin{equation}\label{eq:minimal-centered-orbdim-general-lowerbounds-proof-PLO-invariant-form}
\sigma = \oplus_{n \in \mathbb{N}} \, p_n \id_{\fockn}\text{ for some }p_n \in \mathbb{R}\,.
\end{equation}

Being stabilized by all of $G_{\rm PLO}$ is written equivalently as
\begin{equation}\label{eq:U-sigma-commuting-eq-PLO}
U \,\sigma = \sigma \, U \qquad \forall U \in G_{\fock,\rm PLO} \subset \uni(\fock)\,.
\end{equation}
First, \cref{eq:U-sigma-commuting-eq-PLO} implies the density matrix $\sigma$ must be block diagonal with respect to $\fock = \oplus_{n \in \mathbb{N}} \fockn$. Indeed, denoting by $P_n$ the orthogonal projection onto $\mathcal{H}_{m}^{n}$, the matrix decomposes as
\begin{equation}
\sigma = \sum_{n,n' \in \mathbb{N}} P_n \sigma P_{n'}\,,
\end{equation}
the phase-shifting unitary $U(\theta) := e^{i \theta (\hat{N}_1 + \cdots + \hat{N}_m)}$ (for $\theta \in \mathbb{R}$) acts on it as
\begin{equation}
U(\theta) \sigma U(\theta)^\dagger = \sum_{n,n' \in \mathbb{N}} e^{i \theta (n-n')} P_n \sigma P_{n'}\,,
\end{equation}
and hence the equality $U(\theta) \sigma U(\theta)^\dagger = \sigma$ [because $U(\theta)$ belongs to $ G_{\fock,\rm PLO}$] yields
\begin{equation}
e^{i \theta (n-n')} P_n \sigma P_{n'} = P_n \sigma P_{n'}
\end{equation}
for all $n,n'$ and all $\theta \in \mathbb{R}$,
which implies for all $n \neq n'$ (by picking a $\theta$ such that $e^{i \theta (n-n')} \neq 1$) that $P_n \sigma P_{n'} = 0$.
Thus, since $\sigma(\fockn) \subseteq \fockn$, we may now consider the above \cref{eq:U-sigma-commuting-eq-PLO} but with both operators $U$ and $\sigma$ restricted to the subspace $\fockn$; by Schur's lemma and the irreducibility of the PLO representation on each $\fockn$ (cf. proof of \cref{prop:QO-orbit-span-infinite-dim-conditions}), this implies that $\left.\sigma\right|_{\fockn} = p_n \id_{\fockn}$ for some $p_n \in \mathbb{C}$, and in fact $p_n \in \mathbb{R}$ by Hermiticity. And conversely, if $\sigma$ is of the form of \cref{eq:minimal-centered-orbdim-general-lowerbounds-proof-PLO-invariant-form}, then clearly it is stabilized by all of $G_{\rm PLO}$ (since these unitaries preserve particle number).
\end{proof}

\end{theorem}

We recall \cite{Walschaers-NonGaussianQuantum-2021} that a state $\rho$ is said to be \textit{Gaussian} if its Wigner function is a Gaussian function in phase space, and that if $\rho$ is pure, then it is Gaussian if and only if it is in the Gaussian orbit of the vacuum, i.e., $\rho \sim_{G_{\rm GO}} \ketbra{\bm{0}}$.

Specializing the above \cref{thm:minimal-orbdim-statements} to \textit{pure} states, we can now obtain the following \cref{cor:minimal-orbdim-purestates-characterization};
note that in its statement, we refer to Fock basis states ($\ketbra{\bm{n}}$, $\bm{n} \in \mathbb{N}^m$) as just "Fock states" for short:

\begin{corollary}[Characterization of pure states of minimal orbit dimension]\label{cor:minimal-orbdim-purestates-characterization}
Consider one of the four Lie groups $G \in \{G_{\rm PLO},G_{\rm DPLO},G_{\rm ALO},G_{\rm GO}\}$.
If we restrict to (Schwartz) states $\rho$ that are \emph{pure}, then the following holds:

\begin{enumerate}[label=(\arabic*)]
\item In the single-mode setting ($m=1$):

There are as many $G$-orbits of minimal dimension [i.e., of dimension $\dim(G) - 1$] as there are Fock basis states ($\ketbra{n}$, $n \in \mathbb{N}$), and together they consist of exactly the following states:
\begin{enumerate}[label=(\alph*)]
\item for $G=G_{\rm PLO}$: the Fock states;
\item for $G=G_{\rm DPLO}$: the displaced Fock states;
\item for $G=G_{\rm ALO}$: the squeezed Fock states;
\item for $G=G_{\rm GO}$: the displaced-and-squeezed Fock states \cite{Kral-DisplacedSqueezed-1990}.
\end{enumerate}

\item In the multimode setting ($m\geq2$):

There is a unique $G$-orbit of minimal dimension [i.e., of dimension $\dim(G) - m^2$], and it consists of exactly the following states:
\begin{enumerate}[label=(\alph*)]
\item for $G=G_{\rm PLO}$: the vacuum state;
\item for $G=G_{\rm DPLO}$: the coherent states;
\item for $G=G_{\rm ALO}$: the squeezed vacuum states;
\item for $G=G_{\rm GO}$: the Gaussian (pure) states.
\end{enumerate}

\end{enumerate}
\end{corollary}
\begin{proof}%
This is just specializing \cref{thm:minimal-orbdim-statements}'s characterization of \cref{eq:minimal-orbdim-purestates-characterization-proof-GO-equiv} to the case of pure states. Indeed, suppose that $\rho$ is pure, i.e., $\rank(\rho) = 1$. Unitary evolution preserves the rank of the density operator, but
\begin{equation}\label{eq:minimal-orbdim-purestates-characterization-proof-rank-of-blockdiag}
\rank( \oplus_{n \in \mathbb{N} }p_n \id_{\fockn} ) = \sum_{\substack{n \in \mathbb{N};\\ p_n>0}} \dim(\fockn)\,
\end{equation}
and $\dim(\fockn) = \binom{m+n-1}{n} \geq 1$ is equal to $1$ if and only if $m=1$ or $n=0$.
Hence, \cref{eq:minimal-orbdim-purestates-characterization-proof-rank-of-blockdiag} is equal to $1$ if and only if (i) there is exactly one $n \in \mathbb{N}$ such that $p_n > 0$ \textit{and} (ii) we have $m=1$ or $n=0$.

Consider the single-mode case $(m=1)$ first. Since $\fockn$ is then the one-dimensional subspace spanned by $\ket{n}$, the above means that for pure states, the bottom condition of \cref{eq:minimal-orbdim-purestates-characterization-proof-GO-equiv} reduces (imposing also unit-trace) to $\rho$ being $G$-equivalent to the Fock basis state $\ketbra{n}$ for some $n \in \mathbb{N}$.

For the multimode case $(m\geq 2)$, the above means that the bottom condition of \cref{eq:minimal-orbdim-purestates-characterization-proof-GO-equiv} this time reduces to $\rho$ being $G$-equivalent to the multimode vacuum state $\ketbra{\bm{0}}$, which means exactly \cite{weedbrook_gaussian_2012} that
$\rho$ is the vacuum state (for $G=G_{\rm PLO}$), a coherent state (for $G=G_{\rm DPLO}$), a squeezed vacuum state (for $G=G_{\rm ALO}$), or a pure Gaussian state (for $G=G_{\rm GO}$).
\end{proof}

Looking at the multimode setting of the above \cref{cor:minimal-orbdim-purestates-characterization}, we directly recognize these as the free pure states in the RTs of non-Gaussianity (for $G=G_{\rm GO}$) and of optical nonclassicality (for $G=G_{\rm DPLO}$) introduced in the main text. In fact, we may also recognize vacuum states (for $G=G_{\rm PLO}$) and squeezed vacuum states (for $G=G_{\rm ALO}$) as the free pure states of other existing resource theories of quantum optics, although they are arguably more artificial ways to regard these two classes of states as free, in part because they are resource theories that are restricted overall to just Gaussian states anyways. Still, we include these $G=G_{\rm PLO}$ and $G=G_{\rm ALO}$ cases anecdotally in the following corollary, which encompasses the main text's \cref{cor:minimal-orbdim-purestates-characterization-main}:

\begin{corollary}[Multimode pure states of minimal orbit dimension are exactly the free pure states of existing resource theories]\label{cor:minimal-orbdim-purestates-multimode-and-RTs}
In the multimode setting $(m\geq 2)$, the unique pure $G$-orbit of minimal dimension (cf. \cref{cor:minimal-orbdim-purestates-characterization}) coincides exactly with the \emph{pure free states} (over $m$ modes) of certain resource theories (RTs) of quantum optics, namely,
\begin{enumerate}[label=(\arabic*)]
\item for $G=G_{\rm PLO}$: the RT of \emph{Gaussian coherence} \cite{Xu-QuantifyingCoherence-2016,Streltsov-ColloquiumQuantum-2017,Gianfelici-HierarchyContinuousvariable-2021} and also the RT of \emph{local Gaussian work extraction} \cite{Singh-QuantumThermodynamics-2019};
\item for $G=G_{\rm DPLO}$: the RT of \emph{optical nonclassicality} \cite{Yadin-OperationalResource-2018,Tan-ResourceTheories-2019};
\item for $G=G_{\rm ALO}$: the RT of \emph{Gaussian parity-asymmetry} \cite{Koukoulekidis-SymmetryAsymmetry-2025};
\item for $G=G_{\rm GO}$: the RT of \emph{genuine non-Gaussianity} \cite{Takagi-ConvexResource-2018,Albarelli-ResourceTheory-2018} or the RT of \emph{Wigner negativity} \cite{Albarelli-ResourceTheory-2018}.
\end{enumerate}
\begin{proof}
These observations follow directly from \cref{cor:minimal-orbdim-purestates-characterization} and the definitions and characterization of the \textit{free pure states} in the respective cited resource theories. The cases of the main text ($G=G_{\rm GO}$ and $G=G_{\rm DPLO}$) are immediate from the description of the associated RTs and their simplifications for pure states, given in \cref{subsec:RTs-of-qo,subsec:minimal-orbdims}. For the other two, we have the following:

\begin{enumerate}[label=(\arabic*)]
\item Case $G_{\rm PLO}$

In the RT of Gaussian coherence, the free states are defined as the subset of Gaussian states that are \emph{incoherent} (i.e., that are diagonal) in the Fock basis, and those are shown to exactly be the thermal states \cite{Xu-QuantifyingCoherence-2016}; but the only pure thermal state is the vacuum state.

In the RT of local Gaussian work extraction, the free states are defined as the union of $G_{\rm PLO}$-orbits of thermal states, but again the only pure thermal state is the vacuum state, whose $G_{\rm PLO}$-orbit is just itself.

\item Case $G_{\rm ALO}$

We saw (\cref{cor:minimal-orbdim-purestates-characterization}) that the states in question are the squeezed vacuum states, or in other words the centered Gaussian pure states.
Recently in \cite{Koukoulekidis-SymmetryAsymmetry-2025}, the so-called RT of \emph{Gaussian asymmetry} was defined, with the asymmetry being in terms of a certain fixed subgroup $G$ of $G_{\rm ALO}$ to be chosen. There, the free states are defined as $G$-invariant Gaussian states. If we choose $G$ to be the parity subgroup ($G := \{ \id, \Pi := e^{i \pi (\hat{N}_1 + \cdots + \hat{N}_m)} \}$), then the only parity-invariant Gaussian states are the centered Gaussian states (since the parity operator $\Pi$ acts on Wigner functions as $W(\cdot) \mapsto W(-(\cdot))$). Hence, the pure free states then correspond to the pure centered Gaussian states.
\end{enumerate}
\end{proof}
\end{corollary}

Lastly, we have a third corollary to \cref{thm:minimal-orbdim-statements}, which is exactly the main text's \cref{cor:orbdim-lowerbound-symplectic-rank-main}: %
\begin{corollary}[A relation between orbit dimensions and symplectic rank]\label{cor:orbdim-lowerbound-symplectic-rank}
Let $\rho$ be a Schwartz state, and consider one of the four Lie groups $G \in \{G_{\rm PLO},G_{\rm DPLO},G_{\rm ALO},G_{\rm GO}\}$.
Then,
\begin{equation}\label{eq:orbdim-lowerbound-symplectic-rank-claim}
\dim(\orb_{G}(\rho)) \geq \dim(G) - m^2 + 2 \, \aleph(\rho) \, [m - \aleph(\rho)]\,,
\end{equation}
where $\aleph(\rho)$ denotes the (ordinary) symplectic rank of $\rho$ (see \cref{subsecs:affine-q-optics-introduction}).
\begin{proof}
This is just the first lower bound of \cref{thm:minimal-orbdim-statements}'s \cref{eq:minimal-centered-orbdim-general-lowerbounds}, combined with the following simple upper bound on the quantity $K(\rho)$ in terms of the symplectic rank $\aleph(\rho)$:
\begin{equation}\label{eq:orbdim-lowerbound-symplectic-rank-upperbound-K-of-rho-claim}
K(\rho) \leq m^2 - 2 \, \aleph(\rho) \, [m - \aleph(\rho)]\,.
\end{equation}
To see why \cref{eq:orbdim-lowerbound-symplectic-rank-upperbound-K-of-rho-claim} holds, notice that when $\aleph(\rho) = m$, it becomes the inequality $K(\rho) \leq m^2$, which always holds [cf. \cref{eq:K-of-rho-general-bounds}], while when $\aleph(\rho) < m$, we can write $k_1 = m - \aleph(\rho)$, which lets us obtain, successively:
\begin{align}
K(\rho)
:&= k_1^2 + \cdots + k_r^2\\
&= [m -\aleph(\rho)]^2 + k_2^2 + \cdots + k_r^2\\
&\leq [m -\aleph(\rho)]^2 + (k_2 + \cdots + k_r)^2\\
&= [m -\aleph(\rho)]^2 + \aleph(\rho)^2\\
&= m^2 - 2 \aleph(\rho) [m - \aleph(\rho)]\,.
\end{align}
\end{proof}
\end{corollary}

\subsection{Proof that $G_{\rm GO}$ and $G_{\rm DPLO}$ orbit dimensions are not fixed-mode monotones for the respective resource theories}\label{subsec:SM-fixed-mode-nonmonotonicity}
In this section, we prove the main text's \cref{thm:orbdims-are-not-fixed-mode-monotones}, by counterexample.
\paragraph{Proof of \cref{thm:orbdims-are-not-fixed-mode-monotones}.}
Consider $m=2$ modes, and the (Hong-Ou-Mandel) state
\begin{align}
\ket{\psi_{\rm in}} := \frac{\ket{2,0} + \ket{0,2}}{\sqrt{2}}\,,
\end{align}
which is $G_{\rm PLO}$-equivalent to $\ket{1,1}$.

Consider the (50-50) \textit{pure-loss} one-mode channel \cite{weedbrook_gaussian_2012} $\mathcal{L}_{1/2}$, given by
\begin{align}\label{eq:counterexample-loss-one-mode-channel-def}
\mathcal{L}_{1/2}(\rho) := \Tr_A \left[ U_{1/2} \left( \rho \otimes \ketbra{0}_A \right) U_{1/2}^\dagger \right]\,,
\end{align}
where $U_{1/2} := e^{- i ({\pi/2}) E_{1A}}$ is a beam-splitter unitary between the system's first mode and the ancillary system's mode $A$ [$E_{1A} := \frac{i}{2}(a_1^\dagger a_A - a_A^\dagger a_1)$]. Denote the single-mode identity channel on the system's second mode by $\mathcal{I}$, and let
\begin{align}
\rho_{\rm out} := (\mathcal{L}_{1/2} \otimes \mathcal{I})(\ketbra{\psi_{\rm in}})\,.
\end{align}

We find (by using the Kraus representation of $\mathcal{L}_{1/2}$ \cite{Ivan-OperatorsumRepresentation-2011}) that

\begin{align}\label{eq:counterexample-rhoout-defining-decomp}
\rho_{\rm out} =
\frac{5}{8}\ketbra{\psi}
+\frac{1}{4}\ketbra{1,0}
+\frac{1}{8}\ketbra{0,0}\,,
\end{align}
where
\begin{align}
\ket{\psi} := \frac{\ket{2,0} + 2 \ket{0,2}}{\sqrt{5}}\,.
\end{align}

By the equivalence to $\ket{1,1}$ and \cref{tab:orbit-dimensions}, we have $\dim(\orb_{G_{\rm GO}}(\ketbra{\psi_{\rm in}})) = 12$, and since the state is pure, its convex-roof orbit dimension is also $\operatorname{dimOrb}_{G_{\rm GO}}^{\rm cr}(\ketbra{\psi_{\rm in}}) = 12$.

We now turn to showing that
\begin{align}\label{eq:counterexample-convex-roof-output-claim}
\operatorname{dimOrb}_{G_{\rm GO}}^{\rm cr}(\rho_{\rm out})=13\,.
\end{align}
First, let us study the $G_{\rm GO}$-orbit dimension of an arbitrary pure state in the support of $\rho_{\rm out}$, which can be written as
\begin{align}\label{eq:counterexample-convex-roof-general-support-state}
\ket{\varphi}
=\alpha\ket{\psi}+\beta\ket{1,0}+\gamma\ket{0,0}\,,
\end{align}
for some $\alpha,\beta,\gamma\in\mathbb{C}$, and let us suppose for now that $\alpha\neq0$.

Let us denote by $i\mathcal{S}$ the Lie-algebra stabilizer of $\ketbra{\varphi}$ (for the Lie algebra $\mathfrak{g}_{\rm GO}$ of $G_{\rm GO}$), i.e.,
\begin{equation}
\mathcal{S} := \left\{ H \in i \mathfrak{g}_{\rm GO} \,|\, \left[\,H, \ketbra{\varphi}\,\right] = 0 \right\}\,.
\end{equation}
First, notice that $\mathcal{S}$ is a real vector space, and that $\lambda\id \in \mathcal{S}$ for all $\lambda \in \mathbb{R}$; therefore, it holds that $H \in \mathcal{S}$ if and only if $H - \lambda\id \in \mathcal{S}$.

We now proceed similarly to the proof of \cref{prop:generic-orbit-dim} in \cref{sec:SM-genericity}.
Recall [\cref{eq:generic-orbit-dim-general-GO-Hamiltonian-normal-ordered}] that (up to a factor $i$) an arbitrary Hamiltonian in $\mathfrak{g}_{\rm GO}$ is written, for $m=2$, as
\begin{align}\label{eq:counterexample-convex-roof-general-GO-Hamiltonian}
H
=&\ c\id+\sum_{j,k=1}^2 A_{jk}a_j^\dagger a_k
+\frac{1}{2}\sum_{j,k=1}^2\left(B_{jk}a_j^\dagger a_k^\dagger+\overline{B_{jk}}a_j a_k\right)\nonumber\\
&+\sum_{k=1}^2\left(u_k a_k^\dagger+\overline{u_k}a_k\right)\,,
\end{align}
for $c\in\mathbb{R}$, $A=A^\dagger\in\mathbb{C}^{2\times2}$, $B=B^\top\in\mathbb{C}^{2\times2}$, and $\bm{u}=(u_1,u_2)\in\mathbb{C}^2$.
Suppose that $H \in \mathcal{S}$, which is equivalent to $H\ket{\varphi}=\lambda\ket{\varphi}$ for some $\lambda\in\mathbb{R}$.
As per the above remark, we may suppose here that $\lambda=0$ (and we later keep in mind that the direction $\spa\!{\{\id\}}$ is also always in $\mathcal{S}$). That is, suppose that
\begin{align}\label{eq:counterexample-convex-roof-annihilation-condition}
H\ket{\varphi}=0\,.
\end{align}

As we have
\begin{align}
\ket{\psi}
=\frac{1}{\sqrt{10}}\left((a_1^\dagger)^2+2(a_2^\dagger)^2\right)\ket{0,0}\,,
\end{align}
the four-photon component of $H\ket{\varphi}$ is
\begin{align}
\frac{\alpha}{2\sqrt{10}}
\sum_{j,k=1}^2 B_{jk}a_j^\dagger a_k^\dagger
\left((a_1^\dagger)^2+2(a_2^\dagger)^2\right)\ket{0,0}\,.
\end{align}
Since $\alpha\neq0$, \cref{eq:counterexample-convex-roof-annihilation-condition} implies that $B=0$.

The three-photon component of \cref{eq:counterexample-convex-roof-annihilation-condition} now reduces to
\begin{align}
\frac{\alpha}{\sqrt{10}}
\sum_{k=1}^2 u_k a_k^\dagger
\left((a_1^\dagger)^2+2(a_2^\dagger)^2\right)\ket{0,0}\,,
\end{align}
so the equality implies in turn that $\bm{u}=\bm{0}$.
With $B=0$ and $\bm{u}=\bm{0}$, only the constant ($c$) and passive terms ($A$) remain in $H$.
Since the two-photon component of $H\ket{\varphi}$ is
\begin{align}\label{eq:counterexample-convex-roof-two-photon-component}
\frac{\alpha}{\sqrt{5}}
\big[
&(c+2A_{11})\ket{2,0}
+2(c+2A_{22})\ket{0,2}\nonumber\\
&+\sqrt{2}\big(\overline{A_{12}}+2A_{12}\big)\ket{1,1}
\big]\,,
\end{align}
\cref{eq:counterexample-convex-roof-annihilation-condition} implies on the
$\ket{2,0}$ and $\ket{0,2}$ components that
\begin{align}
A_{11}=A_{22}=-\frac{c}{2}\,,
\end{align}
while it implies on the $\ket{1,1}$ component that
\begin{align}
\overline{A_{12}}+2A_{12}=0\,.
\end{align}
Writing $A_{12}=x+iy$ for $x,y\in\mathbb{R}$, the latter equality becomes $3x+iy=0$, implying that $A_{12}=0$.
We have hence obtained that $A=-\frac{c}{2}\id_2$, and hence that
\begin{align}\label{eq:counterexample-convex-roof-annihilating-H-form}
H=c\id-\frac{c}{2}(N_1+N_2)\,.
\end{align}

We are left to imposing \cref{eq:counterexample-convex-roof-annihilation-condition} on the one- and zero-photon components.
These are now, respectively,
\begin{align}\label{eq:counterexample-convex-roof-small-components}
\frac{c}{2}\beta\ket{1,0}\quad\text{and}\quad c\gamma\ket{0,0}\,.
\end{align}
Suppose first that $\beta=\gamma=0$, i.e., that $\ketbra{\varphi}=\ketbra{\psi}$.
Then no further constraints arise beyond \cref{eq:counterexample-convex-roof-annihilating-H-form}, so we obtained that $H \in \spa\!{\{\id-(N_1+N_2)/2\}}$. Since conversely this direction is indeed in $\mathcal{S}$, and so is $\spa\!{\{\id\}}$, we have established that, in this case,
\begin{equation}\label{eq:counterexample-convex-roof-stablizer-expression-case-1}
\mathcal{S} = \spa\!{\{\id,\ \id-(N_1+N_2)/2\}}\,.
\end{equation}
Since it is two-dimensional [and since $\dim(G_{\rm GO})=15$ for $m=2$], the orbit-stabilizer dimension formula then gives
\begin{align}\label{eq:counterexample-convex-roof-psi-orbit-dim}
\dim(\orb_{G_{\rm GO}}(\ketbra{\psi}))
=15-2=13\,.
\end{align}

Suppose now that $\beta\neq0$ or $\gamma\neq0$.
Then one of the components in \cref{eq:counterexample-convex-roof-small-components} is nonzero, which by \cref{eq:counterexample-convex-roof-annihilation-condition} then implies that $c=0$, and hence, by \cref{eq:counterexample-convex-roof-annihilating-H-form}, that $H=0$. In that case, we thus obtain only
\begin{equation}
\mathcal{S} = \spa\!{\{\id\}}\,,
\end{equation}
and hence
\begin{align}\label{eq:counterexample-convex-roof-general-support-state-orbit-dim}
\dim(\orb_{G_{\rm GO}}(\ketbra{\varphi}))
=15-1=14\,.
\end{align}

In both cases, we have hence shown that every pure state of the form of \cref{eq:counterexample-convex-roof-general-support-state} with $\alpha\neq0$ satisfies
\begin{align}\label{eq:counterexample-convex-roof-support-state-lower-bound}
\dim(\orb_{G_{\rm GO}}(\ketbra{\varphi}))\geq13\,.
\end{align}

Next, consider an arbitrary convex decomposition of $\rho_{\rm out}$ into pure states,
\begin{align}\label{eq:counterexample-convex-roof-arbitrary-decomposition}
\rho_{\rm out}
=\int d\lambda\,p_\lambda\ketbra{\varphi_\lambda}\,.
\end{align}
As claimed elsewhere in the main text, the fact that $p_\lambda$ is a probability density implies that, up to exceptions of measure zero, every state $\ket{\varphi_\lambda}$ in this decomposition must live inside $\operatorname{supp}(\rho_{\rm out})$. This is because, by denoting $P_{\perp}$ as the orthogonal projector onto $\operatorname{supp}(\rho_{\rm out})^\perp$, we have
\begin{align}\label{eq:counterexample-convex-roof-zero-orthogonal-support}
0
=\Tr(P_{\perp}\rho_{\rm out})
=\int d\lambda\,p_\lambda
\|P_{\perp}\ket{\varphi_\lambda}\|^2\,,
\end{align}
and as the integrand is non-negative, it must vanish almost everywhere.
Hence, almost every state in the decomposition \cref{eq:counterexample-convex-roof-arbitrary-decomposition} can be written as
\begin{align}
\ket{\varphi_\lambda}
=\alpha_\lambda\ket{\psi}
+\beta_\lambda\ket{1,0}
+\gamma_\lambda\ket{0,0}\,.
\end{align}
Furthermore, we have
\begin{align}
0 < \frac{5}{8}
=\bra{\psi}\rho_{\rm out}\ket{\psi}
=\int d\lambda\,p_\lambda|\alpha_\lambda|^2\,,
\end{align}
which implies that $\alpha_\lambda\neq0$ on a positive ($p_\lambda$) measure set (for all convex decompositions $\{p_\lambda,\rho_\lambda\}$ of $\rho_{\rm out}$).
Therefore, by \cref{eq:counterexample-convex-roof-support-state-lower-bound}, the essential supremum of the pure-state orbit dimensions in every decomposition \cref{eq:counterexample-convex-roof-arbitrary-decomposition} is at least $13$, and therefore
\begin{align}\label{eq:counterexample-convex-roof-output-lower-bound}
\operatorname{dimOrb}_{G_{\rm GO}}^{\rm cr}(\rho_{\rm out})\geq13\,.
\end{align}

Conversely, \cref{eq:counterexample-convex-roof-psi-orbit-dim} and the first column of \cref{tab:orbit-dimensions} give
\begin{align}
\dim(\orb_{G_{\rm GO}}(\ketbra{\psi}))&=13\,,\\
\dim(\orb_{G_{\rm GO}}(\ketbra{1,0}))&=12\,,\\
\dim(\orb_{G_{\rm GO}}(\ketbra{0,0}))&=10\,.
\end{align}
Since the maximum of these three values is $13$, the particular (defining) decomposition of $\rho_{\rm out}$ of \cref{eq:counterexample-rhoout-defining-decomp} yields that
\begin{align}
\operatorname{dimOrb}_{G_{\rm GO}}^{\rm cr}(\rho_{\rm out}) \leq 13\,.
\end{align}
Combining this with \cref{eq:counterexample-convex-roof-output-lower-bound} establishes \cref{eq:counterexample-convex-roof-output-claim}.

Finally, since $\ketbra{\psi_{\rm in}}$ is pure, its convex-roof orbit dimension agrees with its ordinary orbit dimension, and using again its $G_{\rm PLO}$ equivalence to $\ketbra{1,1}$ and \cref{tab:orbit-dimensions}, we get
\begin{align}
\operatorname{dimOrb}_{G_{\rm GO}}^{\rm cr}(\ketbra{\psi_{\rm in}})=12<13=\operatorname{dimOrb}_{G_{\rm GO}}^{\rm cr}(\rho_{\rm out})\,.
\end{align}

For completeness, note that the calculation leading to \cref{eq:counterexample-convex-roof-psi-orbit-dim} also gives the ordinary orbit dimension of $\rho_{\rm out}$.
Indeed, its three nonzero eigenvalues are all distinct [cf. \cref{eq:counterexample-rhoout-defining-decomp}], so a Hamiltonian commuting with $\rho_{\rm out}$ must also commute with $\ketbra{\psi}$, i.e., lie in the $\mathcal{S}$ of \cref{eq:counterexample-convex-roof-stablizer-expression-case-1}.
Since the latter is spanned by $\id$ and $N_1+N_2$, and since conversely both of these Hamiltonians commute with $\rho_{\rm out}$ as well, we get
\begin{align}
\dim(\orb_{G_{\rm GO}}(\rho_{\rm out}))
=15-2=13\,.
\end{align}

Furthermore, notice that the above derived expression(s) for the stabilizer $\mathcal{S}$ of $\ketbra{\varphi}$ (and hence also of $\rho_{\rm out}$) under $G_{\rm GO}$, are still valid as stabilizers under the subgroup $G_{\rm DPLO} \subset G_{\rm GO}$ (since the directions $\id$ and $N_1+N_2$ found in the former case still belong to the Lie algebra of $G_{\rm DPLO}$). Therefore, when using the orbit-stabilizer dimension formula for $G_{\rm DPLO}$, the obtained orbit dimensions are all those of the $G_{\rm GO}$ orbits but subtracted by $6$ [as for $m=2$, $\dim(G_{\rm GO})=15$ and $\dim(G_{\rm DPLO})=9$, which differ by $6$].

To summarize, we have obtained that for both types of orbit dimensions (standard or convex-roof), those of $\ketbra{\psi_{\rm in}}$ are $12$ for $G_{\rm GO}$ and $6$ for $G_{\rm DPLO}$, while those of $\rho_{\rm out}$ are increased by one, i.e., $13$ for $G_{\rm GO}$ and $7$ for $G_{\rm DPLO}$.

Since the channel $\mathcal{L}_{1/2}$ of \cref{eq:counterexample-loss-one-mode-channel-def} (and hence $\mathcal{L}_{1/2} \otimes \mathcal{I}$ as well) is free in both the RTs of genuine non-Gaussianity \cite{Takagi-ConvexResource-2018} and nonclassicality \cite{Yadin-OperationalResource-2018}, this shows that the $G_{\rm GO}$ and the $G_{\rm DPLO}$ orbit dimensions (both standard and convex-roof) are not fixed-mode monotones for these resource theories.$\qed$

\subsection{Justification that orbit dimensions are not equivalent to the stellar or symplectic ranks}\label{subsec:SM-orbit-dimensions-vs-stellar-symplectic-ranks}

In this section, we justify the fact (stated in the main text in \cref{subsec:relations-with-stellar-or-symplectic-ranks}) that the properties of $G$-orbit dimensions (for $G = G_{\rm PLO},G_{\rm DPLO},G_{\rm ALO},G_{\rm GO}$, in both the ket and density operator pictures) are generally independent of the properties of stellar and symplectic ranks.
We show this by highlighting pairs of states which have equal $G$-orbit dimensions but different stellar rank (resp. symplectic rank), as well as pairs of states where the converse holds. Since these states all have relatively simple forms, we do not detail the proofs for the obtained values.
For all $m\geq1$, let
\begin{align}
\ket{\psi_1^{m}}:&=\ket{1}^{\otimes m}\,,\\
\ket{\psi_2^{m}}:&= \ket{1}^{\otimes(m-1)}\otimes\ket{0}\,,
\intertext{and}
\ket{\psi_3^{m}}:&=\ket{1,0,\dots,0}\,,\\
\ket{\psi_4^{m}}:&=\frac{\ket{0,\dots,0}+\ket{1,0,\dots,0}}{\sqrt{2}}\,.
\end{align}
The states in the first pair have \textit{different} stellar and symplectic ranks,
\begin{align}
r^\star(\psi_1^{m})&=\aleph(\psi_1^{m})&&= m\,,\\
r^\star(\psi_2^{m})&=\aleph(\psi_2^{m})&&= m-1\,,
\end{align}
but being Fock basis states with $0$ or $1$ unoccupied modes, they have \textit{equal} $G$-orbit dimensions (in the ketbra picture) for all four groups $G$, cf. \cref{tab:orbit-dimensions}.

Conversely, the states in the second pair have \textit{equal} stellar and symplectic ranks,
\begin{align}
r^\star(\psi_3^{m})=\aleph(\psi_3^{m})=r^\star(\psi_4^{m})=\aleph(\psi_4^{m})= 1\,,
\end{align}

but by the first two columns of \cref{tab:orbit-dimensions}, they have \textit{different} $G$-orbit dimensions (in the ketbra picture) for all four groups $G$, namely,
\begin{equation}
\dim(\orb_G(\ketbra{\psi_4^{m}}))=\dim(\orb_G(\ketbra{\psi_3^{m}}))+1\,.
\end{equation}

As for orbits in the ket picture instead of the ketbra picture: for $G\in\{G_{\rm DPLO},G_{\rm ALO},G_{\rm GO}\}$, the orbit dimensions of the four states above equally increase by one; therefore, they serve as counterexamples as well in that case.
For $G=G_{\rm PLO}$, we can this time take, for all $m\geq 1$,
\begin{align}
\ket{\varphi_1^{m}}:&=\ket{1,0,\dots,0}\,,\\
\ket{\varphi_2^{m}}:&= \ket{\alpha} \otimes \ket{0,\dots,0}\,,
\intertext{and}
\ket{\varphi_3^{m}}:&=\ket{0,0,\dots,0}\,,\\
\ket{\varphi_4^{m}}:&= \ket{\alpha} \otimes \ket{0,\dots,0} = \ket{\varphi_2^{m}}\,,
\end{align}
for some coherent state $\ket{\alpha}$ with $\alpha \in \mathbb{C} \setminus \{0\}$.
The states in the first pair have \textit{different} stellar and symplectic ranks,
\begin{align}
r^\star(\varphi_1^{m})&=\aleph(\varphi_1^{m})&&= 1\,,\\
r^\star(\varphi_2^{m})&=\aleph(\varphi_2^{m})&&= 0\,,
\end{align}
but they have an \textit{equal} $G_{\rm PLO}$-orbit dimension (in the ket picture) of $2m-1$.
Conversely, the states in the second pair have \textit{equal} stellar and symplectic ranks,
\begin{align}
r^\star(\varphi_3^{m})=\aleph(\varphi_3^{m})=r^\star(\varphi_4^{m})=\aleph(\varphi_4^{m})= 0\,,
\end{align}
but they have \textit{different} $G_{\rm PLO}$-orbit dimensions (in the ket picture), namely,
\begin{align}
\dim(\orb_G(\ket{\varphi_3^{m}})) &= 0\,,\\
\dim(\orb_G(\ket{\varphi_4^{m}})) &= 2m-1\,.
\end{align}

\section{Phase-space representations of quadratic Hamiltonians}\label{sec:SM-phase-space-representations}

\subsection{Passage to a phase-space picture}

First, we need to introduce some notation. If $\mathcal{V}$ and $\mathcal{W}$ are two (possibly infinite-dimensional, complex) vector spaces,
and if
\begin{equation}\label{eq:starting-isomorphism-map-for-phase-space-reprs}
\mathtt{T} : \mathcal{V} \to \mathcal{W}
\end{equation}
is a (complex-)linear isomorphism, then any operator $A : \mathcal{V} \to \mathcal{V}$ can naturally be turned into an operator $\widehat{A} : \mathcal{W} \to \mathcal{W}$ via $\widehat{A} := \mathtt{T} \circ A \circ \mathtt{T}^{-1}$. That is, $\mathtt{T}$ induces another isomorphism $\tilde{\mathtt{T}} : \operatorname{End}(\mathcal{V}) \to \operatorname{End}(\mathcal{W})$ between the two spaces of operators, given by $\tilde{\mathtt{T}}(A) := \widehat{A}$. This map $\tilde{\mathtt{T}}$ is in fact clearly an algebra isomorphism (it preserves addition but also composition and scalar multiplication). If $\mathcal{D} \subseteq \mathcal{V}$ is a vector subspace, then of course the isomorphism $\tilde{\mathtt{T}}$ still provides (by appropriate restrictions) an algebra isomorphism
\begin{equation}\label{eq:pushforward-of-operator-map-restricted-common-domain}
\tilde{\mathtt{T}} : \operatorname{End}(\mathcal{D}) \to \operatorname{End}(\mathtt{T}(\mathcal{D}))\,.
\end{equation}
Concretely, this means that if we are given some operators $\{A_k\}$ on $\mathcal{V}$ that are only defined on a common domain $\mathcal{D} \subseteq \mathcal{V}$ and that all stabilize this domain, then we can turn each of them into the corresponding operator $\widehat{A_k}:=\tilde{\mathtt{T}}(A_k)$ on $\mathcal{W}$, and it then holds that polynomials in the operators $\{A_k\}$ get mapped transparently, i.e., for example the operator $A:= A_1A_2^5 + 2A_3^2$ gets mapped to $\widehat{A} = \widehat{A_1}\widehat{A_2}^5 + 2\widehat{A_3}^2$.
By regarding \cref{eq:starting-isomorphism-map-for-phase-space-reprs} as a real-linear isomorphism between real vector spaces (by restriction of the scalars), and since $\widehat{A_k}\cdot \mathtt{T}(p) = \mathtt{T}(A_k \cdot p)$ and linear isomorphisms preserve the rank, it of course holds that
\begin{align}
&\rank_{\mathbb{R}}(\{ A_1\cdot p,\dots, A_d\cdot p\})\nonumber\\
=
&\rank_{\mathbb{R}}(\{ \widehat{A_1}\cdot \mathtt{T}(p),\dots, \widehat{A_d}\cdot \mathtt{T}(p)\})
\end{align}
for any $p \in \mathcal{D}$.

\subsection{Case of the stellar representation}

The \textit{stellar representation} of (pure) quantum states \cite{chabaud_stellar_2020,Chabaud-ResourcesBosonic-2023} provides us with a (complex) Hilbert-space isomorphism
\begin{equation}
\mathtt{T} : \fock \to \mathcal{H}L^2(\mathbb{C}^m)
\end{equation}
between the Fock space $\fock$ and a certain Hilbert space $\mathcal{H}L^2(\mathbb{C}^m)$ of complex functions called the Segal-Bargmann space \cite{hall_quantum_2013}.
Explicitly, $\mathcal{H}L^2(\mathbb{C}^m) := \{ f : \mathbb{C}^m \to \mathbb{C} \,|\, f \text{ holomorphic },\ \int_{\mathbb{C}^m} |f(\bm{z})|^2 \frac{e^{-|\bm{z}|^2}}{\pi^m} d\bm z < \infty\}\,$,
and
\begin{equation}\label{eq:stellar-isomorphism-explicit}
(\mathtt{T} \ket{\psi})(\bm{z}) := \sum_{\bm{n} \in \mathbb{N}^m} \bra{\bm{n}}\ket{\psi} \frac{\bm{z}^{\bm{n}}}{\sqrt{\bm{n}!}}\,.
\end{equation}
Here, we take the subspace $\mathcal{D} \subseteq \fock$ [in the notation of \cref{eq:pushforward-of-operator-map-restricted-common-domain}] to be the space of Schwartz states.

It is then known that creation and annihilation operators are mapped under $\tilde{\mathtt{T}}$ to multiplication and derivative operators \cite[Sec. 14.4]{hall_quantum_2013}, i.e.,
\begin{align}
\widehat{a_k^\dagger} &= z_k\,,\label{eq:stellar-creation}\\
\widehat{a_k} &= \partial_{z_k}\,,\label{eq:stellar-annihilation}
\end{align}
where $(z_k f)(\bm{z}):= z_k f(\bm{z})$ and $(\partial_{z_k} f)(\bm{z}):= (\frac{\partial f}{\partial z_k})(\bm{z})$.
Since the Hamiltonians $H$ considered in the main text (\cref{tab:Lie-algebra-bases}) are polynomials in $\{a_k^\dagger, a_k\}$, the correspondences of \cref{eq:stellar-creation,eq:stellar-annihilation} suffice (using the mentioned algebra isomorphism property of $\tilde{\mathtt{T}}$) to obtain the corresponding differential operators $\widehat{H}$. We list the obtained operators, for completeness:
\begin{align}
\widehat{e}_{kl} &= \frac{1}{2}\Big(z_k\,\partial_{z_l}+z_l\,\partial_{z_k}\Big)\,,\label{eq:stellar-generator-e}\\
\widehat{E}_{kl} &= \frac{i}{2}\Big(z_k\,\partial_{z_l}-z_l\,\partial_{z_k}\Big)\,,\label{eq:stellar-generator-E}\\
\widehat{r}_{kl} &= \frac{1}{2}\Big(z_k z_l+\partial_{z_k}\partial_{z_l}\Big)\,,\label{eq:stellar-generator-r}\\
\widehat{R}_{kl} &= \frac{i}{2}\Big(z_k z_l-\partial_{z_k}\partial_{z_l}\Big)\,,\label{eq:stellar-generator-R}\\
\widehat{N}_{k} &= z_k\,\partial_{z_k}\,,\label{eq:stellar-generator-N}\\
\widehat{s}_{k} &= \frac{1}{2}\Big(z_k^2+\partial_{z_k}^2\Big)\,,\label{eq:stellar-generator-s}\\
\widehat{S}_{k} &= \frac{i}{2}\Big(z_k^2-\partial_{z_k}^2\Big)\,,\label{eq:stellar-generator-S}\\
\widehat{q}_{k} &= \frac{1}{\sqrt{2}}\Big(z_k+\partial_{z_k}\Big)\,,\label{eq:stellar-generator-q}\\
\widehat{p}_{k} &= \frac{i}{\sqrt{2}}\Big(z_k-\partial_{z_k}\Big)\,,\label{eq:stellar-generator-p}\\
\widehat{\id} &= 1\,,\label{eq:stellar-generator-id}
\end{align}
where in the last line, $1$ denotes the identity operator $f\mapsto f$.

\subsection{Case of the Wigner representation}
The \textit{Wigner representation} of (mixed) quantum states provides us with a (complex) Hilbert-space isomorphism
\begin{equation}
\mathtt{T} : B_2(\fock) \to L^2(\mathbb{R}^{2m})
\end{equation}
between the Hilbert-Schmidt operator space $B_2(\fock)$ and the Hilbert space $L^2(\mathbb{R}^{2m})$ of square-integrable complex functions \cite{Pool-MathematicalAspects-1966} [which become real-valued functions for the operators in $B_2(\fock)$ that are self-adjoint, such as density operators].
Here, we take the subspace $\mathcal{D} \subseteq B_2(\fock)$ [in the notation of \cref{eq:pushforward-of-operator-map-restricted-common-domain}] to be the space of Schwartz operators.
Given an operator $H$ on $\fock$ that is defined on $\mathcal{D} \subseteq \fock$ and that stabilizes $\mathcal{D}$, let us denote by $L_H, R_H \in \operatorname{End}(\mathcal{D})$ the operators that perform left and right multiplication by $H$, i.e., $L_H(X) := H X$ and $R_H(X) := X H$ ($X \in \mathcal{D}$).

It is then known that the operators $L_{H}, R_{H}$ performing left and right multiplication by position and momentum operators ($H = q_k,p_k$) are mapped under $\tilde{\mathtt{T}}$ to the following differential operators (known as \textit{Bopp operators}) \cite[Sec. 2.3]{Hillery-DistributionFunctions-1984},
\begin{align}
\widehat{L_{q_k}} &= \Big(q_k+\frac{i}{2}\,\partial_{p_k}\Big)\,,\label{eq:Wigner-left-multiplication-q}\\
\widehat{R_{q_k}} &= \Big(q_k-\frac{i}{2}\,\partial_{p_k}\Big)\,,\label{eq:Wigner-right-multiplication-q}\\
\widehat{L_{p_k}} &= \Big(p_k-\frac{i}{2}\,\partial_{q_k}\Big)\,,\label{eq:Wigner-left-multiplication-p}\\
\widehat{R_{p_k}} &= \Big(p_k+\frac{i}{2}\,\partial_{q_k}\Big)\,,\label{eq:Wigner-right-multiplication-p}
\end{align}
where
\begin{align}
(q_k W)(\bm{q},\bm{p}) &:= q_k W(\bm{q},\bm{p})\,,\\
(p_k W)(\bm{q},\bm{p}) &:= p_k W(\bm{q},\bm{p})\,,\\
(\partial_{q_k} W)(\bm{q},\bm{p}) &:= \Big(\frac{\partial W}{\partial q_k}\Big)(\bm{q},\bm{p})\,,\\
(\partial_{p_k} W)(\bm{q},\bm{p}) &:= \Big(\frac{\partial W}{\partial p_k}\Big)(\bm{q},\bm{p})\,,
\end{align}
and where indeed all the above operators are well defined on (and preserve) Wigner functions $W$ of Schwartz operators, since such Wigner functions are known to be Schwartz functions \cite{keyl_schwartz_2016,hernandez_rapidly_2022} and hence are smooth.

Here, being in a density operator setting, the linear operators $A \in \operatorname{End}(\mathcal{D})$ that we are interested in casting into operators $\widehat{A}$ acting on Wigner functions are not directly the Hamiltonians $H$ from the main text, but are their associated commutators $A:=[H, \cdot] := \operatorname{ad}_H$ [cf. \cref{eq:concrete-rank-formula-density}].
Since the Hamiltonians $H$ considered in the main text (\cref{tab:Lie-algebra-bases}) can be written as polynomials in $\{q_k, p_k\}$ (cf. \cref{sec:measuring-pure-gram-matrix-entries}), using such rewritings [e.g., \cref{eq:def-generator-e-qpform,eq:def-generator-E-qpform,eq:def-generator-r-qpform,eq:def-generator-R-qpform,eq:def-generator-N-qpform,eq:def-generator-s-qpform,eq:def-generator-S-qpform,eq:def-generator-q-qpform,eq:def-generator-p-qpform,eq:def-generator-id-qpform}], the equation
\begin{equation}
\operatorname{ad}_H = L_H - R_H\,,
\end{equation}
the correspondences of \cref{eq:Wigner-left-multiplication-q,eq:Wigner-right-multiplication-q,eq:Wigner-left-multiplication-p,eq:Wigner-right-multiplication-p}, and the fact that $L_H$, $R_H$, and $\tilde{\mathtt{T}}$ are algebra homomorphisms, again suffice to obtain the corresponding differential operators $\widehat{\operatorname{ad}_H}$.
For completeness, we also list the obtained operators:
\begin{align}
\widehat{\operatorname{ad}_{e_{kl}}}
&= i\left(\frac{q_l}{2}\,\partial_{p_k}-\frac{p_l}{2}\,\partial_{q_k}
 +\frac{q_k}{2}\,\partial_{p_l}-\frac{p_k}{2}\,\partial_{q_l}\right)\,,
\\[2pt]
\widehat{\operatorname{ad}_{E_{kl}}}
&= i\left(-\frac{p_l}{2}\,\partial_{p_k}-\frac{q_l}{2}\,\partial_{q_k}
 +\frac{p_k}{2}\,\partial_{p_l}+\frac{q_k}{2}\,\partial_{q_l}\right)\,,
\\[2pt]
\widehat{\operatorname{ad}_{r_{kl}}}
&= i\left(\frac{q_l}{2}\,\partial_{p_k}+\frac{p_l}{2}\,\partial_{q_k}
 +\frac{q_k}{2}\,\partial_{p_l}+\frac{p_k}{2}\,\partial_{q_l}\right)\,,
\\[2pt]
\widehat{\operatorname{ad}_{R_{kl}}}
&= i\left(\frac{p_l}{2}\,\partial_{p_k}-\frac{q_l}{2}\,\partial_{q_k}
 +\frac{p_k}{2}\,\partial_{p_l}-\frac{q_k}{2}\,\partial_{q_l}\right)\,,
\\[2pt]
\widehat{\operatorname{ad}_{N_k}}
&= i\left(q_k\,\partial_{p_k}-p_k\,\partial_{q_k}\right)\,,\label{eq:Wigner-adjoint-action-N}
\\[2pt]
\widehat{\operatorname{ad}_{s_k}}
&= i\left(q_k\,\partial_{p_k}+p_k\,\partial_{q_k}\right)\,,
\\[2pt]
\widehat{\operatorname{ad}_{S_k}}
&= i\left(p_k\,\partial_{p_k}-q_k\,\partial_{q_k}\right)\,,
\\[2pt]
\widehat{\operatorname{ad}_{q_k}}
&= i\,\partial_{p_k}\,,
\\[2pt]
\widehat{\operatorname{ad}_{p_k}}
&= -i\,\partial_{q_k}\,,
\\[2pt]
\widehat{\operatorname{ad}_{\id}}
&= 0\,.
\end{align}
Indeed, to illustrate with an example, we have
\begin{align*}
\operatorname{ad}_{N_{k}}
&= L_{N_{k}} - R_{N_{k}}\\
&= L_{\frac{1}{2}\bigl(q_k^{2} + p_k^{2} - \id\bigr)} - R_{\frac{1}{2}\bigl(q_k^{2} + p_k^{2} - \id\bigr)}\\
&= \frac{1}{2}L_{q_k}^2 + \frac{1}{2}L_{p_k}^2 - \frac{1}{2}L_{\id} - \left(\frac{1}{2}R_{q_k}^2 + \frac{1}{2}R_{p_k}^2 - \frac{1}{2}R_{\id}\right)\\
&= \frac{1}{2}L_{q_k}^2 + \frac{1}{2}L_{p_k}^2 - \frac{1}{2}R_{q_k}^2 - \frac{1}{2}R_{p_k}^2\,,
\intertext{and hence}
\widehat{\operatorname{ad}_{N_k}}
&= \frac{1}{2}\widehat{L_{q_k}}^2 + \frac{1}{2}\widehat{L_{p_k}}^2 - \frac{1}{2}\widehat{R_{q_k}}^2 - \frac{1}{2}\widehat{R_{p_k}}^2\,,
\end{align*}
which, after injecting \cref{eq:Wigner-left-multiplication-q,eq:Wigner-right-multiplication-q,eq:Wigner-left-multiplication-p,eq:Wigner-right-multiplication-p} and expanding the squares (terms inside the square commute), readily simplifies to \cref{eq:Wigner-adjoint-action-N}.

Other phase-space representations, such as the more general $s$-parametrized quasidistributions (including $Q$ and $P$ functions) which provide vector space isomorphisms $\mathtt{T}$ from $B_2(\fock)$ to suitable function or distribution spaces \cite{Cahill-DensityOperators-1969,Cahill-OrderedExpansions-1969,Brif-PhasespaceFormulation-1999}, can in principle be treated in the same way.

\subsection{Example computations in the stellar representation}\label{subsec:SM-stellar-representation-examples}
Consider the single-mode ($m=1$) coherent state $\ket{\alpha}$ for some $\alpha \in \mathbb{C}$ [cf. \cref{eq:coherent-state-Fock-repr}], and its orbit under the six-dimensional group $G_{\rm GO}$. The stellar function of this state is $f_{\alpha}(z) = e^{-|\alpha|^2/2} \, e^{\alpha z}$, which, discarding the normalization factor (which does not affect the rank computations), we simply write as
\begin{equation}
f_{\alpha}(z) = e^{\alpha z}\,.
\end{equation}
By applying each operator of \cref{eq:stellar-generator-N,eq:stellar-generator-s,eq:stellar-generator-S,eq:stellar-generator-q,eq:stellar-generator-p,eq:stellar-generator-id} to $f_{\alpha}$, one obtains the following six functions:
\begin{align}
\widehat{N} f_{\alpha}(z) &= \alpha z e^{\alpha z}\,,\label{eq:stellar-action-N-on-coherent-state}\\
\widehat{s} f_{\alpha}(z) &= \frac{1}{2}\Big(z^2 + \alpha^2\Big) e^{\alpha z}\,,\label{eq:stellar-action-s-on-coherent-state}\\
\widehat{S} f_{\alpha}(z) &= \frac{i}{2}\Big(z^2 - \alpha^2\Big) e^{\alpha z}\,,\label{eq:stellar-action-S-on-coherent-state}\\
\widehat{q} f_{\alpha}(z) &= \frac{1}{\sqrt{2}}\Big(z + \alpha\Big) e^{\alpha z}\,,\label{eq:stellar-action-q-on-coherent-state}\\
\widehat{p} f_{\alpha}(z) &= \frac{i}{\sqrt{2}}\Big(z - \alpha\Big) e^{\alpha z}\,,\label{eq:stellar-action-p-on-coherent-state}\\
\widehat{\id} f_{\alpha}(z) &= e^{\alpha z}\,.\label{eq:stellar-action-id-on-coherent-state}
\end{align}
There are many ways to evaluate the real rank of this list of six functions.
For instance, here since all these functions are elements of
$\spa_{\mathbb{R}}(\mathcal{B})$ with
\begin{equation}\label{eq:basis-functions-for-coherent-state-calc}
\mathcal{B} = \left\{ e^{\alpha z}, i e^{\alpha z}, z e^{\alpha z}, i z e^{\alpha z}, z^2 e^{\alpha z}, i z^2 e^{\alpha z}\right\}\,,
\end{equation}
and since the functions in $\mathcal{B}$ can easily be shown to be linearly independent over $\mathbb{R}$, the rank of our six functions \labelcref{eq:stellar-action-N-on-coherent-state,eq:stellar-action-s-on-coherent-state,eq:stellar-action-S-on-coherent-state,eq:stellar-action-q-on-coherent-state,eq:stellar-action-p-on-coherent-state,eq:stellar-action-id-on-coherent-state} can be computed as the rank of their matrix in the basis of functions $\mathcal{B}$, i.e., the real rank of
\begin{equation}
M =
\begin{pmatrix}
0 & 0 & a & b & 0 & 0\\[4pt]
\frac{a^2-b^2}{2} & ab & 0 & 0 & \frac{1}{2} & 0\\[6pt]
ab & -\frac{a^2-b^2}{2} & 0 & 0 & 0 & \frac{1}{2}\\[6pt]
\frac{a}{\sqrt{2}} & \frac{b}{\sqrt{2}} & \frac{1}{\sqrt{2}} & 0 & 0 & 0\\[6pt]
\frac{b}{\sqrt{2}} & -\frac{a}{\sqrt{2}} & 0 & \frac{1}{\sqrt{2}} & 0 & 0\\[6pt]
1 & 0 & 0 & 0 & 0 & 0
\end{pmatrix}
\,,
\end{equation}
with $a := \operatorname{Re}(\alpha),\ b := \operatorname{Im}(\alpha)$.
One finds, e.g., with symbolic computation software, that $\rank(M) = 5$ (independently of $a,b \in \mathbb{R}$).
Thus, we obtained through the stellar representation that $\dim(\orb_{G_{\rm GO}}(\ket{\alpha})) = 5$ for all $\alpha \in \mathbb{C}$.

Of course, for coherent states this was the expected result [since those are Gaussian pure states and we already saw in the main text that $\dim(\orb_{G_{\rm GO}}(\ket{0})) = 5$], but the same techniques can be applied to, e.g., infinite stellar-rank states \cite{Walschaers-NonGaussianQuantum-2021}.
Indeed, let us for instance follow the same procedure for the single-mode cat state $\ket{\mathrm{cat}_\alpha} := c(\ket{\alpha} + \ket{-\alpha})$ for some $\alpha \in \mathbb{C}\setminus\{0\}$ and normalization constant $c>0$, whose stellar function writes (up to normalization) as
\begin{equation}
f_{\alpha}(z) = \cosh(\alpha z)\,.
\end{equation}
This time, we get
\begin{align}
\widehat{N} f(z) &= \alpha z\,\sinh(\alpha z),\\
\widehat{s} f(z) &= \frac{1}{2}\Big(z^2+\alpha^2\Big)\cosh(\alpha z),\\
\widehat{S} f(z) &= \frac{i}{2}\Big(z^2-\alpha^2\Big)\cosh(\alpha z),\\
\widehat{q} f(z) &= \frac{1}{\sqrt{2}}\Big(z\cosh(\alpha z)+\alpha\sinh(\alpha z)\Big),\\
\widehat{p} f(z) &= \frac{i}{\sqrt{2}}\Big(z\cosh(\alpha z)-\alpha\sinh(\alpha z)\Big),\\
\widehat{\id} f(z) &= \cosh(\alpha z),
\end{align}
which are hence all elements of $\spa_{\mathbb{R}}(\mathcal{B})$ with
\begin{equation}\label{eq:basis-functions-for-cat-state-calc}
\mathcal{B} = \left\{
\begin{aligned}
\cosh(\alpha z),\ i\cosh(\alpha z),\\
\sinh(\alpha z),\ i\sinh(\alpha z),\\
z\cosh(\alpha z),\ iz\cosh(\alpha z),\\
z\sinh(\alpha z),\ iz\sinh(\alpha z),\\
z^2\cosh(\alpha z),\ iz^2\cosh(\alpha z)
\end{aligned}
\right\}\,,
\end{equation}
whose ten functions are again linearly independent over $\mathbb{R}$.
Thus, we get that
$\dim(\orb_{G_{\rm GO}}(\ket{\mathrm{cat}_\alpha})) = \rank(M)$, with $M$ the $6 \times 10$ real matrix:
\begin{equation}
M =
\begin{pmatrix}
0 & 0 & 0 & 0 & 0 & 0 & a & b & 0 & 0\\[4pt]
\frac{a^2-b^2}{2} & ab & 0 & 0 & 0 & 0 & 0 & 0 & \frac12 & 0\\[6pt]
ab & -\frac{a^2-b^2}{2} & 0 & 0 & 0 & 0 & 0 & 0 & 0 & \frac12\\[6pt]
0 & 0 & \frac{a}{\sqrt2} & \frac{b}{\sqrt2} & \frac{1}{\sqrt2} & 0 & 0 & 0 & 0 & 0\\[6pt]
0 & 0 & \frac{b}{\sqrt2} & -\frac{a}{\sqrt2} & 0 & \frac{1}{\sqrt2} & 0 & 0 & 0 & 0\\[6pt]
1 & 0 & 0 & 0 & 0 & 0 & 0 & 0 & 0 & 0
\end{pmatrix}
\,,
\end{equation}
with $a := \operatorname{Re}(\alpha)$ and $b := \operatorname{Im}(\alpha)$.
Using symbolic computation, we find that $\rank(M) = 6$ for all $(a,b )\in \mathbb{R}^2\setminus\{(0,0)\}$.

Hence we have shown that for all $\alpha\neq0$:
\begin{equation}
\dim(\orb_{G_{\rm GO}}(\ket{\mathrm{cat}_\alpha})) = 6\,,
\end{equation}
which is the maximal possible dimension [for $m=1$, $\dim(G_{\rm GO}) = 6$].

We end the examples by considering and comparing two different flavors of two-mode cat states.
On one hand, consider the two-mode state given by the stellar function
\begin{equation}
f_a(z_1,z_2) = c\,(1 + e^{z_1} + e^{z_2})\,,
\end{equation}
with $c$ a normalization constant [$c=(7 + 2 e)^{-1/2}$].
Note that this corresponds in the Fock representation [inverting \cref{eq:stellar-isomorphism-explicit}] to the (normalized) state
\begin{equation}
\ket{\psi_a} = c\left(
3\ket{0,0} +
\sum_{n\geq1} \frac{1}{\sqrt{n!}} \ket{n,0} +
\sum_{n\geq1} \frac{1}{\sqrt{n!}} \ket{0,n}
\right)\,.
\end{equation}
This time, the actions of the (15) generators of $G_{\rm GO}$ on $f_a$ give (discarding $c$):
\begin{align}
\widehat{\mathrm{id}}\,f_a
&= 1+e^{z_1}+e^{z_2},\\[2pt]
\widehat e_{12}\,f_a
&= \frac12\!\left(z_1 e^{z_2}+z_2 e^{z_1}\right),\\
\widehat E_{12}\,f_a
&= \frac{i}{2}\!\left(z_1 e^{z_2}-z_2 e^{z_1}\right),\\[2pt]
\widehat r_{12}\,f_a
&= \frac12\,z_1 z_2\!\left(1+e^{z_1}+e^{z_2}\right),\\
\widehat R_{12}\,f_a
&= \frac{i}{2}\,z_1 z_2\!\left(1+e^{z_1}+e^{z_2}\right),\\[2pt]
\widehat N_{1}\,f_a
&= z_1 e^{z_1},\\
\widehat N_{2}\,f_a
&= z_2 e^{z_2},\\[2pt]
\widehat s_{1}\,f_a
&= \frac12\!\left[z_1^{2}\!\left(1+e^{z_1}+e^{z_2}\right)+e^{z_1}\right],\\
\widehat s_{2}\,f_a
&= \frac12\!\left[z_2^{2}\!\left(1+e^{z_1}+e^{z_2}\right)+e^{z_2}\right],\\[2pt]
\widehat S_{1}\,f_a
&= \frac{i}{2}\!\left[z_1^{2}\!\left(1+e^{z_1}+e^{z_2}\right)-e^{z_1}\right],\\
\widehat S_{2}\,f_a
&= \frac{i}{2}\!\left[z_2^{2}\!\left(1+e^{z_1}+e^{z_2}\right)-e^{z_2}\right],\\[2pt]
\widehat q_{1}\,f_a
&= \frac{1}{\sqrt2}\!\left[z_1\!\left(1+e^{z_1}+e^{z_2}\right)+e^{z_1}\right],\\
\widehat q_{2}\,f_a
&= \frac{1}{\sqrt2}\!\left[z_2\!\left(1+e^{z_1}+e^{z_2}\right)+e^{z_2}\right],\\[2pt]
\widehat p_{1}\,f_a
&= \frac{i}{\sqrt2}\!\left[z_1\!\left(1+e^{z_1}+e^{z_2}\right)-e^{z_1}\right],\\
\widehat p_{2}\,f_a
&= \frac{i}{\sqrt2}\!\left[z_2\!\left(1+e^{z_1}+e^{z_2}\right)-e^{z_2}\right].
\end{align}
These are all elements of $\spa_{\mathbb{R}}(\mathcal{B})$ with $\mathcal{B} := (\mathcal{B}_{0} \cup \mathcal{B}_{1} \cup \mathcal{B}_{2} \cup i\mathcal{B}_{0} \cup i\mathcal{B}_{1} \cup i\mathcal{B}_{2})\setminus\{i\}$, $g_0 := 1,\ g_1 := e^{z_1},\ g_2 := e^{z_2}$,
\begin{equation}
\mathcal{B}_j := \left\{ g_j,\ g_jz_1,\ g_jz_2,\ g_jz_1z_2,\ g_jz_1^2,\ g_jz_2^2 \right\}\,,
\end{equation}
for $j=0,1,2$, and the 35 functions in $\mathcal{B}$ are again linearly independent over $\mathbb{R}$.

Thus, we get that
$\dim(\orb_{G_{\rm GO}}(\ket{\psi_a})) = \rank(M)$, with $M$ the associated $15 \times 35$ real matrix of the 15 above functions $\widehat{H}f_a$ in the basis $\mathcal{B}$. Using symbolic computation, we find that this (sparse) matrix is full rank, i.e., $\rank(M) = 15$. We have thus shown that
\begin{equation}
\dim(\orb_{G_{\rm GO}}(\ket{\psi_a})) = 15\,.
\end{equation}
On the other hand, proceeding likewise for
\begin{equation}
f_b(z_1,z_2) = c'[\cosh(\alpha z_1 z_2)]\,,
\end{equation}
with $\alpha \in \mathbb{C}$, $0<|\alpha|<1$ (so that the state is normalizable), and with $c'$ a normalization constant [$c'=(1 - |\alpha|^4)^{1/2}$], which in the Fock representation reads
\begin{equation}
\ket{\psi_b} = c'
\sum_{n\geq0} \alpha^{2n} \ket{2n,2n}\,,
\end{equation}
we find that
\begin{equation}
\dim(\orb_{G_{\rm GO}}(\ket{\psi_b})) = 14\,.
\end{equation}

This showcases that $\ket{\psi_a}$ and $\ket{\psi_b}$ are two-mode states that are both of infinite stellar rank, but that have different orbit dimensions under $G_{\rm GO}$. In this way, we obtain that these two states are inequivalent under $G_{\rm GO}$, which could not have been inferred from the concept of the stellar rank alone.

Note that we also chose the forms of $\ket{\psi_a}$ and $\ket{\psi_b}$ to be both nonproduct (and likely nonseparable under $G_{\rm GO}$) states, to make things less trivial.


\begin{thebibliography}{111}%
\makeatletter
\providecommand \@ifxundefined [1]{%
 \@ifx{#1\undefined}
}%
\providecommand \@ifnum [1]{%
 \ifnum #1\expandafter \@firstoftwo
 \else \expandafter \@secondoftwo
 \fi
}%
\providecommand \@ifx [1]{%
 \ifx #1\expandafter \@firstoftwo
 \else \expandafter \@secondoftwo
 \fi
}%
\providecommand \natexlab [1]{#1}%
\providecommand \enquote  [1]{``#1''}%
\providecommand \bibnamefont  [1]{#1}%
\providecommand \bibfnamefont [1]{#1}%
\providecommand \citenamefont [1]{#1}%
\providecommand \href@noop [0]{\@secondoftwo}%
\providecommand \href [0]{\begingroup \@sanitize@url \@href}%
\providecommand \@href[1]{\@@startlink{#1}\@@href}%
\providecommand \@@href[1]{\endgroup#1\@@endlink}%
\providecommand \@sanitize@url [0]{\catcode `\\12\catcode `\$12\catcode `\&12\catcode `\#12\catcode `\^12\catcode `\_12\catcode `\%12\relax}%
\providecommand \@@startlink[1]{}%
\providecommand \@@endlink[0]{}%
\providecommand \url  [0]{\begingroup\@sanitize@url \@url }%
\providecommand \@url [1]{\endgroup\@href {#1}{\urlprefix }}%
\providecommand \urlprefix  [0]{URL }%
\providecommand \Eprint [0]{\href }%
\providecommand \doibase [0]{https://doi.org/}%
\providecommand \selectlanguage [0]{\@gobble}%
\providecommand \bibinfo  [0]{\@secondoftwo}%
\providecommand \bibfield  [0]{\@secondoftwo}%
\providecommand \translation [1]{[#1]}%
\providecommand \BibitemOpen [0]{}%
\providecommand \bibitemStop [0]{}%
\providecommand \bibitemNoStop [0]{.\EOS\space}%
\providecommand \EOS [0]{\spacefactor3000\relax}%
\providecommand \BibitemShut  [1]{\csname bibitem#1\endcsname}%
\let\auto@bib@innerbib\@empty
\bibitem [{\citenamefont {Kok}\ \emph {et~al.}(2007)\citenamefont {Kok}, \citenamefont {Munro}, \citenamefont {Nemoto}, \citenamefont {Ralph}, \citenamefont {Dowling},\ and\ \citenamefont {Milburn}}]{kok_linear_2007}%
  \BibitemOpen
  \bibfield  {author} {\bibinfo {author} {\bibfnamefont {P.}~\bibnamefont {Kok}}, \bibinfo {author} {\bibfnamefont {W.~J.}\ \bibnamefont {Munro}}, \bibinfo {author} {\bibfnamefont {K.}~\bibnamefont {Nemoto}}, \bibinfo {author} {\bibfnamefont {T.~C.}\ \bibnamefont {Ralph}}, \bibinfo {author} {\bibfnamefont {J.~P.}\ \bibnamefont {Dowling}},\ and\ \bibinfo {author} {\bibfnamefont {G.~J.}\ \bibnamefont {Milburn}},\ }\bibfield  {title} {\bibinfo {title} {Linear optical quantum computing with photonic qubits},\ }\href {https://doi.org/10.1103/RevModPhys.79.135} {\bibfield  {journal} {\bibinfo  {journal} {Reviews of Modern Physics}\ }\textbf {\bibinfo {volume} {79}},\ \bibinfo {pages} {135} (\bibinfo {year} {2007})}\BibitemShut {NoStop}%
\bibitem [{\citenamefont {Aaronson}\ and\ \citenamefont {Arkhipov}(2013)}]{Aaronson-ComputationalComplexity-2013}%
  \BibitemOpen
  \bibfield  {author} {\bibinfo {author} {\bibfnamefont {S.}~\bibnamefont {Aaronson}}\ and\ \bibinfo {author} {\bibfnamefont {A.}~\bibnamefont {Arkhipov}},\ }\bibfield  {title} {\bibinfo {title} {The {{Computational Complexity}} of {{Linear Optics}}},\ }\href {https://doi.org/10.4086/toc.2013.v009a004} {\bibfield  {journal} {\bibinfo  {journal} {Theory of Computing}\ }\textbf {\bibinfo {volume} {9}},\ \bibinfo {pages} {143} (\bibinfo {year} {2013})}\BibitemShut {NoStop}%
\bibitem [{\citenamefont {Lloyd}\ and\ \citenamefont {Braunstein}(1999)}]{lloyd_quantum_1999}%
  \BibitemOpen
  \bibfield  {author} {\bibinfo {author} {\bibfnamefont {S.}~\bibnamefont {Lloyd}}\ and\ \bibinfo {author} {\bibfnamefont {S.~L.}\ \bibnamefont {Braunstein}},\ }\bibfield  {title} {\bibinfo {title} {Quantum {Computation} over {Continuous} {Variables}},\ }\href {https://doi.org/10.1103/PhysRevLett.82.1784} {\bibfield  {journal} {\bibinfo  {journal} {Physical Review Letters}\ }\textbf {\bibinfo {volume} {82}},\ \bibinfo {pages} {1784} (\bibinfo {year} {1999})}\BibitemShut {NoStop}%
\bibitem [{\citenamefont {Braunstein}\ and\ \citenamefont {{van Loock}}(2005)}]{Braunstein-QuantumInformation-2005}%
  \BibitemOpen
  \bibfield  {author} {\bibinfo {author} {\bibfnamefont {S.~L.}\ \bibnamefont {Braunstein}}\ and\ \bibinfo {author} {\bibfnamefont {P.}~\bibnamefont {{van Loock}}},\ }\bibfield  {title} {\bibinfo {title} {Quantum information with continuous variables},\ }\href {https://doi.org/10.1103/RevModPhys.77.513} {\bibfield  {journal} {\bibinfo  {journal} {Reviews of Modern Physics}\ }\textbf {\bibinfo {volume} {77}},\ \bibinfo {pages} {513} (\bibinfo {year} {2005})}\BibitemShut {NoStop}%
\bibitem [{\citenamefont {Weedbrook}\ \emph {et~al.}(2012)\citenamefont {Weedbrook}, \citenamefont {Pirandola}, \citenamefont {García-Patrón}, \citenamefont {Cerf}, \citenamefont {Ralph}, \citenamefont {Shapiro},\ and\ \citenamefont {Lloyd}}]{weedbrook_gaussian_2012}%
  \BibitemOpen
  \bibfield  {author} {\bibinfo {author} {\bibfnamefont {C.}~\bibnamefont {Weedbrook}}, \bibinfo {author} {\bibfnamefont {S.}~\bibnamefont {Pirandola}}, \bibinfo {author} {\bibfnamefont {R.}~\bibnamefont {García-Patrón}}, \bibinfo {author} {\bibfnamefont {N.~J.}\ \bibnamefont {Cerf}}, \bibinfo {author} {\bibfnamefont {T.~C.}\ \bibnamefont {Ralph}}, \bibinfo {author} {\bibfnamefont {J.~H.}\ \bibnamefont {Shapiro}},\ and\ \bibinfo {author} {\bibfnamefont {S.}~\bibnamefont {Lloyd}},\ }\bibfield  {title} {\bibinfo {title} {Gaussian quantum information},\ }\href {https://doi.org/10.1103/RevModPhys.84.621} {\bibfield  {journal} {\bibinfo  {journal} {Reviews of Modern Physics}\ }\textbf {\bibinfo {volume} {84}},\ \bibinfo {pages} {621} (\bibinfo {year} {2012})}\BibitemShut {NoStop}%
\bibitem [{\citenamefont {Adesso}\ \emph {et~al.}(2014)\citenamefont {Adesso}, \citenamefont {Ragy},\ and\ \citenamefont {Lee}}]{Adesso-ContinuousVariable-2014}%
  \BibitemOpen
  \bibfield  {author} {\bibinfo {author} {\bibfnamefont {G.}~\bibnamefont {Adesso}}, \bibinfo {author} {\bibfnamefont {S.}~\bibnamefont {Ragy}},\ and\ \bibinfo {author} {\bibfnamefont {A.~R.}\ \bibnamefont {Lee}},\ }\bibfield  {title} {\bibinfo {title} {Continuous variable quantum information: {{Gaussian}} states and beyond},\ }\href {https://doi.org/10.1142/S1230161214400010} {\bibfield  {journal} {\bibinfo  {journal} {Open Systems \& Information Dynamics}\ }\textbf {\bibinfo {volume} {21}},\ \bibinfo {pages} {1440001} (\bibinfo {year} {2014})},\ \Eprint {https://arxiv.org/abs/1401.4679} {1401.4679 [quant-ph]} \BibitemShut {NoStop}%
\bibitem [{\citenamefont {Hamilton}\ \emph {et~al.}(2017)\citenamefont {Hamilton}, \citenamefont {Kruse}, \citenamefont {Sansoni}, \citenamefont {Barkhofen}, \citenamefont {Silberhorn},\ and\ \citenamefont {Jex}}]{hamilton_gaussian_2017}%
  \BibitemOpen
  \bibfield  {author} {\bibinfo {author} {\bibfnamefont {C.~S.}\ \bibnamefont {Hamilton}}, \bibinfo {author} {\bibfnamefont {R.}~\bibnamefont {Kruse}}, \bibinfo {author} {\bibfnamefont {L.}~\bibnamefont {Sansoni}}, \bibinfo {author} {\bibfnamefont {S.}~\bibnamefont {Barkhofen}}, \bibinfo {author} {\bibfnamefont {C.}~\bibnamefont {Silberhorn}},\ and\ \bibinfo {author} {\bibfnamefont {I.}~\bibnamefont {Jex}},\ }\bibfield  {title} {\bibinfo {title} {Gaussian {Boson} {Sampling}},\ }\href {https://doi.org/10.1103/PhysRevLett.119.170501} {\bibfield  {journal} {\bibinfo  {journal} {Physical Review Letters}\ }\textbf {\bibinfo {volume} {119}},\ \bibinfo {pages} {170501} (\bibinfo {year} {2017})}\BibitemShut {NoStop}%
\bibitem [{\citenamefont {Serafini}(2023)}]{Serafini-QuantumContinuous-2023}%
  \BibitemOpen
  \bibfield  {author} {\bibinfo {author} {\bibfnamefont {A.}~\bibnamefont {Serafini}},\ }\href {https://doi.org/10.1201/9781003250975} {\emph {\bibinfo {title} {Quantum {{Continuous Variables}}: {{A Primer}} of {{Theoretical Methods}}}}},\ \bibinfo {edition} {2nd}\ ed.\ (\bibinfo  {publisher} {CRC Press},\ \bibinfo {address} {Boca Raton},\ \bibinfo {year} {2023})\BibitemShut {NoStop}%
\bibitem [{\citenamefont {Walschaers}(2021)}]{Walschaers-NonGaussianQuantum-2021}%
  \BibitemOpen
  \bibfield  {author} {\bibinfo {author} {\bibfnamefont {M.}~\bibnamefont {Walschaers}},\ }\bibfield  {title} {\bibinfo {title} {Non-{{Gaussian Quantum States}} and {{Where}} to {{Find Them}}},\ }\href {https://doi.org/10.1103/PRXQuantum.2.030204} {\bibfield  {journal} {\bibinfo  {journal} {PRX Quantum}\ }\textbf {\bibinfo {volume} {2}},\ \bibinfo {pages} {030204} (\bibinfo {year} {2021})}\BibitemShut {NoStop}%
\bibitem [{\citenamefont {D'Alessandro}(2021)}]{dalessandro_introduction_2021}%
  \BibitemOpen
  \bibfield  {author} {\bibinfo {author} {\bibfnamefont {D.}~\bibnamefont {D'Alessandro}},\ }\href {https://doi.org/10.1201/9781003051268} {{\selectlanguage {en}\emph {\bibinfo {title} {Introduction to {Quantum} {Control} and {Dynamics}}}}},\ \bibinfo {edition} {2nd}\ ed.\ (\bibinfo  {publisher} {Chapman and Hall/CRC},\ \bibinfo {address} {Boca Raton},\ \bibinfo {year} {2021})\BibitemShut {NoStop}%
\bibitem [{\citenamefont {Zanardi}\ and\ \citenamefont {Lloyd}(2004)}]{zanardi_universal_2004}%
  \BibitemOpen
  \bibfield  {author} {\bibinfo {author} {\bibfnamefont {P.}~\bibnamefont {Zanardi}}\ and\ \bibinfo {author} {\bibfnamefont {S.}~\bibnamefont {Lloyd}},\ }\bibfield  {title} {\bibinfo {title} {Universal control of quantum subspaces and subsystems},\ }\href {https://doi.org/10.1103/PhysRevA.69.022313} {\bibfield  {journal} {\bibinfo  {journal} {Physical Review A}\ }\textbf {\bibinfo {volume} {69}},\ \bibinfo {pages} {022313} (\bibinfo {year} {2004})}\BibitemShut {NoStop}%
\bibitem [{\citenamefont {Knill}\ \emph {et~al.}(2001)\citenamefont {Knill}, \citenamefont {Laflamme},\ and\ \citenamefont {Milburn}}]{knill_scheme_2001}%
  \BibitemOpen
  \bibfield  {author} {\bibinfo {author} {\bibfnamefont {E.}~\bibnamefont {Knill}}, \bibinfo {author} {\bibfnamefont {R.}~\bibnamefont {Laflamme}},\ and\ \bibinfo {author} {\bibfnamefont {G.~J.}\ \bibnamefont {Milburn}},\ }\bibfield  {title} {{\selectlanguage {en}\bibinfo {title} {A scheme for efficient quantum computation with linear optics}},\ }\href {https://doi.org/10.1038/35051009} {\bibfield  {journal} {\bibinfo  {journal} {Nature}\ }\textbf {\bibinfo {volume} {409}},\ \bibinfo {pages} {46} (\bibinfo {year} {2001})}\BibitemShut {NoStop}%
\bibitem [{\citenamefont {VanMeter}\ \emph {et~al.}(2007)\citenamefont {VanMeter}, \citenamefont {Lougovski}, \citenamefont {Uskov}, \citenamefont {Kieling}, \citenamefont {Eisert},\ and\ \citenamefont {Dowling}}]{VanMeter-GeneralLinearoptical-2007}%
  \BibitemOpen
  \bibfield  {author} {\bibinfo {author} {\bibfnamefont {N.~M.}\ \bibnamefont {VanMeter}}, \bibinfo {author} {\bibfnamefont {P.}~\bibnamefont {Lougovski}}, \bibinfo {author} {\bibfnamefont {D.~B.}\ \bibnamefont {Uskov}}, \bibinfo {author} {\bibfnamefont {K.}~\bibnamefont {Kieling}}, \bibinfo {author} {\bibfnamefont {J.}~\bibnamefont {Eisert}},\ and\ \bibinfo {author} {\bibfnamefont {J.~P.}\ \bibnamefont {Dowling}},\ }\bibfield  {title} {\bibinfo {title} {General linear-optical quantum state generation scheme: {{Applications}} to maximally path-entangled states},\ }\href {https://doi.org/10.1103/PhysRevA.76.063808} {\bibfield  {journal} {\bibinfo  {journal} {Physical Review A}\ }\textbf {\bibinfo {volume} {76}},\ \bibinfo {pages} {063808} (\bibinfo {year} {2007})}\BibitemShut {NoStop}%
\bibitem [{\citenamefont {Chabaud}\ \emph {et~al.}(2021{\natexlab{a}})\citenamefont {Chabaud}, \citenamefont {Markham},\ and\ \citenamefont {Sohbi}}]{chabaud_quantum_2021}%
  \BibitemOpen
  \bibfield  {author} {\bibinfo {author} {\bibfnamefont {U.}~\bibnamefont {Chabaud}}, \bibinfo {author} {\bibfnamefont {D.}~\bibnamefont {Markham}},\ and\ \bibinfo {author} {\bibfnamefont {A.}~\bibnamefont {Sohbi}},\ }\bibfield  {title} {{\selectlanguage {en-GB}\bibinfo {title} {Quantum machine learning with adaptive linear optics}},\ }\href {https://doi.org/10.22331/q-2021-07-05-496} {\bibfield  {journal} {\bibinfo  {journal} {Quantum}\ }\textbf {\bibinfo {volume} {5}},\ \bibinfo {pages} {496} (\bibinfo {year} {2021}{\natexlab{a}})}\BibitemShut {NoStop}%
\bibitem [{\citenamefont {Monbroussou}\ \emph {et~al.}(2025{\natexlab{a}})\citenamefont {Monbroussou}, \citenamefont {Mamon}, \citenamefont {Thomas}, \citenamefont {Yacoub}, \citenamefont {Chabaud},\ and\ \citenamefont {Kashefi}}]{Monbroussou-QuantumAdvantage-2025}%
  \BibitemOpen
  \bibfield  {author} {\bibinfo {author} {\bibfnamefont {L.}~\bibnamefont {Monbroussou}}, \bibinfo {author} {\bibfnamefont {E.~Z.}\ \bibnamefont {Mamon}}, \bibinfo {author} {\bibfnamefont {H.}~\bibnamefont {Thomas}}, \bibinfo {author} {\bibfnamefont {V.}~\bibnamefont {Yacoub}}, \bibinfo {author} {\bibfnamefont {U.}~\bibnamefont {Chabaud}},\ and\ \bibinfo {author} {\bibfnamefont {E.}~\bibnamefont {Kashefi}},\ }\bibfield  {title} {\bibinfo {title} {Toward quantum advantage with photonic state injection},\ }\href {https://doi.org/10.1103/PhysRevResearch.7.033051} {\bibfield  {journal} {\bibinfo  {journal} {Physical Review Research}\ }\textbf {\bibinfo {volume} {7}},\ \bibinfo {pages} {033051} (\bibinfo {year} {2025}{\natexlab{a}})}\BibitemShut {NoStop}%
\bibitem [{\citenamefont {Dell’Anno}\ \emph {et~al.}(2006)\citenamefont {Dell’Anno}, \citenamefont {De~Siena},\ and\ \citenamefont {Illuminati}}]{dellanno_multiphoton_2006}%
  \BibitemOpen
  \bibfield  {author} {\bibinfo {author} {\bibfnamefont {F.}~\bibnamefont {Dell’Anno}}, \bibinfo {author} {\bibfnamefont {S.}~\bibnamefont {De~Siena}},\ and\ \bibinfo {author} {\bibfnamefont {F.}~\bibnamefont {Illuminati}},\ }\bibfield  {title} {\bibinfo {title} {Multiphoton quantum optics and quantum state engineering},\ }\href {https://doi.org/10.1016/j.physrep.2006.01.004} {\bibfield  {journal} {\bibinfo  {journal} {Physics Reports}\ }\textbf {\bibinfo {volume} {428}},\ \bibinfo {pages} {53} (\bibinfo {year} {2006})}\BibitemShut {NoStop}%
\bibitem [{\citenamefont {Migdał}\ \emph {et~al.}(2014)\citenamefont {Migdał}, \citenamefont {Rodríguez-Laguna}, \citenamefont {Oszmaniec},\ and\ \citenamefont {Lewenstein}}]{migdal_which_2014}%
  \BibitemOpen
  \bibfield  {author} {\bibinfo {author} {\bibfnamefont {P.}~\bibnamefont {Migdał}}, \bibinfo {author} {\bibfnamefont {J.}~\bibnamefont {Rodríguez-Laguna}}, \bibinfo {author} {\bibfnamefont {M.}~\bibnamefont {Oszmaniec}},\ and\ \bibinfo {author} {\bibfnamefont {M.}~\bibnamefont {Lewenstein}},\ }\bibfield  {title} {\bibinfo {title} {Multiphoton states related via linear optics},\ }\href {https://doi.org/10.1103/PhysRevA.89.062329} {\bibfield  {journal} {\bibinfo  {journal} {Physical Review A}\ }\textbf {\bibinfo {volume} {89}},\ \bibinfo {pages} {062329} (\bibinfo {year} {2014})}\BibitemShut {NoStop}%
\bibitem [{\citenamefont {Parellada}\ \emph {et~al.}(2023)\citenamefont {Parellada}, \citenamefont {Gimeno~i Garcia}, \citenamefont {Moyano-Fernández},\ and\ \citenamefont {Garcia-Escartin}}]{parellada_no-go_2023}%
  \BibitemOpen
  \bibfield  {author} {\bibinfo {author} {\bibfnamefont {P.~V.}\ \bibnamefont {Parellada}}, \bibinfo {author} {\bibfnamefont {V.}~\bibnamefont {Gimeno~i Garcia}}, \bibinfo {author} {\bibfnamefont {J.~J.}\ \bibnamefont {Moyano-Fernández}},\ and\ \bibinfo {author} {\bibfnamefont {J.~C.}\ \bibnamefont {Garcia-Escartin}},\ }\bibfield  {title} {\bibinfo {title} {No-go theorems for photon state transformations in quantum linear optics},\ }\href {https://doi.org/10.1016/j.rinp.2023.107108} {\bibfield  {journal} {\bibinfo  {journal} {Results in Physics}\ }\textbf {\bibinfo {volume} {54}},\ \bibinfo {pages} {107108} (\bibinfo {year} {2023})}\BibitemShut {NoStop}%
\bibitem [{\citenamefont {Parellada}\ \emph {et~al.}(2026)\citenamefont {Parellada}, \citenamefont {{Gimeno i Garcia}}, \citenamefont {{Moyano-Fern{\'a}ndez}},\ and\ \citenamefont {{Garcia-Escartin}}}]{Parellada-LieAlgebraic-2026}%
  \BibitemOpen
  \bibfield  {author} {\bibinfo {author} {\bibfnamefont {P.~V.}\ \bibnamefont {Parellada}}, \bibinfo {author} {\bibfnamefont {V.}~\bibnamefont {{Gimeno i Garcia}}}, \bibinfo {author} {\bibfnamefont {J.~J.}\ \bibnamefont {{Moyano-Fern{\'a}ndez}}},\ and\ \bibinfo {author} {\bibfnamefont {J.~C.}\ \bibnamefont {{Garcia-Escartin}}},\ }\bibfield  {title} {\bibinfo {title} {Lie algebraic invariants in quantum linear optics},\ }\href {https://doi.org/10.22331/q-2026-06-12-2132} {\bibfield  {journal} {\bibinfo  {journal} {Quantum}\ }\textbf {\bibinfo {volume} {10}},\ \bibinfo {pages} {2132} (\bibinfo {year} {2026})}\BibitemShut {NoStop}%
\bibitem [{\citenamefont {Chabaud}\ \emph {et~al.}(2020)\citenamefont {Chabaud}, \citenamefont {Markham},\ and\ \citenamefont {Grosshans}}]{chabaud_stellar_2020}%
  \BibitemOpen
  \bibfield  {author} {\bibinfo {author} {\bibfnamefont {U.}~\bibnamefont {Chabaud}}, \bibinfo {author} {\bibfnamefont {D.}~\bibnamefont {Markham}},\ and\ \bibinfo {author} {\bibfnamefont {F.}~\bibnamefont {Grosshans}},\ }\bibfield  {title} {\bibinfo {title} {Stellar {Representation} of {Non}-{Gaussian} {Quantum} {States}},\ }\href {https://doi.org/10.1103/PhysRevLett.124.063605} {\bibfield  {journal} {\bibinfo  {journal} {Physical Review Letters}\ }\textbf {\bibinfo {volume} {124}},\ \bibinfo {pages} {063605} (\bibinfo {year} {2020})}\BibitemShut {NoStop}%
\bibitem [{\citenamefont {Hahn}\ \emph {et~al.}(2026)\citenamefont {Hahn}, \citenamefont {Garnier}, \citenamefont {Ferrini}, \citenamefont {Ferraro},\ and\ \citenamefont {Chabaud}}]{Hahn-AssessingNonGaussian-2026}%
  \BibitemOpen
  \bibfield  {author} {\bibinfo {author} {\bibfnamefont {O.}~\bibnamefont {Hahn}}, \bibinfo {author} {\bibfnamefont {M.}~\bibnamefont {Garnier}}, \bibinfo {author} {\bibfnamefont {G.}~\bibnamefont {Ferrini}}, \bibinfo {author} {\bibfnamefont {A.}~\bibnamefont {Ferraro}},\ and\ \bibinfo {author} {\bibfnamefont {U.}~\bibnamefont {Chabaud}},\ }\bibfield  {title} {\bibinfo {title} {Assessing non-{{Gaussian}} quantum state conversion with the stellar rank},\ }\href {https://doi.org/10.22331/q-2026-05-05-2095} {\bibfield  {journal} {\bibinfo  {journal} {Quantum}\ }\textbf {\bibinfo {volume} {10}},\ \bibinfo {pages} {2095} (\bibinfo {year} {2026})}\BibitemShut {NoStop}%
\bibitem [{\citenamefont {Mele}\ \emph {et~al.}(2026)\citenamefont {Mele}, \citenamefont {Oliviero}, \citenamefont {Upreti},\ and\ \citenamefont {Chabaud}}]{Mele-SymplecticRank-2026}%
  \BibitemOpen
  \bibfield  {author} {\bibinfo {author} {\bibfnamefont {F.~A.}\ \bibnamefont {Mele}}, \bibinfo {author} {\bibfnamefont {S.~F.}\ \bibnamefont {Oliviero}}, \bibinfo {author} {\bibfnamefont {V.}~\bibnamefont {Upreti}},\ and\ \bibinfo {author} {\bibfnamefont {U.}~\bibnamefont {Chabaud}},\ }\bibfield  {title} {\bibinfo {title} {Symplectic {{Rank}} of {{Non-Gaussian Quantum States}}},\ }\href {https://doi.org/10.1103/1rtk-1jsn} {\bibfield  {journal} {\bibinfo  {journal} {PRX Quantum}\ }\textbf {\bibinfo {volume} {7}},\ \bibinfo {pages} {020366} (\bibinfo {year} {2026})}\BibitemShut {NoStop}%
\bibitem [{\citenamefont {Rodari}\ \emph {et~al.}(2025)\citenamefont {Rodari}, \citenamefont {Francalanci}, \citenamefont {Caruccio}, \citenamefont {Hoch}, \citenamefont {Carvacho}, \citenamefont {Giordani}, \citenamefont {Spagnolo}, \citenamefont {Albiero}, \citenamefont {Di~Giano}, \citenamefont {Ceccarelli}, \citenamefont {Corrielli}, \citenamefont {Crespi}, \citenamefont {Osellame}, \citenamefont {Chabaud},\ and\ \citenamefont {Sciarrino}}]{Rodari-ObservationLie-2025}%
  \BibitemOpen
  \bibfield  {author} {\bibinfo {author} {\bibfnamefont {G.}~\bibnamefont {Rodari}}, \bibinfo {author} {\bibfnamefont {T.}~\bibnamefont {Francalanci}}, \bibinfo {author} {\bibfnamefont {E.}~\bibnamefont {Caruccio}}, \bibinfo {author} {\bibfnamefont {F.}~\bibnamefont {Hoch}}, \bibinfo {author} {\bibfnamefont {G.}~\bibnamefont {Carvacho}}, \bibinfo {author} {\bibfnamefont {T.}~\bibnamefont {Giordani}}, \bibinfo {author} {\bibfnamefont {N.}~\bibnamefont {Spagnolo}}, \bibinfo {author} {\bibfnamefont {R.}~\bibnamefont {Albiero}}, \bibinfo {author} {\bibfnamefont {N.}~\bibnamefont {Di~Giano}}, \bibinfo {author} {\bibfnamefont {F.}~\bibnamefont {Ceccarelli}}, \bibinfo {author} {\bibfnamefont {G.}~\bibnamefont {Corrielli}}, \bibinfo {author} {\bibfnamefont {A.}~\bibnamefont {Crespi}}, \bibinfo {author} {\bibfnamefont {R.}~\bibnamefont {Osellame}}, \bibinfo {author} {\bibfnamefont {U.}~\bibnamefont {Chabaud}},\ and\ \bibinfo {author} {\bibfnamefont {F.}~\bibnamefont {Sciarrino}},\ }\bibfield  {title} {\bibinfo {title} {Observation of {{Lie}} algebraic invariants in quantum linear optics},\ }\href {https://doi.org/10.1103/7961-hg2q} {\bibfield  {journal} {\bibinfo  {journal} {Physical Review Research}\ }\textbf {\bibinfo {volume} {7}},\ \bibinfo {pages} {043325} (\bibinfo {year} {2025})}\BibitemShut {NoStop}%
\bibitem [{\citenamefont {Cenci}\ \emph {et~al.}(2010)\citenamefont {Cenci}, \citenamefont {Lyons}, \citenamefont {Snyder},\ and\ \citenamefont {Walck}}]{cenci_symmetric_2010}%
  \BibitemOpen
  \bibfield  {author} {\bibinfo {author} {\bibfnamefont {C.~D.}\ \bibnamefont {Cenci}}, \bibinfo {author} {\bibfnamefont {D.~W.}\ \bibnamefont {Lyons}}, \bibinfo {author} {\bibfnamefont {L.~M.}\ \bibnamefont {Snyder}},\ and\ \bibinfo {author} {\bibfnamefont {S.~N.}\ \bibnamefont {Walck}},\ }\bibfield  {title} {\bibinfo {title} {Symmetric states: local unitary equivalence via stabilizers},\ }\href {https://doi.org/10.26421/QIC10.11-12-12} {\bibfield  {journal} {\bibinfo  {journal} {Quantum Inf. Comput.}\ }\textbf {\bibinfo {volume} {10}},\ \bibinfo {pages} {1029} (\bibinfo {year} {2010})}\BibitemShut {NoStop}%
\bibitem [{\citenamefont {Lyons}\ and\ \citenamefont {Walck}(2005)}]{lyons_minimum_2005}%
  \BibitemOpen
  \bibfield  {author} {\bibinfo {author} {\bibfnamefont {D.~W.}\ \bibnamefont {Lyons}}\ and\ \bibinfo {author} {\bibfnamefont {S.~N.}\ \bibnamefont {Walck}},\ }\bibfield  {title} {\bibinfo {title} {Minimum orbit dimension for local unitary action on n-qubit pure states},\ }\href {https://doi.org/10.1063/1.2048327} {\bibfield  {journal} {\bibinfo  {journal} {Journal of Mathematical Physics}\ }\textbf {\bibinfo {volume} {46}},\ \bibinfo {pages} {102106} (\bibinfo {year} {2005})}\BibitemShut {NoStop}%
\bibitem [{\citenamefont {Luis}\ and\ \citenamefont {Sanchez-Soto}(1995)}]{luis_quantum_1995}%
  \BibitemOpen
  \bibfield  {author} {\bibinfo {author} {\bibfnamefont {A.}~\bibnamefont {Luis}}\ and\ \bibinfo {author} {\bibfnamefont {L.~L.}\ \bibnamefont {Sanchez-Soto}},\ }\bibfield  {title} {{\selectlanguage {en}\bibinfo {title} {A quantum description of the beam splitter}},\ }\href {https://doi.org/10.1088/1355-5111/7/2/005} {\bibfield  {journal} {\bibinfo  {journal} {Quantum and Semiclassical Optics: Journal of the European Optical Society Part B}\ }\textbf {\bibinfo {volume} {7}},\ \bibinfo {pages} {153} (\bibinfo {year} {1995})}\BibitemShut {NoStop}%
\bibitem [{Note1()}]{Note1}%
  \BibitemOpen
  \bibinfo {note} {That is, it is the minimal subset size $k$ for which all vectors not in the subset can be written as \protect \textit {real} linear combinations of vectors in the subset. The real rank of a list of vectors is hence greater or equal to its complex rank.}\BibitemShut {Stop}%
\bibitem [{\citenamefont {Horn}\ and\ \citenamefont {Johnson}(2013)}]{horn_matrix_2013}%
  \BibitemOpen
  \bibfield  {author} {\bibinfo {author} {\bibfnamefont {R.~A.}\ \bibnamefont {Horn}}\ and\ \bibinfo {author} {\bibfnamefont {C.~R.}\ \bibnamefont {Johnson}},\ }\href {https://doi.org/10.1017/CBO9780511810817} {{\selectlanguage {en}\emph {\bibinfo {title} {Matrix {Analysis}}}}},\ \bibinfo {edition} {2nd}\ ed.\ (\bibinfo  {publisher} {Cambridge University Press},\ \bibinfo {address} {New York},\ \bibinfo {year} {2013})\BibitemShut {NoStop}%
\bibitem [{\citenamefont {Hertz}\ and\ \citenamefont {De~Bi{\`e}vre}(2020)}]{Hertz-QuadratureCoherence-2020}%
  \BibitemOpen
  \bibfield  {author} {\bibinfo {author} {\bibfnamefont {A.}~\bibnamefont {Hertz}}\ and\ \bibinfo {author} {\bibfnamefont {S.}~\bibnamefont {De~Bi{\`e}vre}},\ }\bibfield  {title} {\bibinfo {title} {Quadrature {{Coherence Scale Driven Fast Decoherence}} of {{Bosonic Quantum Field States}}},\ }\href {https://doi.org/10.1103/PhysRevLett.124.090402} {\bibfield  {journal} {\bibinfo  {journal} {Physical Review Letters}\ }\textbf {\bibinfo {volume} {124}},\ \bibinfo {pages} {090402} (\bibinfo {year} {2020})}\BibitemShut {NoStop}%
\bibitem [{\citenamefont {Griffet}\ \emph {et~al.}(2023)\citenamefont {Griffet}, \citenamefont {Arnhem}, \citenamefont {De~Bièvre},\ and\ \citenamefont {Cerf}}]{griffet_interferometric_2023}%
  \BibitemOpen
  \bibfield  {author} {\bibinfo {author} {\bibfnamefont {C.}~\bibnamefont {Griffet}}, \bibinfo {author} {\bibfnamefont {M.}~\bibnamefont {Arnhem}}, \bibinfo {author} {\bibfnamefont {S.}~\bibnamefont {De~Bièvre}},\ and\ \bibinfo {author} {\bibfnamefont {N.~J.}\ \bibnamefont {Cerf}},\ }\bibfield  {title} {\bibinfo {title} {Interferometric measurement of the quadrature coherence scale using two replicas of a quantum optical state},\ }\href {https://doi.org/10.1103/PhysRevA.108.023730} {\bibfield  {journal} {\bibinfo  {journal} {Physical Review A}\ }\textbf {\bibinfo {volume} {108}},\ \bibinfo {pages} {023730} (\bibinfo {year} {2023})}\BibitemShut {NoStop}%
\bibitem [{\citenamefont {Glauber}(1963)}]{Glauber-CoherentIncoherent-1963}%
  \BibitemOpen
  \bibfield  {author} {\bibinfo {author} {\bibfnamefont {R.~J.}\ \bibnamefont {Glauber}},\ }\bibfield  {title} {\bibinfo {title} {Coherent and {{Incoherent States}} of the {{Radiation Field}}},\ }\href {https://doi.org/10.1103/PhysRev.131.2766} {\bibfield  {journal} {\bibinfo  {journal} {Physical Review}\ }\textbf {\bibinfo {volume} {131}},\ \bibinfo {pages} {2766} (\bibinfo {year} {1963})}\BibitemShut {NoStop}%
\bibitem [{\citenamefont {Mamon}(2025)}]{Bosonic-Orbit-Dimensions-Github-cff}%
  \BibitemOpen
  \bibfield  {author} {\bibinfo {author} {\bibfnamefont {E.~Z.}\ \bibnamefont {Mamon}},\ }\href {https://github.com/eliott-zm/bosonic-orbit-dimensions} {\bibinfo {title} {{bosonic-orbit-dimensions}}} (\bibinfo {year} {2025})\BibitemShut {NoStop}%
\bibitem [{\citenamefont {Chabaud}\ and\ \citenamefont {Walschaers}(2023)}]{Chabaud-ResourcesBosonic-2023}%
  \BibitemOpen
  \bibfield  {author} {\bibinfo {author} {\bibfnamefont {U.}~\bibnamefont {Chabaud}}\ and\ \bibinfo {author} {\bibfnamefont {M.}~\bibnamefont {Walschaers}},\ }\bibfield  {title} {\bibinfo {title} {Resources for {{Bosonic Quantum Computational Advantage}}},\ }\href {https://doi.org/10.1103/PhysRevLett.130.090602} {\bibfield  {journal} {\bibinfo  {journal} {Physical Review Letters}\ }\textbf {\bibinfo {volume} {130}},\ \bibinfo {pages} {090602} (\bibinfo {year} {2023})}\BibitemShut {NoStop}%
\bibitem [{\citenamefont {Becker}\ \emph {et~al.}(2024)\citenamefont {Becker}, \citenamefont {Datta}, \citenamefont {Lami},\ and\ \citenamefont {Rouze}}]{Becker-ClassicalShadow-2024}%
  \BibitemOpen
  \bibfield  {author} {\bibinfo {author} {\bibfnamefont {S.}~\bibnamefont {Becker}}, \bibinfo {author} {\bibfnamefont {N.}~\bibnamefont {Datta}}, \bibinfo {author} {\bibfnamefont {L.}~\bibnamefont {Lami}},\ and\ \bibinfo {author} {\bibfnamefont {C.}~\bibnamefont {Rouze}},\ }\bibfield  {title} {\bibinfo {title} {Classical {{Shadow Tomography}} for {{Continuous Variables Quantum Systems}}},\ }\href {https://doi.org/10.1109/TIT.2024.3357972} {\bibfield  {journal} {\bibinfo  {journal} {IEEE Transactions on Information Theory}\ }\textbf {\bibinfo {volume} {70}},\ \bibinfo {pages} {3427} (\bibinfo {year} {2024})}\BibitemShut {NoStop}%
\bibitem [{\citenamefont {Gandhari}\ \emph {et~al.}(2024)\citenamefont {Gandhari}, \citenamefont {Albert}, \citenamefont {Gerrits}, \citenamefont {Taylor},\ and\ \citenamefont {Gullans}}]{Gandhari-PrecisionBounds-2024}%
  \BibitemOpen
  \bibfield  {author} {\bibinfo {author} {\bibfnamefont {S.}~\bibnamefont {Gandhari}}, \bibinfo {author} {\bibfnamefont {V.~V.}\ \bibnamefont {Albert}}, \bibinfo {author} {\bibfnamefont {T.}~\bibnamefont {Gerrits}}, \bibinfo {author} {\bibfnamefont {J.~M.}\ \bibnamefont {Taylor}},\ and\ \bibinfo {author} {\bibfnamefont {M.~J.}\ \bibnamefont {Gullans}},\ }\bibfield  {title} {\bibinfo {title} {Precision {{Bounds}} on {{Continuous-Variable State Tomography Using Classical Shadows}}},\ }\href {https://doi.org/10.1103/PRXQuantum.5.010346} {\bibfield  {journal} {\bibinfo  {journal} {PRX Quantum}\ }\textbf {\bibinfo {volume} {5}},\ \bibinfo {pages} {010346} (\bibinfo {year} {2024})}\BibitemShut {NoStop}%
\bibitem [{\citenamefont {Thomas}\ \emph {et~al.}(2025)\citenamefont {Thomas}, \citenamefont {Chabaud},\ and\ \citenamefont {Emeriau}}]{Thomas-SheddingLight-2025}%
  \BibitemOpen
  \bibfield  {author} {\bibinfo {author} {\bibfnamefont {H.}~\bibnamefont {Thomas}}, \bibinfo {author} {\bibfnamefont {U.}~\bibnamefont {Chabaud}},\ and\ \bibinfo {author} {\bibfnamefont {P.-E.}\ \bibnamefont {Emeriau}},\ }\href {https://doi.org/10.48550/arXiv.2510.07240} {\bibinfo {title} {Shedding light on classical shadows: Learning photonic quantum states}} (\bibinfo {year} {2025}),\ \Eprint {https://arxiv.org/abs/2510.07240} {arXiv:2510.07240 [quant-ph]} \BibitemShut {NoStop}%
\bibitem [{\citenamefont {Brun}(2004)}]{brun_measuring_2004}%
  \BibitemOpen
  \bibfield  {author} {\bibinfo {author} {\bibfnamefont {T.~A.}\ \bibnamefont {Brun}},\ }\bibfield  {title} {\bibinfo {title} {Measuring polynomial functions of states},\ }\href {https://doi.org/10.26421/QIC4.5-6} {\bibfield  {journal} {\bibinfo  {journal} {Quantum Inf. Comput.}\ }\textbf {\bibinfo {volume} {4}},\ \bibinfo {pages} {401} (\bibinfo {year} {2004})}\BibitemShut {NoStop}%
\bibitem [{\citenamefont {Volkoff}\ and\ \citenamefont {Subaşı}(2022)}]{volkoff_ancilla-free_2022}%
  \BibitemOpen
  \bibfield  {author} {\bibinfo {author} {\bibfnamefont {T.~J.}\ \bibnamefont {Volkoff}}\ and\ \bibinfo {author} {\bibfnamefont {Y.}~\bibnamefont {Subaşı}},\ }\bibfield  {title} {{\selectlanguage {en-GB}\bibinfo {title} {Ancilla-free continuous-variable {SWAP} test}},\ }\href {https://doi.org/10.22331/q-2022-09-08-800} {\bibfield  {journal} {\bibinfo  {journal} {Quantum}\ }\textbf {\bibinfo {volume} {6}},\ \bibinfo {pages} {800} (\bibinfo {year} {2022})}\BibitemShut {NoStop}%
\bibitem [{\citenamefont {Chitambar}\ and\ \citenamefont {Gour}(2019)}]{Chitambar-QuantumResource-2019}%
  \BibitemOpen
  \bibfield  {author} {\bibinfo {author} {\bibfnamefont {E.}~\bibnamefont {Chitambar}}\ and\ \bibinfo {author} {\bibfnamefont {G.}~\bibnamefont {Gour}},\ }\bibfield  {title} {\bibinfo {title} {Quantum {{Resource Theories}}},\ }\href {https://doi.org/10.1103/RevModPhys.91.025001} {\bibfield  {journal} {\bibinfo  {journal} {Reviews of Modern Physics}\ }\textbf {\bibinfo {volume} {91}},\ \bibinfo {pages} {025001} (\bibinfo {year} {2019})},\ \Eprint {https://arxiv.org/abs/1806.06107} {1806.06107} \BibitemShut {NoStop}%
\bibitem [{\citenamefont {Takagi}\ and\ \citenamefont {Zhuang}(2018)}]{Takagi-ConvexResource-2018}%
  \BibitemOpen
  \bibfield  {author} {\bibinfo {author} {\bibfnamefont {R.}~\bibnamefont {Takagi}}\ and\ \bibinfo {author} {\bibfnamefont {Q.}~\bibnamefont {Zhuang}},\ }\bibfield  {title} {\bibinfo {title} {Convex resource theory of non-{{Gaussianity}}},\ }\href {https://doi.org/10.1103/PhysRevA.97.062337} {\bibfield  {journal} {\bibinfo  {journal} {Physical Review A}\ }\textbf {\bibinfo {volume} {97}},\ \bibinfo {pages} {062337} (\bibinfo {year} {2018})}\BibitemShut {NoStop}%
\bibitem [{\citenamefont {Albarelli}\ \emph {et~al.}(2018)\citenamefont {Albarelli}, \citenamefont {Genoni}, \citenamefont {Paris},\ and\ \citenamefont {Ferraro}}]{Albarelli-ResourceTheory-2018}%
  \BibitemOpen
  \bibfield  {author} {\bibinfo {author} {\bibfnamefont {F.}~\bibnamefont {Albarelli}}, \bibinfo {author} {\bibfnamefont {M.~G.}\ \bibnamefont {Genoni}}, \bibinfo {author} {\bibfnamefont {M.~G.~A.}\ \bibnamefont {Paris}},\ and\ \bibinfo {author} {\bibfnamefont {A.}~\bibnamefont {Ferraro}},\ }\bibfield  {title} {\bibinfo {title} {Resource theory of quantum non-{{Gaussianity}} and {{Wigner}} negativity},\ }\href {https://doi.org/10.1103/PhysRevA.98.052350} {\bibfield  {journal} {\bibinfo  {journal} {Physical Review A}\ }\textbf {\bibinfo {volume} {98}},\ \bibinfo {pages} {052350} (\bibinfo {year} {2018})}\BibitemShut {NoStop}%
\bibitem [{\citenamefont {Sudarshan}(1963)}]{Sudarshan-EquivalenceSemiclassical-1963}%
  \BibitemOpen
  \bibfield  {author} {\bibinfo {author} {\bibfnamefont {E.~C.~G.}\ \bibnamefont {Sudarshan}},\ }\bibfield  {title} {\bibinfo {title} {Equivalence of {{Semiclassical}} and {{Quantum Mechanical Descriptions}} of {{Statistical Light Beams}}},\ }\href {https://doi.org/10.1103/PhysRevLett.10.277} {\bibfield  {journal} {\bibinfo  {journal} {Physical Review Letters}\ }\textbf {\bibinfo {volume} {10}},\ \bibinfo {pages} {277} (\bibinfo {year} {1963})}\BibitemShut {NoStop}%
\bibitem [{\citenamefont {Hillery}(1985)}]{Hillery-ClassicalPure-1985}%
  \BibitemOpen
  \bibfield  {author} {\bibinfo {author} {\bibfnamefont {M.}~\bibnamefont {Hillery}},\ }\bibfield  {title} {\bibinfo {title} {Classical pure states are coherent states},\ }\href {https://doi.org/10.1016/0375-9601(85)90483-9} {\bibfield  {journal} {\bibinfo  {journal} {Physics Letters A}\ }\textbf {\bibinfo {volume} {111}},\ \bibinfo {pages} {409} (\bibinfo {year} {1985})}\BibitemShut {NoStop}%
\bibitem [{\citenamefont {Yadin}\ \emph {et~al.}(2018)\citenamefont {Yadin}, \citenamefont {Binder}, \citenamefont {Thompson}, \citenamefont {Narasimhachar}, \citenamefont {Gu},\ and\ \citenamefont {Kim}}]{Yadin-OperationalResource-2018}%
  \BibitemOpen
  \bibfield  {author} {\bibinfo {author} {\bibfnamefont {B.}~\bibnamefont {Yadin}}, \bibinfo {author} {\bibfnamefont {F.~C.}\ \bibnamefont {Binder}}, \bibinfo {author} {\bibfnamefont {J.}~\bibnamefont {Thompson}}, \bibinfo {author} {\bibfnamefont {V.}~\bibnamefont {Narasimhachar}}, \bibinfo {author} {\bibfnamefont {M.}~\bibnamefont {Gu}},\ and\ \bibinfo {author} {\bibfnamefont {M.~S.}\ \bibnamefont {Kim}},\ }\bibfield  {title} {\bibinfo {title} {Operational {{Resource Theory}} of {{Continuous-Variable Nonclassicality}}},\ }\href {https://doi.org/10.1103/PhysRevX.8.041038} {\bibfield  {journal} {\bibinfo  {journal} {Physical Review X}\ }\textbf {\bibinfo {volume} {8}},\ \bibinfo {pages} {041038} (\bibinfo {year} {2018})}\BibitemShut {NoStop}%
\bibitem [{\citenamefont {Tan}\ and\ \citenamefont {Jeong}(2019{\natexlab{a}})}]{Tan-NonclassicalLight-2019}%
  \BibitemOpen
  \bibfield  {author} {\bibinfo {author} {\bibfnamefont {K.~C.}\ \bibnamefont {Tan}}\ and\ \bibinfo {author} {\bibfnamefont {H.}~\bibnamefont {Jeong}},\ }\bibfield  {title} {\bibinfo {title} {Nonclassical light and metrological power: {{An}} introductory review},\ }\href {https://doi.org/10.1116/1.5126696} {\bibfield  {journal} {\bibinfo  {journal} {AVS Quantum Science}\ }\textbf {\bibinfo {volume} {1}},\ \bibinfo {pages} {014701} (\bibinfo {year} {2019}{\natexlab{a}})}\BibitemShut {NoStop}%
\bibitem [{\citenamefont {Tan}\ and\ \citenamefont {Jeong}(2019{\natexlab{b}})}]{Tan-ResourceTheories-2019}%
  \BibitemOpen
  \bibfield  {author} {\bibinfo {author} {\bibfnamefont {K.~C.}\ \bibnamefont {Tan}}\ and\ \bibinfo {author} {\bibfnamefont {H.}~\bibnamefont {Jeong}},\ }\bibfield  {title} {\bibinfo {title} {Resource {{Theories}} of {{Nonclassical Light}}},\ }\href {https://doi.org/10.3390/quantum1020014} {\bibfield  {journal} {\bibinfo  {journal} {Quantum Reports}\ }\textbf {\bibinfo {volume} {1}},\ \bibinfo {pages} {151} (\bibinfo {year} {2019}{\natexlab{b}})}\BibitemShut {NoStop}%
\bibitem [{\citenamefont {Hudson}(1974)}]{Hudson-WhenWigner-1974}%
  \BibitemOpen
  \bibfield  {author} {\bibinfo {author} {\bibfnamefont {R.~L.}\ \bibnamefont {Hudson}},\ }\bibfield  {title} {\bibinfo {title} {When is the {Wigner} quasi-probability density non-negative?},\ }\href {https://doi.org/10.1016/0034-4877(74)90007-X} {\bibfield  {journal} {\bibinfo  {journal} {Reports on Mathematical Physics}\ }\textbf {\bibinfo {volume} {6}},\ \bibinfo {pages} {249} (\bibinfo {year} {1974})}\BibitemShut {NoStop}%
\bibitem [{\citenamefont {Soto}\ and\ \citenamefont {Claverie}(1983)}]{Soto-WhenWigner-1983}%
  \BibitemOpen
  \bibfield  {author} {\bibinfo {author} {\bibfnamefont {F.}~\bibnamefont {Soto}}\ and\ \bibinfo {author} {\bibfnamefont {P.}~\bibnamefont {Claverie}},\ }\bibfield  {title} {\bibinfo {title} {When is the {{Wigner}} function of multidimensional systems nonnegative?},\ }\href {https://doi.org/10.1063/1.525607} {\bibfield  {journal} {\bibinfo  {journal} {Journal of Mathematical Physics}\ }\textbf {\bibinfo {volume} {24}},\ \bibinfo {pages} {97} (\bibinfo {year} {1983})}\BibitemShut {NoStop}%
\bibitem [{Note2()}]{Note2}%
  \BibitemOpen
  \bibinfo {note} {This appears as a reflection of the fact that the identity Hamiltonian, which appears for these three groups' Lie algebra bases of \protect \cref {tab:Lie-algebra-bases}, never contributes to the orbit dimension in the density operator picture.}\BibitemShut {Stop}%
\bibitem [{\citenamefont {Girardi}\ \emph {et~al.}(2025)\citenamefont {Girardi}, \citenamefont {Witteveen}, \citenamefont {Mele}, \citenamefont {Bittel}, \citenamefont {Oliviero}, \citenamefont {Gross},\ and\ \citenamefont {Walter}}]{Girardi-ItGaussian-2025}%
  \BibitemOpen
  \bibfield  {author} {\bibinfo {author} {\bibfnamefont {F.}~\bibnamefont {Girardi}}, \bibinfo {author} {\bibfnamefont {F.}~\bibnamefont {Witteveen}}, \bibinfo {author} {\bibfnamefont {F.~A.}\ \bibnamefont {Mele}}, \bibinfo {author} {\bibfnamefont {L.}~\bibnamefont {Bittel}}, \bibinfo {author} {\bibfnamefont {S.~F.~E.}\ \bibnamefont {Oliviero}}, \bibinfo {author} {\bibfnamefont {D.}~\bibnamefont {Gross}},\ and\ \bibinfo {author} {\bibfnamefont {M.}~\bibnamefont {Walter}},\ }\href {https://doi.org/10.48550/arXiv.2510.07305} {\bibinfo {title} {Is it {{Gaussian}}? {{Testing}} bosonic quantum states}} (\bibinfo {year} {2025}),\ \Eprint {https://arxiv.org/abs/2510.07305} {arXiv:2510.07305 [quant-ph]} \BibitemShut {NoStop}%
\bibitem [{Note3()}]{Note3}%
  \BibitemOpen
  \bibinfo {note} {If a state supported on an $n\geq 1$ $\protect \mathcal {H}_m^n$ were in $\protect \operatorname {conv}(\protect \mathcal {G})$ or $\protect \mathcal {P}_+$, its associated function $p_\lambda $ or $P(\protect \bm {\alpha })$ in \protect \cref {eq:convex-mixture-Gaussian-states,eq:convex-mixture-classical-states} would be a probability distribution, which would have to be supported in $\protect \mathcal {H}_m^n$ as well [cf. \protect \cref {eq:counterexample-convex-roof-zero-orthogonal-support}], but that is impossible since the only pure Gaussian state in a $\protect \mathcal {H}_m^n$ is the vacuum, with $n=0$.}\BibitemShut {Stop}%
\bibitem [{\citenamefont {Chabaud}\ and\ \citenamefont {Mehraban}(2022)}]{Chabaud-HolomorphicRepresentation-2022}%
  \BibitemOpen
  \bibfield  {author} {\bibinfo {author} {\bibfnamefont {U.}~\bibnamefont {Chabaud}}\ and\ \bibinfo {author} {\bibfnamefont {S.}~\bibnamefont {Mehraban}},\ }\bibfield  {title} {\bibinfo {title} {Holomorphic representation of quantum computations},\ }\href {https://doi.org/10.22331/q-2022-10-06-831} {\bibfield  {journal} {\bibinfo  {journal} {Quantum}\ }\textbf {\bibinfo {volume} {6}},\ \bibinfo {pages} {831} (\bibinfo {year} {2022})}\BibitemShut {NoStop}%
\bibitem [{\citenamefont {Chabaud}\ \emph {et~al.}(2021{\natexlab{b}})\citenamefont {Chabaud}, \citenamefont {Ferrini}, \citenamefont {Grosshans},\ and\ \citenamefont {Markham}}]{Chabaud-ClassicalSimulation-2021}%
  \BibitemOpen
  \bibfield  {author} {\bibinfo {author} {\bibfnamefont {U.}~\bibnamefont {Chabaud}}, \bibinfo {author} {\bibfnamefont {G.}~\bibnamefont {Ferrini}}, \bibinfo {author} {\bibfnamefont {F.}~\bibnamefont {Grosshans}},\ and\ \bibinfo {author} {\bibfnamefont {D.}~\bibnamefont {Markham}},\ }\bibfield  {title} {\bibinfo {title} {Classical simulation of {{Gaussian}} quantum circuits with non-{{Gaussian}} input states},\ }\href {https://doi.org/10.1103/PhysRevResearch.3.033018} {\bibfield  {journal} {\bibinfo  {journal} {Physical Review Research}\ }\textbf {\bibinfo {volume} {3}},\ \bibinfo {pages} {033018} (\bibinfo {year} {2021}{\natexlab{b}})}\BibitemShut {NoStop}%
\bibitem [{\citenamefont {Nicacio}(2021)}]{Nicacio-WilliamsonTheorem-2021}%
  \BibitemOpen
  \bibfield  {author} {\bibinfo {author} {\bibfnamefont {F.}~\bibnamefont {Nicacio}},\ }\bibfield  {title} {\bibinfo {title} {Williamson theorem in classical, quantum, and statistical physics},\ }\href {https://doi.org/10.1119/10.0005944} {\bibfield  {journal} {\bibinfo  {journal} {American Journal of Physics}\ }\textbf {\bibinfo {volume} {89}},\ \bibinfo {pages} {1139} (\bibinfo {year} {2021})}\BibitemShut {NoStop}%
\bibitem [{\citenamefont {Funcke}\ \emph {et~al.}(2021)\citenamefont {Funcke}, \citenamefont {Hartung}, \citenamefont {Jansen}, \citenamefont {K{\"u}hn},\ and\ \citenamefont {Stornati}}]{Funcke-DimensionalExpressivity-2021a}%
  \BibitemOpen
  \bibfield  {author} {\bibinfo {author} {\bibfnamefont {L.}~\bibnamefont {Funcke}}, \bibinfo {author} {\bibfnamefont {T.}~\bibnamefont {Hartung}}, \bibinfo {author} {\bibfnamefont {K.}~\bibnamefont {Jansen}}, \bibinfo {author} {\bibfnamefont {S.}~\bibnamefont {K{\"u}hn}},\ and\ \bibinfo {author} {\bibfnamefont {P.}~\bibnamefont {Stornati}},\ }\bibfield  {title} {\bibinfo {title} {Dimensional {{Expressivity Analysis}} of {{Parametric Quantum Circuits}}},\ }\href {https://doi.org/10.22331/q-2021-03-29-422} {\bibfield  {journal} {\bibinfo  {journal} {Quantum}\ }\textbf {\bibinfo {volume} {5}},\ \bibinfo {pages} {422} (\bibinfo {year} {2021})},\ \Eprint {https://arxiv.org/abs/2011.03532} {arXiv:2011.03532 [hep-lat, physics:math-ph, physics:quant-ph]} \BibitemShut {NoStop}%
\bibitem [{\citenamefont {Haug}\ \emph {et~al.}(2021)\citenamefont {Haug}, \citenamefont {Bharti},\ and\ \citenamefont {Kim}}]{Haug-CapacityQuantum-2021}%
  \BibitemOpen
  \bibfield  {author} {\bibinfo {author} {\bibfnamefont {T.}~\bibnamefont {Haug}}, \bibinfo {author} {\bibfnamefont {K.}~\bibnamefont {Bharti}},\ and\ \bibinfo {author} {\bibfnamefont {M.~S.}\ \bibnamefont {Kim}},\ }\bibfield  {title} {\bibinfo {title} {Capacity and {{Quantum Geometry}} of {{Parametrized Quantum Circuits}}},\ }\href {https://doi.org/10.1103/PRXQuantum.2.040309} {\bibfield  {journal} {\bibinfo  {journal} {PRX Quantum}\ }\textbf {\bibinfo {volume} {2}},\ \bibinfo {pages} {040309} (\bibinfo {year} {2021})}\BibitemShut {NoStop}%
\bibitem [{\citenamefont {Hong}\ \emph {et~al.}(1987)\citenamefont {Hong}, \citenamefont {Ou},\ and\ \citenamefont {Mandel}}]{Hong-MeasurementSubpicosecond-1987}%
  \BibitemOpen
  \bibfield  {author} {\bibinfo {author} {\bibfnamefont {C.~K.}\ \bibnamefont {Hong}}, \bibinfo {author} {\bibfnamefont {Z.~Y.}\ \bibnamefont {Ou}},\ and\ \bibinfo {author} {\bibfnamefont {L.}~\bibnamefont {Mandel}},\ }\bibfield  {title} {\bibinfo {title} {Measurement of subpicosecond time intervals between two photons by interference},\ }\href {https://doi.org/10.1103/PhysRevLett.59.2044} {\bibfield  {journal} {\bibinfo  {journal} {Physical Review Letters}\ }\textbf {\bibinfo {volume} {59}},\ \bibinfo {pages} {2044} (\bibinfo {year} {1987})}\BibitemShut {NoStop}%
\bibitem [{\citenamefont {{Moyano-Fern{\'a}ndez}}\ and\ \citenamefont {{Garcia-Escartin}}(2017)}]{Moyano-Fernandez-LinearOptics-2017}%
  \BibitemOpen
  \bibfield  {author} {\bibinfo {author} {\bibfnamefont {J.~J.}\ \bibnamefont {{Moyano-Fern{\'a}ndez}}}\ and\ \bibinfo {author} {\bibfnamefont {J.~C.}\ \bibnamefont {{Garcia-Escartin}}},\ }\bibfield  {title} {\bibinfo {title} {Linear optics only allows every possible quantum operation for one photon or one port},\ }\href {https://doi.org/10.1016/j.optcom.2016.07.085} {\bibfield  {journal} {\bibinfo  {journal} {Optics Communications}\ }\textbf {\bibinfo {volume} {382}},\ \bibinfo {pages} {237} (\bibinfo {year} {2017})}\BibitemShut {NoStop}%
\bibitem [{\citenamefont {Tichy}(2014)}]{Tichy-InterferenceIdentical-2014}%
  \BibitemOpen
  \bibfield  {author} {\bibinfo {author} {\bibfnamefont {M.~C.}\ \bibnamefont {Tichy}},\ }\bibfield  {title} {\bibinfo {title} {Interference of identical particles from entanglement to boson-sampling},\ }\href {https://doi.org/10.1088/0953-4075/47/10/103001} {\bibfield  {journal} {\bibinfo  {journal} {Journal of Physics B: Atomic, Molecular and Optical Physics}\ }\textbf {\bibinfo {volume} {47}},\ \bibinfo {pages} {103001} (\bibinfo {year} {2014})}\BibitemShut {NoStop}%
\bibitem [{\citenamefont {De~Gliniasty}\ \emph {et~al.}(2024)\citenamefont {De~Gliniasty}, \citenamefont {Bagourd}, \citenamefont {Draux},\ and\ \citenamefont {Bourdoncle}}]{DeGliniasty-SimpleRules-2024}%
  \BibitemOpen
  \bibfield  {author} {\bibinfo {author} {\bibfnamefont {G.}~\bibnamefont {De~Gliniasty}}, \bibinfo {author} {\bibfnamefont {P.}~\bibnamefont {Bagourd}}, \bibinfo {author} {\bibfnamefont {S.}~\bibnamefont {Draux}},\ and\ \bibinfo {author} {\bibfnamefont {B.}~\bibnamefont {Bourdoncle}},\ }\bibfield  {title} {\bibinfo {title} {Simple rules for two-photon state preparation with linear optics},\ }in\ \href {https://doi.org/10.1109/QCE60285.2024.00088} {\emph {\bibinfo {booktitle} {Proceedings of the 2024 {{IEEE}} International Conference on Quantum Computing and Engineering ({{QCE}})}}}\ (\bibinfo  {publisher} {IEEE Computer Society},\ \bibinfo {address} {Los Alamitos},\ \bibinfo {year} {2024})\ pp.\ \bibinfo {pages} {706--711}\BibitemShut {NoStop}%
\bibitem [{\citenamefont {Larocca}\ \emph {et~al.}(2023)\citenamefont {Larocca}, \citenamefont {Ju}, \citenamefont {{Garc{\'i}a-Mart{\'i}n}}, \citenamefont {Coles},\ and\ \citenamefont {Cerezo}}]{Larocca-TheoryOverparametrization-2023}%
  \BibitemOpen
  \bibfield  {author} {\bibinfo {author} {\bibfnamefont {M.}~\bibnamefont {Larocca}}, \bibinfo {author} {\bibfnamefont {N.}~\bibnamefont {Ju}}, \bibinfo {author} {\bibfnamefont {D.}~\bibnamefont {{Garc{\'i}a-Mart{\'i}n}}}, \bibinfo {author} {\bibfnamefont {P.~J.}\ \bibnamefont {Coles}},\ and\ \bibinfo {author} {\bibfnamefont {M.}~\bibnamefont {Cerezo}},\ }\bibfield  {title} {\bibinfo {title} {Theory of overparametrization in quantum neural networks},\ }\href {https://doi.org/10.1038/s43588-023-00467-6} {\bibfield  {journal} {\bibinfo  {journal} {Nature Computational Science}\ }\textbf {\bibinfo {volume} {3}},\ \bibinfo {pages} {542} (\bibinfo {year} {2023})}\BibitemShut {NoStop}%
\bibitem [{\citenamefont {Anschuetz}\ and\ \citenamefont {Kiani}(2022)}]{Anschuetz-QuantumVariational-2022}%
  \BibitemOpen
  \bibfield  {author} {\bibinfo {author} {\bibfnamefont {E.~R.}\ \bibnamefont {Anschuetz}}\ and\ \bibinfo {author} {\bibfnamefont {B.~T.}\ \bibnamefont {Kiani}},\ }\bibfield  {title} {\bibinfo {title} {Quantum variational algorithms are swamped with traps},\ }\href {https://doi.org/10.1038/s41467-022-35364-5} {\bibfield  {journal} {\bibinfo  {journal} {Nature Communications}\ }\textbf {\bibinfo {volume} {13}},\ \bibinfo {pages} {7760} (\bibinfo {year} {2022})}\BibitemShut {NoStop}%
\bibitem [{\citenamefont {Liu}\ \emph {et~al.}(2020)\citenamefont {Liu}, \citenamefont {Yuan}, \citenamefont {Lu},\ and\ \citenamefont {Wang}}]{Liu-QuantumFisher-2019}%
  \BibitemOpen
  \bibfield  {author} {\bibinfo {author} {\bibfnamefont {J.}~\bibnamefont {Liu}}, \bibinfo {author} {\bibfnamefont {H.}~\bibnamefont {Yuan}}, \bibinfo {author} {\bibfnamefont {X.-M.}\ \bibnamefont {Lu}},\ and\ \bibinfo {author} {\bibfnamefont {X.}~\bibnamefont {Wang}},\ }\bibfield  {title} {\bibinfo {title} {Quantum {{Fisher}} information matrix and multiparameter estimation},\ }\href {https://doi.org/10.1088/1751-8121/ab5d4d} {\bibfield  {journal} {\bibinfo  {journal} {Journal of Physics A: Mathematical and Theoretical}\ }\textbf {\bibinfo {volume} {53}},\ \bibinfo {pages} {023001} (\bibinfo {year} {2020})}\BibitemShut {NoStop}%
\bibitem [{\citenamefont {Mihailescu}\ \emph {et~al.}(2026)\citenamefont {Mihailescu}, \citenamefont {Sarkar}, \citenamefont {Bayat}, \citenamefont {Campbell},\ and\ \citenamefont {Mitchell}}]{Mihailescu-MetrologicalSymmetries-2025}%
  \BibitemOpen
  \bibfield  {author} {\bibinfo {author} {\bibfnamefont {G.}~\bibnamefont {Mihailescu}}, \bibinfo {author} {\bibfnamefont {S.}~\bibnamefont {Sarkar}}, \bibinfo {author} {\bibfnamefont {A.}~\bibnamefont {Bayat}}, \bibinfo {author} {\bibfnamefont {S.}~\bibnamefont {Campbell}},\ and\ \bibinfo {author} {\bibfnamefont {A.~K.}\ \bibnamefont {Mitchell}},\ }\bibfield  {title} {\bibinfo {title} {Metrological symmetries in singular quantum multi-parameter estimation},\ }\href {https://doi.org/10.1088/2058-9565/ae1757} {\bibfield  {journal} {\bibinfo  {journal} {Quantum Science and Technology}\ }\textbf {\bibinfo {volume} {11}},\ \bibinfo {pages} {015006} (\bibinfo {year} {2026})}\BibitemShut {NoStop}%
\bibitem [{\citenamefont {Shettell}\ \emph {et~al.}(2022)\citenamefont {Shettell}, \citenamefont {Hassani},\ and\ \citenamefont {Markham}}]{Shettell-PrivateNetwork-2022}%
  \BibitemOpen
  \bibfield  {author} {\bibinfo {author} {\bibfnamefont {N.}~\bibnamefont {Shettell}}, \bibinfo {author} {\bibfnamefont {M.}~\bibnamefont {Hassani}},\ and\ \bibinfo {author} {\bibfnamefont {D.}~\bibnamefont {Markham}},\ }\href {https://doi.org/10.48550/arXiv.2207.14450} {\bibinfo {title} {Private network parameter estimation with quantum sensors}} (\bibinfo {year} {2022}),\ \Eprint {https://arxiv.org/abs/2207.14450} {arXiv:2207.14450 [quant-ph]} \BibitemShut {NoStop}%
\bibitem [{\citenamefont {Bugalho}\ \emph {et~al.}(2025)\citenamefont {Bugalho}, \citenamefont {Hassani}, \citenamefont {Omar},\ and\ \citenamefont {Markham}}]{Bugalho-PrivateRobust-2025}%
  \BibitemOpen
  \bibfield  {author} {\bibinfo {author} {\bibfnamefont {L.}~\bibnamefont {Bugalho}}, \bibinfo {author} {\bibfnamefont {M.}~\bibnamefont {Hassani}}, \bibinfo {author} {\bibfnamefont {Y.}~\bibnamefont {Omar}},\ and\ \bibinfo {author} {\bibfnamefont {D.}~\bibnamefont {Markham}},\ }\bibfield  {title} {\bibinfo {title} {Private and {{Robust States}} for {{Distributed Quantum Sensing}}},\ }\href {https://doi.org/10.22331/q-2025-01-15-1596} {\bibfield  {journal} {\bibinfo  {journal} {Quantum}\ }\textbf {\bibinfo {volume} {9}},\ \bibinfo {pages} {1596} (\bibinfo {year} {2025})}\BibitemShut {NoStop}%
\bibitem [{\citenamefont {Hassani}\ \emph {et~al.}(2025)\citenamefont {Hassani}, \citenamefont {Scheiner}, \citenamefont {Paris},\ and\ \citenamefont {Markham}}]{Hassani-PrivacyNetworks-2025}%
  \BibitemOpen
  \bibfield  {author} {\bibinfo {author} {\bibfnamefont {M.}~\bibnamefont {Hassani}}, \bibinfo {author} {\bibfnamefont {S.}~\bibnamefont {Scheiner}}, \bibinfo {author} {\bibfnamefont {M.~G.~A.}\ \bibnamefont {Paris}},\ and\ \bibinfo {author} {\bibfnamefont {D.}~\bibnamefont {Markham}},\ }\bibfield  {title} {\bibinfo {title} {Privacy in {{Networks}} of {{Quantum Sensors}}},\ }\href {https://doi.org/10.1103/PhysRevLett.134.030802} {\bibfield  {journal} {\bibinfo  {journal} {Physical Review Letters}\ }\textbf {\bibinfo {volume} {134}},\ \bibinfo {pages} {030802} (\bibinfo {year} {2025})}\BibitemShut {NoStop}%
\bibitem [{\citenamefont {{Garc{\'i}a-Mart{\'i}n}}\ \emph {et~al.}(2024)\citenamefont {{Garc{\'i}a-Mart{\'i}n}}, \citenamefont {Larocca},\ and\ \citenamefont {Cerezo}}]{Garcia-Martin-EffectsNoise-2024}%
  \BibitemOpen
  \bibfield  {author} {\bibinfo {author} {\bibfnamefont {D.}~\bibnamefont {{Garc{\'i}a-Mart{\'i}n}}}, \bibinfo {author} {\bibfnamefont {M.}~\bibnamefont {Larocca}},\ and\ \bibinfo {author} {\bibfnamefont {M.}~\bibnamefont {Cerezo}},\ }\bibfield  {title} {\bibinfo {title} {Effects of noise on the overparametrization of quantum neural networks},\ }\href {https://doi.org/10.1103/PhysRevResearch.6.013295} {\bibfield  {journal} {\bibinfo  {journal} {Physical Review Research}\ }\textbf {\bibinfo {volume} {6}},\ \bibinfo {pages} {013295} (\bibinfo {year} {2024})}\BibitemShut {NoStop}%
\bibitem [{\citenamefont {Wu}\ \emph {et~al.}(2024)\citenamefont {Wu}, \citenamefont {Yao}, \citenamefont {Zhang},\ and\ \citenamefont {Li}}]{Wu-RandomnessEnhancedExpressivity-2024}%
  \BibitemOpen
  \bibfield  {author} {\bibinfo {author} {\bibfnamefont {Y.}~\bibnamefont {Wu}}, \bibinfo {author} {\bibfnamefont {J.}~\bibnamefont {Yao}}, \bibinfo {author} {\bibfnamefont {P.}~\bibnamefont {Zhang}},\ and\ \bibinfo {author} {\bibfnamefont {X.}~\bibnamefont {Li}},\ }\bibfield  {title} {\bibinfo {title} {Randomness-{{Enhanced Expressivity}} of {{Quantum Neural Networks}}},\ }\href {https://doi.org/10.1103/PhysRevLett.132.010602} {\bibfield  {journal} {\bibinfo  {journal} {Physical Review Letters}\ }\textbf {\bibinfo {volume} {132}},\ \bibinfo {pages} {010602} (\bibinfo {year} {2024})}\BibitemShut {NoStop}%
\bibitem [{\citenamefont {Folland}(2016)}]{Folland-CourseAbstract-2016}%
  \BibitemOpen
  \bibfield  {author} {\bibinfo {author} {\bibfnamefont {G.~B.}\ \bibnamefont {Folland}},\ }\href {https://doi.org/10.1201/b19172} {\emph {\bibinfo {title} {A {{Course}} in {{Abstract Harmonic Analysis}}}}},\ \bibinfo {edition} {2nd}\ ed.\ (\bibinfo  {publisher} {{Chapman and Hall/CRC}},\ \bibinfo {address} {New York},\ \bibinfo {year} {2016})\BibitemShut {NoStop}%
\bibitem [{\citenamefont {Lee}(2012)}]{lee_introduction_2012}%
  \BibitemOpen
  \bibfield  {author} {\bibinfo {author} {\bibfnamefont {J.~M.}\ \bibnamefont {Lee}},\ }\href {https://doi.org/10.1007/978-1-4419-9982-5} {{\selectlanguage {en}\emph {\bibinfo {title} {Introduction to {Smooth} {Manifolds}}}}},\ \bibinfo {series} {Graduate {Texts} in {Mathematics}}, Vol.\ \bibinfo {volume} {218}\ (\bibinfo  {publisher} {Springer},\ \bibinfo {address} {New York},\ \bibinfo {year} {2012})\BibitemShut {NoStop}%
\bibitem [{\citenamefont {{Arvind}}\ \emph {et~al.}(1995)\citenamefont {{Arvind}}, \citenamefont {Dutta}, \citenamefont {Mukunda},\ and\ \citenamefont {Simon}}]{arvind_real_1995}%
  \BibitemOpen
  \bibfield  {author} {\bibinfo {author} {\bibnamefont {{Arvind}}}, \bibinfo {author} {\bibfnamefont {B.}~\bibnamefont {Dutta}}, \bibinfo {author} {\bibfnamefont {N.}~\bibnamefont {Mukunda}},\ and\ \bibinfo {author} {\bibfnamefont {R.}~\bibnamefont {Simon}},\ }\bibfield  {title} {{\selectlanguage {en}\bibinfo {title} {The real symplectic groups in quantum mechanics and optics}},\ }\href {https://doi.org/10.1007/BF02848172} {\bibfield  {journal} {\bibinfo  {journal} {Pramana}\ }\textbf {\bibinfo {volume} {45}},\ \bibinfo {pages} {471} (\bibinfo {year} {1995})}\BibitemShut {NoStop}%
\bibitem [{\citenamefont {Folland}(1989)}]{Folland-HarmonicAnalysis-1989}%
  \BibitemOpen
  \bibfield  {author} {\bibinfo {author} {\bibfnamefont {G.~B.}\ \bibnamefont {Folland}},\ }\href {https://doi.org/10.1515/9781400882427} {\emph {\bibinfo {title} {Harmonic {{Analysis}} in {{Phase Space}}}}},\ \bibinfo {series} {Annals of {{Mathematics Studies}}}, Vol.\ \bibinfo {volume} {122}\ (\bibinfo  {publisher} {Princeton University Press},\ \bibinfo {address} {Princeton},\ \bibinfo {year} {1989})\BibitemShut {NoStop}%
\bibitem [{\citenamefont {Moretti}(2017)}]{Moretti-SpectralTheory-2017}%
  \BibitemOpen
  \bibfield  {author} {\bibinfo {author} {\bibfnamefont {V.}~\bibnamefont {Moretti}},\ }\href {https://doi.org/10.1007/978-3-319-70706-8} {\emph {\bibinfo {title} {Spectral {{Theory}} and {{Quantum Mechanics}}}}},\ \bibinfo {series} {{UNITEXT}}, Vol.\ \bibinfo {volume} {110}\ (\bibinfo  {publisher} {Springer International Publishing},\ \bibinfo {address} {Cham},\ \bibinfo {year} {2017})\BibitemShut {NoStop}%
\bibitem [{\citenamefont {Varadarajan}(1984)}]{Varadarajan-LieGroups-1984}%
  \BibitemOpen
  \bibfield  {author} {\bibinfo {author} {\bibfnamefont {V.~S.}\ \bibnamefont {Varadarajan}},\ }\href {https://doi.org/10.1007/978-1-4612-1126-6} {\emph {\bibinfo {title} {Lie {{Groups}}, {{Lie Algebras}}, and {{Their Representations}}}}},\ \bibinfo {series} {Graduate {{Texts}} in {{Mathematics}}}, Vol.\ \bibinfo {volume} {102}\ (\bibinfo  {publisher} {Springer},\ \bibinfo {address} {New York},\ \bibinfo {year} {1984})\BibitemShut {NoStop}%
\bibitem [{\citenamefont {Goodman}(1969)}]{Goodman-AnalyticEntire-1969}%
  \BibitemOpen
  \bibfield  {author} {\bibinfo {author} {\bibfnamefont {R.}~\bibnamefont {Goodman}},\ }\bibfield  {title} {\bibinfo {title} {Analytic and {{Entire Vectors}} for {{Representations}} of {{Lie Groups}}},\ }\href {https://doi.org/10.1090/S0002-9947-1969-0248285-6} {\bibfield  {journal} {\bibinfo  {journal} {Transactions of the American Mathematical Society}\ }\textbf {\bibinfo {volume} {143}},\ \bibinfo {pages} {55} (\bibinfo {year} {1969})}\BibitemShut {NoStop}%
\bibitem [{\citenamefont {Tu}(2011)}]{Tu-IntroductionManifolds-2011}%
  \BibitemOpen
  \bibfield  {author} {\bibinfo {author} {\bibfnamefont {L.~W.}\ \bibnamefont {Tu}},\ }\href {https://doi.org/10.1007/978-1-4419-7400-6} {\emph {\bibinfo {title} {An {{Introduction}} to {{Manifolds}}}}},\ Universitext\ (\bibinfo  {publisher} {Springer},\ \bibinfo {address} {New York},\ \bibinfo {year} {2011})\BibitemShut {NoStop}%
\bibitem [{\citenamefont {Aniello}\ \emph {et~al.}(2006)\citenamefont {Aniello}, \citenamefont {Lupo},\ and\ \citenamefont {Napolitano}}]{Aniello-ExploringRepresentation-2006}%
  \BibitemOpen
  \bibfield  {author} {\bibinfo {author} {\bibfnamefont {P.}~\bibnamefont {Aniello}}, \bibinfo {author} {\bibfnamefont {C.}~\bibnamefont {Lupo}},\ and\ \bibinfo {author} {\bibfnamefont {M.}~\bibnamefont {Napolitano}},\ }\bibfield  {title} {\bibinfo {title} {Exploring {{Representation Theory}} of {{Unitary Groups}} via {{Linear Optical Passive Devices}}},\ }\href {https://doi.org/10.1007/s11080-006-9023-1} {\bibfield  {journal} {\bibinfo  {journal} {Open Systems \& Information Dynamics}\ }\textbf {\bibinfo {volume} {13}},\ \bibinfo {pages} {415} (\bibinfo {year} {2006})}\BibitemShut {NoStop}%
\bibitem [{\citenamefont {Harrow}(2013)}]{Harrow-ChurchSymmetric-2013}%
  \BibitemOpen
  \bibfield  {author} {\bibinfo {author} {\bibfnamefont {A.~W.}\ \bibnamefont {Harrow}},\ }\href {https://doi.org/10.48550/arXiv.1308.6595} {\bibinfo {title} {The {{Church}} of the {{Symmetric Subspace}}}} (\bibinfo {year} {2013}),\ \Eprint {https://arxiv.org/abs/1308.6595} {arXiv:1308.6595} \BibitemShut {NoStop}%
\bibitem [{\citenamefont {Hofmann}\ and\ \citenamefont {Morris}(2007)}]{Hofmann-LieTheory-2007}%
  \BibitemOpen
  \bibfield  {author} {\bibinfo {author} {\bibfnamefont {K.~H.}\ \bibnamefont {Hofmann}}\ and\ \bibinfo {author} {\bibfnamefont {S.~A.}\ \bibnamefont {Morris}},\ }\href@noop {} {\emph {\bibinfo {title} {The {{Lie}} Theory of Connected Pro-{{Lie}} Groups}}},\ \bibinfo {series} {{{EMS Tracts}} in Mathematics}\ No.~\bibinfo {number} {2}\ (\bibinfo  {publisher} {European Mathematical Society},\ \bibinfo {address} {Z\"urich},\ \bibinfo {year} {2007})\BibitemShut {NoStop}%
\bibitem [{\citenamefont {Kadison}\ and\ \citenamefont {Ringrose}(1983)}]{Kadison-FundamentalsTheory-1983}%
  \BibitemOpen
  \bibfield  {author} {\bibinfo {author} {\bibfnamefont {R.~V.}\ \bibnamefont {Kadison}}\ and\ \bibinfo {author} {\bibfnamefont {J.~R.}\ \bibnamefont {Ringrose}},\ }\href@noop {} {\emph {\bibinfo {title} {Fundamentals of the Theory of Operator Algebras}}},\ \bibinfo {series} {Pure and Applied Mathematics}, Vol.~\bibinfo {volume} {I}\ (\bibinfo  {publisher} {Academic Press},\ \bibinfo {address} {New York},\ \bibinfo {year} {1983})\BibitemShut {NoStop}%
\bibitem [{\citenamefont {Abraham}\ \emph {et~al.}(1988)\citenamefont {Abraham}, \citenamefont {Marsden},\ and\ \citenamefont {Ratiu}}]{Abraham-ManifoldsTensor-1988}%
  \BibitemOpen
  \bibfield  {author} {\bibinfo {author} {\bibfnamefont {R.}~\bibnamefont {Abraham}}, \bibinfo {author} {\bibfnamefont {J.~E.}\ \bibnamefont {Marsden}},\ and\ \bibinfo {author} {\bibfnamefont {T.}~\bibnamefont {Ratiu}},\ }\href {https://doi.org/10.1007/978-1-4612-1029-0} {\emph {\bibinfo {title} {Manifolds, {{Tensor Analysis}}, and {{Applications}}}}},\ \bibinfo {edition} {3rd}\ ed.,\ edited by\ \bibinfo {editor} {\bibfnamefont {J.~E.}\ \bibnamefont {Marsden}}, \bibinfo {editor} {\bibfnamefont {L.}~\bibnamefont {Sirovich}},\ and\ \bibinfo {editor} {\bibfnamefont {F.}~\bibnamefont {John}},\ \bibinfo {series} {Applied {{Mathematical Sciences}}}, Vol.~\bibinfo {volume} {75}\ (\bibinfo  {publisher} {Springer},\ \bibinfo {address} {New York},\ \bibinfo {year} {1988})\BibitemShut {NoStop}%
\bibitem [{\citenamefont {Bamber}\ and\ \citenamefont {van~Santen}(1985)}]{Bamber-HowMany-1985}%
  \BibitemOpen
  \bibfield  {author} {\bibinfo {author} {\bibfnamefont {D.}~\bibnamefont {Bamber}}\ and\ \bibinfo {author} {\bibfnamefont {J.~P.~H.}\ \bibnamefont {van~Santen}},\ }\bibfield  {title} {\bibinfo {title} {How many parameters can a model have and still be testable?},\ }\href {https://doi.org/10.1016/0022-2496(85)90005-7} {\bibfield  {journal} {\bibinfo  {journal} {Journal of Mathematical Psychology}\ }\textbf {\bibinfo {volume} {29}},\ \bibinfo {pages} {443} (\bibinfo {year} {1985})}\BibitemShut {NoStop}%
\bibitem [{\citenamefont {Husain}(2018)}]{Husain-IntroductionTopological-2018}%
  \BibitemOpen
  \bibfield  {author} {\bibinfo {author} {\bibfnamefont {T.}~\bibnamefont {Husain}},\ }\href@noop {} {\emph {\bibinfo {title} {Introduction to {{Topological Groups}}}}}\ (\bibinfo  {publisher} {Dover Publications},\ \bibinfo {address} {Mineola, New York},\ \bibinfo {year} {2018})\BibitemShut {NoStop}%
\bibitem [{\citenamefont {Monbroussou}\ \emph {et~al.}(2025{\natexlab{b}})\citenamefont {Monbroussou}, \citenamefont {Mamon}, \citenamefont {Landman}, \citenamefont {Grilo}, \citenamefont {Kukla},\ and\ \citenamefont {Kashefi}}]{Monbroussou-TrainabilityExpressivity-2025}%
  \BibitemOpen
  \bibfield  {author} {\bibinfo {author} {\bibfnamefont {L.}~\bibnamefont {Monbroussou}}, \bibinfo {author} {\bibfnamefont {E.~Z.}\ \bibnamefont {Mamon}}, \bibinfo {author} {\bibfnamefont {J.}~\bibnamefont {Landman}}, \bibinfo {author} {\bibfnamefont {A.~B.}\ \bibnamefont {Grilo}}, \bibinfo {author} {\bibfnamefont {R.}~\bibnamefont {Kukla}},\ and\ \bibinfo {author} {\bibfnamefont {E.}~\bibnamefont {Kashefi}},\ }\bibfield  {title} {\bibinfo {title} {Trainability and {{Expressivity}} of {{Hamming-Weight Preserving Quantum Circuits}} for {{Machine Learning}}},\ }\href {https://doi.org/10.22331/q-2025-05-15-1745} {\bibfield  {journal} {\bibinfo  {journal} {Quantum}\ }\textbf {\bibinfo {volume} {9}},\ \bibinfo {pages} {1745} (\bibinfo {year} {2025}{\natexlab{b}})}\BibitemShut {NoStop}%
\bibitem [{\citenamefont {Hall}(2013)}]{hall_quantum_2013}%
  \BibitemOpen
  \bibfield  {author} {\bibinfo {author} {\bibfnamefont {B.~C.}\ \bibnamefont {Hall}},\ }\href {https://doi.org/10.1007/978-1-4614-7116-5} {{\selectlanguage {en}\emph {\bibinfo {title} {Quantum {Theory} for {Mathematicians}}}}},\ \bibinfo {series} {Graduate {Texts} in {Mathematics}}, Vol.\ \bibinfo {volume} {267}\ (\bibinfo  {publisher} {Springer},\ \bibinfo {address} {New York},\ \bibinfo {year} {2013})\BibitemShut {NoStop}%
\bibitem [{\citenamefont {Keyl}\ \emph {et~al.}(2016)\citenamefont {Keyl}, \citenamefont {Kiukas},\ and\ \citenamefont {Werner}}]{keyl_schwartz_2016}%
  \BibitemOpen
  \bibfield  {author} {\bibinfo {author} {\bibfnamefont {M.}~\bibnamefont {Keyl}}, \bibinfo {author} {\bibfnamefont {J.}~\bibnamefont {Kiukas}},\ and\ \bibinfo {author} {\bibfnamefont {R.~F.}\ \bibnamefont {Werner}},\ }\bibfield  {title} {\bibinfo {title} {Schwartz operators},\ }\href {https://doi.org/10.1142/S0129055X16300016} {\bibfield  {journal} {\bibinfo  {journal} {Reviews in Mathematical Physics}\ }\textbf {\bibinfo {volume} {28}},\ \bibinfo {pages} {1630001} (\bibinfo {year} {2016})}\BibitemShut {NoStop}%
\bibitem [{\citenamefont {Hernández}\ and\ \citenamefont {Riedel}(2022)}]{hernandez_rapidly_2022}%
  \BibitemOpen
  \bibfield  {author} {\bibinfo {author} {\bibfnamefont {F.}~\bibnamefont {Hernández}}\ and\ \bibinfo {author} {\bibfnamefont {C.~J.}\ \bibnamefont {Riedel}},\ }\bibfield  {title} {\bibinfo {title} {Rapidly decaying {Wigner} functions are {Schwartz} functions},\ }\href {https://doi.org/10.1063/5.0049581} {\bibfield  {journal} {\bibinfo  {journal} {Journal of Mathematical Physics}\ }\textbf {\bibinfo {volume} {63}},\ \bibinfo {pages} {022104} (\bibinfo {year} {2022})}\BibitemShut {NoStop}%
\bibitem [{\citenamefont {Schaefer}\ and\ \citenamefont {Wolff}(1999)}]{Schaefer-TopologicalVector-1999}%
  \BibitemOpen
  \bibfield  {author} {\bibinfo {author} {\bibfnamefont {H.~H.}\ \bibnamefont {Schaefer}}\ and\ \bibinfo {author} {\bibfnamefont {M.~P.}\ \bibnamefont {Wolff}},\ }\href {https://doi.org/10.1007/978-1-4612-1468-7} {\emph {\bibinfo {title} {Topological {{Vector Spaces}}}}},\ \bibinfo {series} {Graduate {{Texts}} in {{Mathematics}}}, Vol.~\bibinfo {volume} {3}\ (\bibinfo  {publisher} {Springer},\ \bibinfo {address} {New York},\ \bibinfo {year} {1999})\BibitemShut {NoStop}%
\bibitem [{\citenamefont {Coutinho}(1995)}]{Coutinho-PrimerAlgebraic-1995}%
  \BibitemOpen
  \bibfield  {author} {\bibinfo {author} {\bibfnamefont {S.~C.}\ \bibnamefont {Coutinho}},\ }\href {https://doi.org/10.1017/CBO9780511623653} {\emph {\bibinfo {title} {A {{Primer}} of {{Algebraic D-Modules}}}}},\ London {{Mathematical Society Student Texts}}\ (\bibinfo  {publisher} {Cambridge University Press},\ \bibinfo {address} {Cambridge},\ \bibinfo {year} {1995})\BibitemShut {NoStop}%
\bibitem [{\citenamefont {Lee}(1995)}]{Lee-TheoryApplication-1995}%
  \BibitemOpen
  \bibfield  {author} {\bibinfo {author} {\bibfnamefont {H.-W.}\ \bibnamefont {Lee}},\ }\bibfield  {title} {\bibinfo {title} {Theory and application of the quantum phase-space distribution functions},\ }\href {https://doi.org/10.1016/0370-1573(95)00007-4} {\bibfield  {journal} {\bibinfo  {journal} {Physics Reports}\ }\textbf {\bibinfo {volume} {259}},\ \bibinfo {pages} {147} (\bibinfo {year} {1995})}\BibitemShut {NoStop}%
\bibitem [{\citenamefont {Davies}(1976)}]{Davies-QuantumTheory-1976}%
  \BibitemOpen
  \bibfield  {author} {\bibinfo {author} {\bibfnamefont {E.~B.}\ \bibnamefont {Davies}},\ }\href@noop {} {\emph {\bibinfo {title} {Quantum Theory of Open Systems}}}\ (\bibinfo  {publisher} {Academic Press},\ \bibinfo {address} {London},\ \bibinfo {year} {1976})\BibitemShut {NoStop}%
\bibitem [{\citenamefont {Schm{\"u}dgen}(2012)}]{Schmudgen-UnboundedSelfadjoint-2012}%
  \BibitemOpen
  \bibfield  {author} {\bibinfo {author} {\bibfnamefont {K.}~\bibnamefont {Schm{\"u}dgen}},\ }\href {https://doi.org/10.1007/978-94-007-4753-1} {\emph {\bibinfo {title} {Unbounded {{Self-adjoint Operators}} on {{Hilbert Space}}}}},\ \bibinfo {series} {Graduate {{Texts}} in {{Mathematics}}}, Vol.\ \bibinfo {volume} {265}\ (\bibinfo  {publisher} {Springer Netherlands},\ \bibinfo {address} {Dordrecht},\ \bibinfo {year} {2012})\BibitemShut {NoStop}%
\bibitem [{\citenamefont {Busch}\ \emph {et~al.}(2016)\citenamefont {Busch}, \citenamefont {Lahti}, \citenamefont {Pellonp{\"a}{\"a}},\ and\ \citenamefont {Ylinen}}]{Busch-QuantumMeasurement-2016}%
  \BibitemOpen
  \bibfield  {author} {\bibinfo {author} {\bibfnamefont {P.}~\bibnamefont {Busch}}, \bibinfo {author} {\bibfnamefont {P.}~\bibnamefont {Lahti}}, \bibinfo {author} {\bibfnamefont {J.-P.}\ \bibnamefont {Pellonp{\"a}{\"a}}},\ and\ \bibinfo {author} {\bibfnamefont {K.}~\bibnamefont {Ylinen}},\ }\href {https://doi.org/10.1007/978-3-319-43389-9} {\emph {\bibinfo {title} {Quantum {{Measurement}}}}},\ Theoretical and {{Mathematical Physics}}\ (\bibinfo  {publisher} {Springer International Publishing},\ \bibinfo {address} {Cham},\ \bibinfo {year} {2016})\BibitemShut {NoStop}%
\bibitem [{\citenamefont {Lvovsky}\ and\ \citenamefont {Raymer}(2009)}]{Lvovsky-ContinuousvariableOptical-2009}%
  \BibitemOpen
  \bibfield  {author} {\bibinfo {author} {\bibfnamefont {A.~I.}\ \bibnamefont {Lvovsky}}\ and\ \bibinfo {author} {\bibfnamefont {M.~G.}\ \bibnamefont {Raymer}},\ }\bibfield  {title} {\bibinfo {title} {Continuous-variable optical quantum-state tomography},\ }\href {https://doi.org/10.1103/RevModPhys.81.299} {\bibfield  {journal} {\bibinfo  {journal} {Reviews of Modern Physics}\ }\textbf {\bibinfo {volume} {81}},\ \bibinfo {pages} {299} (\bibinfo {year} {2009})}\BibitemShut {NoStop}%
\bibitem [{\citenamefont {D'Ariano}\ \emph {et~al.}(2003)\citenamefont {D'Ariano}, \citenamefont {Paris},\ and\ \citenamefont {Sacchi}}]{MauroDAriano-QuantumTomography-2003}%
  \BibitemOpen
  \bibfield  {author} {\bibinfo {author} {\bibfnamefont {G.~M.}\ \bibnamefont {D'Ariano}}, \bibinfo {author} {\bibfnamefont {M.~G.~A.}\ \bibnamefont {Paris}},\ and\ \bibinfo {author} {\bibfnamefont {M.~F.}\ \bibnamefont {Sacchi}},\ }\bibfield  {title} {\bibinfo {title} {Quantum {{Tomography}}},\ }in\ \href {https://doi.org/10.1016/S1076-5670(03)80065-4} {\emph {\bibinfo {booktitle} {Advances in {{Imaging}} and {{Electron Physics}}}}},\ Vol.\ \bibinfo {volume} {128},\ \bibinfo {editor} {edited by\ \bibinfo {editor} {\bibfnamefont {P.~W.}\ \bibnamefont {Hawkes}}}\ (\bibinfo  {publisher} {Academic Press},\ \bibinfo {address} {San Diego, California},\ \bibinfo {year} {2003})\ pp.\ \bibinfo {pages} {205--308}\BibitemShut {NoStop}%
\bibitem [{\citenamefont {Mityagin}(2020)}]{Mityagin-ZeroSet-2020}%
  \BibitemOpen
  \bibfield  {author} {\bibinfo {author} {\bibfnamefont {B.~S.}\ \bibnamefont {Mityagin}},\ }\bibfield  {title} {\bibinfo {title} {The {{Zero Set}} of a {{Real Analytic Function}}},\ }\href {https://doi.org/10.1134/S0001434620030189} {\bibfield  {journal} {\bibinfo  {journal} {Mathematical Notes}\ }\textbf {\bibinfo {volume} {107}},\ \bibinfo {pages} {529} (\bibinfo {year} {2020})}\BibitemShut {NoStop}%
\bibitem [{\citenamefont {Amann}\ and\ \citenamefont {Escher}(2009)}]{Amann-AnalysisIII-2009}%
  \BibitemOpen
  \bibfield  {author} {\bibinfo {author} {\bibfnamefont {H.}~\bibnamefont {Amann}}\ and\ \bibinfo {author} {\bibfnamefont {J.}~\bibnamefont {Escher}},\ }\href {https://doi.org/10.1007/978-3-7643-7480-8} {\emph {\bibinfo {title} {Analysis {{III}}}}}\ (\bibinfo  {publisher} {Birkh{\"a}user},\ \bibinfo {address} {Basel},\ \bibinfo {year} {2009})\BibitemShut {NoStop}%
\bibitem [{\citenamefont {Lee}(2018)}]{Lee-IntroductionRiemannian-2018}%
  \BibitemOpen
  \bibfield  {author} {\bibinfo {author} {\bibfnamefont {J.~M.}\ \bibnamefont {Lee}},\ }\href {https://doi.org/10.1007/978-3-319-91755-9} {\emph {\bibinfo {title} {Introduction to {{Riemannian Manifolds}}}}},\ \bibinfo {edition} {2nd}\ ed.,\ \bibinfo {series} {Graduate {{Texts}} in {{Mathematics}}}, Vol.\ \bibinfo {volume} {176}\ (\bibinfo  {publisher} {Springer International Publishing},\ \bibinfo {address} {Cham},\ \bibinfo {year} {2018})\BibitemShut {NoStop}%
\bibitem [{\citenamefont {Kr{\'a}l}(1990)}]{Kral-DisplacedSqueezed-1990}%
  \BibitemOpen
  \bibfield  {author} {\bibinfo {author} {\bibfnamefont {P.}~\bibnamefont {Kr{\'a}l}},\ }\bibfield  {title} {\bibinfo {title} {Displaced and {{Squeezed Fock States}}},\ }\href {https://doi.org/10.1080/09500349014550941} {\bibfield  {journal} {\bibinfo  {journal} {Journal of Modern Optics}\ }\textbf {\bibinfo {volume} {37}},\ \bibinfo {pages} {889} (\bibinfo {year} {1990})}\BibitemShut {NoStop}%
\bibitem [{\citenamefont {Xu}(2016)}]{Xu-QuantifyingCoherence-2016}%
  \BibitemOpen
  \bibfield  {author} {\bibinfo {author} {\bibfnamefont {J.}~\bibnamefont {Xu}},\ }\bibfield  {title} {\bibinfo {title} {Quantifying coherence of {{Gaussian}} states},\ }\href {https://doi.org/10.1103/PhysRevA.93.032111} {\bibfield  {journal} {\bibinfo  {journal} {Physical Review A}\ }\textbf {\bibinfo {volume} {93}},\ \bibinfo {pages} {032111} (\bibinfo {year} {2016})}\BibitemShut {NoStop}%
\bibitem [{\citenamefont {Streltsov}\ \emph {et~al.}(2017)\citenamefont {Streltsov}, \citenamefont {Adesso},\ and\ \citenamefont {Plenio}}]{Streltsov-ColloquiumQuantum-2017}%
  \BibitemOpen
  \bibfield  {author} {\bibinfo {author} {\bibfnamefont {A.}~\bibnamefont {Streltsov}}, \bibinfo {author} {\bibfnamefont {G.}~\bibnamefont {Adesso}},\ and\ \bibinfo {author} {\bibfnamefont {M.~B.}\ \bibnamefont {Plenio}},\ }\bibfield  {title} {\bibinfo {title} {{\emph{Colloquium}}: {{Quantum}} coherence as a resource},\ }\href {https://doi.org/10.1103/RevModPhys.89.041003} {\bibfield  {journal} {\bibinfo  {journal} {Reviews of Modern Physics}\ }\textbf {\bibinfo {volume} {89}},\ \bibinfo {pages} {041003} (\bibinfo {year} {2017})}\BibitemShut {NoStop}%
\bibitem [{\citenamefont {Gianfelici}\ \emph {et~al.}(2021)\citenamefont {Gianfelici}, \citenamefont {Kampermann},\ and\ \citenamefont {Bru{\ss}}}]{Gianfelici-HierarchyContinuousvariable-2021}%
  \BibitemOpen
  \bibfield  {author} {\bibinfo {author} {\bibfnamefont {G.}~\bibnamefont {Gianfelici}}, \bibinfo {author} {\bibfnamefont {H.}~\bibnamefont {Kampermann}},\ and\ \bibinfo {author} {\bibfnamefont {D.}~\bibnamefont {Bru{\ss}}},\ }\bibfield  {title} {\bibinfo {title} {Hierarchy of continuous-variable quantum resource theories},\ }\href {https://doi.org/10.1088/1367-2630/ac2f90} {\bibfield  {journal} {\bibinfo  {journal} {New Journal of Physics}\ }\textbf {\bibinfo {volume} {23}},\ \bibinfo {pages} {113008} (\bibinfo {year} {2021})}\BibitemShut {NoStop}%
\bibitem [{\citenamefont {Singh}\ \emph {et~al.}(2019)\citenamefont {Singh}, \citenamefont {Jabbour}, \citenamefont {Van~Herstraeten},\ and\ \citenamefont {Cerf}}]{Singh-QuantumThermodynamics-2019}%
  \BibitemOpen
  \bibfield  {author} {\bibinfo {author} {\bibfnamefont {U.}~\bibnamefont {Singh}}, \bibinfo {author} {\bibfnamefont {M.~G.}\ \bibnamefont {Jabbour}}, \bibinfo {author} {\bibfnamefont {Z.}~\bibnamefont {Van~Herstraeten}},\ and\ \bibinfo {author} {\bibfnamefont {N.~J.}\ \bibnamefont {Cerf}},\ }\bibfield  {title} {\bibinfo {title} {Quantum thermodynamics in a multipartite setting: {{A}} resource theory of local {{Gaussian}} work extraction for multimode bosonic systems},\ }\href {https://doi.org/10.1103/PhysRevA.100.042104} {\bibfield  {journal} {\bibinfo  {journal} {Physical Review A}\ }\textbf {\bibinfo {volume} {100}},\ \bibinfo {pages} {042104} (\bibinfo {year} {2019})}\BibitemShut {NoStop}%
\bibitem [{\citenamefont {Koukoulekidis}\ and\ \citenamefont {Marvian}(2025)}]{Koukoulekidis-SymmetryAsymmetry-2025}%
  \BibitemOpen
  \bibfield  {author} {\bibinfo {author} {\bibfnamefont {N.}~\bibnamefont {Koukoulekidis}}\ and\ \bibinfo {author} {\bibfnamefont {I.}~\bibnamefont {Marvian}},\ }\href {https://doi.org/10.48550/arXiv.2510.25719} {\bibinfo {title} {Symmetry and {{Asymmetry}} in {{Bosonic Gaussian Systems}}: {{A Resource-Theoretic Framework}}}} (\bibinfo {year} {2025}),\ \Eprint {https://arxiv.org/abs/2510.25719} {arXiv:2510.25719 [quant-ph]} \BibitemShut {NoStop}%
\bibitem [{\citenamefont {Ivan}\ \emph {et~al.}(2011)\citenamefont {Ivan}, \citenamefont {Sabapathy},\ and\ \citenamefont {Simon}}]{Ivan-OperatorsumRepresentation-2011}%
  \BibitemOpen
  \bibfield  {author} {\bibinfo {author} {\bibfnamefont {J.~S.}\ \bibnamefont {Ivan}}, \bibinfo {author} {\bibfnamefont {K.~K.}\ \bibnamefont {Sabapathy}},\ and\ \bibinfo {author} {\bibfnamefont {R.}~\bibnamefont {Simon}},\ }\bibfield  {title} {\bibinfo {title} {Operator-sum representation for bosonic {{Gaussian}} channels},\ }\href {https://doi.org/10.1103/PhysRevA.84.042311} {\bibfield  {journal} {\bibinfo  {journal} {Physical Review A}\ }\textbf {\bibinfo {volume} {84}},\ \bibinfo {pages} {042311} (\bibinfo {year} {2011})}\BibitemShut {NoStop}%
\bibitem [{\citenamefont {Pool}(1966)}]{Pool-MathematicalAspects-1966}%
  \BibitemOpen
  \bibfield  {author} {\bibinfo {author} {\bibfnamefont {J.~C.}\ \bibnamefont {Pool}},\ }\bibfield  {title} {\bibinfo {title} {Mathematical aspects of the {{Weyl}} correspondence},\ }\href {https://doi.org/10.1063/1.1704817} {\bibfield  {journal} {\bibinfo  {journal} {Journal of Mathematical Physics}\ }\textbf {\bibinfo {volume} {7}},\ \bibinfo {pages} {66} (\bibinfo {year} {1966})}\BibitemShut {NoStop}%
\bibitem [{\citenamefont {Hillery}\ \emph {et~al.}(1984)\citenamefont {Hillery}, \citenamefont {O'Connell}, \citenamefont {Scully},\ and\ \citenamefont {Wigner}}]{Hillery-DistributionFunctions-1984}%
  \BibitemOpen
  \bibfield  {author} {\bibinfo {author} {\bibfnamefont {M.}~\bibnamefont {Hillery}}, \bibinfo {author} {\bibfnamefont {R.~F.}\ \bibnamefont {O'Connell}}, \bibinfo {author} {\bibfnamefont {M.~O.}\ \bibnamefont {Scully}},\ and\ \bibinfo {author} {\bibfnamefont {E.~P.}\ \bibnamefont {Wigner}},\ }\bibfield  {title} {\bibinfo {title} {Distribution functions in physics: {{Fundamentals}}},\ }\href {https://doi.org/10.1016/0370-1573(84)90160-1} {\bibfield  {journal} {\bibinfo  {journal} {Physics Reports}\ }\textbf {\bibinfo {volume} {106}},\ \bibinfo {pages} {121} (\bibinfo {year} {1984})}\BibitemShut {NoStop}%
\bibitem [{\citenamefont {Cahill}\ and\ \citenamefont {Glauber}(1969{\natexlab{a}})}]{Cahill-DensityOperators-1969}%
  \BibitemOpen
  \bibfield  {author} {\bibinfo {author} {\bibfnamefont {K.~E.}\ \bibnamefont {Cahill}}\ and\ \bibinfo {author} {\bibfnamefont {R.~J.}\ \bibnamefont {Glauber}},\ }\bibfield  {title} {\bibinfo {title} {Density {{Operators}} and {{Quasiprobability Distributions}}},\ }\href {https://doi.org/10.1103/PhysRev.177.1882} {\bibfield  {journal} {\bibinfo  {journal} {Physical Review}\ }\textbf {\bibinfo {volume} {177}},\ \bibinfo {pages} {1882} (\bibinfo {year} {1969}{\natexlab{a}})}\BibitemShut {NoStop}%
\bibitem [{\citenamefont {Cahill}\ and\ \citenamefont {Glauber}(1969{\natexlab{b}})}]{Cahill-OrderedExpansions-1969}%
  \BibitemOpen
  \bibfield  {author} {\bibinfo {author} {\bibfnamefont {K.~E.}\ \bibnamefont {Cahill}}\ and\ \bibinfo {author} {\bibfnamefont {R.~J.}\ \bibnamefont {Glauber}},\ }\bibfield  {title} {\bibinfo {title} {Ordered {{Expansions}} in {{Boson Amplitude Operators}}},\ }\href {https://doi.org/10.1103/PhysRev.177.1857} {\bibfield  {journal} {\bibinfo  {journal} {Physical Review}\ }\textbf {\bibinfo {volume} {177}},\ \bibinfo {pages} {1857} (\bibinfo {year} {1969}{\natexlab{b}})}\BibitemShut {NoStop}%
\bibitem [{\citenamefont {Brif}\ and\ \citenamefont {Mann}(1999)}]{Brif-PhasespaceFormulation-1999}%
  \BibitemOpen
  \bibfield  {author} {\bibinfo {author} {\bibfnamefont {C.}~\bibnamefont {Brif}}\ and\ \bibinfo {author} {\bibfnamefont {A.}~\bibnamefont {Mann}},\ }\bibfield  {title} {\bibinfo {title} {Phase-space formulation of quantum mechanics and quantum-state reconstruction for physical systems with {{Lie-group}} symmetries},\ }\href {https://doi.org/10.1103/PhysRevA.59.971} {\bibfield  {journal} {\bibinfo  {journal} {Physical Review A}\ }\textbf {\bibinfo {volume} {59}},\ \bibinfo {pages} {971} (\bibinfo {year} {1999})}\BibitemShut {NoStop}%
\end{thebibliography}
\end{document}